\documentclass[12pt]{article}

\usepackage[colorlinks, linkcolor=blue, citecolor=blue]{hyperref}           
\usepackage{color}
\usepackage{graphicx,subfigure,amsmath,amssymb,amsfonts,bm,epsfig,epsf,url,dsfont}
\usepackage{amsthm}

\newtheorem{thm}{Theorem}[section]
\newtheorem{lem}[thm]{Lemma}
\newtheorem{assum}[thm]{Assumption}
\newtheorem{proposition}[thm]{Proposition}
\newtheorem{remark}[thm]{Remark}
\newtheorem{definition}[thm]{Definition}
\newtheorem{corollary}[thm]{Corollary}

\usepackage{tikz}
\usepackage{bbm}
\usepackage[small,bf]{caption}
\usepackage[top=1in,bottom=1in,left=1in,right=1in]{geometry}
\usepackage{fancybox}
\usepackage{hyperref}
\usepackage{algorithm}
\usepackage{verbatim}

\usepackage{threeparttable}
\usepackage[titletoc,toc]{appendix}

\usepackage{mathtools}

\usepackage{scalerel,stackengine}
\stackMath
\newcommand\reallywidehat[1]{%
\savestack{\tmpbox}{\stretchto{%
  \scaleto{%
    \scalerel*[\widthof{\ensuremath{#1}}]{\kern-.6pt\bigwedge\kern-.6pt}%
    {\rule[-\textheight/2]{1ex}{\textheight}}
  }{\textheight}%
}{0.5ex}}%
\stackon[1pt]{#1}{\tmpbox}%
}

\newcommand*{\rom}[1]{\expandafter\@slowromancap\romannumeral #1@}

\newcommand{\Cov}{\mathrm{Cov}}

\DeclareMathOperator{\diag}{diag}

\newcommand{\abs}[1]{\left|#1\right|}
 
\newcommand{\N}{\mathbb{N}}

\newcommand{\Z}{\mathbb{Z}}

\newcommand{\E}{\mathbb{E}}

\newcommand{\cN}{\mathcal{N}}

\DeclareMathOperator*{\argmax}{arg max}
\DeclareMathOperator*{\argmin}{arg min}

\newcommand{\p}[1]{\left(#1\right)}
\newcommand{\pp}[1]{\left[#1\right]}
\newcommand{\ppp}[1]{\left\{#1\right\}}
\newcommand{\norm}[1]{\left\|#1\right\|}

\newcommand{\s}[1]{\mathsf{#1}}
\usepackage{xspace}

\numberwithin{equation}{section}

\newcommand{\calN}{{\cal N}}

\newcommand{\calP}{{\cal P}}

\newcommand{\calY}{{\cal Y}}

\allowdisplaybreaks
\usepackage[utf8]{inputenc}
\usepackage{authblk}
\usepackage{url}
\usepackage[left=1in,right=1in,top=1in,bottom=1in]{geometry}

\usepackage{titlesec}
\titleformat{\part}[display]
  {\centering\bfseries\LARGE}
  {}{0pt}{}
\titlespacing*{\part}{0pt}{1.0em}{1.0em}

\usepackage{booktabs}
\usepackage{array}
\usepackage{enumitem}
\newcolumntype{L}[1]{>{\raggedright\arraybackslash}p{#1}}
\setlist[itemize]{leftmargin=*,nosep,topsep=0pt,parsep=0pt,partopsep=0pt}

\makeatletter
\renewcommand\paragraph{\@startsection{paragraph}{4}{\z@}%
  {1ex}
  {-1em}
  {\normalfont\normalsize\bfseries}}
\makeatother

\begin{document}

\title{Expectation-Maximization for Low-SNR Multi-Reference Alignment }

\author[1]{Amnon Balanov\thanks{Corresponding author: \url{amnonba15@gmail.com}}}
\author[1]{Wasim Huleihel}
\author[1]{Tamir Bendory}

\affil[1]{\normalsize School of Electrical and Computer Engineering, Tel Aviv University, Tel Aviv 69978, Israel}

\maketitle

\begin{abstract}
We study the multi-reference alignment (MRA) problem of recovering a signal from noisy observations acted on by unknown random circular shifts. While the information-theoretic limits of MRA are well characterized in many settings, the algorithmic behavior at low signal-to-noise ratio (SNR), the regime of practical interest, remains poorly understood. In this paper, we analyze the expectation–maximization (EM) algorithm, a widely used method for MRA, and characterize its convergence dynamics and initialization dependence in the low-SNR limit.

On the convergence side, we prove a two-phase phenomenon near the ground truth as $\mathrm{SNR}\to 0$: an initial contraction with error decaying as $\exp(-\, \mathrm{SNR} \cdot t)$ followed by a much slower phase scaling as $\exp(- \,\mathrm{SNR}^2 \cdot t)$, where $t$ is the iteration number. This yields an iteration-complexity lower bound $T \gtrsim \mathrm{SNR}^{-2}$ to reach a small fixed target accuracy, revealing a severe computational bottleneck at low SNR. We also identify a finite-sample instability, which we term \emph{Ghost of Newton}, in which EM  initially approaches the ground truth but later diverges, degrading reconstruction quality.

On the bias side, we analyze EM in the noise-only setting ($\mathrm{SNR}=0$), a regime  referred to as Einstein from Noise, to highlight its pronounced sensitivity to initialization. We prove that the EM map preserves the Fourier phases of the initialization across all iterations, while the corresponding Fourier magnitudes contract toward zero at a slow rate of $(1+T)^{-1/2}$. Consequently, although the amplitudes vanish in the limit of $T \to \infty$ iterations, the reconstructed structure continues to reflect the geometry encoded by the template's Fourier phases. We complement this analysis by studying model bias in the low-SNR regime, as well as finite-sample effects. 
Together, these results expose fundamental computational and initialization-driven limitations of EM for MRA in the low-SNR regime.

\end{abstract}

\newpage
\tableofcontents

\newpage
\section{Introduction}

Reconstructing signals from noisy observations subject to unknown group transformations is a fundamental challenge in numerous scientific fields, including signal processing, computer vision, medical imaging, and structural biology. A widely studied abstraction of this challenge is the multi-reference alignment (MRA) problem \cite{bandeira2014multireference,bendory2017bispectrum,romanov2021multi,bandeira2023estimation}.

\subsection{Problem description} 
In this work, we study the classical MRA problem, where the objective is to recover an unknown signal $x^\star \in \mathbb{R}^d$ from noisy measurements consisting of randomly circularly shifted copies of $x^\star$.
For each shift $\ell \in \{0,1,\ldots,d-1\}$, we define the circular-shift operator $\mathcal{T}_\ell:\mathbb{R}^d \to \mathbb{R}^d$ by
\[
    [\mathcal{T}_\ell z]_j \triangleq z_{(j - \ell) \bmod d}, \quad \text{for any } z \in \mathbb{R}^d \text{ and } 0 \leq j \leq d-1.
\]
Under this model, the observations are given by,
\begin{align}
    (\text{The MRA Model}) \quad y_i = \mathcal{T}_{\ell_i} x^\star + \xi_i, \label{eqn:mainModel1D}
\end{align}
for $1\leq i\leq n$, where $\{\xi_i\}_i\stackrel{\text{i.i.d.}}{\sim} \mathcal{N}\p{0, \sigma^2 I_{d}}$ are independent and identically distributed (i.i.d.) Gaussian random variables with zero mean and variance $\sigma^2$. The shifts $\ppp{\ell_i}_i$ are drawn independently and uniformly at random, each over the set $\{0,1,\ldots, d-1\}$, independently of $\{\xi_i\}_i$. The objective is to recover the underlying signal $x^\star$,  up to a global circular shift, from the noisy and randomly shifted observations $\{y_i\}_{i=1}^{n}$. Throughout, we define the signal-to-noise ratio (SNR) by 
\begin{align}
    \mathrm{SNR} \triangleq \frac{\norm{x^\star}^2}{d \sigma^2}. \label{eqn:snr-def}
\end{align}

At high SNR, one can often recover the underlying structure by first aligning the observations and then averaging them, for example via synchronization techniques~\cite{singer2011three,singer2011angular,perry2018message}.
However, in the low-SNR regime, which is common in many real-world applications and the main focus of this work, reliable alignment becomes infeasible, as the noise overwhelms the signal~\cite{dadon2024detection}. Despite this challenge in the low-SNR regime, it is still possible to recover the underlying structure by modeling the unknown shifts as nuisance variables and applying statistical inference techniques to marginalize them out. 
  
\paragraph{The expectation-maximization algorithm.}
To estimate the unknown signal $x^\star$ in the MRA model~\eqref{eqn:mainModel1D}, a widely used approach is the expectation-maximization (EM) algorithm~\cite{dempster1977maximum,sigworth1998maximum,scheres2012relion}. 
In practice, EM underpins many MRA pipelines, most notably in cryo-EM~\cite{bai2015cryo,nogales2016development,renaud2018cryo,yip2020atomic} and cryo-electron tomography (cryo-ET)~\cite{chen2019complete,turk2020promise,watson2024advances}. 

EM is an iterative refinement scheme that alternates between (i) estimating, for each observation, a probability distribution over all possible shifts (E-step), and (ii) updating the current signal estimate as the corresponding probability-weighted average (M-step). Each iteration is guaranteed to monotonically (but not necessarily strictly) increase the likelihood, and the procedure converges to a stationary point of the likelihood.

A well-known limitation of EM is its sensitivity to initialization~\cite{bubeck2012initialization, balanov2025confirmation}. In practice, the iterations start from an initial template $x_{\mathrm{template}}$, chosen either at random or using domain knowledge. While a good template can accelerate convergence and improve reconstruction quality~\cite{singer2020computational,bendory2020single}, finding one is often challenging. Moreover, handcrafted templates can introduce confirmation bias, potentially steering reconstructions toward preconceived structures~\cite{henderson2013avoiding, balanov2024einstein,balanov2025confirmation}. This issue is especially pronounced in cryo-electron microscopy (cryo-EM), where substantial effort has been devoted to designing effective initialization strategies for EM, e.g.,~\cite{barnett2017rapid,levin20183d}.

\paragraph{The sample complexity of the MRA problem.}
Recent theoretical advances have shown that, even when individual translations cannot be reliably estimated, the underlying signal in the MRA problem can still be recovered provided that the number of observations is sufficiently large~\cite{abbe2018estimation,perry2019sample,bandeira2023estimation}.  
In the asymptotic regime where both the number of observations $n \to \infty$ and the $\mathrm{SNR} \to 0$, a necessary condition for unique recovery is that the number of samples satisfies $n = \omega(\mathrm{SNR}^{-3})$,  that is, $n / \mathrm{SNR}^{-3} \to \infty$, as $n \to \infty$, and $\mathrm{SNR} \to 0$;  see Section~\ref{sec:MRAreconstructionRegimes} for further discussion. Moreover, this information-theoretic scaling is achieved by method-of-moments based estimators, which recover $x^\star$ via low-order moments such as autocorrelations and the bispectrum~\cite{perry2019sample,bendory2017bispectrum}.  
However, these sample-complexity guarantees address identifiability rather than algorithmic performance: they do not imply that the EM algorithm will converge to the maximum-likelihood estimate (MLE) or yield accurate estimates.
Despite this gap, EM remains a widely used workhorse for MRA and related latent-group models, making it important to understand its behavior in the high-noise, low-SNR regime, where its practical performance is still far from fully characterized.

\paragraph{The EM algorithm in the low SNR regime.}
Recent studies have begun to shed light on this gap. Fan \emph{et~al.}~\cite{fan2023likelihood} showed that the likelihood landscape of discrete orbit-recovery models, i.e., latent-group models in which observations are transformed by an unknown element of a finite group (such as cyclic shifts in MRA), can contain spurious local optima that may trap EM iterations.
Extending this analysis to the cryo-EM setting, they further demonstrated that the Fisher information captures intrinsic group-invariant degeneracies that slow or prevent convergence~\cite{fan2024maximum}.  
Complementary work by Katsevich and Bandeira~\cite{katsevich2023likelihood} revealed that in the low-SNR regime, EM behaves approximately as a moment-matching algorithm, offering additional insight into its local dynamics.  
Despite these advances, the quantitative convergence rate of EM at low SNR remains insufficiently characterized. 

\subsection{Main results}
In this work, we develop a unified theory of EM for MRA in low-SNR regimes, showing that the EM dynamics exhibit complementary behaviors as a function of both the SNR and the initialization scale.
Table~\ref{tab:regimes_contribution} provides a schematic overview of the regimes studied, and summarizes the corresponding contributions in both the population limit ($n\to\infty$) and the finite-sample setting (finite $n < \infty$). 

We next elaborate on the key results in each regime below. Throughout, we write $\hat{x}^{(t)}$ for the $t$-th EM iterate, $\hat{x}^{(0)}$ for the initialization, and $\hat{\s{X}}^{(t)}[k]$ for the $k$-th discrete Fourier coefficient of $\hat{x}^{(t)}$.
We complement the theoretical analysis with population-level numerical experiments.
In these simulations, we evaluate the expectations defining the population EM operator via Gauss--Hermite quadrature, and we restrict to small ambient dimensions (typically $d=5$) to keep the quadrature computations tractable.

\begin{table*}[t]
\centering
\begingroup
\setlength{\tabcolsep}{3.5pt}
\renewcommand{\arraystretch}{1.05}
\footnotesize

\caption{\textbf{Regimes and main contributions of this work.}
We summarize the behaviors studied as a function of signal to noise ratio (SNR) and initialization.}
\label{tab:regimes_contribution}

\begin{tabular}{|L{2.7cm}||L{2.3cm}|L{2.4cm}|L{4cm}|L{3.5cm}|}
\hline
\textbf{Regime} &
\textbf{SNR} &
\textbf{Initialization} &
\textbf{Main results (population)} &
\textbf{Main results (finite-sample)} \\
\hline\hline

\textbf{Convergence in the vicinity of ground-truth (Section~\ref{sec:fundemental-properties})} &
$0<\mathrm{SNR}\ll 1$ &
In the basin of attraction &
Two-phase convergence dynamics (Theorem~\ref{thm:lowSNR-two-phase}); Iteration complexity $T \gtrsim \mathrm{SNR}^{-2}$ (Corollary~\ref{cor:flat-iter-lower}).
& The Ghost of Newton phenomenon (Section~\ref{sec:GoN}); sample complexity $n \gtrsim \mathrm{SNR}^{-3}$ (Theorem~\ref{thm:tight-1d-branch}).
\\
\hline

\textbf{Bias to initialization (Section~\ref{sec:bias-to-initalization})} &
$0<\mathrm{SNR}\ll 1$ &
$\|\hat{x}^{(0)}\| \ll \sigma$ &
Iteration barrier: $T \gtrsim \mathrm{SNR}^{-1}$ to lose dependence on initialization (Theorem~\ref{thm:mra-lowSNR-iteration-bias-to-init}).
&

\\
\hline

\textbf{Einstein from Noise (Section~\ref{sec:EfN})} &
$\mathrm{SNR} = 0$ &
$\|\hat{x}^{(0)}\| > 0$ & Fourier phases stay at initialization (Corollary~\ref{thm:3.5});  Fourier magnitudes contract toward $0$ at a non-geometric rate~\eqref{eqn:asymp-low-mag-main} (Theorem~\ref{thm:low-mag-rate}).
&
Fourier phases drift due to finite-sample (Proposition~\ref{prop:finite_sample_one_step} and Proposition~\ref{thm:lowerBoundAfterTiterations_rewrite}).
\\
\hline

\end{tabular}
\endgroup
\end{table*}

\subsubsection{Convergence in the vicinity of ground-truth}

The first part of this work (Section~\ref{sec:fundemental-properties}) studies the behavior of EM in a \emph{neighborhood of the ground-truth} in the low-SNR regime, focusing on the population EM map. 

\paragraph{Two phases of convergence.} 
Figure~\ref{fig:2}(a) highlights a striking phenomenon: EM exhibits two phases of convergence in the vicinity of the ground-truth: first, a moderate convergence rate regime with decay proportional to $\exp \p{-c \, \mathrm{SNR} \cdot t}$ followed by a markedly slower phase scaling as $\exp\p{- c' \, \mathrm{SNR}^2 \cdot  t}$. We provide a theoretical explanation of this behavior in Section~\ref{sec:fundemental-properties} and Theorem~\ref{thm:lowSNR-two-phase}.

\paragraph{Iteration complexity.}
As a direct consequence of this two-phase convergence behavior, Section~\ref{sec:iterationComplexity} derives iteration-complexity bounds for EM in MRA. In particular, even within the vicinity of the ground-truth, EM requires at least $T \gtrsim \mathrm{SNR}^{-2}$ iterations to attain any fixed accuracy, so the runtime degrades significantly as the noise level increases. We summarize these findings informally below; see Corollary~\ref{cor:flat-iter-lower} for iteration complexity. 

\begin{thm}[Informal: Iteration complexity of EM in MRA]
In the low-SNR regime of the MRA model~\eqref{eqn:mainModel1D} (i.e., $\mathrm{SNR}\to 0$), EM needs at least $T\gtrsim \mathrm{SNR}^{-2}$ iterations to reach a small, fixed target accuracy.
\end{thm}

\begin{figure*}[!t]
    \centering
    \includegraphics[width=1.0 \linewidth]{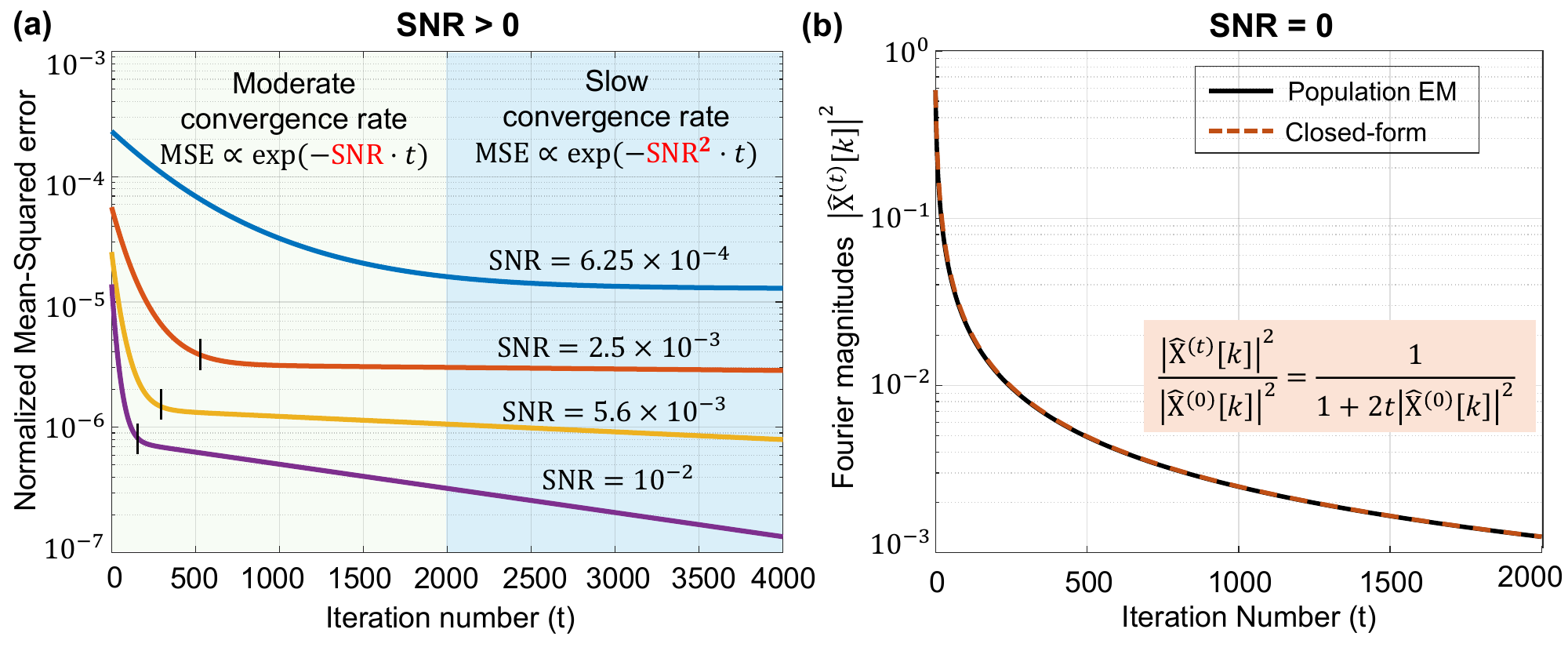}
    \caption{\textbf{Population EM dynamics in the low-SNR regime.}
    \textbf{(a)} Two-phase convergence at positive $0 < \mathrm{SNR} \ll 1$. The plot shows the normalized population MSE $\|\hat{x}^{(t)}-x^\star\|_2^2/\|x\|_2^2$ versus iteration $t$ for several low SNR values (annotated on the curves).
    The population EM map exhibits a clear two-phase behavior:
    an initial moderate-rate geometric decay, $\mathrm{MSE}\propto \exp\p{-c \, \mathrm{SNR} \cdot t}$, followed by a slow phase where the error decreases much more slowly,  $\mathrm{MSE}\propto\exp\p{- c' \, \mathrm{SNR}^2 \cdot t}$. 
    \textbf{(b)} Fourier magnitudes contraction to 0 at $\mathrm{SNR}=0$.
    The Fourier magnitudes $|\s{\hat{X}}^{(t)}[k]|^2$ are shown for a representative non-mean Fourier mode under the population EM operator.
    The dashed curve is the closed-form prediction from Theorem~\ref{thm:low-mag-rate}, $|\s{\hat{X}}^{(t)}[k]|^2/|\s{\hat{X}}^{(0)}[k]|^2 = (1+2t|\s{\hat{X}}^{(0)}[k]|^2)^{-1}$.
    The agreement corroborate that, in the regime of $\mathrm{SNR} = 0$, Fourier magnitudes decay toward zero at a rate $\asymp(1+t)^{-1}$, while the corresponding Fourier phases remain frozen along the population trajectory (not shown in the Figure).
    Simulation parameters: Population expectations defining the EM update are evaluated numerically via Gauss-Hermite quadrature over the Gaussian noise with a fixed signal dimension $d=5$. }
    \label{fig:2} 
\end{figure*}

\subsubsection{Bias to initialization at low-SNR}
In addition to the local convergence analysis, we characterize a complementary \emph{bias to initialization} regime of EM in the low (but non-vanishing) SNR setting (Section~\ref{sec:bias-to-initalization}). 
Whereas the convergence results describe the dynamics once the iterate is already close to the ground truth, the bias to initialization regime applies more broadly whenever the EM initialization has small norm relative to the noise level.
Theorem~\ref{thm:mra-lowSNR-iteration-bias-to-init} shows that while the mean component is estimated after a single iteration, the remaining population drift in the other non-mean components is extremely weak: the displacement per iteration is $O(\mathrm{SNR})$. 
Consequently, the EM ``forgets'' its initialization only on a long time scale: achieving a constant fraction change relative to $\|\hat{x}^{(0)}\|$ requires at least $T\gtrsim \mathrm{SNR}^{-1}$ iterations.

\subsubsection{Model bias at vanishing SNR}

A complementary contribution of this work is a quantitative analysis of the Einstein from Noise phenomenon as a form of algorithmic confirmation bias. In the extreme case where $\mathrm{SNR} = 0$ (i.e., when the signal is absent), the reconstruction tends to align with the \emph{initial template}  (the user’s prior guess) over many iterations. This phenomenon may also be viewed as the limiting case of vanishing signal strength: $\mathrm{SNR}\to 0$ while the initialization $\hat{x}^{(0)}$ remains nonvanishing (i.e., does not scale to zero with the signal).
While this effect was previously observed and studied in single-iteration hard-assignment algorithms \cite{shatsky2009method, sigworth1998maximum,balanov2024einstein}, in this work, we extend the analysis to the EM algorithm across multiple iterations.

Our analysis characterizes the Einstein from Noise dynamics both for finite-dimensional signals $d \in \mathbb{N}$ and the asymptotic signal length $d \to \infty$. For fixed dimension $d$, Theorem~\ref{thm:1} shows that each population EM step acts diagonally in the Fourier basis: all non-mean Fourier phases are preserved from one iteration to the next, while their magnitudes contract. 
Consequently, Corollary~\ref{thm:3.5} establishes that along the entire population trajectory the phases remain identical to those of the initialization, even as the iterates converge to the origin. 
Moreover, Theorem~\ref{thm:low-mag-rate} yields an explicit contraction law: for any fixed non-mean frequency $k$ and iteration $T$,
\begin{align}
    \frac{|\s{\hat{X}}^{(T)}[k]|}{|\s{\hat{X}}^{(0)}[k]|} = \frac{1}{\sqrt{1+2T|\s{\hat{X}}^{(0)}[k]|^2}}.
    \label{eqn:asymp-low-mag-main}
\end{align}
This result empirically matches the convergence of the population EM as seen in Figure~\ref{fig:2}(b). 
Thus, at $\mathrm{SNR}=0$, the Fourier magnitudes decay toward zero as $T\to\infty$, but only at the slow rate $\asymp(1+T)^{-1/2}$, while the Fourier phases remain identical to their initialized values.

\subsubsection{Finite-sample analysis}

We also study the finite-sample manifestations of the regimes described above.
For finite-sample numerical studies (Monte-Carlo experiments), we use a running example that appears throughout the paper (Figure~\ref{fig:1}). The ground-truth signal $x^\star$ is a grayscale image of Newton; each observation is generated by applying a random 2D cyclic shift to $x^\star$ and then adding Gaussian noise.
We compare two iterative procedures: (i) the standard EM algorithm, and (ii) a hard-assignment variant~\cite{bendory2020single} which, at each iteration, assigns each $y_i$ to its most likely transformation $\mathcal{T}_{\ell_i}$ given the current estimate, and then updates the signal by averaging the aligned observations. Both methods require initialization, and in all experiments we use a grayscale image of Einstein as a representative initial template.

\begin{figure*}
    \centering
     \includegraphics[width=0.9 \linewidth]{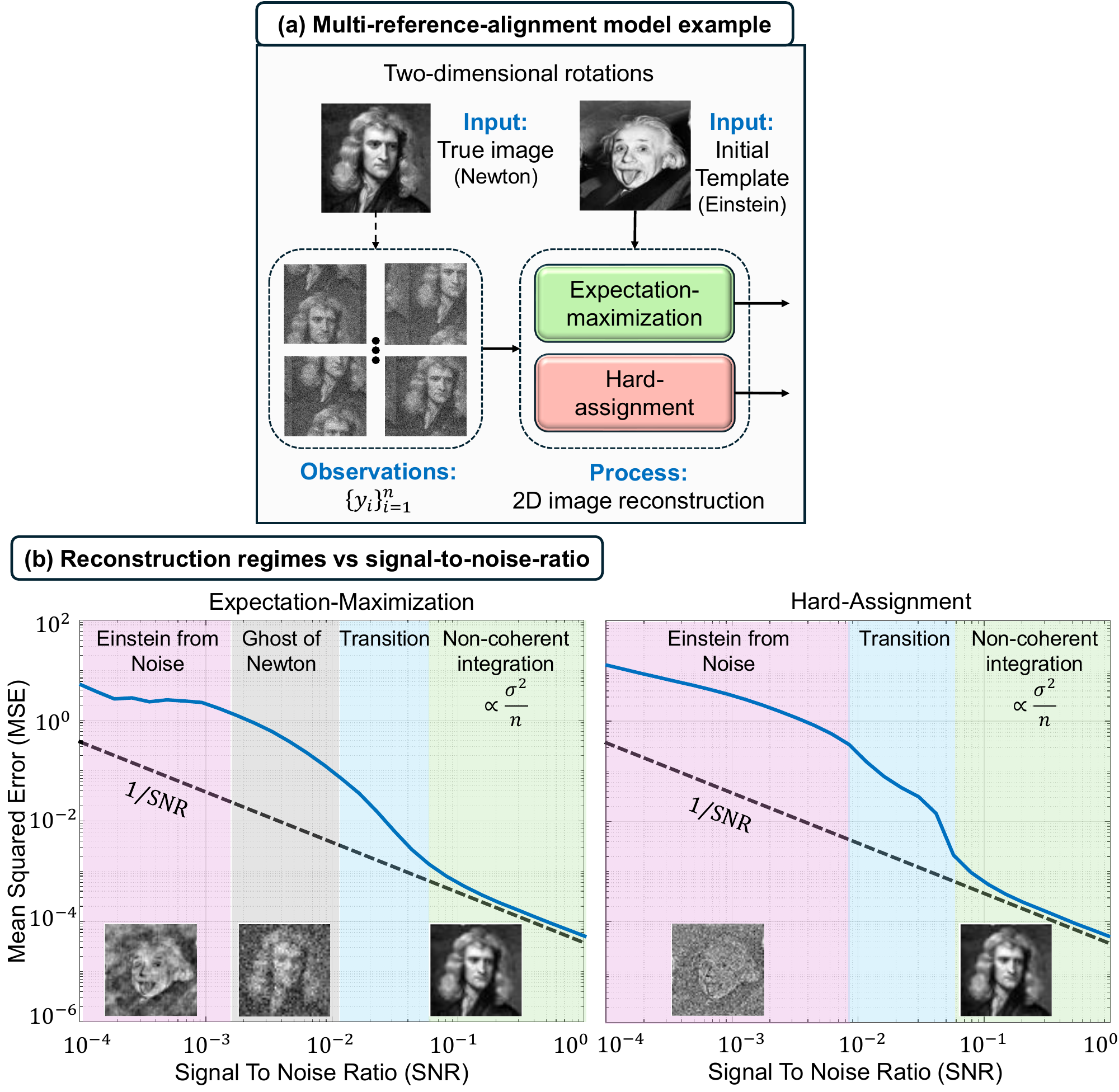}
    \caption{ \textbf{Empirical setup used in this work and reconstruction regimes.}
    \textbf{(a)} Motivating application: Illustration of the MRA model \eqref{eqn:mainModel1D} in a two-dimensional setting. 
    The goal is to reconstruct a 2D Newton image from noisy, randomly shifted observations, starting from an initial template of Einstein's image. We study the \emph{expectation-maximization (EM) algorithm as the primary method analyzed in this work}, and include a \emph{hard-assignment method as a baseline for comparison}.
    \textbf{(b)} Reconstruction performance for both methods is quantified by the MSE, defined in~\eqref{eqn:mseDef}, and plotted as a function of SNR under cyclic 2D shifts. Four distinct regimes are observed:  
    1) \textit{Einstein from Noise}: at very low SNR, the reconstruction resembles the initial template (Einstein);  
    2) \textit{Ghost of Newton}: at moderately low SNR, the EM algorithm initially converges to the true structure (Newton), but in later iterations, the noise term causes the reconstruction to gradually deteriorate, leading to increasingly corrupted estimates (see Figure \ref{fig:4} for further illustration). 
    3) a \textit{transition region}, 
    4) \textit{high-SNR regime}, where alignment is accurate and performance approaches the baseline case, following the scaling law $\s{MSE} \propto 1/\mathrm{SNR}$.  
    All experiments use $d = 64 \times 64$ images, $n = 2 \times 10^4$ observations, $T = 200$ iterations, and average results over 30 Monte Carlo trials per data point.  
    }
    \label{fig:1} 
\end{figure*}

\paragraph{The Ghost of Newton phenomenon.}
We identify a finite-sample artifact, which we term the \emph{Ghost of Newton} (named after the Newton image used as ground truth in our running example). This phenomenon appears at intermediate SNR levels and can be seen when the EM algorithm initially converges toward the correct structure (Newton), but subsequent iterations introduce a degradation in accuracy (see Figure~\ref{fig:4}). This results in a non-monotonic MSE curve: the error initially decreases as the reconstruction improves, then increases as the estimate begins to diverge from the ground truth (Figure~\ref{fig:4}(a)).

The underlying mechanism is the discrepancy between the empirical and population EM operators. In the population setting ($n \to \infty$), and with an initialization in the vicinity of the ground-truth, EM converges to the true signal $x^\star$. For finite $n$, however, the empirical EM map converges instead to a sample-dependent fixed point $\hat{x}_n^\star$ that maximizes the finite-sample likelihood. Consequently, when initialized near $x^\star$, the empirical trajectory initially tracks the population dynamics but eventually peels away under accumulated finite-sample perturbations, drifting toward $\hat{x}_n^\star$. An analysis of this effect is provided in Section~\ref{sec:GoN}.

\paragraph{Finite-sample Einstein from Noise.}
Finally, we turn to the finite-sample Einstein from Noise regime. 
Proposition~\ref{prop:finite_sample_one_step} shows that each empirical EM step introduces a small but nonzero phase drift: the phase MSE between successive iterates is of order $1/n$. 
Proposition~\ref{thm:lowerBoundAfterTiterations_rewrite} further shows that these per-iteration drifts accumulate at most quadratically in the iteration count, yielding a template-referenced phase error that scales as $\asymp T^2/n$ over the first $T$ iterations. 

\section{Preliminaries and fundamental properties}
\label{sec:problem_formulation}

This section introduces the notation, problem setup, and EM framework used throughout the paper. We first present the EM algorithm and define the associated fixed-point maps: the population operator $M$ and its empirical counterpart $M_n$. For context, we also describe a widely used alternative, the hard-assignment method, which can be viewed as a simplified (“poor-man’s”) version of EM. We then prove a one-step consistency result in the high-SNR regime, which leads to global convergence from arbitrary initializations. 

\subsection{Notation}
Throughout, $\xrightarrow[]{\mathcal{D}}$, $\xrightarrow[]{\mathcal{P}}$, $\xrightarrow[]{\text{a.s.}}$, and $\xrightarrow[]{\mathcal{L}^p}$ denote convergence in distribution, in probability, almost surely, and in $\mathcal{L}^p$, respectively, for sequences of random variables. We denote the index set $\ppp{0, 1, \dots, d-1}$ by $\pp{d}$, and $\mathbb{Z}_d$ for the cyclic group of order $d$.
For vectors and matrices over $\mathbb{R}$ or $\mathbb{C}$, $(\cdot)^\top$ denotes the (ordinary) transpose. The notation $(\cdot)^{\ast}$ denotes the Hermitian transpose (conjugate transpose). The Euclidean inner product is written as either $\langle a,b\rangle$ or $a^\top b$ when $a,b\in\mathbb{R}^d$, and as $\langle a,b\rangle\triangleq a^{\ast} b$ when $a,b\in\mathbb{C}^d$.
We use $\|\cdot\|_2$ for the vector Euclidean norm, $\|x\|_2=\sqrt{x^{\ast}x}$, and $\|\cdot\|_F$ for the matrix Frobenius norm, $\|A\|_F=\sqrt{\mathrm{trace}(A^{\ast}A)}$; for vectors, $\|x\|_F=\|x\|_2$. When no confusion can arise, we  write $\|\cdot\|$ for $\|\cdot\|_2$. 

\paragraph{Asymptotic notations.}
We employ the following standard asymptotic notation. For nonnegative sequences ${a_n}$ and ${b_n}$, we write $a_n=O(b_n)$ if there exists a constant $C<\infty$ and $n_0$ such that $a_n\le C b_n$ for all $n\ge n_0$, and $a_n=o(b_n)$ if $a_n/b_n\to 0$. We write $a_n=\omega(b_n)$ if $a_n/b_n\to\infty$, and $a_n=\Theta(b_n)$ if both $a_n=O(b_n)$ and $b_n=O(a_n)$ hold. 

We also use the notation $a_n \gtrsim b_n$ to denote $b_n = O(a_n)$; equivalently, there exist constants $c>0$ and $n_0$ such that $a_n \ge c b_n$ for all $n\ge n_0$. Finally, we write $a_n \asymp b_n$ as a synonym for $a_n = \Theta(b_n)$, that is, when $a_n \gtrsim b_n$ and $b_n \gtrsim a_n$ both hold. All asymptotic statements are taken in the appropriate limit specified in context (e.g., $n\to\infty$, $\mathrm{SNR}\to 0$, or $d\to\infty$).

\paragraph{Fourier-space notation.}
To facilitate the analysis of EM for the MRA model, we work in the discrete Fourier domain. 
For a complex number $\s{Z}\in\mathbb{C}$, let $\phi_{\s{Z}}\triangleq \arg(\s{Z})$ denote its phase.  
Let $F\in\mathbb{C}^{d\times d}$ be the unitary discrete Fourier transform (DFT) matrix with columns $\{f_k\}_{k=0}^{d-1}$, whose entries are
\begin{align}
    [f_k]_\ell  =  \frac{1}{\sqrt d}\,e^{j\frac{2\pi}{d}k\ell},
    \qquad k,\ell\in\{0,\dots,d-1\}. \label{eqn:DFT-col}
\end{align}
For a signal $y\in\mathbb{R}^d$, we write its DFT as $\s{Y}=F^\ast y$, i.e.,
\begin{align}
    \s{Y}[k]
    \triangleq
    \frac{1}{\sqrt d}\sum_{\ell=0}^{d-1} y_\ell\, 
    e^{-j\frac{2\pi}{d}k\ell},
    \qquad k\in\{0,\dots,d-1\},
    \label{eqn:DFT-def}
\end{align}
where $j\triangleq\sqrt{-1}$.  
Throughout, we use the notational convention $\s{Y}_k=\s{Y}[k]$ interchangeably. 

\paragraph{Metrics.} 
Recall the MRA model introduced in \eqref{eqn:mainModel1D}, where the objective is to recover the underlying signal $x^\star$, up to a global shift, from $n$ observations $\ppp{y_i}_{i=1}^{n}$.  Let $\mathcal{Y} = \{y_i\}_{i=1}^{n}$ denote the set of observations. We denote the estimator of $x^\star\in\mathbb{R}^d$ based on $\mathcal{Y}$ as $\hat{x} (\mathcal{Y})$. Because the statistics of $y_i$ is invariant to cyclic shifts, i.e., it is the same for $x^\star$ and $\mathcal{T}_\ell x^\star$, for any $\ell \in \pp{d}$, we aim to recover the orbit of $x^\star$, namely, $\{\mathcal{T}_\ell  x^\star |  \mathcal{T}_\ell \in \mathbb{Z}_d\}$; we thus define the group-invariant distance
\begin{align}
d_{\mathrm{orb}}(u,v)
      \triangleq 
    \min_{\widetilde{\mathcal{T}}_\ell \in \mathbb{Z}_d} \| u - \widetilde{\mathcal{T}}_\ell v \|, \label{eqn:orbit-distance}
\end{align}
for any $u,v\in\mathbb{R}^d$, and the normalized mean-square-error (MSE) by,
\begin{align}
    \text{MSE}\p{\hat{x}, x} \triangleq \frac{1}{\norm{x}^2} \mathbb{E} \pp{d_{\mathrm{orb}}^2 (\hat{x}, x)},\label{eqn:mseDef}
\end{align}
where the expectation is with respect to the joint distribution of $\calY$.

\subsection{Reconstruction algorithms}
We present the EM algorithm as the primary iterative method studied in this work. For comparison, we also introduce the popular hard-assignment approach. The two methods differ in the objective functions they effectively optimize and, more fundamentally, in how they treat uncertainty in the latent variables (i.e., the unknown cyclic shifts): EM uses soft responsibilities, whereas hard assignment commits to a single shift estimate at each iteration.
\paragraph{EM algorithm.}  
The EM algorithm, described in Algorithm~\ref{alg:generalizedEfNsoft}, adopts a probabilistic approach to shift estimation~\cite{dempster1977maximum}. At each iteration $t$, with current estimate $\hat{x}^{(t)}$, the algorithm proceeds in two steps. First, it computes the posterior distribution over all shifts $\mathcal{T}_\ell$ for each observation $y_i$, conditioned on $\hat{x}^{(t)}$ (E-step). Second, it updates the signal estimate by computing a weighted average of the observations, where the weights reflect the shift probabilities (M-step), resulting in the next iterate $\hat{x}^{(t+1)}$.

\begin{algorithm}[t!]
  \caption{Expectation-Maximization (soft-assignment)} \label{alg:generalizedEfNsoft}
\textbf{Input:} Initial template $\hat{x}^{(0)} = x_{\s{template}}$, observations $\{y_i\}_{i=1}^{n}$, noise level $\sigma^2$, and number of iterations $T$.\\
\textbf{Output:} EM estimate after $T$ iterations, $\hat{x}^{(T)}$. \\
\textbf{Each Iteration, for $t=0,\ldots,T-1$:}
\begin{enumerate}
  \item \textbf{E-step:} For $i=1,\ldots,n$ and $\ell=0,\ldots,d-1$, compute the responsibilities
  \begin{align}
    \gamma_{\ell} (\hat{x}^{(t)}, y_i)
      \triangleq 
    \frac{\exp \p{y_i^\top(\mathcal{T}_\ell  \hat{x}^{(t)})/\sigma^2}}
         {\sum_{r=0}^{d-1}\exp \p{y_i^\top(\mathcal{T}_r  \hat{x}^{(t)})/\sigma^2}}.
    \label{eqn:softmaxGamma}
  \end{align}
  \item \textbf{M-step:} Update the estimate
  \begin{align}
    \hat{x}^{(t+1)}
      \triangleq 
    \frac{1}{n}\sum_{i=1}^{n}\sum_{\ell=0}^{d-1}
    \gamma_{\ell}(\hat{x}^{(t)}, y_i) \p{\mathcal{T}_\ell^{-1}  y_i}.
    \label{eqn:generalizedEfNEqnSoft}
  \end{align}
\end{enumerate}
\end{algorithm}

The EM algorithm maximizes the log-likelihood of the data under the MRA model:
\begin{align}
    \mathcal{L}(x, \mathcal{Y}) \triangleq \sum_{i=1}^{n} \log \left( \sum_{\ell=0}^{d-1} p(\mathcal{T}_\ell) \cdot \exp\left(-\frac{\|y_i - \mathcal{T}_\ell  x\|^2}{2\sigma^2} \right) \right) - \frac{dn}{2} \log(2\pi\sigma^2), \label{eqn:logLikelihoodMRA}
\end{align}
where $p(\mathcal{T}_\ell) = 1/d$ reflects a uniform prior over shifts. The maximum likelihood estimator is defined as 
\begin{align}
    \hat{x}_{\text{ML}}(\mathcal{Y}) \triangleq \argmax_{x \in \mathbb{R}^d} \mathcal{L}(x, \mathcal{Y}).
\end{align}
Although alternative optimization strategies may be used to minimize $\mathcal{L}$, the EM algorithm is a practical and widely adopted method. It enjoys the desirable property of non-decreasing likelihood across iterations~\cite{dempster1977maximum}.
By accounting for the full posterior distribution over shifts, the EM algorithm offers greater robustness to noise than the hard-assignment method. This increased robustness, however, comes at the cost of higher computational complexity, since the algorithm requires accounting for all possible shifts for each observation. Additional details and derivation of the EM update rule are provided in Appendix~\ref{sec:softAssignmentUpdateStep}.

\paragraph{Hard-assignment algorithm.}  
The hard-assignment method, described in Algorithm~\ref{alg:generalizedEfNhard}, provides a simple iterative procedure for estimating the signal $x^\star$. At each iteration $t$, it assigns to each observation $y_i$ the cyclic shift $\mathcal{T}_{\ell_i}$ that best aligns it with the current estimate $\hat{x}^{(t)}$, that is, the shift minimizing the Euclidean distance $\|y_i - \mathcal{T}_\ell  \hat{x}^{(t)}\|$. After aligning all observations according to their assigned shifts, the estimate is updated by averaging the aligned data to produce $\hat{x}^{(t+1)}$.

This algorithm can be interpreted as an alternating minimization scheme for the following non-convex objective function:
\begin{align}
    \mathcal{L}_{\text{hard}}(x, \mathcal{Y}) \triangleq \sum_{i=1}^{n} \min_{\mathcal{T}_{\ell_i} \in \mathbb{Z}_d} \|y_i - \mathcal{T}_{\ell_i} x\|^2. \label{eqn:mraHardTargetFunction}
\end{align}
The objective is to optimize the function with respect to the signal:
\begin{align}
    \hat{x}_{\text{hard}}(\mathcal{Y}) \triangleq \argmin_{x \in \mathbb{R}^d} \mathcal{L}_{\text{hard}}(x, \mathcal{Y}).
\end{align}
In particular, the algorithm alternates between (i) estimating, for each observation, the shift that best aligns it with the current signal estimate, and (ii) updating the signal estimate given these inferred shifts.
As with the EM case, other optimization strategies may be used to maximize $\mathcal{L}_{\text{hard}}$, but the hard-assignment algorithm is widely used in practice due to its simplicity. The hard-assignment algorithm is computationally efficient and performs well in moderate-noise regimes, where the selection of the most likely shift per observation is typically correct. However, its performance degrades in high-noise settings due to its reliance on a single shift per observation.

\begin{algorithm}[t!]
  \caption{\texttt{Hard-assignment} 
  \label{alg:generalizedEfNhard}}
\textbf{Input:} Initial template $\hat{x}^{(0)} = x_{\s{template}}$, observations $\ppp{y_i}_{i=1}^{n}$, and number of iterations $T$.\\
\textbf{Output:} Hard assignment estimate after $T$ iterations, $\hat{x}^{(T)}$. \\
\textbf{Each Iteration, for $t=0\ldots,T-1$:}
\begin{enumerate}
    \item For $i=1,\ldots,n$, compute
     \begin{align}
        \s{\hat{R}}_i^{(t)}\triangleq \underset{0 \leq \ell \leq d-1}{\argmax} {    \langle{y_i}, \mathcal{T}_\ell  \hat{x}^{(t)} \rangle}.
    \label{eqn:OptShiftRealSpace}
    \end{align}
\item  Compute: 
    \begin{align}
        \hat{x}^{(t+1)}\triangleq \frac{1}{n} \sum_{i=1}^{n} \mathcal{T}_{\s{\hat{R}}_i^{(t)}}^{-1} \cdot y_i \label{eqn:generalizedEfNEqn}.
    \end{align}
\end{enumerate}
\end{algorithm}

\subsection{Population and finite-sample EM dynamics}
\label{subsec:pop-vs-finite-1d-mra}

In modern analyses of the EM algorithm, it is standard to distinguish between two related operators: the \emph{population} EM map, in which all expectations are taken under the true data-generating distribution (equivalently, the limit corresponding to infinitely many samples, $n \to \infty$), and the \emph{finite-sample} or \emph{empirical} EM map, in which these expectations are replaced by empirical averages over the observed dataset~\cite{balakrishnan2017statistical,daskalakis2017ten}. 
This distinction enables a precise separation between the deterministic population dynamics and the random fluctuations arising from finite-sample estimation.

\paragraph{Finite-sample EM operator.}
Recall that $\mathcal{Y}=\{y_i\}_{i=1}^{n}$ denote $n$ i.i.d. observations drawn from the MRA model~\eqref{eqn:mainModel1D}.  
The empirical EM operator $M_n:\mathbb{R}^d\to\mathbb{R}^d$, corresponds to the empirical update rule introduced in~\eqref{eqn:generalizedEfNEqnSoft}, is defined by
\begin{align}
    M_n(x;\mathcal{Y}) = \frac{1}{n}\sum_{i=1}^{n}\sum_{\ell=0}^{d-1} \gamma_\ell(x;y_i) \p{\mathcal{T}_\ell^{-1}  y_i},
    \label{eq:Phi-emp-1d}
\end{align}
where $\gamma_\ell(x;y_i)$, 
defined in~\eqref{eqn:softmaxGamma}, denotes the posterior responsibility of shift $\ell$ given the current signal estimate $x$, and $\mathcal{T}_\ell$ is the circular-shift operator.  
The finite-sample EM iteration then proceeds as
\begin{align}
    \hat{x}^{(t+1)} = M_n    \p{\hat{x}^{(t)};\mathcal{Y}},
    \qquad t=0,1,\ldots. \label{eqn:empirical-EM-update}
\end{align}

\paragraph{Population EM operator.}
The population EM operator is defined by replacing the empirical average in~\eqref{eq:Phi-emp-1d} with expectation under the true data-generating distribution, yielding a deterministic update.  
For any fixed $x\in\mathbb{R}^d$, we set
\begin{align}
    M(x) \triangleq \E \left[\sum_{\ell=0}^{d-1} \gamma_\ell(x;Y)\, \p{\mathcal{T}_\ell^{-1}Y}\right],
    \label{eq:Phi-pop-1d}
\end{align}
where $Y$ is an independent draw from the MRA model~\eqref{eqn:mainModel1D}. 
The corresponding population EM dynamics are
\begin{align}
    \hat{x}_{\mathrm{pop}}^{(t+1)} = M \p{\hat{x}_{\mathrm{pop}}^{(t)}}. \label{eqn:population-EM-update}
\end{align}
By the strong law of large numbers, for a fixed $x$,
\begin{align}
    M_n(x;\mathcal{Y}) \xrightarrow[n\to\infty]{\mathrm{a.s.}} M(x),
\end{align}
so the empirical operator coincides almost surely with the population operator in the infinite-sample limit. 
Accordingly, all randomness in $M_n$ stems from the finite dataset $\mathcal{Y}$, whereas $M$ is deterministic.

Throughout the population-level analysis (Sections~\ref{sec:fundemental-properties}-\ref{sec:EfN}), we write $\hat{x}^{(t)}$ for the population EM iterates generated by $M(\cdot)$, rather than $\hat{x}^{(t)}_{\mathrm{pop}}$, for simplicity. In Section~\ref{sec:finite-sample-analysis}, where we analyze the finite-sample behavior, we explicitly distinguish between the empirical and population EM updates using the notation in~\eqref{eqn:empirical-EM-update} and~\eqref{eqn:population-EM-update}, respectively.

\subsection{Global convergence at high-SNR} \label{sec:asymptoticSNRproperties}
Here, we analyze the asymptotic high-SNR regime, i.e., in the low-noise limit $\sigma \to 0$ (equivalently, $\mathrm{SNR}\to\infty$) with $d$ fixed. In this limit, the posterior over shifts concentrates on the maximizer of the cross-correlation\footnote{In our setting, the hard-assignment rule in~\eqref{eqn:mraHardTargetFunction} can be written equivalently as either minimizing an $\ell_2$ distance or maximizing a cross-correlation since each shift operator $\mathcal{T}_\ell$ is orthonormal, and it preserves the Euclidean norm, $\|\mathcal{T}_\ell x\|_2=\|x\|_2$.}, so the EM update reduces to a hard alignment step. As a consequence, EM and hard assignment become asymptotically identical, and both recover the ground-truth orbit after a single iteration. This asymptotic prediction is consistent with the empirical behavior at moderately high SNR (small but nonzero $\sigma$): in that regime both algorithms converge rapidly after a few iterations and exhibit essentially indistinguishable reconstructions.

\begin{proposition} [Global convergence at high-SNR]
\label{prop:relationBetweenSoftAndHard}
Let $\ppp{y_i}_{i=1}^{n}$ be observations drawn from the MRA model in \eqref{eqn:mainModel1D}, with ground-truth signal $x^\star$. Let $\hat{x}_{\mathrm{hard}}^{(t+1)}$ and $\hat{x}_{\mathrm{EM}}^{(t+1)}$ denote the reconstructions after $t+1$ iterations of the hard-assignment (Algorithm~\ref{alg:generalizedEfNhard}) and  EM (Algorithm~\ref{alg:generalizedEfNsoft}), respectively. Denote $\hat{x}^{(0)} = x_{\s{template}}$ as the initialization used by both algorithms. Assume that the cross-correlation between $x^\star$ and $\hat{x}^{(t)}$ attains a unique maximum, that is $\argmax_{\ell \in \pp{d}} \langle x^\star, \mathcal{T}_\ell\hat{x}^{(t)} \rangle$ is unique. Then,
\begin{enumerate}
    \item For any $n \in \mathbb{N}$, the EM and hard-assignment reconstructions are asymptotically equivalent in the low-noise limit, i.e.,
    \begin{align}
        \lim_{\sigma \to 0} \ \hat{x}_{\mathrm{EM}}^{(t)} - \hat{x}_{\mathrm{hard}}^{(t)} = 0, \label{eqn:equivalenceSoftAndHardHighSNR}
    \end{align}
    almost surely, for every $t \geq 0$.
    \item For any $n \in \mathbb{N}$, 
    \begin{align}
        \lim_{\sigma \to 0} \ d_{\mathrm{orb}}(\hat{x}_{\mathrm{hard}}^{(1)}, x^\star) = 0, \label{eqn:highSNRconvergneToSignal}
    \end{align}
    almost surely.
\end{enumerate}
\end{proposition}
The proof of Proposition~\ref{prop:relationBetweenSoftAndHard} is given in Appendix~\ref{sec:highSNRasymptotics} for a broad family of models that includes~\eqref{eqn:mainModel1D}. It builds on~\cite[Theorem~5.10]{robert1999monte} and extends the argument from cyclic shifts to more general group actions. 
While in the high-SNR limit, both EM and hard assignment recover the ground-truth orbit from arbitrary initializations, the remainder of the paper is devoted to the low-SNR regime, where global convergence fails and the dynamics are instead governed by initialization bias and finite-sample drift.


\section{EM convergence at low SNR}\label{sec:fundemental-properties}

In this section we characterize the behavior of the \emph{population} EM estimator in the low-SNR regime. 
We begin by establishing a spectral-radius criterion for local linear convergence of the population EM map $M$ via the Jacobian $J(x^\star)$ evaluated at the ground truth $x^\star$. 
Building on this, we derive a low-SNR Taylor expansion of $J(x^\star)$ and obtain a spectral decomposition that governs the dynamics of EM in this regime. 
These tools yield two main results: Theorem~\ref{thm:lowSNR-two-phase}, which explains the two-phase convergence behavior observed in Figure~\ref{fig:2}, and Corollary~\ref{cor:flat-iter-lower}, which quantifies the iteration complexity of EM at low SNR. 
For completeness, Appendix~\ref{sec:lowSNRapproxMain} also provides a global (i.e., not necessarily in the vicinity of the ground-truth) low-SNR approximation of the population EM map.

\subsection{Local convergence of the population EM }

We first characterize the local convergence behavior of the population EM dynamics defined in~\eqref{eq:Phi-pop-1d}.
It is well established that the local behavior of the EM algorithm around a fixed point is governed by the spectral properties of the Jacobian of the population update operator.
Classical analyzes (see, for instance,~\cite{wu1983convergence,balakrishnan2017statistical,daskalakis2017ten,mclachlan2008algorithm,redner1984mixture}), show that if the spectral radius of the Jacobian at the true parameter is strictly smaller than one, the population EM iteration converges linearly to the fixed point. 
Moreover, directions corresponding to eigenvalues close to one give rise to slow convergence modes.
Our goal in this section is to establish an analogous characterization for the MRA model, demonstrating that its population EM operator exhibits the same spectral radius-based local convergence behavior.

\paragraph{Jacobian, spectral radius, and symmetry considerations.}
Before analyzing the convergence of the population EM dynamics, we formalize the notions of the Jacobian and its spectral radius, which govern the local behavior of iterative algorithms. Detailed preliminaries and results on nonlinear dynamical systems can be found in Appendix~\ref{subsec:dynamical-preliminaries}.

\begin{definition}[Jacobian and spectral radius]
\label{def:Jacobian-spectral-radius}
Let $F:\mathbb{R}^d\to\mathbb{R}^d$ be a continuously differentiable multivariate function in a neighborhood of a point $x^\star\in\mathbb{R}^d$. The \emph{Jacobian matrix} of $F$ at $x^\star$ is the linear map
\begin{align}
    J_F(x^\star) \triangleq \nabla F(x)\big|_{x=x^\star}\in\mathbb{R}^{d\times d},
\end{align}
characterized by the first-order expansion: for any vector $h\in\mathbb{R}^d$ with $\|h\|_2\to 0$,
\begin{align}
    F(x^\star+h)-F(x^\star) = J_F(x^\star)\,h  +  o(\|h\|_2).
\end{align}
For a square 
matrix $J\in\mathbb{R}^{d\times d}$, we denote by $\mathrm{spec}(J)\subset\mathbb{C}$ its spectrum, i.e., the set of eigenvalues of $J$. The \emph{spectral radius} of $J$ is then defined as
\begin{align}
    \rho(J) \triangleq  \max\{\,|\lambda|:\lambda\in\mathrm{spec}(J)\,\}.
\end{align}
\end{definition}

\paragraph{Symmetry of the EM operator.}
In the MRA model~\eqref{eqn:mainModel1D} with uniform distribution of cyclic shifts, the log-likelihood and hence the EM update operator are equivariant with respect to circular shifts (see Lemma~\ref{lem:shift-equivariance-M}), namely,
\begin{align}
    M(\mathcal{T}_\ell  x) =  \mathcal{T}_\ell  M(x),
    \qquad \forall \ell \in \{0,1, \dots d-1 \}.
\end{align}
and the ground-truth signal $x^\star$ is not an isolated fixed point but lies on an orbit of equivalent solutions
\begin{align}
    \mathrm{Orb}(x^\star)
      = 
    \{\mathcal{T}_\ell  x^\star : \ell = 0,1,\ldots,d-1\}.
\end{align}
This inherent symmetry implies that the EM fixed points are not unique:
every circularly shifted version of $x^\star$ yields the same likelihood value and defines an equivalent solution.  
To analyze local convergence meaningfully, one must therefore disambiguate among these equivalent representatives.  
Thus, we restrict the dynamics to a canonical neighborhood of a single representative of the orbit by introducing a condition that guarantees the identity shift is the nearest orbit element in a small neighborhood. 

\begin{assum}
\label{assump:local-tiebreak-finite}
Let $x^\star \in \mathbb{R}^d$ denote the ground–truth signal, and let $\{\mathcal{T}_\ell\}_{\ell=0}^{d-1}$ be the finite set of circular shifts acting on $\mathbb{R}^d$.  
Assume there exists a radius $r > 0$ such that, for every $x$ in the open ball $\mathcal{B}(x^\star, r) \triangleq  \{ x \in \mathbb{R}^d : \|x - x^\star\|_2 < r \}$, the nearest shifted version of $x^\star$ is uniquely attained at the identity shift:
\begin{align}
    \argmin_{0 \le \ell < d} \|x - \mathcal{T}_\ell x^\star\|_2 = \{0\}, \qquad \forall x\in\mathcal{B}(x^\star, r).
\end{align}
In other words, $x^\star$ is locally separated from all its nontrivial circular shifts.
\end{assum}

We are now ready to state the local linear convergence theorem for the MRA model; its proof is given in Appendix~\ref{sec:proofOfEmSpectralRadius}. The first part of the theorem is specific to the MRA model~\eqref{eqn:mainModel1D}, whereas the second and third parts follow from standard results in nonlinear dynamical systems; see Appendix~\ref{subsec:dynamical-preliminaries} for background.

\begin{thm}[Local linear convergence of the population EM operator]
\label{thm:EM-MRA-spectral-radius}
Assume that the observation $Y$ is generated according to the MRA model~\eqref{eqn:mainModel1D} with a uniform distribution over cyclic shifts, and let the ground-truth signal $x^\star \in \mathbb{R}^d$.  
Suppose further that Assumption~\ref{assump:local-tiebreak-finite} holds.  
Let $M:\mathbb{R}^d \to \mathbb{R}^d$ denote the population EM operator defined in~\eqref{eq:Phi-pop-1d}.  
Then:
\begin{enumerate}
    \item \emph{(Regularity and Jacobian of $M$).}
    The mapping $M \in C^\infty (\mathbb{R}^d)$ is smooth on $\mathbb{R}^d$.  
    For any $x \in \mathbb{R}^d$, its Jacobian $J \in \mathbb{R}^{d \times d}$ is
    \begin{align}
        J(x) \triangleq \nabla_x M(x) = \frac{1}{\sigma^2} \mathbb{E} \left[\sum_{\ell=0}^{d-1} \gamma_\ell(x;Y) (z_\ell - \bar{z})(z_\ell-\bar{z})^\top\right],
        \label{eq:Jacobian-MRA-proof}
    \end{align}
    where $z_\ell = \mathcal{T}_\ell^{-1}  Y$ and $\bar{z}(x;Y) = \sum_{\ell=0}^{d-1}\gamma_\ell(x;Y) z_\ell$. In particular,  $J(x^\star)$ is a real symmetric matrix, hence orthogonally diagonalizable, and $\mathrm{spec}(J(x^\star))\subset[0,1]$.
    
    \item \emph{(Local linear convergence).}
    If $\rho (J(x^\star)) < 1$, then for any constant $q \in \p{\rho(J(x^\star)),1}$, there exists $\delta > 0$ and a neighborhood $\, \mathcal{U} = \mathcal{B}(x^\star, \delta)$ of $x^\star$, such that for any initialization $\hat{x}^{(0)} \in \mathcal{U}$, the EM iterates $\hat{x}^{(t+1)}=M(\hat{x}^{(t)})$ satisfy
    \begin{align}
        \|\hat{x}^{(t)}-x^\star\|_2 \leq q^{ t} \|\hat{x}^{(0)}-x^\star\|_2, \qquad t=0,1,2,\ldots,
    \end{align}
    and, asymptotically for $t \to \infty$,
    \begin{align}
        \limsup_{t\to\infty} \frac{\|\hat{x}^{(t+1)}-x^\star\|_2}{\|\hat{x}^{(t)}-x^\star\|_2} \le \rho \p{J(x^\star)}.
    \end{align}
    \item \emph{(Infinitesimal initialization regime).}
    Let $u_1$ be a unit eigenvector associated with the largest eigenvalue of $J(x^\star)$, $\lambda_1 = \rho\p{J(x^\star)}$.  
    Consider EM iterates initialized at $\hat{x}^{(0)}$ with a small but generic error, in the sense that $\hat{x}^{(0)} \to x^\star$ and $\langle \hat{x}^{(0)} - x^\star,\, u_1 \rangle \neq 0$.
    Then,    \begin{align}\label{eq:Lyap-EM-MRA}
        \lim_{t\to\infty}     \,\lim_{\substack{\hat{x}^{(0)} \to x^\star \\ \langle \hat{x}^{(0)}-x^\star,\,u_1\rangle \neq 0}}       \frac{1}{t}\,        \log\frac{\|\hat{x}^{(t)} - x^\star\|_2}{\|\hat{x}^{(0)} - x^\star\|_2} = \log \rho\p{J(x^\star)}.
    \end{align}
    \end{enumerate}
\end{thm}

Theorem~\ref{thm:EM-MRA-spectral-radius} characterizes the local behavior of the population EM operator through the eigenstructure of the Jacobian at the fixed point $x^\star$. 
Item~(1) establishes smoothness of $M$ and provides an explicit expression for the Jacobian, implying in particular that $J(x^\star)$ is symmetric with spectrum contained in $[0,1]$. 
Item~(2) is a uniform local result: if $\rho\p{J(x^\star)}<1$, then for any $q \in \p{\rho\!\p{J(x^\star)},\,1}$ there exists $\delta>0$ such that $M$ is a contraction on the ball $\mathcal{U} = \mathcal{B}(x^\star,\delta)$; hence, for any initialization $\hat{x}^{(0)} \in \mathcal{U}$, the EM iterates converge linearly to $x^\star$, and the asymptotic contraction rate is at most $\rho\!\p{J(x^\star)}$.
Item~(3) is a sharp asymptotic statement in the infinitesimal-initialization regime: for initializations approaching $x^\star$ with a generic component along the top eigenvector $u_1$, the per-iteration decay rate converges to $\log \rho(J(x^\star))$, i.e., the long-run behavior is dominated by the slowest contracting eigendirection.
Overall, these guarantees are inherently local: they apply only to initializations within a sufficiently small neighborhood around $x^\star$ (a local basin contained in the basin of attraction).
The size and geometry of this neighborhood depend on higher-order terms beyond the Jacobian (i.e., on nonlinear terms of $M$) and are not captured by the linearization alone.

\subsection{Spectral decomposition of the Jacobian} \label{subsec:local-convergence}

Having established in Theorem~\ref{thm:EM-MRA-spectral-radius} that the local convergence rate of the population EM operator is governed by the spectral radius of its Jacobian at the ground-truth $x^\star$, we now quantify this rate in the low-SNR regime of the MRA model~\eqref{eqn:mainModel1D}. We parameterize the ground truth as
\begin{align}\label{eqn:def-v}
    x^\star =  \beta  v,
    \qquad \|v\|_2 = 1,
\end{align}
and consider the asymptotic low-SNR regime in which the noise level $\sigma^2$ and the dimension $d$ are fixed while $\beta \to 0$. 
For notational brevity, write
\begin{align}
    J(\beta) \triangleq J(\beta v) = J(x^\star).
\end{align}
In this regime, the responsibilities $\gamma_\ell(x;Y)$ are nearly uniform over shifts, and the Jacobian admits a Taylor expansion in the parameter $\beta/\sigma$. The results established below formalize this small-$\beta$ expansion and yield an explicit leading-order contraction rate in this low SNR regime.

Recall that $F$ denotes the unitary DFT matrix on $\mathbb{C}^d$ with columns $\{f_k\}_{k=0}^{d-1}$ defined in~\eqref{eqn:DFT-col}. Define the complex Fourier coefficients of $v$ by $\s{V} \triangleq F^\ast {v} \in \mathbb{C}^d$. 
Since $v$ is real, the coefficients satisfy the conjugate symmetry $\s{V}_{-k} = \overline{\s{V}_k}$, for $k\in[d]$. We work under the standard  non-vanishing assumption $\s{V}_k \neq 0$, for all $k\neq  0$, commonly used in the MRA literature~\cite{bendory2017bispectrum,bandeira2023estimation,perry2019sample}. 
For notational convenience, we assume $d$ is odd; all statements extend to even $d$ with only minor modifications. 

The next proposition analyzes the spectral properties of the population EM Jacobian in the low-SNR regime. A proof of the proposition appears in Appendix~\ref{sec:proofOfJacobStructure}; empirical demonstration of the findings is shown in Figure~\ref{fig:3}(a). 

\begin{proposition}
[Second-order Jacobian expansion in the low-SNR regime]
\label{prop:K2-MRA}
Let $d\ge 3$ and consider the population EM map for the MRA model~\eqref{eqn:mainModel1D}, evaluated at the ground truth $x^\star=\beta v$ with fixed noise variance $\sigma^2>0$.
Define for $k=1,\dots,d-1$ the scalars
\begin{align}
    a_k \triangleq 1-\frac{\beta^2}{\sigma^2} |\s{V}_k|^2,
    \qquad b_k \triangleq -\frac{\beta^2}{\sigma^2} \s{V}_k^2.
\end{align}
Then, 
\begin{align}
    \widehat{J}(\beta)
     \triangleq  F^*J(\beta)F
    =
    \begin{pmatrix}
    0      & 0      & 0      & 0      & \cdots & 0      & 0     \\
    0      & a_1    & 0      & 0      & \cdots & 0      & b_1    \\
    0      & 0      & a_2    & 0      & \cdots & b_2    & 0     \\
    0      & 0      & 0      & a_3    & \cdots & 0      & 0     \\
    \vdots & \vdots & \vdots & \vdots & \ddots & \vdots & \vdots \\
    0      & 0      & \overline{b_2} & 0 & \cdots & a_{d-2} & 0  \\
    0      & \overline{b_1} & 0 & 0 & \cdots & 0      & a_{d-1} \\   
    \end{pmatrix}
     + 
    O \p{\frac{\beta^4}{\sigma^4}},
    \label{eq:Jhat_explicit_matrix_prop}
\end{align}
as $\beta\to 0$. In particular, for $k,\ell\in\{0,\dots,d-1\}$, as $\beta\to 0$,
\begin{align}
    \big[\widehat{J}(\beta)\big]_{k\ell}=
    \begin{cases}
        0, & k=\ell=0,\\[2pt]
        1-\frac{\beta^2}{\sigma^2} |\s{V}_k|^2, & k=\ell\neq 0,\\[2pt]
        -\frac{\beta^2}{\sigma^2} \s{V}_k^2, & \ell\equiv -k\pmod d,\ k\neq 0,\\[2pt]
        0, & \text{otherwise},
    \end{cases}
    \qquad
    + O \p{\frac{\beta^4}{\sigma^4}}.
\end{align}
\end{proposition}

We note that the expansion in~\eqref{eq:Jhat_explicit_matrix_prop} couples each Fourier mode $k$ only with its conjugate mode $-k$ (and leaves $k=0$ decoupled). Consequently, $\widehat{J}(\beta)$ preserves each subspace $\mathrm{span}\{f_0\}$ and $\mathrm{span}\{f_k,f_{-k}\}$, and hence it is \emph{block-diagonal with respect to} the orthogonal decomposition
\begin{align}\label{eq:Fourier-splitting}
    \mathbb{R}^d = \mathrm{span}\{f_0\}\ \oplus 
    \bigoplus_{k=1}^{(d-1)/2} \mathrm{span}\{f_k,f_{-k}\}.
\end{align}
Equivalently, consider the matrix representation of $\widehat{J}(\beta)$ in the permuted Fourier basis ordered as $(f_0,\ f_1,\ f_{-1},\ f_2,\ f_{-2},\ \ldots,\ f_{(d-1)/2},\ f_{-(d-1)/2})$.
In this reordered basis, $\widehat{J}(\beta)$ is block-diagonal, consisting of a single $1\times 1$ block associated with $f_0$ and $(d-1)/2$ blocks of size $2\times 2$ associated with the conjugate pairs $(f_k,f_{-k})$.
In particular, on each two-dimensional block $\mathrm{span}\{f_k,f_{-k}\}$ ($k\neq 0$) one has
\begin{align}\label{eq:block-2x2}
    \widehat{J}(\beta)\big|_{\{k,-k\}}
    = I_2 - \frac{\beta^2}{\sigma^2} 
    \begin{bmatrix}
    |\s{V}_k|^2 & \s{V}_k^2\\[2pt]
    \overline{\s{V}_k}^{\,2} & |\s{V}_k|^2
    \end{bmatrix}
    + O\p{\frac{\beta^4}{\sigma^4}},
\end{align}
whose eigenvalues are $1$ and $1-2(\beta^2/\sigma^2)|\s{V}_k|^2$ to second order, with corresponding eigenvectors $w_k$ and $u_k$ defined below. Diagonalizing these $2\times2$ blocks yields the following corollary, proved in Appendix~\ref{sec:proof-of-eigenstructure}.

\begin{corollary}[Jacobian spectral decomposition]\label{cor:block-spectral}
The eigenvalues of $J(\beta)$ decompose as follows.
\begin{enumerate}
    \item \emph{(Mean block.)} For the mean component, $k=0$, of the Jacobian, 
    \begin{align}
        \lambda_0\p{J(\beta)}
          = 0.
    \end{align}
    \item \emph{(Non-mean pairs.)} For each pair $\{k,-k\}$ with $1\le k\le (d-1)/2$, write $\s{V}_k=|\s{V}_k|e^{i\phi_k}$, where $\phi_k \triangleq \phi_{\s{V}}[k]$, and define the following vectors in the basis $\{f_k,f_{-k}\}$,
    \begin{align}\label{eq:uk-wk-def}
        u_k =  \frac{1}{\sqrt{2}} \begin{bmatrix} e^{i\phi_k}\\ \  e^{-i\phi_k}\end{bmatrix}, \qquad
        w_k = \frac{1}{i\sqrt{2}} \begin{bmatrix} e^{i\phi_k}\\ - e^{-i\phi_k}\end{bmatrix}.
    \end{align}
    Then, 
    \begin{align}        
        F^\ast J(\beta)F u_k
        &   =  \p{1-2 \frac{\beta^2}{\sigma^2} |\s{V}_k |^2+O\left(\frac{\beta^4}{\sigma^4}\right)} u_k, \\
        F^\ast J(\beta)F w_k
        & = \p{1+O\left(\frac{\beta^4}{\sigma^4}\right)} w_k.
    \end{align}

    \item \emph{(Spectral radius.)} There exists a number $\kappa_{\max}\geq 0$, depending on $v$ and $d$, such that
    \begin{align}
        \rho\p{J(\beta)}
          = 1 - \kappa_{\max} \frac{\beta^4}{\sigma^4} + O\p{\frac{\beta^6}{\sigma^6}}.
    \end{align}    
\end{enumerate}
\end{corollary}

\begin{figure*}[!t]
    \centering
    \includegraphics[width=1.0\linewidth]{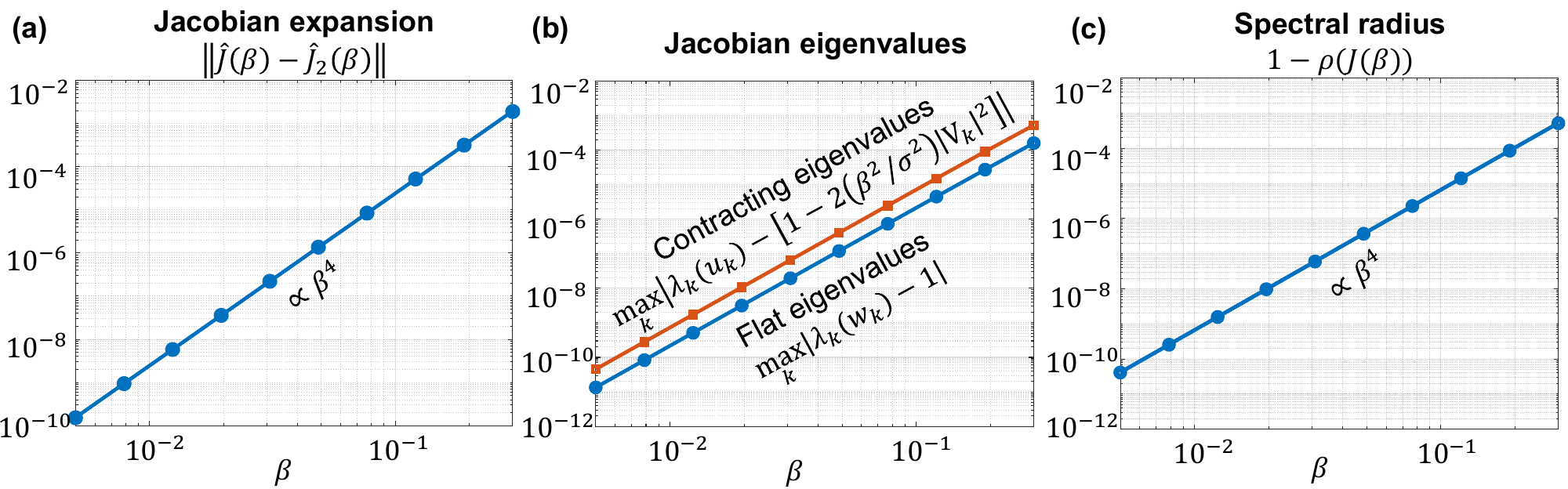}
    \caption{\textbf{Empirical demonstration of the Jacobian's second-order expansion, block spectrum, and spectral radius.}
    \textbf{(a)} Numerical support for the second-order expansion in~\eqref{eq:Jhat_explicit_matrix_prop}: the operator-norm discrepancy between the exact Jacobian $\widehat{J}(\beta)$ and its second-order approximation $\widehat{J}_2(\beta)$ (given by~\eqref{eq:Jhat_explicit_matrix_prop}) decays proportionally to $\beta^4$, consistent with a remainder term of order $O(\beta^4/\sigma^4)$.
    \textbf{(b)} Numerical support for the blockwise spectral characterization in Corollary~\ref{cor:block-spectral}. In the $\{k,-k\}$ contracting blocks, the eigenvalues satisfy
    $\lambda_k(u_k)=1-2(\beta^2/\sigma^2)|\widehat{V}_k|^2+O(\beta^4/\sigma^4)$; correspondingly, the plotted deviation    $\max_k\bigl|\lambda_k(u_k)-\bigl(1-2(\beta^2/\sigma^2)|\widehat{V}_k|^2\bigr)\bigr|$ exhibits slope $4$ on the log--log scale.
    Along the flat directions, $\lambda_k(w_k)=1-O(\beta^4/\sigma^4)$, and the quantity   $\max_k|\lambda_k(w_k)-1|$ also scales as $\beta^4$, in agreement with the corollary.
    \textbf{(c)} Scaling of the spectral radius: the spectral gap $1-\rho({J}(\beta))$ follows $O(\beta^4/\sigma^4)$ and displays slope $4$ on the log--log scale, consistent with Corollary~\ref{cor:block-spectral}(c).
    All curves are computed using Gauss--Hermite quadrature to evaluate the population Jacobian, with signal dimension $d=5$ and noise level $\sigma^2=1$.}
    \label{fig:3}
\end{figure*}

\subsection{Two-phase convergence and iteration complexity} \label{sec:iterationComplexity}
By Corollary~\ref{cor:block-spectral}, for each frequency pair $\{k,-k\}$, the population EM Jacobian $J(\beta)$ has a $2 \times 2$ DFT block with orthonormal eigenvectors $\{u_k,w_k\}$ and eigenvalues
\begin{align}
    \lambda_{u_k}^{\mathrm{ctr}}(J(\beta))   =  1 - 2 \frac{\beta^2}{\sigma^2} |\s{V}_k |^2 + O \p{\frac{\beta^4}{\sigma^4}},
    \qquad
    \lambda_{w_k}^{\mathrm{flat}}(J(\beta))   =  1 + O \p{\frac{\beta^4}{\sigma^4}}. \label{eqn:Jacobian-eigenvalues-def}
\end{align}
Hence, each block decomposes into a \emph{contracting} signal-dependent direction $u_k$ with geometric contraction factor $1-\Theta(\beta^2/\sigma^2)=1-\Theta(\mathrm{SNR})$, and a \emph{flat} direction $w_k$ whose gap to $1$ is only $O(\beta^4/\sigma^4)=O(\mathrm{SNR}^2)$.
The resulting local dynamics are two-phase: (i) a \emph{transient fast phase} while the error has nontrivial components along $\{u_k\}$, yielding per-step shrinkage $\approx 1-2(\beta^2/\sigma^2)|\s{V}_k|^2$; (ii) a \emph{slow tail} once the error mass is concentrated in the flat subspace $\mathrm{span}\{w_k\}$, for which the asymptotic contraction factor becomes $1-O(\beta^4/\sigma^4)=1-O(\mathrm{SNR}^2)$.
Consequently, the spectral radius is governed by the flat branches (see Figure~\ref{fig:3}(c)).
We summarize this two-phase behavior in the following theorem, proved in Appendix~\ref{sec:proof-of-two-phases}.

\begin{thm}[Two-phase convergence]\label{thm:lowSNR-two-phase}
Consider the setting described in Theorem~\ref{thm:EM-MRA-spectral-radius} with $x^\star=\beta v$, $\|v\|_2=1$, fixed $\sigma^2>0$, and varying $0 < \beta \ll \sigma$. Recall the eigenvectors $\{ u_k, w_k \}$ of the Jacobian $J(\beta)$ from~\eqref{eq:uk-wk-def}.
Define the subspaces
\begin{align}
    \mathcal{V}_{\mathrm{ctr}} &   \triangleq  \mathrm{span}\{F u_k: 1\le k\le (d-1)/2\}, \label{eqn:Vctr-def}
    \\
    \mathcal{V}_{\mathrm{flat}} &   \triangleq  \mathrm{span}\{F w_k: 1\le k\le (d-1)/2\}, \label{eqn:Vflat-def}
\end{align}
which are the subspaces from Corollary~\ref{cor:block-spectral}.
Write $e^{(t)}=\hat{x}^{(t)}-x^\star$ and decompose $e^{(t)}=e_{\mathrm{mean}}^{(t)} +  e_{\mathrm{ctr}}^{(t)}+e_{\mathrm{flat}}^{(t)}$ with $e_{\mathrm{ctr}}^{(t)}\in \mathcal{V}_{\mathrm{ctr}}$,  $e_{\mathrm{flat}}^{(t)}\in\mathcal{V}_{\mathrm{flat}}$, and $e_{\mathrm{mean}}^{(t)}\in\mathrm{span}\{\mathbf{1}\}$. Then, there exist a neighborhood $\mathcal{U}_\beta$ of $x^\star$ such that, for any initialization $\hat{x}^{(0)} \in \mathcal{U}_\beta$, the following holds:

\begin{enumerate}
\item \underline{Early transient (contracting modes).}
Fix any $\delta_0\in(0,1)$ and define the index
\begin{align}
    t_\ast (\delta_0) \triangleq \inf\Big\{t \ge 0: \ \|e_{\mathrm{ctr}}^{(t)}\|_2 \le \delta_0 \|e^{(t)}\|_2\Big\}\ \in \ \mathbb{N}. \label{eqn:t_ast_def}
\end{align}
Then, there exist a constant $C(\delta_0)>0$, and a constant $\beta_0>0$, such that for all $0 < \beta \le \beta_0$ and for all integers $t$ with $0\le t<t_\ast(\delta_0)$,
\begin{align}
    \frac{\|e_{\mathrm{ctr}}^{(t+1)}\|_2}{\|e_{\mathrm{ctr}}^{(t)}\|_2}
     \in \pp{ 1-2 \frac{\beta^2}{\sigma^2}\max_{k\neq 0} |\s{V}_k|^2 - C\frac{\beta^4}{\sigma^4} ,\
    1-2 \frac{\beta^2}{\sigma^2}\min_{k\neq 0} |\s{V}_k|^2 + C\frac{\beta^4}{\sigma^4} }.
\end{align}

\item \underline{Asymptotic tail (flat modes).}
Denote by $k_{\max}\in\{1,\dots,(d-1)/2\}$ the index that satisfies $\lambda_{w_{k_{\max}}}(\beta) =\rho(J(\beta))$, so that $F w_{k_{\max}}$ spans an eigendirection attaining the spectral radius of $J(\beta)$, and assume that the initial error has a nonzero component along this direction, i.e., $\langle e^{(0)}, F w_{k_{\max}}\rangle \neq 0$. Then, 
there exists $\kappa_{\max}\ge 0$ (depending on $v$ and $d$), 
\begin{align}
    \lim_{t \to \infty} \lim_{e^{(0)} \to 0} \frac{\|e^{(t+1)}\|_2}{\|e^{(t)}\|_2} = 1-\kappa_{\max} \frac{\beta^4}{\sigma^4} + O\p{\frac{\beta^6}{\sigma^6}}.
\end{align}
\end{enumerate}
\end{thm}

\begin{remark}[Basin of attraction]
\label{rem:basin-shrinks}
At $\beta=0$ (equivalently,  $\mathrm{SNR}=0$), the population EM Jacobian has eigenvalues $1$ for all Fourier blocks $k \neq 0$. Hence, the population map is no longer locally contractive in any direction orthogonal to the mean, and the basin of attraction $\mathcal{U}_\beta$ in Theorem~\ref{thm:lowSNR-two-phase} cannot be taken uniformly in $\beta$ as $\beta \to 0$. However, for every $\beta>0$, $J(\beta)$ admits a strict spectral gap below $1$, and therefore there exists a $\beta$-dependent neighborhood $\mathcal{U}_\beta$ of $x^\star$ on which the Jacobian governs the local dynamics. All results in Sections~\ref{sec:fundemental-properties}-\ref{sec:finite-sample-analysis} are proved on such a local neighborhood. 
\end{remark}

This theorem explains the two-phase convergence visible in Figure~\ref{fig:2}: a fast transient driven by contracting modes, followed by a slow tail governed by the flat modes. A direct result of Theorem~\ref{thm:lowSNR-two-phase} is the iteration complexity of the EM in the MRA problem, as stated next, and proved in Appendix~\ref{sec:proof-of-iteration-complexity}. 

\begin{corollary}[Tail iteration complexity of EM in MRA]\label{cor:flat-iter-lower}
Consider the setting described in Theorem~\ref{thm:lowSNR-two-phase}. Assume that the initialization is generic in the sense that the error $e^{(0)}=\hat{x}^{(0)}-x^\star$ has nonzero projection onto the eigenvector of $J(\beta)$ corresponding to the largest eigenvalue $\rho\p{J(\beta)}$ (the slowest flat eigenspace). Then,  there exist constants $\beta_0>0$ and $C>0$ such that the following holds: for all $0<\beta\le\beta_0$ and any target accuracy $\varepsilon_{\mathrm{abs}}\in(0,\|e^{(0)}\|_2]$, any iteration index $t$ for which $\|e^{(t)}\|_2 \le \varepsilon_{\mathrm{abs}}$ must satisfy
\begin{align}\label{eq:necessary-iter-flat}
  t \ge \frac{1}{C}\,\frac{\sigma^{4}}{\beta^{4}}\, \log  \p{\frac{\|e^{(0)}\|_2}{\varepsilon_{\mathrm{abs}}}}.
\end{align}
\end{corollary}
Corollary~\ref{cor:flat-iter-lower} shows that, in the low–SNR regime $\beta/\sigma\ll 1$, the iteration complexity is governed by the flat Fourier modes whose one-step contraction is only $1-O(\beta^4/\sigma^4) = 1 - O(\mathrm{SNR}^2)$. Consequently, even when initialized inside the local basin, EM requires at least on the order of $(\sigma^4/\beta^4) \asymp \mathrm{SNR}^{-2}$ iterations to reduce the error to a fixed tolerance. This slow rate creates a pronounced computational bottleneck: as SNR decreases, the number of iterations (and hence total runtime) increases at  least quadratically in $\mathrm{SNR}^{-1}$.

\subsection{Fourier magnitudes and phases at low SNR} \label{sec:Fourier-mag-and-phases-convergence}

The next result refines the picture obtained in Corollary~\ref{cor:block-spectral} and  Theorem~\ref{thm:lowSNR-two-phase} by working directly in the Fourier basis and decomposing, on each nonzero block $\{k,-k\}$, the error into radial and tangential components, that is, into errors in the Fourier magnitudes and phases.

\begin{proposition}[Local expansion of Fourier magnitudes/phases around $x^\star$]
\label{prop:mag-phase-first-order}
Consider the setting of Theorem~\ref{thm:EM-MRA-spectral-radius} and Corollary~\ref{cor:block-spectral}.
Let $x^\star$ denote the ground truth and let $x$ be any signal in $\mathbb{R}^d$.
Let $F$ be the unitary DFT matrix and define $\s{X}^\star \triangleq F^\ast x^\star$ and $\s{X}\triangleq F^\ast x$.
Assume that $\s{X}^\star[k]\neq 0$ for all $1\le k\le (d-1)/2$.
Let $\{u_k,w_k\}_{k=1}^{(d-1)/2}$ be the eigenvectors from Corollary~\ref{cor:block-spectral} and set $e\triangleq x-x^\star$.

Then, there exist constants $r_0>0$ and $C>0$, depending only on $x^\star$ and $d$, such that whenever $\|e\|_2\le r_0$, the following first-order expansions hold for every $k=1,\dots,(d-1)/2$:
\begin{align}
    \label{eq:delta-rk-alpha-statement}
    \text{(Magnitude error)} \qquad
    |\s{X}[k]| - |\s{X}^\star[k]|
    &= \frac{1}{\sqrt{2}}\;\langle e, u_k\rangle + R^{\mathrm{mag}}_k(e),
      \\[4pt]
    \label{eq:delta-phik-beta-statement}
    \text{(Phase error)} \qquad
    \phi_{\s{X}}[k] - \phi_{\s{X}^\star}[k]
    &= -\,\frac{1}{\sqrt{2}\,|\s{X}^\star[k]|}\;\langle e, w_k\rangle + R^{\mathrm{ph}}_k(e),
\end{align}
where the remainder terms satisfy the uniform bound
\begin{align}
    \max_{1\le k\le (d-1)/2}\bigl(|R^{\mathrm{mag}}_k(e)| + |R^{\mathrm{ph}}_k(e)|\bigr) \le C\,\|e\|_2^2.
\end{align}
\end{proposition}

Proposition~\ref{prop:mag-phase-first-order} (proved in Appendix~\ref{sec:proof-Fourier-phases-magnitudes-convergence}) may be viewed as a Fourier-coordinate refinement of the two-phase convergence result in Theorem~\ref{thm:lowSNR-two-phase}.
On each nonzero block $\{k,-k\}$, the eigenvector $u_k$ is aligned, to first order, with the radial direction in the complex plane: the magnitude error $|\s{\hat{X}}^{(t)}[k]| - |\s{X}^\star[k]|$ is a fixed multiple of $\langle e^{(t)},u_k\rangle$ up to $O(\|e^{(t)}\|_2^2)$ corrections.
Thus, the eigenvalue $\lambda_{u_k} = 1 - O(\mathrm{SNR})$, defined in~\eqref{eqn:Jacobian-eigenvalues-def}, governs the leading-order decay of the magnitude error. In contrast, $w_k$ spans the tangential direction corresponding to infinitesimal rotations of $\s{\hat{X}}^{(t)}[k]$, and Proposition~\ref{prop:mag-phase-first-order} shows that the phase error $\phi_{\s{\hat{X}}^{(t)}}[k] - \phi_{\s{X}^\star}[k] $ is controlled, to first order, by $\langle e^{(t)},w_k\rangle$, again up to $O(\|e^{(t)}\|_2^2)$ terms.
Its eigenvalue satisfies $\lambda_{w_k} = 1 - O (\mathrm{SNR}^2)$, so any contraction of the phase error can occur, at best, on the much slower $\mathrm{SNR}^2$ scale at the level of the Jacobian.

This first-order dichotomy is consistent with the classical moment-based identifiability picture in MRA. In the Fourier domain, the second-order autocorrelation (equivalently, the power spectrum) determines the Fourier magnitudes $\{|\s{X}^\star[k]|^2\}$ but contains no phase information. Accordingly, in the low-SNR expansion, the magnitude (radial) directions already appear in the EM linearization at order $O(\mathrm{SNR})$, leading to a comparatively faster contraction. In contrast, identifying the phases requires genuinely third-order information, such as the bispectrum~\cite{bendory2017bispectrum,bandeira2023estimation,perry2019sample}; correspondingly, the phase (tangential) directions enter the EM linearization only through higher-order terms in $\mathrm{SNR}$, and thus can contract only on a much slower scale. This result is in line with the result established in~\cite{katsevich2023likelihood}.

\section{EM bias to initialization at low SNR}\label{sec:bias-to-initalization}

In contrast to the previous section, which studied convergence to the ground truth once the iterate enters a basin of attraction, here we focus on the complementary phenomenon of \emph{bias to initialization} in the population low-SNR limit. In this regime, both the true signal and the initialization are small relative to the noise level, and as a result the population EM dynamics exhibit only a very weak per-iteration drift away from the initialization.

We implement the low-SNR limit via the scaling
\begin{align}\label{eq:mra-beta-scaling}
    x^\star = \beta v, \qquad \hat{x}^{(t)} = \beta \hat{v}^{(t)}, \qquad \beta\to 0,
\end{align}
and assume that the normalized ground truth and initialization are uniformly bounded, $\|v\|\le P$, and $\|\hat{v}^{(0)}\|\le P$, for some fixed constant $P<\infty$ (independent of $\beta$). 
To state the result, let $\mathbf{1}\in\mathbb{R}^d$ denote the all-ones vector and define the \emph{mean projector} $\Pi_{\mathrm{mean}} \triangleq \frac{1}{d}\mathbf{1}\mathbf{1}^\top$, so that $\Pi_{\mathrm{mean}}x$ is the constant vector whose entries equal the sample mean $\frac{1}{d}\sum_{j=1}^d x_j$.

\begin{thm}[Bias to initialization at low SNR]\label{thm:mra-lowSNR-iteration-bias-to-init}
Consider the population EM iteration \eqref{eq:Phi-pop-1d} for the MRA model \eqref{eqn:mainModel1D}.
Fix $T \in \mathbb{N}$ and assume the low-SNR scaling \eqref{eq:mra-beta-scaling}, such that $\mathrm{SNR} \asymp \beta^2 / \sigma^2$~\eqref{eqn:snr-def} with $\|v\|\le P$ and $\|\hat{v}^{(0)}\|\le P$
for some fixed $P<\infty$.
Then, there exists a constant $C \triangleq C(d,P)<\infty$, depending only on $(d,P)$, such that
\begin{align}\label{eq:mra-iter-bound}
    \limsup_{\beta\to 0} \; \frac{1}{\mathrm{SNR}}\,  \frac{\big\|\hat{x}^{(T)}-\hat{x}^{(0)}-\Pi_{\mathrm{mean}}\big(x^\star-\hat{x}^{(0)}\big)\big\|}{\|\hat{x}^{(0)}\|}  \leq  C\,T.
\end{align}
In particular:
\begin{enumerate}
\item \emph{(One-time mean  correction.)} The mean component is corrected in a single iteration:
\begin{align}\label{eq:mean-one-step}
    \Pi_{\mathrm{mean}}\hat{x}^{(1)}=\Pi_{\mathrm{mean}}x^\star,
\end{align}
independently of the initialization $\hat{x}^{(0)}$.
\item \emph{(Mean-aligned initialization.)} If the initialization is mean-aligned,
$\Pi_{\mathrm{mean}}\hat{x}^{(0)}=\Pi_{\mathrm{mean}}x^\star$, then the offset in \eqref{eq:mra-iter-bound} vanishes and
\begin{align}\label{eq:mra-aligned-bound}
    \limsup_{\beta\to 0} \;
    \frac{1}{\mathrm{SNR}}\, \frac{\|\hat{x}^{(T)}-\hat{x}^{(0)}\|}{\|\hat{x}^{(0)}\|} \leq  C\,T.
\end{align}
\end{enumerate}
\end{thm}

Theorem~\ref{thm:mra-lowSNR-iteration-bias-to-init}, proved in Appendix~\ref{sec:proof-of-bias-to-initalization}, formalizes an intrinsic \emph{iteration-complexity barrier} underlying the bias to initialization in low SNR.
The population EM normalized drift in the non-mean components is bounded by $O(\mathrm{SNR})$ per iteration. Thus, a constant-fraction displacement from a generic initialization requires at least $T \gtrsim \mathrm{SNR}^{-1}$ iterations.

\section{Einstein from Noise at vanishing SNR} \label{sec:EfN}

In this section, we analyze the Einstein from Noise phenomenon, extending previous analyses from the single-iteration hard-assignment algorithm~\cite{balanov2024einstein} to the more general setting of multiple iterations in the EM algorithm. 

The original Einstein from Noise phenomenon illustrates a striking instance of model bias in scientific data analysis. Consider the following scenario: researchers aim to reconstruct an underlying signal based on the assumption that each observation is a noisy, randomly shifted copy of a known template. 
The phenomenon was first demonstrated using an image of Einstein as the template, which gave rise to its name~\cite{shatsky2009method,henderson2013avoiding}. However, in reality, the observations are purely noise and do not contain a signal. This corresponds to the extreme case of $\mathrm{SNR} = 0$, where the data no longer follow the MRA model described in~\eqref{eqn:mainModel1D}, but instead adhere to the distribution:
\begin{align} \label{eqn:calPModel} 
    y_1, y_2, \ldots, y_{n} \stackrel{\text{i.i.d.}}{\sim} \mathcal{N}(\mathbf{0}, \sigma^2 I_{d}).
\end{align}
Despite the absence of a signal, researchers, who mistakenly believe the data contain meaningful information (i.e., an image of Einstein), may still apply standard hard-assignment methods, such as aligning observations via cross-correlation with the template and averaging the aligned outputs. Interestingly, empirical results have shown that when the data is pure noise, the reconstructed image resembles the initial template (e.g., Einstein's image)~\cite{shatsky2009method, sigworth1998maximum}. This effect, where structure is seemingly recovered from pure noise, is not just a theoretical phenomenon, but has played a prominent role in real-world scientific debates, most notably the structural analysis of an HIV molecule~\cite{mao2013molecular, henderson2013avoiding, van2013finding, subramaniam2013structure, mao2013reply}, as recently revisited in~\cite{balanov2024einstein}. Here, we extend this analysis beyond the single-step hard-assignment setting to understand how the iterative EM algorithm interacts with this form of confirmation bias in the absence of a true signal.

\paragraph{Summary of prior work.} 
Previous studies investigated the behavior of the Einstein from Noise phenomenon in the context of the hard-assignment estimator. It was shown that in the asymptotic regime, as $n \to \infty$, the Fourier phases of the estimator converge to those of the template \cite{balanov2024einstein}. Specifically, the MSE between the template's and the estimator's Fourier phases decays at a rate of $1/n$. Since Fourier phases are key to defining geometric structures like edges and contours, this explains why the reconstructed image retains structural similarities to the template. In high-dimensional settings, where the signal dimension increases, the rate of convergence of the Fourier phases was found to be inversely proportional to the square of the template's Fourier magnitudes. Additionally, in this regime, the Fourier magnitudes of the estimator converge to a scaled version of the template's Fourier magnitudes. Similar results were extended to Gaussian mixture models, showing that the initialization of the centroids in algorithms such as $K$-means and EM can induce confirmation bias \cite{balanov2025confirmation}. See~\cite{silva2025observation} for a related recent result along similar lines.

\paragraph{Notations.}
For clarity of notation, and since the observations contain no signal, we denote them as $y_i = \xi_i$ for all $i \in \{1,2, \dots, n\}$, where $\xi_i \sim \mathcal{N}(\mathbf{0}, \sigma^2 I_{d})$ represents isotropic Gaussian noise, according to~\eqref{eqn:calPModel}. To further simplify the exposition, we assume throughout this section that the noise variance is $\sigma^2 = 1$.

Let  $\s{\hat{X}}^{(t)}$, and $\s{N}_i$, denote the unitary DFTs of $\hat{x}^{(t)}$, and $\xi_i$, respectively, for $i\in \{1,2, \ldots, n\}$. These DFT sequences can be equivalently represented in polar coordinates as follows, 
\begin{align}\label{FourierSpace}
        \s{\hat{X}}^{(t)}  = \{|\s{\hat{X}}^{(t)}[k]|e^{j\phi_{\s{\hat{X}}^{(t)}}[k]}\}_{k=0}^{d-1}, \quad 
        \s{N}_i = \{\abs{\s{N}_i[k]} e^{j\phi_{\s{N}_i}[k]}\}_{k=0}^{d-1}.
\end{align}
Note that the random variables $\ppp{|\s{N}_i[k]|}_{k=0}^{d/2}$ and $\ppp{\phi_{\s{N}_i}[k]}_{k=0}^{d/2}$ are two independent sequences of i.i.d. random variables, such that, $|\s{N}_i[k]| \sim \s{Rayleigh} \left({\sigma^2}\right)$ has Rayleigh distribution, and the phase $\phi_{\s{N}_i}[k] \sim \s{Unif}[-\pi,\pi)$ is uniformly distributed over $[-\pi,\pi)$.


In this section, we analyze the behavior of Einstein from Noise in the population limit, where $n \to \infty$, so all quantities are understood in this limit, and the EM update is deterministic (i.e., there is no sampling randomness). The finite-sample analysis can be found in Section~\ref{sec:finite-sample-EfN}.
We begin with the finite-dimensional setting (fixed $d$), and then treat the high-dimensional regime (where $d\to\infty$) in Section~\ref{sec:high-dim-EfN}.

\subsection{Fixed dimension: phase invariance and magnitude decay}

We start with analyzing the Einstein from Noise population EM for successive iterations. 

\begin{thm}[Finite dimension: Successive population iterations]
\label{thm:1}
Fix $d\geq 2$, and consider the population EM operator $M$ associated with the model~\eqref{eqn:calPModel}, obtained in the limit $n\to\infty$. Let $\hat{x}^{(t)}$ denote the $t$-th population iterate $\hat{x}^{(t+1)} = M(\hat{x}^{(t)})$ for $t=0,1,\dots,T-1$, and let ${\s{\hat{X}}}^{(t)}[k]$ denote the discrete Fourier transform of $x^{(t)}$. Assume that $\s{\hat{X}}^{(t)}[k] \neq 0$ for all $0<k\leq d-1$ and $t=0,\dots,T-1$.
Then:
\begin{enumerate}
    \item \emph{(Contraction of Fourier components.)} 
    For each $t\in\{0,\dots,T-1\}$ and each $0<k\leq d-1$, there exists a deterministic real constant $\alpha_k^{(t)}\in(0,1)$ such that
    \begin{align}
        \s{\hat{X}}^{(t+1)}[k] = \alpha_k^{(t)} \, \s{\hat{X}}^{(t)}[k].        \label{eqn:magnitudeConvergenceAsymptoticM}
    \end{align}

    \item \emph{(Fourier phases preservation.)} 
    For any $0<k\leq d-1$ and any $t\in\{0,\dots,T-1\}$,
    \begin{align}
        \phi_{{\s{\hat{X}}}^{(t+1)}}[k]  =  \phi_{{\s{\hat{X}}}^{(t)}}[k].
        \label{eqn:fourierPhasesConvergence}
    \end{align}

    \item \emph{(Positive correlation.)} 
    For any two successive population iterates,
    \begin{align}
        \langle \hat{x}^{(t+1)}, \hat{x}^{(t)} \rangle  >  0, \label{eqn:positiveCorrelation}
    \end{align}
    for all $t=0,\dots,T-1$.
\end{enumerate}
\end{thm}

Theorem~\ref{thm:1} (proved in Appendix~\ref{thm:proofs1}) shows that, at the population level and for fixed~$d$, the EM operator acts diagonally in the Fourier basis: each non-mean Fourier mode is multiplied by a real contraction factor $\alpha_k^{(t)}\in(0,1)$. Consequently, the Fourier phases of all non-mean components remain the same as in the previous iteration estimate, while their magnitudes are strictly decreased at each iteration. The next corollary extends the population Einstein from Noise to the multi-iteration case, and is proved in Appendix~\ref{sec:proofOfMultiIterationEfN}.

\begin{corollary}
[Finite dimension: Multi-iteration population behavior]
\label{thm:3.5}
Fix $d\geq 2$, and consider the population EM operator $M$ associated with the model~\eqref{eqn:calPModel}, obtained in the limit $n\to\infty$. Let $\hat{x}^{(t)}$ denote the $t$-th population iterate $\hat{x}^{(t+1)} = M\p{\hat{x}^{(t)}}$ for $t=0,1,\dots,T-1$, and let ${\s{\hat{X}}}^{(t)}[k]$ denote the discrete Fourier transform of $x^{(t)}$. Assume that $\s{\hat{X}}^{(0)}[k] \neq 0$ for all $0<k\leq d-1$.
Then:
\begin{enumerate}
    \item \emph{(Multi-iteration Fourier phases.)} 
    For any fixed $T\in\mathbb{N}$ and any $0<k\leq d-1$, the Fourier phases are preserved along the population trajectory:
    \begin{align}
        \phi_{\s{\hat{X}}^{(T)}}[k] = \phi_{\s{\hat{X}}^{(0)}}[k].
    \end{align}

    \item \emph{(Magnitudes converge to zero as $t\to\infty$.)}
    For every $0<k\leq d-1$,
    \begin{align}
        \lim_{t\to\infty} |{\s{\hat{X}}}^{(t)}[k]| = 0.
    \end{align}
    In particular, $\hat{x}^{(t)} \to 0$ as $t\to\infty$.
\end{enumerate}
\end{corollary}

Corollary~\ref{thm:3.5} highlights a distinctive feature of the Einstein from Noise regime: independently of the initialization, the population EM trajectory converges to 0, while retaining the initialization phases. However, the contraction towards zero becomes extremely weak near the origin (as shown in Figure~\ref{fig:1}(b)), which motivates the statement below (proved in Appendix~\ref{sec:proofOfEfNlowSNR}).  

\begin{thm}[Low-magnitude regime: contraction expansion and global rate]
\label{thm:low-mag-rate}
Assume the conditions of Corollary~\ref{thm:3.5}, and suppose that the initialization is in the low-magnitude region:
$\|\hat{x}^{(0)}\|\le \delta$, where $\delta>0$ is sufficiently small. 
Then:
\begin{enumerate}
    \item \emph{(Successive-step expansion.)}
    For every $t\ge 0$ and every $1\le k\le d-1$,
    \begin{align}
        \alpha_k^{(t)}
         = 1-|{\s{\hat{X}}}^{(t)}[k]|^2
        + E_{t,k},
        \label{eqn:alpha-low-mag-expansion}
    \end{align}
    where $\alpha_k^{(t)}$ is the contraction coefficient of the $k$-th Fourier component defined in~\eqref{eqn:magnitudeConvergenceAsymptoticM}, with a uniform remainder satisfying $|E_{t,k}| \le C\,|{\s{\hat{X}}}^{(t)}[k]|^4$, for a global constant $C$.

    \item \emph{(Uniform multi-iteration bounds.)}
    For any $\varepsilon\in(0,1)$,  there exists $\delta(\varepsilon)>0$ such that if
    $\big| \s{\hat{X}}^{(0)}[k]\big|^2\le \delta(\varepsilon)$, then for all $T\ge 0$,
    \begin{align}
        \frac{1}{\sqrt{1 + 2T(1+\varepsilon)\big|  \s{\hat{X}}^{(0)}[k]\big|^2}}
         \le 
        \frac{\big| \s{\hat{X}}^{(T)}[k]\big|}{\big|\s{\hat{X}}^{(0)}[k]\big|}
         \le 
        \frac{1}{\sqrt{1 + 2T(1-\varepsilon)\big|  \s{\hat{X}}^{(0)}[k]\big|^2}}.
        \label{eqn:multi-iter-logistic-bounds}
    \end{align}
\end{enumerate}
\end{thm}

The bounds \eqref{eqn:multi-iter-logistic-bounds} imply the asymptotic decay
\begin{align}
    \frac{\big|\s{\hat{X}}^{(T)}[k]\big|}{\big|\s{\hat{X}}^{(0)}[k]\big|}
     \asymp 
    \frac{1}{\sqrt{1+2T\big|  \s{\hat{X}}^{(0)}[k]\big|^2}},
    \label{eqn:asymp-low-mag}
\end{align}
uniformly over all iteration counts $T$ as the initialization magnitude tends to zero. 
Thus, the EM iterates produces a convergence of the Fourier magnitudes toward $0$ in a slow rate of $1/\sqrt{T}$, while leaving phases unchanged. 

\paragraph{Comparison with geometric convergence near the ground truth.}
The bounds in~\eqref{eqn:multi-iter-logistic-bounds}  exhibit a qualitatively different convergence mechanism from the Jacobian-based analysis around a nonzero ground truth, as analyzed in Section~\ref{sec:fundemental-properties}. When $x^\star=\beta v$ with $\beta>0$, the Jacobian expansion in Corollary~\ref{thm:3.5} shows that the population EM map is locally contractive around $x^\star$ so that, as long as the iterates remain in a sufficiently small neighborhood of $x^\star$, the error decays geometrically at rate $\rho(J(x^\star))^t$. In the Einstein  from Noise settings, where $\mathrm{SNR} = 0$, the Jacobian is flat in the sense that $\rho(J(0))=1$ and the first nontrivial term in the Taylor expansion of the EM map is quadratic in the Fourier magnitudes, as captured by the expansion~\eqref{eqn:alpha-low-mag-expansion}. Consequently, the linear term provides no contraction, and the dynamics are governed by the second-order term only, which integrates into the bound~\eqref{eqn:multi-iter-logistic-bounds}.

\subsection{High dimension asymptotics}\label{sec:high-dim-EfN}

We now study the population Einstein from Noise dynamics in a high-dimensional regime where the signal length $d\to\infty$.
For each dimension $d$, let $\hat{x}^{(t)}_d\in\mathbb{R}^d$ denote the $t$-th population iterate. 
For each $d$ and $t$, we define the shift-autocorrelations of $\hat{x}^{(t)}_d$ by
\begin{align}
    \rho^{(t)}_{d,\ell}
     \triangleq  
    \big\langle \mathcal{T}_{\ell} \hat{x}^{(t)}_d, \hat{x}^{(t)}_d\big\rangle,
    \qquad \ell\in\{0,\dots,d-1\}.
    \label{eqn:app_H1_1_rewrite}
\end{align}
We impose a mild asymptotic decorrelation condition: far-apart circular shifts of the iterate become orthogonal as $d$ grows.
Formally, for every fixed $t\in\{0,\dots,T\}$,
\begin{align}
    \lim_{L\to\infty} \limsup_{d\to\infty} \max_{L\le \ell\le d-L}\big|\rho^{(t)}_{d,\ell}\big| = 0.
    \label{eqn:rho_decay_assump}
\end{align}
Finally, we assume that the initialization has dimension-independent energy: there exists $\tau \ge 0$ such that for all $d \ge 2$,
\begin{align}
    \| \hat{x}^{(0)}_d\|_2 = \tau.
    \label{eqn:init_fixed_energy}
\end{align}
We are now ready to state the high-dimensional Einstein from Noise theorem, which is proved in Appendix~\ref{sec:proofOfHighDimensionEfN}.

\begin{thm}[High dimension population iterations]
\label{thm:2}
Fix $1\le T<\infty$ and let $ \hat{x}^{(t)}_d$ denote the $t$-th population EM iterate produced by Algorithm~\ref{alg:generalizedEfNsoft} in dimension $d$. 
Assume \eqref{eqn:rho_decay_assump} holds for all $t\in\{0,\dots,T\}$ and that the initialization satisfies \eqref{eqn:init_fixed_energy}. 
Then, 
\begin{enumerate}
    \item \emph{(Successive iterations.)} For every $t\in\{0,\dots,T-1\}$,
    \begin{align}
        \lim_{d\to\infty}\ 
        \big\| \hat{x}^{(t+1)}_d- \hat{x}^{(t)}_d\big\|_2 = 0 .
        \label{eqn:succ_iter_highdim}
    \end{align}

    \item \emph{(Multi-iteration.)} For any fixed $1\le T<\infty$,
    \begin{align}
        \lim_{d\to\infty}\ 
        \big\| \hat{x}^{(T)}_d- \hat{x}^{(0)}_d\big\|_2
         = 0 .
        \label{eqn:multi_iter_highdim}
    \end{align}
\end{enumerate}
\end{thm}

Theorem~\ref{thm:2}(1) shows that, under asymptotic decorrelation of distant shifts, each population update becomes increasingly
close to the identity as $d\to\infty$. 
In particular, the contraction factors from the finite-dimensional analysis
(see \eqref{eqn:magnitudeConvergenceAsymptoticM}) satisfy $\alpha^{(t)}_{d,k}\to 1$ for every fixed Fourier mode $k$,
so the Fourier magnitudes of $ \hat{x}^{(t)}_d$ remain essentially those of the initialization over any fixed number of iterations.
As a direct consequence, Theorem~\ref{thm:2}(2) shows a high-dimensional stagnation phenomenon, in which the population EM trajectory does not move appreciably away from its initialization over any finite number of iterations. 
Thus, unlike the finite-dimensional Einstein from Noise behavior of Corollary~\ref{thm:3.5}, the high-dimensional population dynamics do not exhibit perceptible magnitude convergence to 0.

\section{Finite-sample analysis}\label{sec:finite-sample-analysis}

This section develops a complementary, finite-sample analysis of EM. Throughout this section, we use the notation $\hat{x}^{(t)}$ for the empirical EM iterates and $\hat{x}^{(t)}_{\mathrm{pop}}$ for the corresponding population iterates initialized at the same point.
That is, for a fixed number of iterations $T\in\mathbb{N}$, and for $t=0,1,\ldots,T-1$,
\begin{align}
\label{eq:EM-iters}
    \hat{x}^{(0)}_{\mathrm{pop}}= \hat{x}^{(0)},\qquad \hat{x}^{(t+1)}_{\mathrm{pop}} = M\big(\hat{x}^{(t)}_{\mathrm{pop}}\big),\qquad \hat{x}^{(t+1)} = M_n\big(\hat{x}^{(t)};\mathcal{Y}\big) .
\end{align}
Our objective is to control the deviation $\|\hat{x}^{(t)}-\hat{x}^{(t)}_{\mathrm{pop}}\|$ uniformly over
$t\le T$, and to translate such bounds into quantitative statements about finite-sample effects on the EM
trajectory in the regimes studied in this work.

Section~\ref{sec:finite-sample-tracking} provides \emph{sufficient} conditions under which the empirical EM operator $M_n$ tracks its population counterpart $M$: under local contraction and a uniform deviation bound, the empirical iterates follow the population trajectory geometrically up to a statistical neighborhood. 
Section~\ref{sec:GoN} then shows how this same finite-sample perturbation can qualitatively reshape the dynamics at moderately low SNR, producing the \emph{Ghost of Newton} phenomenon: the iterates initially improve by tracking the contracting population map, but after a crossover time the accumulated deviation dominates and the trajectory peels away toward a sample-dependent fixed point.
Section~\ref{sec:finite-sample-necessary} complements this with a \emph{necessary} perspective, proving that even inside the local basin there is an intrinsic finite-sample error floor for any empirical fixed point.
Finally, Section~\ref{sec:finite-sample-EfN} turns to the $\mathrm{SNR}=0$ setting and quantifies how finite-sample fluctuations manifest as a per-iteration $n^{-1}$ Fourier-phase drift and its accumulation across iterations.

\subsection{Finite-sample EM tracking} \label{sec:finite-sample-tracking}
The following theorem formalizes the finite-sample tracking behavior of EM in a local basin of attraction under standard regularity conditions. It adapts a classical result for Gaussian mixture models (GMMs) to the MRA setting~\cite{balakrishnan2017statistical}.
For a map $M:\mathbb{R}^d\to\mathbb{R}^d$ and a set $\mathcal{S}\subseteq\mathbb{R}^d$, we denote the image of $\mathcal{S}$ under $M$ by $M(\mathcal{S})\triangleq \{M(x):x\in\mathcal{S}\}$.

\begin{assum}[Finite-sample EM tracking conditions]
\label{assump:finite-sample-EM}
Let $M$ and $M_n(\cdot;\mathcal{Y})$ be the population and empirical EM operators defined in
\eqref{eq:Phi-pop-1d} and \eqref{eq:Phi-emp-1d}. Fix a compact set $\mathcal{B}\subset\mathbb{R}^d$ containing the ground-truth $x^\star$, such that $M(x^\star) = x^\star$. 
Assume that the following hold:

\begin{enumerate}
    \item \emph{(A1) Local contraction of the population map.}
    There exists $\kappa\in[0,1)$ such that for all $x\in\mathcal{B}$, \begin{align}
        \|M(x)-x^\star\| \le \kappa\,\|x-x^\star\|.  \label{eq:A1_local_contraction}
    \end{align}

    \item \emph{(A2) Uniform finite-sample accuracy.}
    For every sample size $n\in\N$ and confidence level $\delta\in(0,1)$, there exists a deterministic tolerance $\varepsilon_M(n,\delta)>0$ such that 
    \begin{align}
       \mathbb{P} \pp{\sup_{x\in\mathcal{B}}  \|M_n(x;\mathcal{Y})-M(x)\|\le \varepsilon_M(n,\delta)} \ge 1-\delta.
        \label{eq:En_delta_def}
    \end{align}

    \item \emph{(A3) Initialization and invariance.}
    The initialization satisfies $\hat{x}^{(0)}\in\mathcal{B}$, and on the event in \eqref{eq:En_delta_def} the set $\mathcal{B}$ is invariant under both maps:
    \begin{align}
        M(\mathcal{B}) \subseteq \mathcal{B},
        \qquad M_n(\mathcal{B};\mathcal{Y}) \subseteq \mathcal{B}.
        \label{eq:A3_invariance}
    \end{align}
\end{enumerate}
\end{assum}

Assumption~\ref{assump:finite-sample-EM} postulates a local basin $\mathcal{B}$ that is invariant under both the population map $M$ and the empirical map $M_n$, and assumes that $M$ is contractive on $\mathcal{B}$.
In practice, the contraction factor $\kappa$ may be chosen using the spectral-radius bound on the Jacobian $J(x^\star)$ from Theorem~\ref{thm:EM-MRA-spectral-radius}.
The tolerance $\varepsilon_M(n,\delta)$ is the deterministic quantity appearing in~\eqref{eq:En_delta_def} that upper-bounds, with probability at least $1-\delta$, the uniform deviation of the empirical operator from its population counterpart over the basin.
Under these conditions, the empirical EM iterates follow the population iterates at a geometric rate until they reach a neighborhood whose radius is proportional to $\varepsilon_M(n,\delta)$.

\begin{thm}[Finite-sample EM tracking]
\label{thm:finite-sample-tracking-groundtruth}
Recall the definition of the empirical EM iterates~\eqref{eq:EM-iters}. Then, under Assumption~\ref{assump:finite-sample-EM}, with probability at least $1-\delta$, 
\begin{align}
    \| \hat{x}^{(t)} - x^\star\| \leq \kappa^t \|\hat{x}^{(0)} - x^\star \| + \frac{1-\kappa^t}{1-\kappa}  \varepsilon_M(n,\delta).
    \label{eq:tracking-to-groundtruth}
\end{align}
In particular, the asymptotic tracking bias satisfies
\begin{align}
    \limsup_{t\to\infty} \|\hat{x}^{(t)} - x^\star\| \leq \frac{\varepsilon_M(n,\delta)}{1-\kappa}.
    \label{eq:asymptotic-bias-floor}
\end{align}
\end{thm}

Theorem~\ref{thm:finite-sample-tracking-groundtruth} shows that the gap between the empirical and population EM trajectories is governed by $\varepsilon_M(n,\delta)$. 
To make this quantitative, we next bound $\varepsilon_M(n,\delta)$ for MRA observations~\eqref{eqn:mainModel1D}; the proof appears in Appendix~\ref{sec:proofOfUniformBoundGauss}.

\begin{proposition}[Uniform sample-population deviation]\label{lem:subgaussian-epsn}
Let $\mathcal{Y} = \{ y_i\}_{i=1}^{n}$ be the observations drawn according to 
the MRA model~\eqref{eqn:mainModel1D}, with a ground-truth signal 
$x^\star \in \mathbb{R}^d$.  
Fix a compact ball for the parameter space $\mathcal{B}=\{x:\|x\|_2\le R\}$ and assume $\|x^\star\|\lesssim \sigma\sqrt d$.
Then, there exist universal constants $C_0,C_1>0$ such that for every $\delta\in(0,1)$, if
$n  \geq C_0 \p{\log(2/\delta)}^3$, then with probability at least $1-\delta$,
\begin{align}
    \sup_{x\in\mathcal{B}} \big\| M_n(x; \mathcal{Y}) - M(x)\big\| \leq  C_1 \left(\sigma \sqrt{\frac{d}{n}} \sqrt{d+\log\frac{2}{\delta}} + \frac{R d^{3/2}}{\sqrt n} \right). \label{eq:MRA-uniform-rate}
\end{align}
\end{proposition}

\subsection{Ghost of Newton: finite-sample artifacts} 
\label{sec:GoN}

At moderately low SNR, we observe an unexpected finite-sample artifact in the EM dynamics. 
Figure~\ref{fig:4} illustrates the effect. Specifically, Figure~\ref{fig:4}(a) plots the reconstruction MSE (defined in~\eqref{eqn:mseDef} with respect to the ground truth) as a function of the iteration index $t$. 
Starting from the initial template, the EM iterates first move towards the true signal, yielding a clear initial decrease in error. 
After a transient period, however, the trajectory reverses direction: the error begins to increase and the estimate gradually drifts away from the ground truth. 
We emphasize that this behavior is not tied to the specific Newton/Einstein pair used in the running example; analogous non-monotone error curves appear for other ground-truth signals and various initial templates. 
We adopt the term ``Ghost of Newton" for consistency with the illustrative example in Figure~\ref{fig:4}.

Empirically, this drift-after-improvement regime is specific to EM. 
The hard-assignment algorithm does not exhibit a comparable intermediate phase; instead, as the SNR increases, it undergoes a sharp transition from the Einstein from Noise behavior (template reconstruction at very low SNR) to accurate recovery of the ground truth, without passing through a gradual deterioration stage. 
Additional numerical evidence and parameter sweeps supporting this distinction are provided in Appendix~\ref{sec:empiricalDetails}.

\begin{figure*}[!t]
    \centering
    \includegraphics[width=1.0 \linewidth]{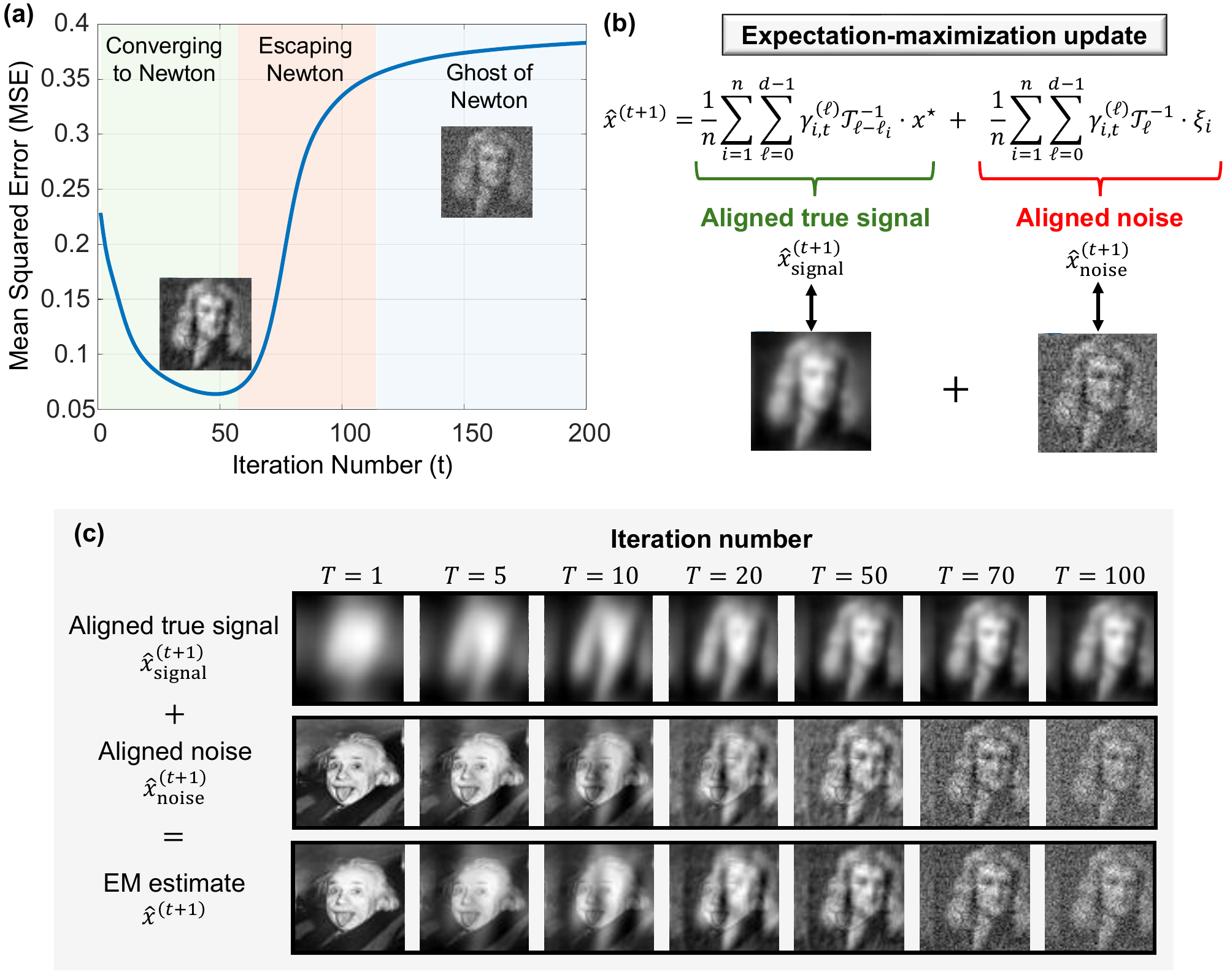}
    \caption{\textbf{The Ghost of Newton phenomenon.} \textbf{(a)} MSE, as defined in \eqref{eqn:mseDef}, between the EM estimator at iteration $t$ and the true signal $x^\star$, plotted as a function of iteration $t$. Initially, the algorithm converges toward the true signal of Newton, but after several iterations, it begins to diverge, ultimately yielding a noisy reconstruction.
    \textbf{(b)} The EM update consists of two components: the aligned signal, $\hat{x}_{\mathrm{signal}}^{(t+1)}$, and the aligned noise, $\hat{x}_{\mathrm{noise}}^{(t+1)}$. In the Ghost of Newton regime, the aligned noise, despite originating from pure noise, resembles a noisy version of the Newton image. This misleading signal steers the algorithm away from the ground truth.     
    \textbf{(c)} Visualization of the aligned signal $\hat{x}_{\mathrm{signal}}^{(t+1)}$, aligned noise $\hat{x}_{\mathrm{noise}}^{(t+1)}$, and the total estimate $\hat{x}^{(t+1)}$ at representative iterations. In the early stages, the aligned noise and the estimate resemble the initial template of Einstein. As the iterations progress, the Newton image gradually emerges. However, after approximately 50 iterations, the influence of the aligned noise dominates, steering the reconstruction away from Newton and resulting in a noisy and degraded estimate.
    The simulation corresponds to a fixed signal-to-noise ratio of $\mathrm{SNR} = 5 \times 10^{-3}$, with an image size of $d = 64 \times 64$ and $n = 2 \times 10^4$ observations.
    }   
    \label{fig:4} 
\end{figure*}

\paragraph{Finite-sample artifact.}
The empirical EM map $M_n$ is a small random perturbation of its deterministic population counterpart $M$.  
While population EM converges to $x^\star$ whenever initialized in its basin of attraction, the perturbed map $M_n$ converges instead to a sample-dependent fixed point $\hat{x}_n^\star= M_n(\hat{x}_n^\star; \mathcal{Y})$ that maximizes the empirical log-likelihood.  
Thus, if $\hat{x}^{(0)}$ lies in the population basin, the empirical iterates initially track the contracting population dynamics, but eventually peel away as the small per-iteration perturbations accumulate, drifting toward~$\hat{x}_n^\star$.

Formally, let $\hat{x}^{(t+1)}=M_n(\hat{x}^{(t)};\mathcal{Y})$ denote the empirical EM iterates.  
Under the tracking guarantee of Theorem~\ref{thm:finite-sample-tracking-groundtruth}, the group-invariant error satisfies, for all $t\ge0$,
\begin{align}
    \|\hat{x}^{(t)}-x^\star\|
     \le 
    \kappa^t \|\hat{x}^{(0)}-x^\star\|
     + 
    \frac{1-\kappa^t}{1-\kappa}\,\varepsilon_M(n,\delta),
    \label{eq:ghost-to-truth}
\end{align}
where the first term captures population contraction and the second term reflects the amplified finite-sample discrepancy between $M_n$ and $M$.  
Consequently, the Ghost-of-Newton regime begins with the population term dominating, but after sufficiently many iterations the finite-sample term overtakes it, producing a rebound away from the truth.

Let $e_0\triangleq\|\hat{x}^{(0)}-x^\star\|$ and define $t_{\mathrm{eq}}$ as the iteration index where the two terms in~\eqref{eq:ghost-to-truth} balance:
\begin{align}
    \kappa^{t_{\mathrm{eq}}} e_0
     = \frac{1-\kappa^{t_{\mathrm{eq}}}}{1-\kappa}\,\varepsilon_M(n,\delta). \label{eqn:two-terms-balance}
\end{align}
Solving~\eqref{eqn:two-terms-balance} yields
\begin{align}
    \kappa^{t_{\mathrm{eq}}}
     = \frac{\frac{\varepsilon_M}{1-\kappa}}{e_0+\frac{\varepsilon_M}{1-\kappa}}
     = \frac{1}{1+\frac{(1-\kappa)e_0}{\varepsilon_M}}
     \quad \longrightarrow
    \quad
    t_{\mathrm{eq}}
     = \frac{\log \p{1+\frac{(1-\kappa)e_0}{\varepsilon_M}}}{-\log(\kappa)} .
    \label{eq:teq-exact-nolin}
\end{align}
This expression interpolates between two regimes:
\begin{align}\label{eqn:EM-GoN-regimes}
    t_{\mathrm{eq}}
     = 
    \begin{cases}
        \dfrac{e_0}{\varepsilon_M(n,\delta)}\,[1+o(1)],
        & \text{if }\dfrac{(1-\kappa)e_0}{\varepsilon_M}\ll 1,\\[0.8em]
        \dfrac{1}{1-\kappa}\,
        \log \p{\dfrac{(1-\kappa)e_0}{\varepsilon_M}}\,[1+o(1)],
        & \text{if }\dfrac{(1-\kappa)e_0}{\varepsilon_M}\gg 1.
    \end{cases}
\end{align}

Since $\varepsilon_M(n,\delta)\asymp n^{-1/2}$ by~\eqref{eq:MRA-uniform-rate}, increasing the number of observations simultaneously delays the rebound (through a larger $t_{\mathrm{eq}}$ in~\eqref{eq:teq-exact-nolin}) and lowers the eventual error floor $\varepsilon_\infty\asymp \varepsilon_M/(1-\kappa)$.  
Conversely, decreasing the SNR increases $\kappa$ toward $1$, weakening population contraction, reducing $-\log(\kappa)$, and thereby advancing the rebound while inflating the floor by the factor $(1-\kappa)^{-1}$.

\paragraph{Empirical demonstration.}
Figures~\ref{fig:5} demonstrate that the Ghost of Newton behavior is governed by the number of observations $n$ and the SNR.
Equation~\eqref{eqn:EM-GoN-regimes} predicts that, in the finite-sample-dominated regime, the crossover time satisfies $t_{\mathrm{eq}}\propto \sqrt{n}$, in agreement with the empirical scalings in Figures~\ref{fig:5}(a-b).
A striking additional signature is that the iteration at which the MSE begins to increase coincides with a sharp rise in the empirical log-likelihood; this alignment persists across SNR levels and sample sizes.
Crucially, although EM monotonically increases the log-likelihood at every step~\cite{little2019statistical}, likelihood ascent alone does not guarantee monotone improvement in group-invariant MSE relative to the ground truth.

\begin{figure*}[!t]
    \centering
    \includegraphics[width=0.9 \linewidth]{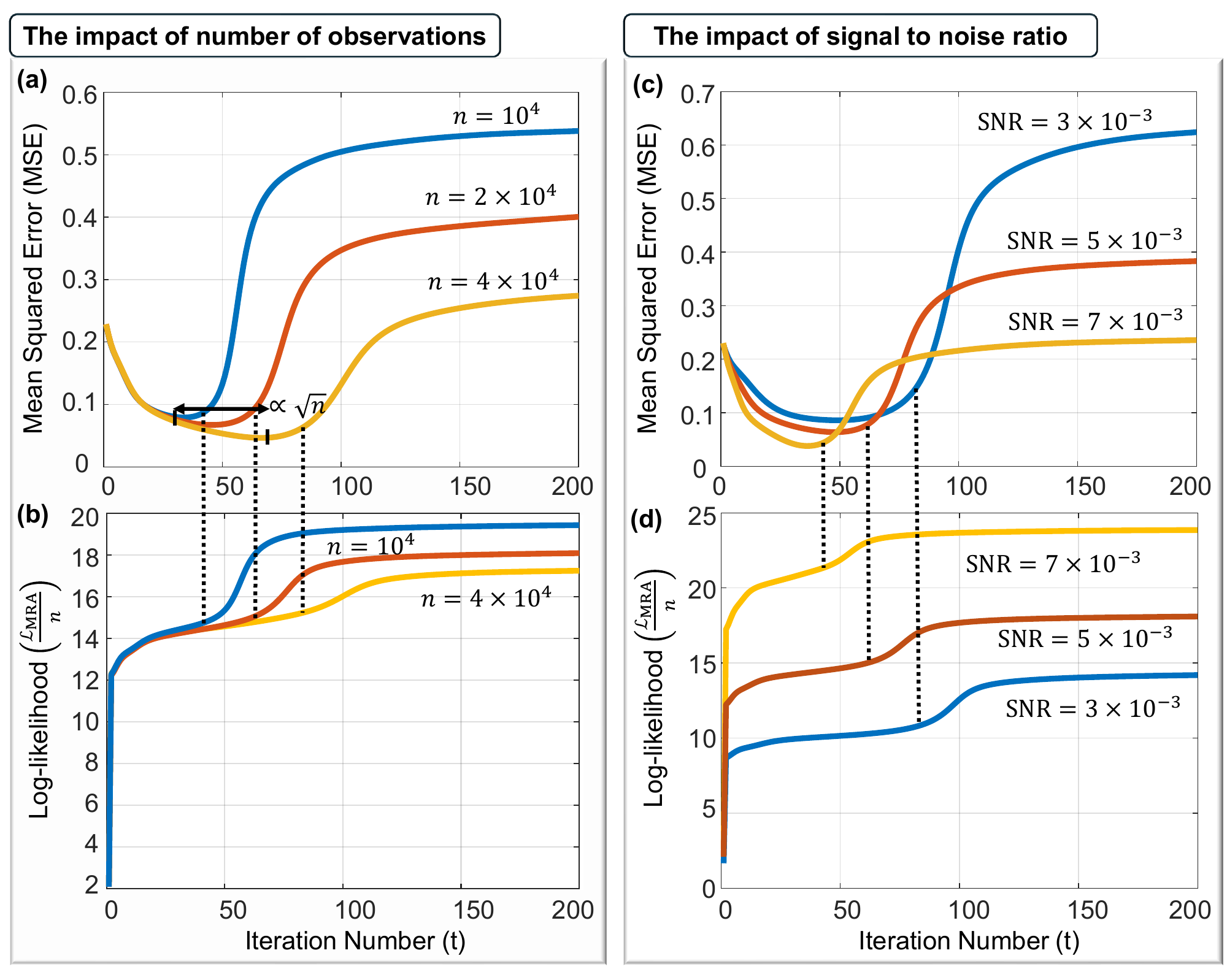}
    \caption{\textbf{The Impact of the number of observations and SNR on the Ghost of Newton phenomenon.}
    \textbf{(a, b)} Effect of increasing the number of observations at a fixed SNR. Parameters: image size $d = 64 \times 64$, SNR $= 5 \times 10^{-3}$. A larger number of observations reduces the impact of noise and mitigates the Ghost of Newton effect.
    \textbf{(c, d)} Effect of increasing the SNR at a fixed number of observations. Parameters: image size $d = 64 \times 64$, $n = 2 \times 10^4$ observations. Higher SNR leads to better alignment and convergence to the true signal.
    The $\text{MSE}$ plotted in both cases is defined in~\eqref{eqn:mseDef}, while the log-likelihood corresponds to~\eqref{eqn:logLikelihoodMRA}, with the constant final term omitted as it does not depend on the iteration number.
    In both scenarios, a sharp increase in the log-likelihood ($\mathcal{L}_{\s{MRA}}$) is observed at the region where the estimator begins to diverge from Newton’s image---this transition is indicated by dashed lines. Notably, this pattern consistently appears across different numbers of observations and SNR levels.   }
    \label{fig:5} 
\end{figure*}

\subsection{Sample complexity threshold for EM} \label{sec:finite-sample-necessary}

Next, we turn to \emph{necessary} sample-size requirements for EM in MRA.
Up to this point, we have derived \emph{sufficient} conditions under which the empirical EM operator $M_n$ closely tracks its population counterpart $M$, and hence inherits its local contraction behavior.
Here we prove a converse statement: under mild regularity assumptions (listed in Appendix~\ref{sec:app-sample-complexity} and verified for the MRA model), finite-sample fluctuations prevent any empirical fixed point inside the local basin from approaching $x^\star$ beyond the intrinsic statistical noise floor.

\begin{thm}[Necessary sample complexity of EM]
\label{thm:tight-1d-branch}
Consider the MRA model in the low-$\mathrm{SNR}$ regime $x^\star=\beta v$, $\|v\|_2=1$, with fixed $\sigma^2>0$ and varying $0<\beta\ll\sigma$.
Fix any target absolute accuracy $\varepsilon_{\mathrm{abs}}\in(0,R_0]$, where $R_0>0$ is chosen so that the Euclidean ball $\mathcal{B}(x^\star,R_0)$ lies within the local neighborhood $\mathcal U$ guaranteed by Theorem~\ref{thm:EM-MRA-spectral-radius}. 
Then, there exists a constant $C_{\mathrm{nec}}>0$ (depending only on the normalized spectrum of $v$ and on $d$, but not on $n,\beta,\sigma$) such that, with constant positive probability, any empirical fixed point $\hat{x}_n^\star = M_n(\hat{x}_n^\star; \, \mathcal{Y})$ lying in a sufficiently small basin of $x^\star$ and satisfying
\begin{align}
    \frac{\|\hat{x}_n^\star-x^\star\|_2}{\|x^\star\|_2} \leq \varepsilon_{\mathrm{abs}}
\end{align}
must obey the necessary sample-size lower bound
\begin{align}\label{eq:necessary-sample-1d}
    n  \geq  C_{\mathrm{nec}}\,\frac{\sigma^{6}}{\beta^{6}}\,\frac{1}{\varepsilon_{\mathrm{abs}}^{2}}.
\end{align}
\end{thm}

The proof of Theorem~\ref{thm:tight-1d-branch} appears in Appendix~\ref{sec:proof-of-sample-complexity-mra}.
The necessary condition
\eqref{eq:necessary-sample-1d} can be rewritten as
\begin{align}
    n \gtrsim  \mathrm{SNR}^{-3}\,\varepsilon_{\mathrm{abs}}^{-2}.
\end{align}
Thus, even within the best-case local basin where the population EM map is contractive,
finite-sample fluctuations impose an intrinsic statistical floor: achieving relative accuracy
$\varepsilon_{\mathrm{abs}}$ requires on the order of $\mathrm{SNR}^{-3}$ samples.
In particular, Theorem~\ref{thm:tight-1d-branch} matches the information-theoretic
$\mathrm{SNR}^{-3}$ scaling~\cite{bandeira2023estimation}.

\subsection{Finite-sample Einstein from Noise: phase drift}
\label{sec:finite-sample-EfN}

We now turn to the finite-sample Einstein from Noise regime, focusing on the drift of Fourier phases induced by finite-sample fluctuations.
Recall the Fourier-domain notation introduced in~\eqref{FourierSpace}: we write $\{\hat{\s{X}}^{(t)}[k]\}_{k=0}^{d-1}$ for DFT coefficients of the \emph{empirical} EM iterate $\hat{x}^{(t)}$, and $\{\hat{\s{N}}_i[k]\}_{k=0}^{d-1}$ for the DFT coefficients of the noise vector $\xi_i$, for $i = \{1,2, \ldots n\}$.
We further use the polar representation $\hat{\s{X}}^{(t)}[k]=|\hat{\s{X}}^{(t)}[k]|\,e^{i\phi_{\hat{\s{X}}^{(t)}}[k]}$ to denote Fourier magnitudes and phases.
Throughout this subsection, we use the term MSE to refer to the initialization-discrepancy error, namely, the expected squared difference between the Fourier phases of the iterate and those of the initialization. 
Specifically, for a frequency $k$ and iteration $t$, we define 
\begin{align}
    \mathrm{MSE}_{\mathrm{phase}}^{(t)}(k) \triangleq \mathbb{E} \pp{|\phi_{\s{\hat{X}}^{(t)}}[k]-\phi_{\s{\hat{X}}^{(0)}}[k]|^{2}}. \label{eqn:fourier-phases-MSE}
\end{align}

\paragraph{Fourier-phases drift: Single iteration.}
Figure~\ref{fig:6} illustrates the finite-sample bias after a single Einstein from Noise step, comparing the hard-assignment update (Algorithm~\ref{alg:generalizedEfNhard}) and the EM update (Algorithm~\ref{alg:generalizedEfNsoft}). 
The upper panels show the appearance of structural similarity to initialization as the number of observations $n$ increases. 
The lower panels quantify this effect through the MSE of the Fourier phases~\eqref{eqn:fourier-phases-MSE} between successive iterates. 
Empirically, the phase MSE decays inversely proportional to the number of observations at a rate $n^{-1}$ for both algorithms, with stronger Fourier modes converging faster. 
Notably, EM exhibits a larger bias toward the template than hard assignment, as reflected both visually and in the phase-MSE curves.

\begin{figure*}[!t]
    \centering
    \includegraphics[width=0.95 \linewidth]{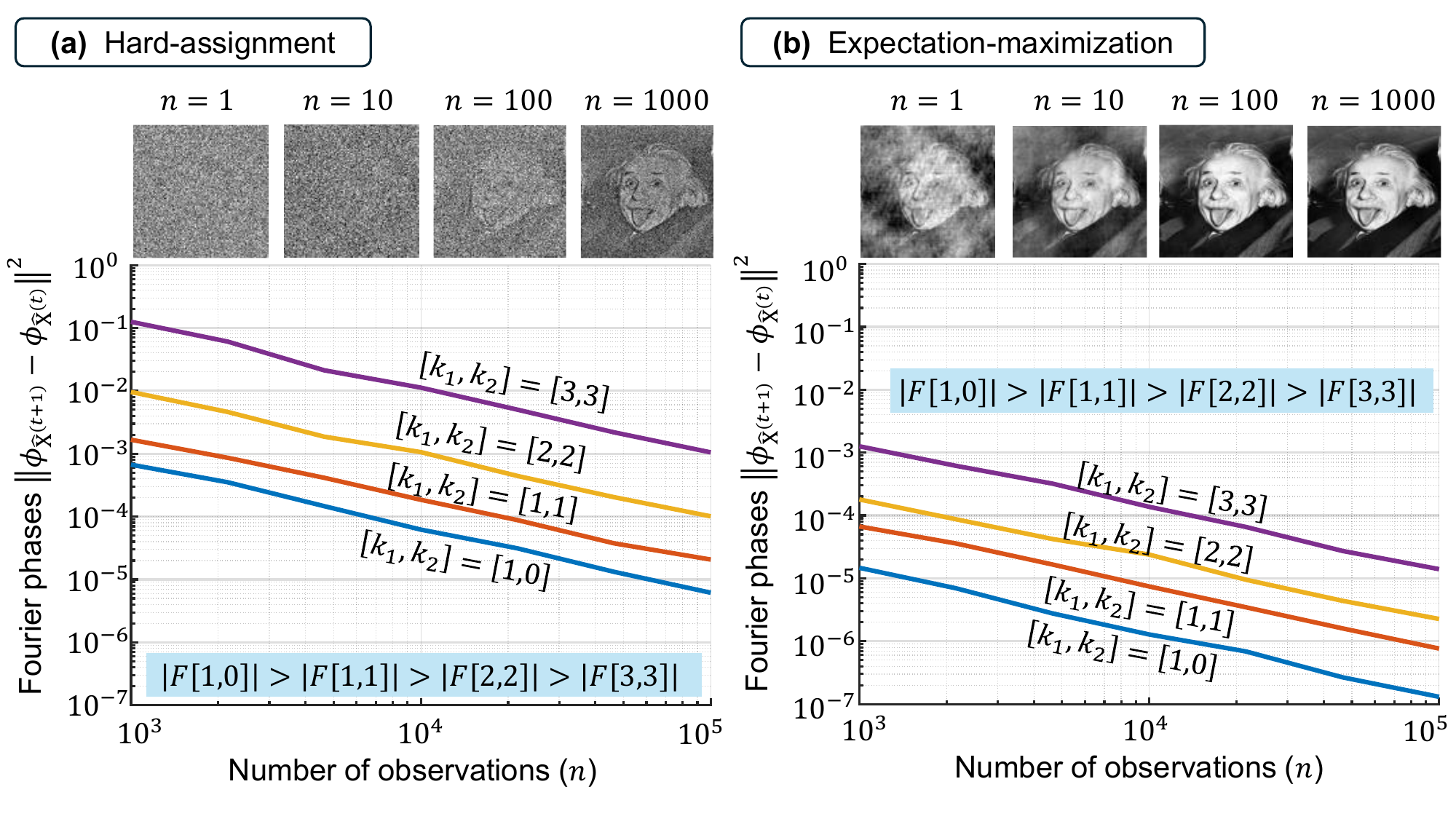}
    \caption{\textbf{Single Einstein from Noise iteration at finite sample size.}  
    \textbf{(a)} Hard assignment (Algorithm~\ref{alg:generalizedEfNhard}). 
    \textbf{(b)} Expectation-maximization (Algorithm~\ref{alg:generalizedEfNsoft}). 
    Top: Structural similarity between the Einstein from Noise estimate and the initialization as a function of $n$. 
    Bottom: MSE between the Fourier phases of the estimate and those of the template, shown versus $n$ across frequencies. 
    In both algorithms the phase MSE scales as $\asymp n^{-1}$, and dominant spectral components exhibit smaller errors (that is, they are more biased). 
    EM is more strongly biased toward the template, producing a visible Einstein reconstruction after only a few observations. 
    Both panels (a)--(b) averages $300$ Monte Carlo trials with an Einstein template of size $d=32\times 32$.}
    \label{fig:6} 
\end{figure*}

These trends are consistent with the finite-sample guaranties developed in the following for EM; the corresponding single-iteration analysis for hard assignment appears in~\cite{balanov2024einstein}.
\begin{proposition}[Finite-sample one-step Fourier phase drift]
\label{prop:finite_sample_one_step}
Fix a finite number of iterations $T \in \mathbb{N}$. Let $\{\hat{x}^{(t)}\}_{t=0}^T$ be the empirical EM iterates defined in~\eqref{eq:EM-iters}, and write ${\s{\hat{X}}}^{(t)}[k]$ for the DFT coefficients of $\hat{x}^{(t)}$ as in~\eqref{FourierSpace}, with polar representation $\s{\hat{X}}^{(t)}[k]=|\s{\hat{X}}^{(t)}[k]|e^{i\phi_{\hat{\s X}^{(t)}}[k]}$.
Assume the regularity conditions of Theorem~\ref{thm:1}.
Then, for every frequency $1\le k\le d-1$ with $|\hat{\s{X}}^{(t)}[k]|>0$, there exists a finite constant $C_k \big(\hat{x}^{(t)}\big)<\infty$ such that
\begin{align}\label{eq:phase_mse_conditional_limit}
   \lim_{n\to\infty} n\,\E\left[\,\Big|\phi_{\hat{\s X}^{(t+1)}}[k]-\phi_{\hat{\s X}^{(t)}}[k]\Big|^2 \,\Big|\, \hat{x}^{(t)}\right]\;=\; C_k(\hat{x}^{(t)}).
\end{align}
The explicit expression for $C_k \big(\hat{x}^{(t)}\big)$ is provided in Remark~\ref{remark:constant-Ck}.
\end{proposition}

Proposition~\ref{prop:finite_sample_one_step}, proved in Appendix~\ref{sec:finite_sample_one_step_proof}, characterizes the one-step finite-sample phase drift at frequency $k$.
In particular, conditional on the current iterate $\hat{x}^{(t)}$, the per-iteration phase mean-squared error is of order $1/n$, with leading term $C_k(\hat{x}^{(t)})/n$.
Next, we specialize this constant to the low-magnitude scaling regime of Theorem~\ref{thm:low-mag-rate} and derive an explicit approximation for $C_k(\hat{x}^{(t)})$ in that setting.

\begin{proposition}[Low-magnitude scaling of $C_k(\hat{x}^{(t)})$] \label{prop:low-magnitude-scaling-of-phases-convergence-rate}
Fix a finite iteration $T \in \mathbb{N}$, a confidence level $\delta\in(0,1)$, and an index $t\in\{0,\dots,T-1\}$.
Let $\{\hat{x}^{(t)}\}_{t=0}^T$ be the empirical EM iterates defined in~\eqref{eq:EM-iters}.
Assume in addition that $\hat{x}^{(t)}$ lies in the low-magnitude regime of Theorem~\ref{thm:low-mag-rate}. Then, 
\begin{align}\label{eq:Ck_lowmag_scaling}
    C_k(\hat{x}^{(t)})=\Theta(\|\hat{x}^{(t)} \|^2).
\end{align}
Moreover, if the sample size is large enough so that $n \gtrsim \|\hat{x}^{(t)}\|^{-6} \log(1/\delta)$, then with probability at least $1-\delta$, 
\begin{align}\label{eq:slow_variation}
    \big|C_k(\hat{x}^{(t+1)})-C_k(\hat{x}^{(t)})\big| \le c\,\|\hat{x}^{(t)}\|^4,
\end{align}
where the probability is with respect to the randomness in forming the next empirical update $\hat{x}^{(t+1)}$, given the current iterate $\hat{x}^{(t)}$, and $c$ is a constant independent of $\|\hat{x}^{(t)}\|$.
\end{proposition}

Proposition~\ref{prop:low-magnitude-scaling-of-phases-convergence-rate}, proved in Appendix~\ref{sec:proof-of-low-magnitude-scaling-of-phases}, describes the behavior of the one-step phase-drift constant in the low-magnitude regime. When $\|\hat{x}^{(t)}\|$ is small, it yields the envelope estimate $C_k(\hat{x}^{(t)})=\Theta(\|\hat{x}^{(t)}\|^2)$. In particular, $\mathbb{E}|\Delta\phi^{(t)}[k]|^2=\Theta(\|\hat{x}^{(t)} \|^2/n)$.
Moreover, the slow-variation bound in~\eqref{eq:slow_variation} implies that the constants $C_k(\hat{x}^{(t)})$ change only gradually across iterations; that is, the phase-drift rate remains essentially stable over a moderate number of steps while the iterate magnitude stays small.

\paragraph{Fourier-phases drift: Multi-iteration.}
In the population limit, Proposition~\ref{thm:3.5} shows that the EM phases remain identical to those of the initialization along the entire trajectory. 
However, at  finite $n$, each iteration introduces a small random phase drift, which accumulates over time. 
Figure~\ref{fig:7} demonstrates this multi-iteration behavior. 
Panel~(a) plots the Pearson cross-correlation (PCC) between the estimate at iteration $t$ and the initialization; the correlation increases with $n$, indicating a strengthening bias toward the template as more observations are available. 
Panel~(b) shows the Fourier phase MSE~\eqref{eqn:fourier-phases-MSE} between the estimate and the template, which empirically scales as $\propto t^2/n$ over the transient regime. 

\begin{figure*}[!t]
    \centering
    \includegraphics[width=0.95 \linewidth]{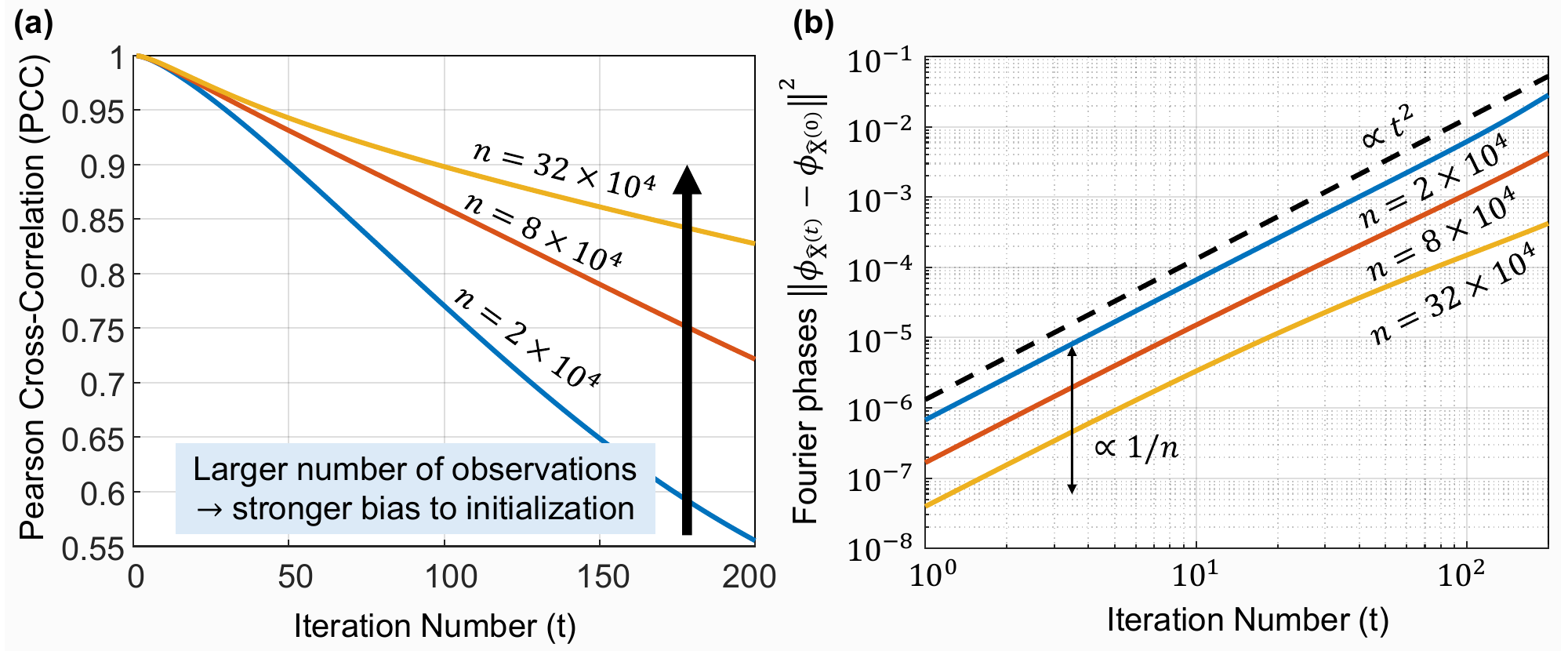}
    \caption{\textbf{Finite-sample Einstein from Noise over multiple EM iterations.}
    \textbf{(a)} PCC between $ \hat{x}^{(t)}$ and the initialization $x_{\mathrm{template}}$ versus iteration $t$ for several sample sizes $n$. Larger $n$ yields stronger and more persistent bias toward the template. 
    \textbf{(b)} Fourier-phase MSE between $ \hat{x}^{(t)}$ and $x_{\mathrm{template}}$ versus $t$, illustrating approximate scaling $\propto t^2/n$. 
    Simulation settings: $d=16\times 16$; $16$ Monte Carlo trials per curve. Panel~(b) displays the phase MSE for the $[1,2]$ Fourier component.}
    \label{fig:7} 
\end{figure*}

We now formalize the accumulation of finite-sample phase errors across iterations, which is proved in Appendix~\ref{sec:proofOfPhaseErrorAfterTiterations}.

\begin{proposition}[Accumulated Fourier phase error across EM iterations]
\label{thm:lowerBoundAfterTiterations_rewrite}
Fix $d\ge 2$ and let $\{y_i\}_{i=1}^n$ follow the model~\eqref{eqn:calPModel}. 
Let $ \hat{x}^{(t)}$ be the EM iterates from Algorithm~\ref{alg:generalizedEfNsoft}, with DFT $\hat{\s{X}}^{(t)}[k]$ and initialization $\hat{\s{X}}^{(0)}[k]=\s{X}_{\mathrm{template}}[k]$. 
Assume that for a given frequency $k$ and
all $t\in\{0,\dots,T-1\}$, 
\begin{align}
    \mathbb{E}| \, \phi_{\hat{\s{X}}^{(t+1)}}[k]-\phi_{\hat{\s{X}}^{(t)}}[k]|^2 \le \epsilon^{(t)}. \label{eqn:eps_assump_rewrite}
\end{align}
Then:
\begin{enumerate}
    \item \emph{(General bound.)} For any $T\ge 1$,
    \begin{align}
        \mathbb{E}\,\big|\phi_{\hat{\s{X}}^{(T)}}[k]-\phi_{\hat{\s{X}}^{(0)}}[k]\big|^2 \le \p{\sum_{t=0}^{T-1}\sqrt{\epsilon^{(t)}}}^2 .        \label{eqn:lowerBoundMultipleIterations_rewrite}
    \end{align}

    \item \emph{(Slow variation.)}
    If $\epsilon^{(t)}$ varies slowly in the sense that
    \begin{align}
        \epsilon^{(t)}=\epsilon^{(0)}+\alpha t + O((\alpha t)^2),
        \qquad \alpha t\ll 1,        \label{eqn:slowly_varying_envelope_rewrite}
    \end{align}
    then
    \begin{align}
        \mathbb{E}\,\big|\phi_{\hat{\s{X}}^{(T)}}[k]-\phi_{\hat{\s{X}}^{(0)}}[k]\big|^2
         \le 
        T^2\,\epsilon^{(0)} + O(\alpha^2T^3).
        \label{eqn:T2eps_rewrite}
    \end{align}
\end{enumerate}
\end{proposition}

In particular, in the regime of Propositions~\ref{prop:finite_sample_one_step}-\ref{prop:low-magnitude-scaling-of-phases-convergence-rate}, the one-step phase error satisfies the slowly varying condition~\eqref{eqn:slowly_varying_envelope_rewrite} with
\begin{align}
    \epsilon^{(t)} = \frac{C_k(\hat{x}^{(0)})}{n} \Big(1+O(t \,  \| \hat{x}^{(t)}\|^2)\Big) \qquad\text{for } t \ll \| \hat{x}^{(t)}\|^{-2},
    \label{eq:eps-slow-envelope}
\end{align}
Hence, as long as $T \ll \| \hat{x}^{(T)}\|^{-2}$, the cumulative phase MSE scales as $\propto T^2/n$, matching Figure~\ref{fig:7}(b). 
This provides a quantitative mechanism for the persistence of initialization bias in large low-SNR datasets.

As the Fourier-phase MSE scales as $1/n$, Appendix~\ref{sec:adam} explores mitigation strategies based on mini-batching, where each iteration uses only a smaller subset of $n$ observations. Our results show that such mini-batching schemes reduce sensitivity to initialization and improve computational efficiency, while achieving reconstruction accuracy comparable to full-batch EM.

\section{Discussion and outlook}
\label{sec:discussion}

\subsection{Extension to additional group actions}

Although our analysis focuses on the classical MRA model~\eqref{eqn:mainModel1D} with a cyclic-shift group action, the structural features revealed by our study of the EM dynamics suggest that similar phenomena may extend to more general latent group actions.
In a general MRA formulation, one considers a compact group $G$ acting continuously and orthogonally on a finite-dimensional real vector space $V$, with an unknown signal $x\in V$ 
and observations
\begin{align}
    \text{(General MRA model)} \qquad
    y_i = g_i \cdot x + \xi_i,\quad g_i \in G, \label{eqn:mainModel}
\end{align}
for $i\in\{0,\dots,n-1\}$, where $\xi_i$ are i.i.d. noise terms and $(g\cdot x)(w)\triangleq x(g^{-1}w)$.
In this setting, the natural goal is to recover the $G$-orbit of $x$ from the noisy samples $\{y_i\}_{i=1}^n$.

It is useful to distinguish which aspects of our analysis rely on the special structure of cyclic shifts, namely that the acting group is discrete and abelian, and which reflect more general features of EM under group actions. 
The finite-sample bounds and the Ghost of Newton phenomenon rely primarily on the unitarity of the group action and on the resulting covariance structure of the per-sample EM update.
We therefore expect these aspects of the theory to extend, with suitable technical modifications, to general MRA models in which $G$ acts unitarily on $V$.
By contrast, the explicit second-order expansion of the Jacobian and its block-diagonal Fourier structure make essential use of the commutative representation theory of the cyclic group.
Extending these spectral calculations to more general compact groups, especially non-abelian groups, such as $\mathrm{SO}(n)$ for $n>2$, would both broaden the range of applications and clarify which geometric features of the group control the local contraction properties of EM.

We also conjecture that the two-phase convergence behavior established in Section~\ref{sec:fundemental-properties} is not specific to the cyclic-shift MRA model~\eqref{eqn:mainModel1D}. Heuristically, the initial fast phase reflects contraction along magnitude-like directions that are identifiable from low-order invariants, whereas the slow tail is governed by nearly flat directions corresponding to phase-like degrees of freedom.
Making this intuition precise for general compact groups would require identifying the appropriate analog of Fourier magnitudes and phases, which is closely related to the decomposition of $V$ into irreducible representations, and tracking EM on each component.
For $\mathrm{SO}(n)$, for example, natural analogs of these quantities are the second-order invariants (Gram matrices of the spherical harmonic coefficients~\cite{bendory2024sample}) and the third-order invariants (bispectrum)~\cite{bendory2025orbit}, which respectively play the roles of magnitude-and phase-type information.
Such a theory could lead to a general iteration-complexity bounds and explain why certain representation components dominate the asymptotic behavior. 

Our analysis of the Einstein from Noise phenomenon in the cyclic-shift MRA model suggests a broader mechanism of algorithmic confirmation bias in latent group-action models, which is also supported by classical results in the cryo-EM literature~\cite{henderson2013avoiding}. In the noise-only setting, where no signal present, the population EM dynamics drive the Fourier magnitudes toward zero while leaving the Fourier phases equal to those of the initialization, so the iterate retains a persistent imprint of the initial template. More generally, we conjecture that for other group-action models there exist analogous nearly invariant directions
that remain stable over many iterations and likewise preserve memory of the initialization. Characterizing these invariant directions for general compact groups, and relating them to the hierarchy of invariant features required to identify the orbit, may provide a general principle for predicting when and how EM exhibits confirmation bias toward its starting point.

\subsection{Mitigation strategies} 

\paragraph{Mitigation strategies of the Ghost of Newton phenomenon.}
Addressing the Ghost of Newton phenomenon remains a challenging open problem with practical implications. One potential solution is to apply early termination to the EM algorithm. However, this approach is difficult to implement reliably because the optimal stopping point depends on both the SNR and the number of observations, making it challenging to generalize across different scenarios. 
One possible approach is to apply change point detection methods~\cite{aminikhanghahi2017survey} in combination with gap statistics~\cite{tibshirani2001estimating}. 

\paragraph{Beyond EM: second-order and acceleration schemes.}
A natural question is whether second-order or acceleration techniques can mitigate the slow-convergence bottleneck we identify for EM in the low-SNR regime, both analytically and numerically. There are many plausible variants to explore, including Newton--Raphson and quasi-Newton methods~\cite{nocedal2006numerical, nelder1972generalized}, as well as accelerated EM schemes based on extrapolation or fixed-point acceleration (e.g., Anderson-type)~\cite{walker2011anderson,zhou2011quasi}, and momentum-based first-order methods (e.g., Nesterov acceleration)~\cite{nesterov1983method}. A systematic study of these alternatives, however, is beyond the scope of the present paper and merits dedicated investigation. That said, our analysis suggests that the bottleneck is rooted in the local geometry of the objective EM map in low SNR (e.g., near-flat directions of the Jacobian), and thus we expect that purely first-order schemes may continue to exhibit a similar slowdown.

\subsection{Comparison with the method of moments} 
Our discouraging findings for EM in the low-SNR regime suggest considering alternative estimation approaches, such as the method of moments. 
A key distinction is that several limitations of EM persist even at the population level and therefore do not vanish with more observations; for instance, the slow convergence phenomenon in low SNR holds regardless of the sample size.
In contrast, method of moments can leverage additional data directly: by collecting more samples one can estimate the relevant moments with increasing accuracy, and in the limit obtain essentially exact moment estimates.

The main drawback of the method of moments is computational rather than statistical, namely the challenge of inverting the (typically nonlinear) moment equations to recover the underlying signal.
For the MRA models considered in this paper, however, provable inversion algorithms are available; see, e.g.,~\cite{bendory2017bispectrum,perry2019sample}, and further extensions to more general group actions such as $\mathrm{SO}(3)$ and $\mathrm{SE}(3)$ in~\cite{bendory2025orbit,balanov2025orbit}.

\section*{Data Availability}
The detailed implementation and code are available at \href{https://github.com/AmnonBa/EM-MRA-low-SNR}{https://github.com/AmnonBa/EM-MRA-low-SNR}.

\section*{Acknowledgment}
T.B. is supported in part by BSF under Grant 2020159, in part by NSF-BSF under Grant 2019752, in part by ISF under Grant 1924/21, and in part by a grant from The Center for AI and Data Science at Tel Aviv University (TAD).
W.H. is supported by ISF under Grant 1734/21. We thank Amit Singer for pointing out possible extensions to second-order methods, as discussed in Section~\ref{sec:discussion}.

\bibliographystyle{plain}

\begin{appendices}

{\centering{\section*{Appendix}}}

\paragraph{Appendix organization.}
This appendix is organized as follows. Appendix~\ref{sub:Preliminaries} provides foundational background, including a derivation of the likelihood function for the MRA model, the EM update rule, and a brief overview of prior theoretical work relevant to our setting. Appendix~\ref{subsec:dynamical-preliminaries} introduces the nonlinear dynamical systems framework and collects general properties that are later applied to the EM population map; these are used in Appendix~\ref{sec:proofOfEmSpectralRadius}, where we establish the local spectral properties and convergence behavior of the population EM operator. In Appendix~\ref{sec:proofOfJacobStructure} we derive the population EM Jacobian in the low-SNR regime and prove the associated statements, including the two-phase convergence behavior and the tail iteration complexity.
Appendix~\ref{sec:appendix-bias-to-initalization} provides the proof of the bias-to-initialization property (Theorem~\ref{thm:mra-lowSNR-iteration-bias-to-init}).

Appendix~\ref{sec:preliminaries-finite-sample} presents the preliminaries for the finite-sample analysis, relying on tools from empirical process theory, and Appendix~\ref{sec:finite-sample-analysis-results} contains the corresponding finite-sample convergence results. Appendix~\ref{sec:highSNRasymptotics} focuses on the high-SNR regime, while Appendix~\ref{sec:lowSNRapproxMain} provides approximation formulas and expansions tailored to the low-SNR regime.

Finally, Appendix~\ref{sec:EfNpreliminaries} introduces the framework and auxiliary lemmas for the Einstein from Noise phenomenon. Appendix~\ref{sec:finite-dimentional-EfN} contains the proofs for the finite-dimensional Einstein from Noise results, Appendix~\ref{sec:proofOfHighDimensionEfN} addresses the high-dimensional case, and Appendix~\ref{sec:finite-sample-EfN-Main} establishes the corresponding finite-sample guarantees. Appendix~\ref{sec:adam} concludes by presenting a mitigation strategy for Einstein from Noise based on mini-batch EM and adaptive optimization.

\section{Preliminaries}\label{sub:Preliminaries}

\subsection{The log-likelihood function and the EM algorithm for MRA} \label{sec:softAssignmentUpdateStep}
\paragraph{The likelihood function.} 
Let $x$ denote the underlying signal in the MRA model in \eqref{eqn:mainModel1D}. The probability density function (pdf) of this model is given by,
\begin{align}
    f_{\s{MRA}}(y; x) = \frac{1}{(2\pi\sigma^2)^{d/2}} \sum_{\ell=0}^{d-1} p(\mathcal{T}_\ell) \cdot \exp\left(-\frac{\norm{y - \mathcal{T}_\ell  x}^2}{2\sigma^2} \right), \label{eqn:A0}
\end{align}
where $ p(\mathcal{T}_\ell) $ is the prior probability of the cyclic shift $ \mathcal{T}_\ell $. Recall that we assume a uniform distribution over the shifts, i.e., $ p(\mathcal{T}_\ell) = 1/d $. Accordingly, the likelihood function, given all observations $ \mathcal{Y} = \{y_i\}_{i=1}^{n} $, is given by,
\begin{align}
    \mathcal{L}_{\s{MRA}}(x; \mathcal{Y}) = \frac{1}{(2\pi\sigma^2)^{dM/2}} \prod_{i=1}^{n} \sum_{\ell=0}^{d-1} p(\mathcal{T}_\ell) \cdot \exp\left(-\frac{\norm{y_i - \mathcal{T}_\ell  x}^2}{2\sigma^2} \right), \label{eqn:A1}
\end{align}
and the corresponding log-likelihood is,
\begin{align}
    \log \mathcal{L}_{\s{MRA}}(x; \mathcal{Y}) 
    = \sum_{i=1}^{n} \log \pp{\sum_{\ell=0}^{d-1} p(\mathcal{T}_\ell) \cdot \exp\left(-\frac{\norm{y_i - \mathcal{T}_\ell  x}^2}{2\sigma^2} \right)} - \frac{dM}{2} \log(2\pi\sigma^2). \label{eqn:A2}
\end{align}
Thus, the maximum likelihood estimator (MLE) of $ x $ is,
\begin{align}
    \hat{x} = \argmax_{x\in\mathbb{R}^d} \log \ \mathcal{L}_{\s{MRA}}(x; \mathcal{Y}). \label{eqn:A3}
\end{align}
Since the norm is invariant under cyclic shifts, i.e., 
 \begin{align}
     \norm{\mathcal{T}_\ell  x} = \norm{x}, \label{eqn:A3_1}
 \end{align}
for every $\ell \in \pp{d}$, the MLE is equivalent to
\begin{align}
    \hat{x} = \argmax_{x\in\mathbb{R}^d} \left\{ \frac{1}{n} \sum_{i=1}^{n} \log \pp{\sum_{\ell=0}^{d-1} p(\mathcal{T}_\ell) \cdot \exp\left(\frac{y_i^\top (\mathcal{T}_\ell  x)}{\sigma^2} \right)} - \frac{\norm{x}^2}{2\sigma^2} \right\}. \label{eqn:A4}
\end{align}

\paragraph{The EM update rule.}
We use the EM algorithm to compute soft assignments. EM is widely used for Gaussian mixture models (GMMs) \cite{dempster1977maximum}, and has been adapted for the MRA setting \cite{sigworth1998maximum,scheres2012relion}. The EM update at iteration $t$ is given by,
\begin{align}
    \hat{x}^{(t+1)}
    = \arg\max_{x \in \mathbb{R}^d} \,   \mathbb{E}_{\mathcal{T}_{\ell_1},\dots,\mathcal{T}_{\ell_n} \sim p(\cdot \mid \mathcal{Y}, \hat{x}^{(t)})}
    \left[\log p \p{\mathcal{Y}, \{\mathcal{T}_{\ell_i}\}_{i=1}^{n} \mid x }\right].
    \label{eqn:A5-MRA}
\end{align}
The posterior distribution over shifts given the data and the current estimate is given by,
\begin{align}
    p(\mathcal{T}_\ell \vert y_i, \hat{x}^{(t)}) 
    &= \frac{\exp\left(-\frac{\norm{y_i - \mathcal{T}_\ell  \hat{x}^{(t)}}^2}{2\sigma^2} \right)}{\sum_{r=0}^{d-1} \exp\left(-\frac{\norm{y_i - \mathcal{T}_r  \hat{x}^{(t)}}^2}{2\sigma^2} \right)}= \frac{\exp\left(\frac{y_i^\top (\mathcal{T}_\ell  \hat{x}^{(t)})}{\sigma^2} \right)}{\sum_{r=0}^{d-1} \exp\left(\frac{y_i^\top (\mathcal{T}_r  \hat{x}^{(t)})}{\sigma^2} \right)}, \label{eqn:A6}
\end{align}
where the second equality  follows from \eqref{eqn:A3_1}. This expression corresponds to the softmax probability, $\gamma_\ell(\hat{x}^{(t)}, y_i)$, defined in \eqref{eqn:softmaxGamma}. Thus,
\begin{align}
    p(\mathcal{T}_\ell \vert y_i, \hat{x}^{(t)}) = \gamma_\ell(\hat{x}^{(t)}, y_i).
\end{align}

In the M–step, we maximize the expected complete–data log-likelihood with respect to $x^\star$.  
For a single observation $y_i$ and shift $\mathcal{T}_\ell$, the complete–data log-likelihood (up to an additive constant independent of $x$) is
\begin{align}
    \log p(y_i, \mathcal{T}_\ell \mid x)
    = -\frac{\|y_i - \mathcal{T}_\ell x\|_2^2}{2\sigma^2}.
    \label{eqn:A7}
\end{align}
Plugging \eqref{eqn:A6} and \eqref{eqn:A7} into the EM formulation \eqref{eqn:A5-MRA}, we obtain the following update,
\begin{align}
    \hat{x}^{(t+1)} 
    &= \argmin_{x} \sum_{i=1}^{n} \sum_{\ell=0}^{d-1} \frac{\norm{y_i - \mathcal{T}_\ell  x}^2}{2\sigma^2} \cdot \gamma_\ell(\hat{x}^{(t)}, y_i) \\
    &= \argmin_{x} \sum_{\ell=0}^{d-1} \sum_{i=1}^{n} \norm{y_i - \mathcal{T}_\ell  x}^2 \cdot \gamma_\ell(\hat{x}^{(t)}, y_i) \nonumber. \label{eqn:softAssignmentUpdateStepInNonExplicitForm}
\end{align}
Define the objective function,
\begin{align}
    \mathcal{H}(x,\calY) = \sum_{\ell=0}^{d-1} \sum_{i=1}^{n} \norm{y_i - \mathcal{T}_\ell  x}^2 \cdot \gamma_\ell(\hat{x}^{(t)}, y_i).
\end{align}
Then, the gradient with respect to $x$ is given by,
\begin{align}
    \frac{1}{2} \frac{\partial \mathcal{H}(x,\calY)}{\partial x} = \sum_{i=1}^{n} \left(x - \sum_{\ell=0}^{d-1} \gamma_\ell(\hat{x}^{(t)}, y_i) \cdot \p{\mathcal{T}_\ell^{-1}  y_i} \right).
\end{align}
Setting this derivative to zero yields the update rule,
\begin{align}
    \hat{x}^{(t+1)} = \frac{1}{n} \sum_{i=1}^{n} \sum_{\ell=0}^{d-1} \gamma_\ell(\hat{x}^{(t)}, y_i) \cdot \p{\mathcal{T}_\ell^{-1}   y_i},
\end{align}
as in \eqref{eqn:generalizedEfNEqnSoft}. We note that the derivation for the cyclic group actions $\mathbb{Z}_d$ presented here can be extended to the general MRA model in \eqref{eqn:mainModel} with a general compact group action $G$.

\subsection{The sample complexity of the MRA model}
\label{sec:MRAreconstructionRegimes}

The sample complexity of the MRA problem quantifies how the number of observations $n$ must scale with the noise variance $\sigma^2$ in order to reliably recover the underlying signal.  
Figure~\ref{fig:1}(b) summarizes the different reconstruction regimes as a function of the SNR.
In the high-SNR regime, direct observation alignment is feasible, enabling accurate alignment and averaging. The sample complexity in this setting is similar to that of incoherent integration, where the underlying group actions (e.g., shifts or rotations) are assumed to be known. 

As the SNR decreases, however, alignment becomes increasingly unreliable, and the sample complexity required for accurate recovery grows rapidly.  
In the asymptotic regime where both the number of samples $n \to \infty$ and the noise variance $\sigma^2 \to \infty$, recent theoretical analyses show that a necessary condition for unique recovery is
\begin{align}
    n = \omega(\sigma^{6}),
\end{align}
as established in~\cite{perry2019sample,bandeira2023estimation,bendory2017bispectrum}.  
The notation $n = \omega(\sigma^{6})$ means that $n / \sigma^{6} \to \infty$, implying that the number of observations must grow faster than $\sigma^{6}$ for consistent estimation of the signal.  
This scaling law characterizes an information-theoretic lower bound that applies to any estimator, regardless of computational constraints. In particular, the authors in~\cite{bendory2017bispectrum} provide an algorithm that uses this estimation rate, thus providing a matching upper bound.  
These results are derived for the maximum-likelihood objective in~\eqref{eqn:logLikelihoodMRA}, and need not apply to alternative formulations such as the hard-assignment objective in~\eqref{eqn:mraHardTargetFunction}.


\section{Preliminaries: Nonlinear dynamical systems}
\label{subsec:dynamical-preliminaries}

To analyze the local convergence properties of the population EM iteration in the MRA problem, we first provide several basic facts from the theory of nonlinear dynamical systems near a fixed point. Throughout this appendix we work in a general deterministic setting: all results apply to an arbitrary nonlinear fixed-point iteration with a symmetric Jacobian, and in particular to the population EM operator considered below.

\subsection{Local linearization and contraction}

Let $M:\mathbb{R}^d\to\mathbb{R}^d$ be $C^2$ in a neighborhood of a fixed point $x^\star$, with $M(x^\star)=x^\star$. 
We consider the associated fixed-point iteration
$\hat{x}^{(t+1)} = M(\hat{x}^{(t)})$, for $t=0,1,\dots$, and define the error sequence
\begin{align}
    e_t \triangleq \hat{x}^{(t)} - x^\star,\qquad t\ge 0.
\end{align}
Thus, by the definition of $e_t$,
\begin{align}\label{eq:error-recursion}
    e_{t+1} = M(x^\star + e_t) - x^\star.
\end{align}
In addition, by Taylor's theorem with integral remainder, there exist constants $c>0$ and $r_0>0$ such that the errors satisfy
\begin{align}\label{eq:nonlinear-perturb-flow}
    e_{t+1} = J e_t + R(e_t),
\end{align}
where $J \triangleq J_M(x^\star)$ denotes the Jacobian of $M$ in $x^\star$ (see, Definition~\ref{def:Jacobian-spectral-radius}), and the remainder $R:\mathbb{R}^d\to\mathbb{R}^d$ satisfies the quadratic bound
\begin{align}\label{eq:remainder-quadratic}
    \|R(e)\|_2 \le  c \|e\|_2^2
    \qquad\text{for all }\|e\|_2\le r_0.
\end{align}
In particular, $\lim_{\|e\|_2\to 0} \frac{\|R(e)\|_2}{\|e\|_2} = 0$, so the dynamics is locally governed by the linear map $J$, up to a perturbation of second order.

We start with a local contraction result in the vicinity of the fixed-point, proved in Appendix~\ref{sec:proofOfLemmaLocalContraction}.
\begin{lem}[Local contraction near the fixed point]
\label{lem:local-contraction}
Let $M:\mathbb{R}^d\to\mathbb{R}^d$ be $C^2$ with a fixed point $M(x^\star)=x^\star$, and let $J(x)$ denote its Jacobian. Assume that $J(x^\star)$ is symmetric with spectral radius $\rho(J(x^\star))<1$.
Then:
\begin{enumerate}
    \item \emph{(Local contraction.)} For every 
    $\rho(J(x^\star))<q<1$, there exists a neighborhood $\mathcal{U}$ of $x^\star$ such that
    \begin{align}\label{eq:local-contractivity}
        \|M(x)-M(x^\star)\|_2 \le q \|x-x^\star\|_2, \qquad \forall x\in\mathcal{U}.
    \end{align}
    Consequently, $M$ is a strict contraction on $\mathcal{U}$, and any iterates $\hat{x}^{(t+1)} = M(\hat{x}^{(t)})$ initialized in $\mathcal{U}$ satisfy
    \begin{align}
        \|\hat{x}^{(t)}-x^\star\|_2 \le q^t \|\hat{x}^{(0)}-x^\star\|_2,
        \qquad t=0,1,2,\dots.
    \end{align}
    \item \emph{(Asymptotic rate bound.)} The asymptotic linear convergence rate is upper bounded by the spectral radius of the Jacobian, namely
    \begin{align}\label{eq:asymptotic-rate}
        \limsup_{t\to\infty}\frac{\|\hat{x}^{(t+1)}-x^\star\|_2}{\|\hat{x}^{(t)}-x^\star\|_2} \le \rho(J(x^\star)).
    \end{align}
\end{enumerate}
\end{lem}

\subsection{Spectral structure and asymptotic rate}

We henceforth assume that the Jacobian $J \succeq 0$ is real, positive semidefinite (PSD), symmetric and admits the spectral decomposition
\begin{align}\label{eq:spectral-decomp-dyn}
    J = U\Lambda U^\top,
\end{align}
where $U\in\mathbb{R}^{d\times d}$ is orthogonal with eigenvectors $\{u_1,\dots,u_d\}$, and nonnegative eigenvalues $\{\lambda_1,\dots,\lambda_d \}$ satisfying
\begin{align}
    \Lambda=\mathrm{diag}(\lambda_1,\dots,\lambda_d),\qquad 1 > \lambda_1 \ge \lambda_2 \ge \dots \ge \lambda_d \ge 0.
\end{align}
In particular, the spectral radius is $\rho(J) = \lambda_1$. Throughout, we assume that $\lambda_1$ is simple, i.e., $\lambda_1 > \lambda_2$, for notational simplicity; the main results extend to the non-simple case by replacing the eigenvector $u_1$ with the dominant eigenspace $E_1 \triangleq \ker(J-\lambda_1 I)$ (and all projection/genericity conditions with the orthogonal projection onto $E_1$).

It is convenient to express the error in this eigenbasis. Writing $e_t = U \alpha_t$, or equivalently, $\alpha_t = U^\top e_t$, and substituting into~\eqref{eq:nonlinear-perturb-flow} gives
\begin{align}\label{eq:alpha-recursion}
    \alpha_{t+1} = \Lambda \alpha_t + U^\top R(U\alpha_t).
\end{align}
Thus, in the absence of the remainder $R$, the coordinates $\{\alpha_{t,i}\}_{i=1}^d$ would be decoupled according to the scalar recursions $\alpha_{t+1,i}= \lambda_i \alpha_{t,i}$, with no interaction between the different indices $i$. The term $U^\top R(U\alpha_t)$ represents a higher-order nonlinear perturbation that generally mixes the coordinates (i.e., introduces cross-terms between different modes). However, this mixing vanishes asymptotically: by~\eqref{eq:remainder-quadratic}, its magnitude is $O(\|\alpha_t\|_2^2)$, and hence it becomes negligible compared to the linear term as $\|e_t\|\to 0$.

\paragraph{Linearized Lyapunov exponent.}

We next formalize the fact that, in the infinitesimal regime $e_0 \to 0$, the asymptotic exponential convergence rate is governed by the spectral radius $\rho(J)$. This is naturally expressed in a Lyapunov-exponent type limit (see, e.g., ~\cite{katok1995introduction, ortega2000iterative}), in which we first consider infinitesimal initial errors and then let $t \to \infty$. In what follows, we say that a nonzero vector $v$ is \emph{generic} (with respect to the dominant eigenvalue) if $\langle v, u_1 \rangle \neq 0$, where $u_1$ is a unit eigenvector of $J$ associated with the eigenvalue $\lambda_1$.

For a given initialization $e_0\in\mathbb{R}^d$, let $e_t(e_0)$ denote the solution of~\eqref{eq:error-recursion} at time $t$.

\begin{thm}[Linearized Lyapunov exponent]\label{thm:Lyap-symmetric}
Let $J$ be symmetric with spectral decomposition as in~\eqref{eq:spectral-decomp-dyn}, and let $u_1$ be a unit eigenvector associated with $\lambda_1=\rho(J)$. Then
\begin{align}\label{eq:Lyap-flow}
    \lim_{t\to\infty}\,\lim_{\substack{e_0\to 0\\ |\langle u_1,e_0\rangle|\ge \theta_0\|e_0\|_2}} \frac{1}{t}\log\frac{\|e_t(e_0)\|_2}{\|e_0\|_2} = \log \rho(J),
\end{align}
for any fixed $\theta_0\in(0,1]$.
\end{thm}

Theorem~\ref{thm:Lyap-symmetric} is a standard result in dynamical systems~\cite{katok1995introduction, ortega2000iterative}, and is a direct consequence of Theorem~\ref{thm:joint-t-e0}, as stated next. 

\subsection{Finite-time linearization and iteration complexity}

We now quantify how close the finite-time convergence rate is to the linearized rate, as a joint function of the time $t$ and the size of the initialization $\|e_0\|_2$. The next theorem (proved in Appendix~\ref{subsec:proofOfIterationComplexityMainThm}) shows that, for sufficiently small $\|e_0\|_2$, there is a substantial time window in which the observed exponential decay rate is arbitrarily close to $\log\rho(J)$.

\begin{thm}[Joint control of $t$ and $\|e_0\|$]\label{thm:joint-t-e0}
Let $J$ be symmetric with a spectral decomposition as in~\eqref{eq:spectral-decomp-dyn}, and let $u_1$ be a unit eigenvector associated with $\lambda_1=\rho(J)$. Assume that the initialization $e_0$ satisfies $|\langle e_0, u_1 \rangle| \ge \theta_0\|e_0\|_2$, i.e., $e_0$ has a nontrivial component along the slow eigenvector $u_1$ associated with $\lambda_1$. 
Fix any tolerance $\eta>0$. Then there exist  $r(\eta) > 0$, and a minimal time $T_{\min}(\eta)$, such that whenever $e_0$ satisfies $0<\|e_0\|_2\le r(\eta)$, and $t$ is an integer obeying
\begin{align}\label{eq:joint-time-window-updated}
    T_{\min}(\eta) \le t \le \frac{1}{|\log \rho(J)|} \log \p{\frac{\theta_0}{2K\|e_0\|_2}},
\end{align}
where $K \triangleq \frac{c}{1- \rho(J)}$, we have
\begin{align}\label{eq:joint-rate-eta}
    \left| \frac{1}{t}\log\frac{\|e_t\|_2}{\|e_0\|_2} - \log \rho(J) \right| \le \eta.
\end{align}
In particular, for sufficiently small $\|e_0\|_2$, the interval in~\eqref{eq:joint-time-window-updated} is nonempty and has length of order $\log(1/\|e_0\|_2)$, and on this interval the observed exponential decay rate of $\|e_t\|_2$ stays within $\eta$ of the linearized rate $\log\rho(J)$.
\end{thm}

The lower bound $t \ge T_{\min}(\eta)$ in~\eqref{eq:joint-time-window-updated} reflects the time needed for the faster modes (those with eigenvalues $<\lambda_1$) to decay, so that the dynamics are effectively dominated by the slow eigendirection $u_1$.
The upper bound on $t$ ensures that the error remains small enough for the quadratic remainder $R(e_t)$ to stay negligible compared with the linear term $J e_t$: the linear part decays exponentially at rate $\rho(J)^t$, whereas the accumulated nonlinear remainder is uniformly of order $\|e_0\|_2^2$, so the validity of the linear approximation persists only up to times of order $\log(1/\|e_0\|_2)$.
Thus, the window in~\eqref{eq:joint-time-window-updated} is precisely the regime in which (i) the slowest eigenmode dominates the dynamics and (ii) higher-order nonlinear effects remain subdominant.

Finally, we extract a generic lower bound on the number of iterations required to reach a prescribed accuracy, showing that the tail iteration complexity is governed by the spectral radius $\rho(J)$. The following corollary is proved in Appendix~\ref{subsec:proofOfCorollaryIterationComplexity}.

\begin{corollary}[Tail iteration complexity in the linear regime]\label{cor:iter-complex-flow}
Suppose the assumptions of Theorem~\ref{thm:joint-t-e0} hold, and let $\varepsilon_{\mathrm{abs}}\in(0,\|e_0\|_2]$ be a target accuracy. Then, for $\|e_0\|_2$ sufficiently small so that we remain in the regime of Theorem~\ref{thm:joint-t-e0}, any integer $t$ with $\|e_t\|_2 \le \varepsilon_{\mathrm{abs}}$ must satisfy
\begin{align}\label{eq:iter-complex-rough}
    t \gtrsim \frac{1}{|\log\rho(J)|}\,\log \left(\frac{\|e_0\|_2}{\varepsilon_{\mathrm{abs}}}\right).
\end{align}
Thus, the tail iteration complexity is governed by the spectral radius $\rho(J)$.
\end{corollary}

\subsection{Proofs}

\subsubsection{Proof of Lemma~\ref{lem:local-contraction}} \label{sec:proofOfLemmaLocalContraction}
We prove first the local contraction result, and then the asymptotic rate expression.

\paragraph{Local contraction.}
As $M \in C^2$, $J$ is continuous at $x^\star$. Thus, for every $\varepsilon > 0$ there exists a neighborhood $\mathcal{U}$ of $x^\star$ such that
\begin{align}
    \|J(x) - J(x^\star)\|_2 < \varepsilon, \qquad \forall x \in \mathcal{U}.
\end{align}
Choose $q$ satisfying $\rho(J(x^\star)) < q < 1$.
Because $J(x^\star)$ is symmetric, its spectral radius equals its operator norm,
$\rho(J(x^\star)) = \|J(x^\star)\|_2$.
Set $\varepsilon \triangleq  q - \|J(x^\star)\|_2 > 0$.
Then, by continuity, there exists a neighborhood $\mathcal{U}_\varepsilon$ of $x^\star$ such that
\begin{align}
    \|J(x) - J(x^\star)\|_2 \le q - \|J(x^\star)\|_2, \qquad \forall x \in \mathcal{U}_\varepsilon.
\end{align}
Consequently, by triangle inequality, we have,
\begin{align}
    \|J(x)\|_2 \le \|J(x^\star)\|_2 + \|J(x) - J(x^\star)\|_2 \le \|J(x^\star)\|_2 + (q - \|J(x^\star)\|_2) = q < 1,
\end{align}
which ensures that the operator norm of the Jacobian remains strictly below one throughout $\mathcal{U}_\varepsilon$, that is,
\begin{align}\label{eq:local-norm-bound}
    \sup_{x\in\mathcal{U}_\varepsilon }\|J(x)\|_2 \leq  q < 1.
\end{align}

To establish the local contraction property, we apply the mean-value form of the first-order Taylor expansion for vector-valued functions.
By the fundamental theorem of calculus applied along the line segment connecting $x^\star$ and $x$, we have,
\begin{align}
    M(x)-M(x^\star) = \int_0^1 \frac{d}{dt}M(x^\star+t(x-x^\star)) dt = \int_0^1 J(x^\star+t(x-x^\star))(x-x^\star) dt.
\end{align}
Thus, for any $x\in\mathcal{U}_\varepsilon$, using the submultiplicative property of the operator norm,
\begin{align}
    \|M(x)-M(x^\star)\|_2 &\le \int_0^1 \|J(x^\star+t(x-x^\star))\|_2 \|x-x^\star\|_2 dt 
    \notag\\
    &\le \sup_{z\in\mathcal{U}_\varepsilon}\|J(z)\|_2 \|x-x^\star\|_2. \label{eq:jacobian-bound}
\end{align}
By construction of the neighborhood $\mathcal{U}_\varepsilon$ (see~\eqref{eq:local-norm-bound}), $\sup_{z\in\mathcal{U}_\varepsilon}\|J(z)\|_2\le q<1$.
Substituting into~\eqref{eq:jacobian-bound} gives
\begin{align}\label{eq:M-contraction}
    \|M(x)-M(x^\star)\|_2  \le q \|x-x^\star\|_2,
\end{align}
which shows that $M$ is a strict contraction on $\mathcal{U}_\varepsilon$ with contraction constant $q<1$.

Finally, the contractive bound~\eqref{eq:local-contractivity} implies geometric convergence by direct iteration. Indeed, if $\hat{x}^{(0)}\in\mathcal{U}$, then applying~\eqref{eq:local-contractivity} with $x=\hat{x}^{(t)}$ yields
\begin{align}
    \|\hat{x}^{(t+1)}-x^\star\|_2  = \|M(\hat{x}^{(t)})-M(x^\star)\|_2 \le q\,\|\hat{x}^{(t)}-x^\star\|_2,
    \qquad t=0,1,2,\dots.
\end{align}
Iterating this inequality gives
\begin{align}\label{eq:geo-conv}
    \|\hat{x}^{(t)}-x^\star\|_2 \le q^t \|\hat{x}^{(0)}-x^\star\|_2, \qquad t=0,1,2,\dots.
\end{align}
In particular, $\hat{x}^{(t)} \to x^\star$ at a linear (geometric) rate with contraction factor $q$.

\paragraph{Asymptotic rate bound.} This result follows from the classical perturbation analysis of nonlinear iterations with asymptotically vanishing remainder terms (see, e.g.,~\cite{ortega2000iterative}).

\subsubsection{Proof of Theorem~\ref{thm:joint-t-e0}} \label{subsec:proofOfIterationComplexityMainThm}
Recall that the error sequence satisfies the Taylor expansion~\eqref{eq:nonlinear-perturb-flow}
\begin{align}\label{eq:taylor-split-lemma}
    e_{t+1} = J e_t + R(e_t),
\end{align}
where the remainder term obeys the quadratic bound
\begin{align}\label{eq:R-quadratic-lemma}
    \|R(e)\|_2 \le c \,\|e\|_2^2 \quad\text{for all}\quad\|e\|_2 \le r_0,
\end{align}
for some constants $c>0$ and $r_0>0$.

Throughout, we denote by $\lambda \triangleq \rho(J)$ the spectral radius of $J$. Since $J$ is real symmetric, we have $\|J\|_2 = \lambda < 1$, so $J$ is a strict contraction.
In addition, we work in a regime where the nonlinear dynamics remain in the local neighborhood of $x^\star$ and are non-expansive (by Lemma~\ref{lem:local-contraction}), namely
\begin{align}\label{eq:contraction-assumption}
    \|e_t\|_2 \le \|e_0\|_2 \le r_0
    \qquad\text{for all }t\ge 0.
\end{align}

We start by stating and proving two auxiliary lemmas that will be used to prove the theorem.

\begin{lem}\label{lem:uniform-remainder}
For every $t\ge 0$ there exists a vector $\widetilde R^{(t)}$ such that
\begin{align}
    e_t &= J^t e_0 + \widetilde R^{(t)}, \label{eq:linear-plus-rem-lemma} 
\end{align}
where $\|\widetilde R^{(t)}\|_2$ satisfies,
\begin{align}
    \|\widetilde R^{(t)}\|_2 &\le \frac{c}{1-\lambda}\,\|e_0\|_2^2.
    \label{eq:Rtilde-min-bound}
\end{align}

\end{lem}

\begin{proof}[Proof of Lemma~\ref{lem:uniform-remainder}]
Unrolling the recursion \eqref{eq:taylor-split-lemma} yields for $t\ge 1$,
\begin{align}\label{eq:VOC-lemma}
    e_t = J^t e_0 + \sum_{k=0}^{t-1} J^{t-1-k} R(e_k),
\end{align}
while for $t=0$ we have $e_0 = J^0 e_0$.
Define the right-hand-side of~\eqref{eq:VOC-lemma} by
\begin{align}\label{eq:Rtilde-def}
    \widetilde R^{(t)} \triangleq \sum_{k=0}^{t-1} J^{t-1-k} R(e_k) \quad (t\ge 1), \qquad \widetilde R^{(0)} \triangleq 0.
\end{align}
Then \eqref{eq:linear-plus-rem-lemma} holds by construction.

Next we bound $\|\widetilde R^{(t)}\|_2$.  
Using \eqref{eq:contraction-assumption} and \eqref{eq:R-quadratic-lemma}, we have $\|R(e_k)\|_2 \le c \|e_k\|_2^2 \le c \|e_0\|_2^2$, for all $k\ge 0$. Combining this with $\|J^{t-1-k}\|_2 \le \|J\|_2^{\,t-1-k} \le \lambda^{t-1-k}$ gives
\begin{align}
    \|\widetilde R^{(t)}\|_2 \, \le \, \sum_{k=0}^{t-1} \|J^{t-1-k}\|_2 \,\|R(e_k)\|_2 \,  \, \le \, c \|e_0\|_2^2 \sum_{k=0}^{t-1} \lambda^{t-1-k}.
    \label{eq:Rtilde-sum}
\end{align}
The sum in \eqref{eq:Rtilde-sum} is a geometric series:
\begin{align}
    \sum_{k=0}^{t-1} \lambda^{t-1-k} = \sum_{j=0}^{t-1} \lambda^j \le \frac{1}{1-\lambda}.
\end{align}
Substituting into \eqref{eq:Rtilde-sum} gives
\begin{align}\label{eq:Rtilde-uniform}
    \|\widetilde R^{(t)}\|_2 \le \frac{c}{1-\lambda}\,\|e_0\|_2^2 \qquad\text{for all } t\ge 0.
\end{align}
Combining \eqref{eq:Rtilde-def} and \eqref{eq:Rtilde-uniform} yields the stated bound \eqref{eq:Rtilde-min-bound}, completing the proof.
\end{proof}

\begin{lem}\label{lem:linear-bracket}
Recall that by~\eqref{eq:Rtilde-min-bound}, there exists a constant $0 < K = \frac{c}{1-\lambda}$ such that for every $t\ge 0$,
\begin{align}\label{eq:lem-uniform-Rtilde}
    e_t = J^t e_0 + \widetilde R^{(t)}, 
    \qquad \|\widetilde R^{(t)}\|_2 \le K \|e_0\|_2^2.
\end{align}
Assume the initialization $e_0$ has nontrivial projection onto the dominant eigenspace, in the sense that $0 < \|e_0\|_2$, and $|\langle e_0, u_1\rangle| \ge \theta_0 \,\|e_0\|_2$, for some fixed $\theta_0\in(0,1]$.
Then for every 
\begin{align}\label{eq:lem-linear-regime-cond}
    0 < t \le \frac{1}{|\log\lambda|} \log \p{\frac{\theta_0}{2K\|e_0\|_2}}
\end{align}
we have,
\begin{align}\label{eq:lem-et-bracket}
    \frac{\theta_0}{2}\,\lambda^t \|e_0\|_2 \le \|e_t\|_2 \le \frac{3}{2}\,\lambda^t \|e_0\|_2.
\end{align}
\end{lem}

\begin{proof}[Proof of Lemma~\ref{lem:linear-bracket}]
Recall the notation for the largest eigenvalue $\lambda = \lambda_1 = \rho(J)$. Write $e_0$ in the eigenbasis of $J$:
\begin{align}
    e_0 = \sum_{i=1}^d a_i u_i,
    \qquad a_i = \langle e_0, u_i\rangle.
\end{align}
Then $J^t e_0 = \sum_{i=1}^d \lambda_i^t a_i u_i$, and hence
\begin{align}\label{eq:lem-Jt-lower}
    \|J^t e_0\|_2^2 = \sum_{i=1}^d \lambda_i^{2t} a_i^2 \ge \lambda_1^{2t} a_1^2 = \lambda^{2t} |\langle e_0, u_1\rangle|^2.
\end{align}
Using the genericity condition $|\langle e_0, u_1\rangle| \ge \theta_0 \,\|e_0\|_2$, we obtain
\begin{align}\label{eq:lem-Jt-lower-final}
    \|J^t e_0\|_2 \ge \lambda_1^t |\langle e_0, u_1\rangle| \ge \theta_0 \lambda^t \|e_0\|_2.
\end{align}
On the other hand, since $J$ is symmetric with eigenvalues bounded by $\lambda_1$, we have
\begin{align}\label{eq:lem-Jt-upper}
    \|J^t e_0\|_2 \le \|J^t\|_2 \|e_0\|_2  \le \lambda_1^t \|e_0\|_2 = \lambda^t \|e_0\|_2.
\end{align}

By the uniform remainder representation~\eqref{eq:lem-uniform-Rtilde},
\begin{align}
    e_t = J^t e_0 + \widetilde R^{(t)},
    \qquad \|\widetilde R^{(t)}\|_2 \le K\|e_0\|_2^2.
\end{align}
Using the inverse triangle inequality together with \eqref{eq:lem-Jt-lower-final}, \eqref{eq:lem-Jt-upper}, and the condition in~\eqref{eq:lem-linear-regime-cond}, we obtain
\begin{align}
    \|e_t\|_2
    &\ge \|J^t e_0\|_2 - \|\widetilde R^{(t)}\|_2
      \notag\\
    &\ge \theta_0 \lambda_1^t \|e_0\|_2 - K\|e_0\|_2^2
      \notag\\
    &\ge \theta_0 \lambda_1^t \|e_0\|_2 - \frac{\theta_0}{2}\lambda_1^t\|e_0\|_2
      \notag\\
    &= \frac{\theta_0}{2}\,\lambda_1^t\|e_0\|_2 = \frac{\theta_0}{2}\,\lambda^t\|e_0\|_2,
    \label{eq:lem-et-lower}
\end{align}
and similarly
\begin{align}
    \|e_t\|_2
    &\le \|J^t e_0\|_2 + \|\widetilde R^{(t)}\|_2
      \notag\\
    &\le \lambda_1^t \|e_0\|_2 + K\|e_0\|_2^2
      \notag\\
    &\le \lambda_1^t \|e_0\|_2 + \frac{\theta_0}{2}\lambda_1^t\|e_0\|_2
      \notag\\
    &=  \p{1 + \frac{\theta_0}{2}}\lambda_1^t\|e_0\|_2
    \le \frac{3}{2}\,\lambda^t\|e_0\|_2,
    \label{eq:lem-et-upper}
\end{align}
where in the last step we used $\theta_0\le 1$.  
Combining \eqref{eq:lem-et-lower} and \eqref{eq:lem-et-upper} yields
\eqref{eq:lem-et-bracket}, completing the proof.
\end{proof}

\begin{proof} [Proof of Theorem~\ref{thm:joint-t-e0}]

By Lemma~\ref{lem:linear-bracket}, for every 
\begin{align}\label{eq:thm-t-upper}
    t \le \frac{1}{|\log\lambda|} \log \p{\frac{\theta_0}{2K\|e_0\|_2}}.
\end{align}
we have the two-sided estimate
\begin{align}\label{eq:thm-et-bracket}
    \frac{\theta_0}{2}\,\lambda^t\|e_0\|_2 \le \|e_t\|_2 \le \frac{3}{2}\,\lambda^t\|e_0\|_2.
\end{align}

\paragraph{Step 1: Control the exponential rate.}
Divide \eqref{eq:thm-et-bracket} by $\|e_0\|_2$ and take logarithms:
\begin{align}
    \log \p{\frac{\theta_0}{2}} + t\log\lambda \le \log\frac{\|e_t\|_2}{\|e_0\|_2} \le \log \p{\frac{3}{2}} + t\log\lambda.
\end{align}
Dividing by $t>0$ gives
\begin{align}\label{eq:thm-rate-bracket}
    \log\lambda + \frac{1}{t}\log \p{\frac{\theta_0}{2}} \le \frac{1}{t}\log\frac{\|e_t\|_2}{\|e_0\|_2} \le \log\lambda + \frac{1}{t}\log \p{\frac{3}{2}}.
\end{align}
Hence
\begin{align}\label{eq:thm-rate-error}
    \left|\frac{1}{t}\log\frac{\|e_t\|_2}{\|e_0\|_2} - \log\lambda \right| \le \frac{1}{t} \max\Bigl\{ \bigl|\log(\theta_0/2)\bigr|, \bigl|\log(3/2)\bigr| \Bigr\}.
\end{align}
Fix $\eta>0$. Define
\begin{align}\label{eq:thm-Tmin-def}
    T_{\min}(\eta) \triangleq \max\Bigl\{1, \frac{1}{\eta} \max\bigl\{ \bigl|\log(\theta_0/2)\bigr|, \bigl|\log(3/2)\bigr| \bigr\} \Bigr\}.
\end{align}
Then, for all integers $t\ge T_{\min}(\eta)$, the right-hand side of \eqref{eq:thm-rate-error} is at most $\eta$, and we obtain \eqref{eq:joint-time-window-updated} assuming $t$ satisfies the linear-regime condition \eqref{eq:thm-t-upper}.

\paragraph{Step 2: Existence of a nonempty time window and choice of $r(\eta)$.}
We now choose $r(\eta)$ so that, whenever $\|e_0\|_2\le r(\eta)$, the interval
\begin{align}\label{eq:thm-window}
    \ppp{t\in\mathbb{N}: T_{\min}(\eta) \le t \le \frac{1}{|\log\lambda|} \log(\frac{\theta_0}{2K\|e_0\|_2}) }
\end{align}
is nonempty.
This is equivalent to requiring
\begin{align}
    T_{\min}(\eta) \le \frac{1}{|\log\lambda|} \log \p{\frac{\theta_0}{2K\|e_0\|_2}},
\end{align}
or, rearranging,
\begin{align}
    \|e_0\|_2 \le \frac{\theta_0}{2K}\, \exp\, \p{-|\log\lambda|\,T_{\min}(\eta)}.
\end{align}
Define
\begin{align}\label{eq:thm-r-eta-def}
    r(\eta) \triangleq \min\Bigl\{r_0,\frac{\theta_0}{2K}\, \exp\, \p{-|\log\lambda|\,T_{\min}(\eta)}\Bigr\}.
\end{align}
Then, whenever $0<\|e_0\|_2\le r(\eta)$, the interval described in \eqref{eq:joint-time-window-updated} is nonempty, and for any integer $t$ in this interval, both the linear-regime condition \eqref{eq:thm-t-upper} and the lower bound $t\ge T_{\min}(\eta)$ hold, completing the proof.
\end{proof}

\subsubsection{Proof of Corollary~\ref{cor:iter-complex-flow}} \label{subsec:proofOfCorollaryIterationComplexity}
Fix $\eta>0$, and assume we are in the regime of Theorem~\ref{thm:joint-t-e0}, so that for all sufficiently small $\|e_0\|_2$ there is a nonempty time window \eqref{eq:joint-time-window-updated} on which the linear approximation is valid. In particular, for any integer $t$ in this window, Lemma~\ref{lem:linear-bracket} and Theorem~\ref{thm:joint-t-e0} yields the lower bound
\begin{align}\label{eq:cor-lower-bracket}
    \|e_t\|_2 \ge \frac{\theta_0}{2}\,\lambda^t \|e_0\|_2,
\end{align}
where $\lambda = \rho(J)$.

Let $\varepsilon_{\mathrm{abs}}\in(0,\|e_0\|_2]$ be the target accuracy, and assume $\|e_0\|_2$ and $\varepsilon_{\mathrm{abs}}$ are small enough that the first time $t$ with $\|e_t\|_2 \le \varepsilon_{\mathrm{abs}}$ lies inside the window of Theorem~\ref{thm:joint-t-e0}.
For such a $t$ we then have
\begin{align}
    \frac{\theta_0}{2}\,\lambda^t \|e_0\|_2 \le \|e_t\|_2 \le \varepsilon_{\mathrm{abs}} .
\end{align}
Rearranging gives
\begin{align}
    \lambda^t \le \frac{2}{\theta_0} \, \frac{\varepsilon_{\mathrm{abs}}}{\|e_0\|_2},
\end{align}
and taking logarithms (using $\lambda\in(0,1)$ so that $\log\lambda<0$) yields
\begin{align}
    t \ge \frac{1}{|\log\lambda|} \p{\log\frac{\|e_0\|_2}{\varepsilon_{\mathrm{abs}}} - \log\frac{2}{\theta_0} }.
\end{align}
Thus, up to an additive constant depending only on $\theta_0$, we have
\begin{align}
    t  \gtrsim \frac{1}{|\log\lambda|}\, \log \left(\frac{\|e_0\|_2}{\varepsilon_{\mathrm{abs}}}\right),
\end{align}
where the notation $\gtrsim$ denotes universal constants independent of $\|e_0\|_2$ and $\varepsilon_{\mathrm{abs}}$.
Recalling that $\lambda = \rho(J)$ gives \eqref{eq:iter-complex-rough}, and shows that the tail iteration complexity is governed by the spectral radius $\rho(J)$.

\section{Local linear convergence of the population EM operator} \label{sec:proofOfEmSpectralRadius}

In this appendix, we show that the population EM map $M$ is $C^\infty$ and derive an explicit covariance-form expression for its Jacobian $J(x)$ (Lemma~\ref{lem:Jacobian-M}), which in particular implies that $J(x^\star)$ is symmetric, positive semidefinite, and its spectrum is contained in $[0,1]$. Finally, we combine these results to conclude the proof of Theorem~\ref{thm:EM-MRA-spectral-radius}.

Throughout, $S\sim \mathrm{Unif}\{0,\dots,d-1\}$ denotes the latent shift distributed uniformly, and the random vector $Y$ has the following distribution,
\begin{align}
    Y=\mathcal{T}_S x^\star + \xi, 
    \qquad \xi\sim\mathcal{N}(0,\sigma^2 I_d),
\end{align}
where $x^\star$ is the ground-truth signal. Recall the definition of $\gamma_\ell$ in \eqref{eqn:softmaxGamma},
\begin{align}
    \gamma_\ell(x;Y) =\frac{\exp\p{\langle Y,\mathcal{T}_\ell x\rangle/\sigma^2}} {\sum_{r=0}^{d-1}\exp\p{\langle Y,\mathcal{T}_r  x\rangle/\sigma^2}}, \label{eqn:app-B2}
\end{align}
and of the population EM map in~\eqref{eq:Phi-pop-1d}
\begin{align}
    M(x)=\mathbb{E} \left[\sum_{\ell=0}^{d-1}\gamma_\ell(x;Y) z_\ell\right], \label{eq:M-def-proof}
\end{align}
where we define
\begin{align}
    z_\ell=\mathcal{T}_\ell^{-1} Y,
    \qquad
    \bar{z}(x;Y)=\sum_{\ell=0}^{d-1}\gamma_\ell(x;Y) z_\ell. \label{eqn:app-B4}
\end{align}

\subsection{Shift-equivariance of the population EM operator}
\begin{lem}[Shift-equivariance of the population EM operator]
\label{lem:shift-equivariance-M}
Let the population EM operator $M:\mathbb{R}^d \to \mathbb{R}^d$ be defined by~\eqref{eq:M-def-proof}. Then:
\begin{align}
    \label{eq:shift-equivariance}
    M(\mathcal{T}_s x) = \mathcal{T}_s   M(x),
    \qquad \forall s \in\{0,\dots,d-1\}.
\end{align}
\end{lem}

\begin{proof}[Proof of Lemma~\ref{lem:shift-equivariance-M}]
First recall that each shift operator $\mathcal{T}_s$ is unitary: $\langle \mathcal{T}_s u,\mathcal{T}_s v\rangle=\langle u,v\rangle$ and $\mathcal{T}_s^{*}=\mathcal{T}_s^{-1}=\mathcal{T}_{-s}$. Moreover, the family $\{\mathcal{T}_s\}$ forms a commutative cyclic group representation, so
\begin{align}
    \mathcal{T}_a \mathcal{T}_b   =  \mathcal{T}_{a+b},
\end{align}
with indices understood modulo $d$.

Fix $s \in\{0,\dots,d-1\}$. We show two identities.
\begin{enumerate}
    \item \emph{Invariance of the data distribution.}
    Set $Y' \triangleq  \mathcal{T}_s Y$. Using the model and group law,
    \begin{align}
        Y' = \mathcal{T}_s (\mathcal{T}_S x^\star + \xi) = \mathcal{T}_{S+s}   x^\star + \mathcal{T}_s \xi.
    \end{align}
    Since $S+s \bmod d \sim \mathrm{Unif}\{0,\dots,d-1\}$ and $\mathcal{T}_s \cdot \xi \stackrel{\mathcal{D}}{=} \xi$ (because $\xi$ is isotropic Gaussian), it follows that $Y'$ has the same distribution as $Y$, that is $Y' \stackrel{\mathcal{D}}{=} Y$.
    
    \item \emph{Equivariance of the E-step ingredients.}
    For any $s$, using the unitarity of $\mathcal{T}_s$, we have,
    \begin{align}
        \nonumber 
        \gamma_\ell(\mathcal{T}_s x; \mathcal{T}_s Y)
        &= \frac{
            \exp \p{\langle \mathcal{T}_s Y, \mathcal{T}_\ell \mathcal{T}_s x\rangle/\sigma^2}
        }{
            \sum_{r=0}^{d-1}\exp \p{\langle \mathcal{T}_s Y, \mathcal{T}_r \mathcal{T}_s x\rangle/\sigma^2}
        } \\[1ex]
        &= 
        \frac{
            \exp \p{\langle Y, \mathcal{T}_\ell x\rangle/\sigma^2}
        }{
            \sum_{r=0}^{d-1}\exp \p{\langle Y, \mathcal{T}_r x\rangle/\sigma^2}
        }
        = \gamma_\ell(x;Y), \label{eqn:app-B7}
    \end{align}
    where indices are taken modulo $d$ and we used $\mathcal{T}_a^{-1}=\mathcal{T}_{-a}$ and $\mathcal{T}_a\mathcal{T}_b=\mathcal{T}_{a+b}$. Similarly, we have,
    \begin{align}
        z_\ell(\mathcal{T}_s Y) = \mathcal{T}_\ell^{-1}(\mathcal{T}_s Y) = \mathcal{T}_{-\ell}\mathcal{T}_s Y = \mathcal{T}_s \mathcal{T}_{-\ell} Y = \mathcal{T}_s z_\ell(Y). \label{eqn:app-B8}
    \end{align}
\end{enumerate}

We now compute $M(\mathcal{T}_s x)$ by a change of variables in the expectation.
Since $Y'=\mathcal{T}_s Y \stackrel{\mathcal{D}}{=} Y$, we have
\begin{align}
    M(\mathcal{T}_s x)
    &= \mathbb{E}_Y \left[\sum_{\ell=0}^{d-1} \gamma_\ell(\mathcal{T}_s x;Y)  z_\ell(Y)\right] \\[1ex]
    &= \mathbb{E}_{Y} \left[\sum_{\ell=0}^{d-1} \gamma_\ell(\mathcal{T}_s x;\mathcal{T}_s Y)  z_\ell(\mathcal{T}_s  Y)\right]. \label{eqn:app-B9}
\end{align}
Now, since $\gamma_\ell(\mathcal{T}_s x;\mathcal{T}_s Y)=\gamma_\ell(x;Y)$ (by~\eqref{eqn:app-B7}) and $z_\ell(\mathcal{T}_s Y)=\mathcal{T}_s z_\ell(Y)$ (by~\eqref{eqn:app-B8}), we substitute these identities into the right-hand-side of~\eqref{eqn:app-B9},
\begin{align}
    M(\mathcal{T}_s x) & = \mathbb{E}_Y \left[\sum_{\ell=0}^{d-1} \gamma_\ell(x;Y)  \mathcal{T}_s z_\ell(Y)\right] 
    \\[1ex]
    &= \mathcal{T}_s  \mathbb{E}_Y \left[\sum_{\ell=0}^{d-1} \gamma_\ell(x;Y)  z_\ell(Y)\right] 
    \\[1ex] &= \mathcal{T}_s  M(x)
\end{align}

This establishes the shift-equivariance identity~\eqref{eq:shift-equivariance}.
\end{proof}

\begin{remark}
The uniform prior on $S$ and isotropy of the Gaussian noise are used to guarantee $Y$ and $\mathcal{T}_s  Y$ have the same distribution, enabling the change of variables above. 
\end{remark}

\subsection{The Jacobian of the population EM operator}
\begin{lem}[The Jacobian of the population EM operator]
\label{lem:Jacobian-M}
Let the population EM operator $M:\mathbb{R}^d \to \mathbb{R}^d$ be defined by~\eqref{eq:M-def-proof}. Then:
\begin{enumerate}
    \item The map $M \in C^\infty (\mathbb{R}^d)$ is smooth on $\mathbb{R}^d$.
    \item Its Jacobian is given explicitly by
    \begin{align}
        J(x) \triangleq \nabla_x M(x) = \frac{1}{\sigma^2} \mathbb{E} \left[\sum_{\ell=0}^{d-1} \gamma_\ell(x;Y) (z_\ell - \bar{z})(z_\ell - \bar{z})^\top \right],
        \label{eq:Jacobian-MRA}
    \end{align}
    where $z_\ell, \bar{z}$ are defined in~\eqref{eqn:app-B4}.
    In particular, $J(x)$ is symmetric for all $x$.
    \item At the ground-truth point $x^\star$, the spectrum of $J(x^\star)$ lies in $\mathrm{spec}(J(x^\star))\subset[0,1]$.
\end{enumerate}
\end{lem}

\begin{proof}[Proof of Lemma~\ref{lem:Jacobian-M}]
We prove the lemma in steps.

\paragraph{Smoothness of $M$.}
For each fixed realization of $Y$, the map $x \mapsto \gamma_\ell(x;Y)$ is a composition of analytic mappings: a linear form $x \mapsto \langle Y,\mathcal{T}_\ell x\rangle$ followed by the exponential and a log-sum-exp normalization. Hence,  $\gamma_\ell(\cdot;Y)$ is $C^\infty$ in $x$. For a multi-index $\alpha = (\alpha_1,\dots,\alpha_d)\in\mathbb{N}^d$ we use the notation
\begin{align}
    |\alpha|=\sum_{j=1}^d \alpha_j,  \qquad \partial_x^\alpha  = \frac{\partial^{|\alpha|}}{\partial x_1^{\alpha_1}\cdots \partial x_d^{\alpha_d}} .
\end{align}
We next bound $\partial_x^\alpha\gamma_\ell(x;Y)$ uniformly over a compact set $K\subset\mathbb{R}^d$. The first derivative of the softmax responsibilities gives
\begin{align}
    \nabla_x\gamma_\ell(x;Y) = \frac{1}{\sigma^2}\gamma_\ell(x;Y) \p{z_\ell-\bar{z}(x;Y)}. \label{eqn:first-derivative-of-softmax}
\end{align}
Since $0\le \gamma_r\le 1$ and $\|z_r\|=\|Y\|$ for all $r$, we have $\|\bar{z}(x;Y)\| =\bigl\|\sum_{r=0}^{d-1}\gamma_r(x;Y)z_r\bigr\| \le \sum_{r=0}^{d-1}\gamma_r(x;Y)\|z_r\| = \|Y\|$.
Therefore,
\begin{align}
    \sup_{x\in K}\|\nabla_x\gamma_\ell(x;Y)\| \le \frac{1}{\sigma^2}\sup_{x\in K}\gamma_\ell(x;Y)\,\|z_\ell-\bar{z}(x;Y)\| \le \frac{2}{\sigma^2}\|Y\|.
\end{align}

Higher-order derivatives follow by repeatedly differentiating \eqref{eqn:first-derivative-of-softmax}.
Each differentiation contributes at most one additional factor of $\sigma^{-2}$ and one factor of a vector of the form $(z_r-\bar{z})$,  whose norm is bounded by $2\|Y\|$, while all remaining coefficients are polynomials in $\{\gamma_r\}_{r=0}^{d-1}$ with magnitudes bounded by a constant depending only on $\alpha$.
Thus, by induction on $|\alpha|$, there exist constants $C_{\alpha,K}>0$ such that for all $x\in K$,
\begin{align}\label{eq:uniform-deriv-bound}
    \bigl|\partial_x^\alpha\gamma_\ell(x;Y)\bigr| \le C_{\alpha,K}\,\sigma^{-2|\alpha|} \p{1+\|Y\|}^{|\alpha|}.
\end{align}
Because $Y$ is sub-Gaussian with finite moments of all orders, the right-hand side of \eqref{eq:uniform-deriv-bound} is integrable, uniformly over $x\in K$. Since $z_\ell(Y)=\mathcal{T}_\ell^{-1}Y$ does not depend on $x$, dominated convergence allows differentiation under the expectation in \eqref{eq:M-def-proof} for every multi-index $\alpha$:
\begin{align}
    \partial_x^\alpha M(x) = \mathbb{E} \left[ \sum_{\ell=0}^{d-1}\partial_x^\alpha \p{\gamma_\ell(x;Y)\,z_\ell(Y)}\right].
\end{align}
Therefore $M\in C^\infty(\mathbb{R}^d)$.

\paragraph{Deriving the Jacobian.}
Define the expression within the exponent of $\gamma_\ell$ in~\eqref{eqn:app-B2} by
\begin{align}
    a_\ell(x) \triangleq  \frac{\langle Y,  \mathcal{T}_\ell  x\rangle}{\sigma^2}.
\end{align}
Its derivative is given by,
\begin{align}
    \nabla_x a_\ell(x) = \frac{1}{\sigma^2}  \mathcal{T}_\ell^{-1} Y = \frac{1}{\sigma^2} z_\ell. \label{eqn:a_l_derivative}
\end{align}
The derivative of the softmax function~\eqref{eqn:app-B2} satisfies the identity (see, e.g., ~\cite[Eq.~(A.31)]{balanov2025confirmation})
\begin{align}
    \nabla_x\gamma_\ell(x;Y) = \gamma_\ell(x;Y) \sum_{r=0}^{d-1} (\delta_{\ell r} - \gamma_r(x;Y)) \nabla_x a_r(x). \label{eqn:softmax-derivative}
\end{align}
Substituting the expression for $\nabla_x a_r(x)$ in~\eqref{eqn:a_l_derivative} into~\eqref{eqn:softmax-derivative} gives
\begin{align}
    \nabla_x \gamma_\ell(x;Y) = \frac{1}{\sigma^2} \gamma_\ell(x;Y) \p{z_\ell - \bar{z}(x;Y)}, \label{eq:dGamma-proof}
\end{align}
where,
\begin{align}
    \bar{z}(x;Y) =  \sum_{r=0}^{d-1}\gamma_r(x;Y) z_r.
\end{align}
Substituting~\eqref{eq:dGamma-proof} into~\eqref{eq:Jacobian-MRA}, while exchanging between the expectation and derivative, justified by the dominated convergence theorem, we have,
\begin{align}
    J(x) = \nabla_x M(x) = \frac{1}{\sigma^2}  \mathbb{E} \left[\sum_{\ell=0}^{d-1} \gamma_\ell(x;Y) z_\ell (z_\ell - \bar{z}(x;Y))^\top\right]. \label{eqn:app-B12}
\end{align}
Expanding $z_\ell = (z_\ell - \bar{z}) + \bar{z}$ in the definition of $J(x)$ gives
\begin{align}
    \sum_{\ell=0}^{d-1} \gamma_\ell(x;Y)  z_\ell (z_\ell - \bar{z})^\top
    &= \sum_{\ell=0}^{d-1} \gamma_\ell(x;Y) \pp{(z_\ell - \bar{z}) + \bar{z}}(z_\ell - \bar{z})^\top \notag\\
    &= \sum_{\ell=0}^{d-1} \gamma_\ell(x;Y) (z_\ell - \bar{z})(z_\ell - \bar{z})^\top + \bar{z} \sum_{\ell=0}^{d-1}\gamma_\ell(x;Y)(z_\ell - \bar{z})^\top. \label{eq:cov-expand}
\end{align}
The second term vanishes because $\sum_\ell \gamma_\ell(x;Y)(z_\ell - \bar{z})=0$, and we obtain
\begin{align}
    \sum_{\ell=0}^{d-1} \gamma_\ell(x;Y)  z_\ell (z_\ell - \bar{z})^\top = \sum_{\ell=0}^{d-1} \gamma_\ell(x;Y) (z_\ell - \bar{z})(z_\ell - \bar{z})^\top. \label{eqn:app-B14}
\end{align}
Substituting~\eqref{eqn:app-B14} into~\eqref{eqn:app-B12} proves~\eqref{eq:Jacobian-MRA}.
Since each element in the sum in~\eqref{eqn:app-B14} is symmetric and positive semidefinite, their expectation is symmetric and positive semidefinite as well. Therefore $J(x) = J(x)^\top$ for all $x \in \mathbb{R}^d$.

\paragraph{The spectrum of the Jacobian $\mathrm{spec}(J(x^\star))\subset[0,1]$.}
Define $z$ to be a discrete random vector taking values in the set $\{z_\ell\}_{\ell=0}^{d-1}$, where $z_\ell$'s are defined in~\eqref{eqn:app-B4}, with conditional probabilities $\mathbb{P}(z=z_\ell\vert Y)=\gamma_\ell(x^\star;Y)$. Equivalently, we have
\begin{align}
    z=\mathcal{T}_S^{-1}Y = x^\star + \mathcal{T}_S^{-1}\xi, \label{eqn:app-B27}
\end{align}
where, $S\sim\mathrm{Unif}\{0,\dots,d-1\}$, and $\xi\sim\mathcal{N}(0,\sigma^2 I_d)$. 

Now, at the ground truth, the responsibilities $\gamma_\ell(x^\star;Y)$ coincide with the posterior probabilities $\mathbb{P}(S=\ell\vert Y)$ under the data model $Y=\mathcal{T}_S x^\star+\xi$, where $S\sim\mathrm{Unif}\{0,\dots,d-1\}$ and $\xi\sim\mathcal{N}(0,\sigma^2 I_d)$.
Hence, conditioned on $Y$,
\begin{align}
    \mathbb{E}[z\vert Y] =\sum_{\ell=0}^{d-1}\gamma_\ell(x^\star;Y) z_\ell = \bar{z}(x^\star;Y),
\end{align}
and,
\begin{align}
    \operatorname{Cov}(z\vert Y) =\sum_{\ell=0}^{d-1}\gamma_\ell(x^\star;Y) (z_\ell-\bar{z}(x^\star;Y))(z_\ell-\bar{z}(x^\star;Y))^\top.
\end{align}
Substituting this expression into~\eqref{eq:Jacobian-MRA} shows that $J(x^\star)$ is the expectation of a scaled conditional covariance:
\begin{align}
    J(x^\star) = \frac{1}{\sigma^2} \mathbb{E}\pp{\operatorname{Cov}(z\vert Y)}.
\end{align}
Since each conditional covariance $\operatorname{Cov}(z\vert Y)$ is symmetric and positive semidefinite, their expectation is also symmetric and positive semidefinite.
Therefore $J(x^\star)=J(x^\star)^\top$, and,
\begin{align}
    J(x^\star)\succeq 0.\label{eqn:app-B30}
\end{align}
To upper bound the spectrum of $J(x^\star)$, we apply the law of total variance:
\begin{align}\label{eq:totalvar-operator}
    \operatorname{Cov}(z) =\mathbb{E}   \pp{\operatorname{Cov}(z\vert Y)} +\operatorname{Cov} \p{\mathbb{E}[z\vert Y]} \succeq \mathbb{E} \pp{\operatorname{Cov}(z\vert Y)},
\end{align}
where both terms on the right-hand side are symmetric positive semidefinite. By the definition of $z$ in~\eqref{eqn:app-B27}, we have,
\begin{align}
    \operatorname{Cov}(z) = \operatorname{Cov}(\mathcal{T}_S^{-1}\xi)=\mathbb{E} \pp{\mathcal{T}_S^{-1}(\sigma^2 I_d)\mathcal{T}_S} = \sigma^2 I_d. \label{eqn:app-B33}
\end{align}
Substituting~\eqref{eqn:app-B33} into~\eqref{eq:totalvar-operator} yields
\begin{align}
    \mathbb{E}[\operatorname{Cov}(z\vert Y)] \preceq \sigma^2 I_d.
\end{align}
Combining with~\eqref{eqn:app-B30} gives
\begin{align}
    0 \preceq J(x^\star) =\frac{1}{\sigma^2} \mathbb{E}[\operatorname{Cov}(z\vert Y)] \preceq I_d,
\end{align}
so all eigenvalues of $J(x^\star)$ lie in $[0,1]$.
In particular, $J(x^\star)$ is symmetric, positive semidefinite, and orthogonally diagonalizable with real nonnegative eigenvalues not exceeding~1.

\end{proof}

\subsection{Proof of Theorem~\ref{thm:EM-MRA-spectral-radius}}

The proof builds on Lemmas~\ref{lem:Jacobian-M} and~\ref{lem:local-contraction}.
By Lemma~\ref{lem:Jacobian-M}, $M\in C^\infty(\mathbb{R}^d)$ and its Jacobian is given by~\eqref{eq:Jacobian-MRA}. At $x^\star$, $J(x^\star)$ is symmetric, positive semidefinite, and $\mathrm{spec}(J(x^\star))\subset[0,1]$.

\paragraph{Local contraction.}
Fix any $q\in\p{\rho(J(x^\star)), 1}$. By Lemma~\ref{lem:local-contraction}(1), there exists a neighborhood $\mathcal{U}$ of $x^\star$ such that
\begin{align}\label{eq:local-contractivity-2}
    \|M(x)-M(x^\star)\|_2  \leq  q \|x-x^\star\|_2,
    \qquad \forall x\in\mathcal{U}.
\end{align}
In particular, if $\hat{x}^{(0)}\in\mathcal{U}$ and $\hat{x}^{(t+1)}=M(\hat{x}^{(t)})$, then applying~\eqref{eq:local-contractivity-2} with $x=\hat{x}^{(t)}$ yields
\begin{align}
    \|\hat{x}^{(t+1)}-x^\star\|_2 =\|M(\hat{x}^{(t)})-M(x^\star)\|_2 \le q\,\|\hat{x}^{(t)}-x^\star\|_2, \qquad t=0,1,2,\dots.
\end{align}
Iterating this inequality gives the geometric bound
\begin{align}
    \|\hat{x}^{(t)}-x^\star\|_2  \le q^t \|\hat{x}^{(0)}-x^\star\|_2,
    \qquad t=0,1,2,\dots.
\end{align}

\paragraph{Asymptotic rate.}
Let $e_t\triangleq \hat{x}^{(t)}-x^\star$. By Lemma~\ref{lem:local-contraction} (2),
\begin{align}
    \limsup_{t\to\infty}\frac{\|e_{t+1}\|_2}{\|e_t\|_2} \le \rho \p{J(x^\star)}.
\end{align}

\paragraph{Infinitesimal initialization regime. }
Apply Theorem~\ref{thm:Lyap-symmetric} with $J=J(x^\star)$ and $e_0 = e^{(0)} = \hat{x}^{(0)}-x^\star$, and note that $e_t(e_0)=e^{(t)}=\hat{x}^{(t)}-x^\star$.
The genericity assumption $\langle \hat{x}^{(0)}-x^\star,u_1\rangle\neq 0$ is exactly the condition $\langle u_1,e_0\rangle\neq 0$ in Theorem~\ref{thm:Lyap-symmetric}.
Therefore,
\begin{align}
    \lim_{t\to\infty} \,\limsup_{\substack{\hat{x}^{(0)} \to x^\star \\ \langle \hat{x}^{(0)}-x^\star,\,u_1\rangle \neq 0}} \frac{1}{t}\, \log\frac{\|\hat{x}^{(t)} - x^\star\|_2}{\|\hat{x}^{(0)} - x^\star\|_2} =  \log \rho\p{J(x^\star)},
\end{align}
which is exactly~\eqref{eq:Lyap-EM-MRA}. This proves item~(3).

\section{EM Jacobian at low-SNR: Spectral decomposition} \label{sec:proofOfJacobStructure}

In this section, we derive the population EM Jacobian expansion in the low-SNR regime $x^\star=\beta v$ with a fixed noise variance $\sigma^2$ and a varying signal amplitude $0 < |\beta| \ll \sigma$.
First, we expand the softmax responsibilities $\gamma_\ell(x^\star)$ and obtain a second-order Taylor expansion with uniform remainders (Lemma~\ref{lem:resp-near-uniform-MRA-gamma}).
Next, we insert this expansion into the covariance representation of the Jacobian and organize the $O(\beta^2)$ contribution (Lemma~\ref{lem:J-up-to-Delta2-correct} and Lemma~\ref{lem:Delta2-closed}).
Passing to the Fourier basis diagonalizes the circular–shift operators, so $\widehat{J}(\beta)=F^\ast J(\beta)F$ is block–diagonal with a $2\times 2$ block for each non-mean frequency pair $\{k,-k\}$.
This block structure yields the spectral decomposition and eigenvalue asymptotics stated in Proposition~\ref{prop:K2-MRA} and Corollary~\ref{cor:block-spectral}. 
In turn, these results directly imply the two-phase convergence (Theorem~\ref{thm:lowSNR-two-phase}) and the iteration–complexity bounds (Corollary~\ref{cor:flat-iter-lower}).

\paragraph{Notations.}
Throughout this section, we denote the mean of any collection $\{a_\ell\}_{\ell=0}^{d-1}\subset\mathbb{R}^{m}$ (scalars or vectors) by,
\begin{align}
    \overline{a} \triangleq \frac{1}{d}\sum_{\ell=0}^{d-1} a_\ell .
\end{align}
We work in the low–SNR scaling
\begin{align}\label{eq:low-snr-scaling}
    x^\star = \beta v, \qquad \|v\|_2=1, 
\end{align}
with a fixed unit-norm template $v\in\mathbb{R}^d$.
For each observation $Y$ and shift index $\ell\in\{0,\ldots,d-1\}$, set
\begin{align}
    z_\ell \triangleq \mathcal{T}_\ell^{-1}Y, \label{eq:zl}
\end{align}
and define the (unnormalized) logits and their mean–centered versions by
\begin{align}
    s_\ell(Y) & \triangleq \frac{1}{\sigma^2} \big\langle Y, \mathcal{T}_\ell x^\star \big\rangle= \frac{\beta}{\sigma^2} \big\langle z_\ell, v \big\rangle,
    \label{eq:s-logits}
    \\[0.3em]
    \overline{s}(Y) & \triangleq \frac{1}{d}\sum_{r=0}^{d-1} s_r(Y),
    \label{eq:s-bar}
    \\[0.3em]
    \eta_\ell(Y) & \triangleq s_\ell(Y) - \overline{s}(Y),
    \label{eq:eta-def}
\end{align}
so that $\sum_{\ell=0}^{d-1}\eta_\ell(Y)=0$ identically. 
According to the definition in~\eqref{eq:s-logits}--\eqref{eq:eta-def}, define the responsibilities
\begin{align}\label{eq:gamma-def}
    \gamma_\ell^{(\beta)}(Y) \triangleq  \frac{\exp \p{s_\ell(Y)}}{\sum_{r=0}^{d-1}\exp \p{s_r(Y)}} = \frac{\exp \p{\eta_\ell(Y)}}{\sum_{r=0}^{d-1}\exp \p{\eta_r(Y)}}, \qquad \ell=0,\dots,d-1,
\end{align}
where the second equality uses the softmax invariance to additive shifts
$s_\ell \mapsto s_\ell - c$ for any $c\in\mathbb{R}$. We also introduce quadratic and mixed empirical averages:
\begin{align}
    \overline{\eta^2}(Y) &  \triangleq  \frac{1}{d}\sum_{r=0}^{d-1}\eta_r(Y)^2,
    \label{eq:eta2-bar}
    \\[0.3em]
    \overline{\eta(Y) z} &  \triangleq  \frac{1}{d}\sum_{\ell=0}^{d-1}\eta_\ell(Y) z_\ell .
    \label{eq:eta-z-bar}
\end{align}

For $Y=\beta \mathcal{T}_S v+\xi$, we have
\begin{align}
    z_\ell = \mathcal{T}_\ell^{-1} (\xi + \beta \mathcal{T}_S v) = \mathcal{T}_\ell^{-1}\xi + \beta \mathcal{T}_\ell^{-1}\mathcal{T}_S v.
\end{align}
Thus, for the centered logits, the linear term in $\beta$ of the expansion around $\beta=0$ is
\begin{align}
    \eta_\ell(Y)   =  \beta \eta_\ell^{(0)}  +  O(\beta^2), \qquad   \eta_\ell^{(0)} \triangleq  \frac{1}{\sigma^2}\p{\langle \xi,\mathcal{T}_\ell v\rangle - \overline{\langle \xi,\mathcal{T} v\rangle}}.
\end{align}
We also denote
\begin{align}
    \overline{\mathcal{T}^{-1}\xi}=\frac{1}{d}\sum_{\ell=0}^{d-1} \mathcal{T}_\ell^{-1}\xi,\qquad  \overline{(\eta^{(0)})^2}=\frac{1}{d}\sum_{r=0}^{d-1}\p{\eta_r^{(0)}}^2,\qquad   \overline{\eta^{(0)} \mathcal{T}^{-1}\xi}=\frac{1}{d}\sum_{\ell=0}^{d-1}\eta_\ell^{(0)} \p{ \mathcal{T}_\ell^{-1}\xi}.
\end{align}
Our first goal is to derive the spatial-domain expansion of the Jacobian:
\begin{align} \label{eqn:Jacobian-spatial-domain-exp}
    J(\beta) &= (I-\Pi_{\mathrm{mean}}) - \frac{\beta^2}{\sigma^2} \frac{1}{d} \left[\sum_{s=0}^{d-1}\p{\mathcal{T}_s \tilde{v}} \tilde{v}^\top \mathcal{T}_s + \sum_{s=0}^{d-1}\big\langle \tilde{v},\mathcal{T}_s \tilde{v}\big\rangle \mathcal{T}_s \right] + O \p{\frac{\beta^4}{\sigma^4}}. 
\end{align}

\subsection{Second-order expansion of responsibilities}

\begin{lem}[Responsibilities near uniform: second--order expansion with uniform remainders]
\label{lem:resp-near-uniform-MRA-gamma}
Fix an observation $Y \in \mathbb{R}^d$ under the MRA model~\eqref{eqn:mainModel1D}.
Recall the definitions of the logits $s_\ell(Y)$, their centered versions $\eta_\ell(Y)$, and the responsibilities $\gamma_\ell^{(\beta)}(Y)$, in~\eqref{eq:s-logits}--\eqref{eq:gamma-def}.
Then, for a fixed $Y$, each $\gamma_\ell^{(\beta)}(Y)$ admits the second--order Taylor expansion:
\begin{align}
    \gamma_\ell^{(\beta)}(Y) &= \frac{1}{d} + \frac{1}{d} \eta_\ell(Y) + \frac{1}{2d}\p{\eta_\ell(Y)^2-\overline{\eta^2}(Y)} + R_\ell(\beta;Y),
    \label{eq:MRA-gamma-eta}
\end{align}
where the uniform remainder satisfies
\begin{align}
    |R_\ell(\beta;Y)| \leq C \frac{|\beta|^3}{\sigma^6} \|Y\|^3 \|v\|^3,
    \label{eq:MRA-R-bound}
\end{align}
uniformly in $(\ell,Y)$, for some constant $C>0$, which depends on $d$, but is independent of $(\beta,\sigma,v)$. 
\end{lem}

\begin{proof}[Proof of Lemma~\ref{lem:resp-near-uniform-MRA-gamma}]
We expand the softmax to second order around $\beta = 0$ using centered coordinates, then substitute the MRA logits and derive uniform remainder bounds. 

\paragraph{Step 1: Softmax map and its derivatives.}
For $u\in\mathbb{R}^d$, set the softmax function
\begin{align}
    f_\ell(u) \triangleq  \frac{e^{u_\ell}}{\sum_{r=0}^{d-1} e^{u_r}},\qquad \ell=0,\dots,d-1.
\end{align}
The gradient and Hessian of $f_\ell$ are
\begin{align}
    \frac{\partial f_\ell}{\partial u_s}(u) &= f_\ell(u) \p{\delta_{\ell s}-f_s(u)},\\
    \frac{\partial^2 f_\ell}{\partial u_s\partial u_t}(u) &= f_\ell(u)\p{\delta_{\ell s}-f_s(u)}\p{\delta_{\ell t}-f_t(u)} - f_\ell(u) f_s(u)\p{\delta_{st}-f_t(u)}.
\end{align}
At $u=0$, since $f_r(0)=1/d$ for all $r$,
\begin{align}
\label{eq:softmax-grad-0-gamma}
    \frac{\partial f_\ell}{\partial u_s}(0) &= \frac{1}{d}\p{\delta_{\ell s}-\frac{1}{d}},\\
\label{eq:softmax-hess-0-gamma}
    \frac{\partial^2 f_\ell}{\partial u_s\partial u_t}(0) &= \frac{1}{d}\p{\delta_{\ell s}-\frac{1}{d}}\p{\delta_{\ell t}-\frac{1}{d}} - \frac{1}{d^2}\p{\delta_{st}-\frac{1}{d}}.
\end{align}

\paragraph{Step 2: Taylor expansion with Lagrange remainder.}
By Taylor expansion,
\begin{align}
\label{eq:taylor-gamma}
    f_\ell(u) =  f_\ell(0) +  \sum_{s} \frac{\partial f_\ell}{\partial u_s}(0) u_s + \frac{1}{2}\sum_{s,t} \frac{\partial^2 f_\ell}{\partial u_s\partial u_t}(0) u_s u_t + R_\ell(u),
\end{align}
with
\begin{align}
\label{eq:remainder-gamma}
    R_\ell(u) = \frac{1}{6}\sum_{s,t,w} \frac{\partial^3 f_\ell}{\partial u_s\partial u_t\partial u_w}\p{\xi u} u_s u_t u_w, \qquad \text{for some }\xi\in(0,1).
\end{align}
Since $f(u)$ satisfies $0\le f_r(u)\le 1$ and $\sum_{r=0}^{d-1}f_r(u)=1$ for all $u\in\mathbb{R}^d$, and every third partial derivative $\partial^3 f_\ell/\partial u_s\partial u_t\partial u_w$ is a finite linear combination of products of the form $f_{r_1}(u)\cdots f_{r_m}(u)$ with coefficients depending only on $d$, these third derivatives are globally bounded on $\mathbb{R}^d$: there exists a constant $B_d<\infty$, depending only on $d$, such that
\begin{align}
    \sup_{u\in\mathbb{R}^d}\max_{\ell,s,t,w} \Bigl|\frac{\partial^3 f_\ell}{\partial u_s\partial u_t\partial u_w}(u)\Bigr| \le  B_d.
\end{align}
Substituting this bound into the Lagrange form of the remainder
\eqref{eq:remainder-gamma} yields, for all $u\in\mathbb{R}^d$,
\begin{align}
\label{eq:remainder-bound-gamma}
    |R_\ell(u)| \le \frac{B_d}{6}\sum_{s,t,w}|u_su_tu_w| \le C_d\,\|u\|_2^3,
\end{align}
where $C_d$ depends only on $d$ (e.g.,  $C_d=\frac{B_d}{6}\,d^{3/2}$ using $\sum_s|u_s|\le\sqrt d\,\|u\|_2$).

\paragraph{Step 3: Centering.}
Let $\bar{u}\triangleq \frac{1}{d}\sum_r u_r$ and $\tilde{u} \triangleq  u-\bar{u} \mathbf{1}$, where $\mathbf{1} \in \mathbb{R}^d$ is the all-ones vector. Then,  $f_\ell(u)=f_\ell(\tilde{u})$ and $\sum_\ell \tilde{u}_\ell=0$. Plugging \eqref{eq:softmax-grad-0-gamma}–\eqref{eq:softmax-hess-0-gamma} into \eqref{eq:taylor-gamma} with $u=\tilde{u}$ gives
\begin{align}
\label{eq:softmax-centered-gamma}
    f_\ell(\tilde{u}) = \frac{1}{d} + \frac{1}{d}\tilde{u}_\ell + \frac{1}{2d}\p{\tilde{u}_\ell^2 - \overline{\tilde{u}^{ 2}}} + R_\ell(\tilde{u}),
    \qquad \overline{\tilde{u}^{ 2}} \triangleq  \frac{1}{d}\sum_s \tilde{u}_s^2 .
\end{align}

\paragraph{Step 4: Substitute MRA logits and bound remainders.}
In the MRA, the Taylor expansion in~\eqref{eq:taylor-gamma} is applied to,
\begin{align}
    u \triangleq s(Y), \qquad \tilde{u} \triangleq \eta(Y), \label{eqn:app-C21}
\end{align}
where 
\begin{align}
    u_\ell = s_\ell (Y) = \frac{\beta}{\sigma^2} \langle Y,\mathcal{T}_\ell v\rangle, \qquad \tilde{u}_\ell \triangleq \eta_\ell (Y) = s_\ell (Y) - \overline{s} (Y), \label{eqn:app-C22}
\end{align}
for $\ell \in \{0 ,1, \ldots, d-1\}$. Substituting~\eqref{eqn:app-C21} into~\eqref{eq:softmax-centered-gamma} gives
\begin{align}
    f_\ell(\eta(Y)) = \frac{1}{d} + \frac{1}{d}\eta_\ell(Y) + \frac{1}{2d}\p{\eta_\ell(Y)^2 - \overline{\eta^{ 2}(Y)}} + R_\ell(\eta(Y)),
\end{align}
which is~\eqref{eq:MRA-gamma-eta}. It remains to bound the remainder $R_\ell(\eta(Y))$. To that end, applying Cauchy–Schwarz inequality we get
\begin{align}
    \|u\|_2 = \big\|s(Y)\big\|_2 &= \frac{|\beta|}{\sigma^2}\p{\sum_{\ell=0}^{d-1}\langle Y,\mathcal{T}_\ell v\rangle^2}^{1/2} \leq  \frac{|\beta|}{\sigma^2} \sqrt{d} \|Y\|_2 \|v\|_2. \label{eqn:app-C24}
\end{align}
Hence, using the definition of $\eta(Y)$ in \eqref{eqn:app-C21}--~\eqref{eqn:app-C22}, we have $\|\eta(Y)\|_2\le \|s(Y)\|_2$. Together with the bound in~\eqref{eqn:app-C24}, we get
\begin{align}
    \|\eta(Y)\|_2 \le \|s(Y)\|_2 \le  \frac{|\beta|}{\sigma^2}\sqrt d\|Y\|_2\|v\|_2.
\end{align}
Since the remainder estimate \eqref{eq:remainder-bound-gamma} holds for all $u\in\mathbb{R}^d$, we obtain
\begin{align}
\label{eq:rem-final-gamma}
    |R_\ell(\eta(Y))| \le C_d \|\eta(Y)\|_2^3 \le C_d\, d^{3/2}\frac{|\beta|^3}{\sigma^6}\|Y\|_2^3\|v\|_2^3 \triangleq  \bar C_d\frac{|\beta|^3}{\sigma^6}\|Y\|_2^3\|v\|_2^3,
\end{align}
where $\bar C_d=C_d d^{3/2}$ depends on $d$ only.
\end{proof}

\subsection{EM Jacobian for MRA: second--order expansion}

In this section we derive the second--order expansion of the population EM Jacobian using the second-order expansion of the EM responsibilities, derived in Lemma~\ref{lem:resp-near-uniform-MRA-gamma}. 

\paragraph{Notations and assumptions.}
Let $\mathbf{1}\in\mathbb{R}^d$ be the all-ones vector and define the (rank-one) orthogonal projector onto the mean subspace by
\begin{align}
    \Pi_{\mathrm{mean}} = \frac{1}{d} \mathbf{1}\mathbf{1}^\top,
    \qquad \Pi_{\mathrm{mean}} x = \frac{1}{d}(\mathbf{1}^\top x) \mathbf{1},
\end{align}
and let $P_0$ be the projections on the complement subspace of the mean subspace, namely,
\begin{align}
    P_0\triangleq I-\Pi_{\mathrm{mean}}.
\end{align}
We first state the following lemma, which gives the expansion of the Jacobian up to order $O(\beta^4/\sigma^4)$.
\begin{lem}[Second-order expansion of the EM Jacobian]
\label{lem:J-up-to-Delta2-correct}
Let $v$ be the normalized parameter vector $x^\star = \beta v \in \mathbb{R}^d$. Recall the EM Jacobian~\eqref{eq:Jacobian-MRA-proof}.
Then, the EM Jacobian admits the expansion
\begin{align}
    \label{eq:J-second-order-main}
    J(\beta) = J(0) + \frac{\beta^2}{\sigma^2}\p{K_{2,\mathrm{sig}}+K_{2,\mathrm{resp}}} + O \p{\frac{\beta^4}{\sigma^4}},
\end{align}
with $J(0)=I-\Pi_{\mathrm{mean}}$ and
\begin{align}
    \label{eq:K2sig}
    K_{2,\mathrm{sig}} =\frac{1}{d}\sum_{m=0}^{d-1} (\mathcal{T}_m v)(\mathcal{T}_m v)^\top -(\Pi_{\mathrm{mean}}v)(\Pi_{\mathrm{mean}}v)^\top,
\end{align}
and, 
\begin{align}\label{eq:K2resp}
    K_{2,\mathrm{resp}} = \mathbb{E}_\xi \left[\frac{1}{2d}\sum_{\ell=0}^{d-1} \p{(\eta_\ell^{(0)})^2-\overline{(\eta^{(0)})^2}} \p{ \mathcal{T}_\ell^{-1}\xi}\p{\mathcal{T}_\ell^{-1}\xi}^\top - \p{\overline{\eta^{(0)}  \mathcal{T}^{-1}\xi}}\p{\overline{\eta^{(0)}  \mathcal{T}^{-1}\xi}}^\top \right].
\end{align}
In particular, all odd (in $\beta$) contributions cancel by Gaussian parity, and the remainder is quartic:
\begin{align}
    \big\|J(\beta)-J(0)-\frac{\beta^2}{\sigma^2} (K_{2,\mathrm{sig}}+K_{2,\mathrm{resp}})\big\|_F \leq C \frac{\beta^4}{\sigma^4} \|v\|^4,
\end{align}
for some constant $C$, which depends on $d$, but is independent of $(\beta,\sigma,v)$.
\end{lem}

\begin{proof}[Proof of Lemma~\ref{lem:J-up-to-Delta2-correct}]
Recall that conditioned on the latent shift $S\sim \mathrm{Unif} \{0,1, \ldots d-1 \}$, we have
\begin{align}
    Y = \beta \mathcal{T}_S v + \xi, \qquad \xi \sim \mathcal{N}(0,\sigma^2 I_d),
\end{align}
and 
\begin{align}
    z_\ell = \mathcal{T}_\ell^{-1}Y = \mathcal{T}_\ell^{-1}\xi + \beta \mathcal{T}_\ell^{-1}\mathcal{T}_S v,
    \qquad \eta_\ell(Y) = \beta \eta_\ell^{(0)} + O(\beta^2),
\end{align}
with $\eta_\ell^{(0)}$ begin linear in $\xi$. Furthermore, recall the definition of $\gamma_\ell^{(\beta)}$ in ~\eqref{eq:gamma-def}. We start with the covariance representation of the population Jacobian given in Lemma~\ref{lem:Jacobian-M},
\begin{align}
    J(\beta) &=\frac{1}{\sigma^2}  \mathbb{E} \pp{\sum_{\ell=0}^{d-1}\gamma_\ell^{(\beta)}(Y) z_\ell z_\ell^\top -\p{\sum_{\ell=0}^{d-1}\gamma_\ell^{(\beta)}(Y) z_\ell}\p{\sum_{r=0}^{d-1}\gamma_r^{(\beta)}(Y) z_r}^\top }.
    \label{eq:J-as-covar-rep}
\end{align}
By Lemma~\ref{lem:resp-near-uniform-MRA-gamma}, the responsibilities admit the second-order softmax expansion,
\begin{align}
    \gamma_\ell^{(\beta)}(Y) &=\frac{1}{d} +\frac{1}{d} \eta_\ell(Y) +\frac{1}{2d}\p{\eta_\ell(Y)^2-\overline{\eta^2}(Y)} +R_\ell(\beta;Y),
    \label{eq:resp-expansion}
\end{align}
with the remainder bounded by
\begin{align}
    |R_\ell(\beta;Y)| \leq C \|\eta(Y)\|_\infty^3 \leq  C \frac{|\beta|^3}{\sigma^6} \|Y\|_2^3\|v\|_2^3 .
    \label{eq:resp-remainder}
\end{align}
We now insert \eqref{eq:resp-expansion} into \eqref{eq:J-as-covar-rep} and retain terms up to order $O(\beta^4)$.

\paragraph{Step 1: $\beta=0$.}
Setting $\beta=0$ gives $Y=\xi\sim\mathcal{N}(0,\sigma^2 I_d)$, $\gamma_\ell^{(0)}= 1/d$.
By using the orthonormality of the shifts $\{ \mathcal{T}_\ell \}_{\ell = 0}^{d-1}$ and $\mathbb{E}[\xi\xi^\top]=\sigma^2 I$ we obtain
\begin{align}
    \frac{1}{\sigma^2} \mathbb{E} \pp{\frac{1}{d}\sum_{\ell=0}^{d-1} (\mathcal{T}_\ell^{-1}\xi) (\mathcal{T}_\ell^{-1}\xi)^\top} &= I,
    \label{eq:baseline-second}
\end{align}
and, using $\Pi_{\mathrm{mean}}^2=\Pi_{\mathrm{mean}}$, we get
\begin{align}
    \frac{1}{\sigma^2} \mathbb{E} \pp{\p{\frac{1}{d}\sum_{\ell=0}^{d-1} \mathcal{T}_\ell^{-1}\xi}\p{\frac{1}{d}\sum_{r=0}^{d-1} \mathcal{T}_r^{-1}\xi}^\top} &= \Pi_{\mathrm{mean}}.
    \label{eq:baseline-first}
\end{align}
Subtracting \eqref{eq:baseline-first} from \eqref{eq:baseline-second} in \eqref{eq:J-as-covar-rep} yields
\begin{align}
    J(0)=I-\Pi_{\mathrm{mean}}.
    \label{eq:J0}
\end{align}

\paragraph{Step 2: Odd orders in $\beta$ vanish.}
Any contribution to $J(\beta)$ of odd order in $\beta$ can be expressed (after expanding $\gamma_\ell^{(\beta)}(Y)$ and $z_\ell$ in powers of $\beta$) as a finite linear combination of conditional expectations of the form $\mathbb{E}\left[P(\xi)\,\middle|\,S\right]$, where $P$ is a polynomial in the noise vector $\xi$ of odd total degree (with coefficients depending on $v$ and the shifts, and possibly on $S$). Conditioned on $S$, the random vector $\xi\sim\mathcal{N}(0,\sigma^2 I_d)$ remains a centered Gaussian, hence $\mathbb{E}\left[P(\xi)\,\middle|\,S\right]=0$ for every such odd polynomial $P$ by Gaussian symmetry. Averaging over $S$ preserves zero, so all odd-order terms in the expansion of $J(\beta)$ vanish.

\paragraph{Step 3: Organizing the $O(\beta^2)$ terms.}
Keeping only the $O(\beta^2)$ pieces produces two non-vanishing contributions:

\begin{enumerate}
    \item \emph{Responsibility perturbation with fixed signal amplitude.}
    Use $\eta_\ell=\beta\eta_\ell^{(0)}+O(\beta^2)$ in \eqref{eq:resp-expansion} and set $z_\ell=\mathcal{T}_\ell^{-1}\xi$.
    \begin{align}
        \frac{1}{\sigma^2} &         \mathbb{E} \left[\frac{1}{2d}\sum_{\ell=0}^{d-1}\p{(\beta\eta_\ell^{(0)})^2-\overline{(\beta\eta^{(0)})^2}} \p{ \mathcal{T}_\ell^{-1}\xi}\p{ \mathcal{T}_\ell^{-1}\xi}^\top - \p{\overline{\beta\eta^{(0)}  \mathcal{T}^{-1}\xi}}\p{\overline{\beta\eta^{(0)}  \mathcal{T}^{-1}\xi}}^\top \right]
      \nonumber  \\ &  = \frac{\beta^2}{\sigma^2} K_{2,\mathrm{resp}} .
        \label{eq:resp-block}
    \end{align}

    \item \emph{Pure–signal shift under uniform responsibilities.}
    With $\gamma_\ell^{(0)} = 1/d$ and $z_\ell=\mathcal{T}_\ell^{-1}\xi+\beta\mathcal{T}_\ell^{-1} \mathcal{T}_S v$,
    \begin{align}
        \frac{1}{\sigma^2} &        \mathbb{E} \left[\frac{1}{d}\sum_{\ell=0}^{d-1}(\beta\mathcal{T}_\ell^{-1} \mathcal{T}_S v)(\beta\mathcal{T}_\ell^{-1} \mathcal{T}_S v)^\top - \p{\frac{1}{d}\sum_{\ell=0}^{d-1}\beta\mathcal{T}_\ell^{-1} \mathcal{T}_S v}\p{\frac{1}{d}\sum_{r=0}^{d-1}\beta\mathcal{T}_r^{-1} \mathcal{T}_S v}^\top\right]
       \nonumber \\ & \qquad   =  \frac{\beta^2}{\sigma^2} K_{2,\mathrm{sig}},
        \label{eq:sig-block}
    \end{align}
    and averaging over the uniform random shift $S\sim \mathrm{Unif} \{0,1, \ldots d-1 \}$ yields the closed form for $K_{2,\mathrm{sig}}$ stated in \eqref{eq:K2sig}.

\end{enumerate}

Combining \eqref{eq:resp-block}–\eqref{eq:sig-block} with the baseline \eqref{eq:J0} gives
\begin{align}
    J(\beta) = J(0) + \frac{\beta^2}{\sigma^2}\p{ K_{2,\mathrm{sig}} + K_{2,\mathrm{resp}}} + R(\beta, v).
    \label{eq:J-quad}
\end{align}

\paragraph{Step 4: Quartic remainder.}
Step (2) implies that all cubic terms vanish by Gaussian parity. We now bound $R(\beta,v)$ at order $\beta^4$. From the softmax expansion in Lemma~\ref{lem:resp-near-uniform-MRA-gamma} and the decomposition $z_\ell = \mathcal{T}_\ell^{-1}\xi + \beta \mathcal{T}_\ell^{-1}\mathcal{T}_S v$, every term contributing to $R(\beta,v)$ is of the form
\begin{align}
    \frac{\beta^4}{\sigma^8}\mathbb{E}\pp{ Q(\xi,v,S)},
\end{align}
where $Q$ is a polynomial in $(\xi,v)$ of total degree at most $4$ in $\xi$ and at most $4$ in $v$.
In particular, every matrix entry of $R(\beta,v)$ is a linear combination of expectations of products of the form $\mathbb{E}\pp{\xi_{i_1}\xi_{i_2}\xi_{i_3}\xi_{i_4}}$. By Isserlis' theorem for centered Gaussian vectors, there exists a constant $C_0=C_0(d)>0$, depending on $d$ only, such that
\begin{align}
    \abs{\mathbb{E}\pp{\xi_{i_1}\xi_{i_2}\xi_{i_3}\xi_{i_4}}}\le C_0\sigma^4,
\end{align}
for all quadruples $(i_1,i_2,i_3,i_4)$. Moreover, every coefficient multiplying such moment terms is a polynomial in $v$ of degree at most $4$, and is therefore bounded by $C_1 \|v\|_2^4$ for some $C_1=C_1(d)>0$.

Combining these bounds and summing over the finitely many index patterns, we obtain an entrywise estimate
\begin{align}
    \big|[R(\beta,v)]_{ij}\big| \le C'\,\frac{\beta^4}{\sigma^4}\,\|v\|_2^4, \qquad 1\le i,j\le d,
\end{align}
for a constant $C'>0$ depending only on $d$.
Since the Frobenius norm satisfies
\begin{align}
    \|R(\beta,v)\|_F^2 = \sum_{i,j} [R(\beta,v)]_{ij}^2 \le d^2 \p{\max_{i,j} |[R(\beta,v)]_{ij}|}^2,
\end{align}
we conclude that there exists a constant $C>0$, depending on $d$ only, such that
\begin{align}
    \|R(\beta,v)\|_F = \big\|J(\beta)-J(0)-\frac{\beta^2}{\sigma^2}\p{K_{2,\mathrm{sig}}+K_{2,\mathrm{resp}}}\big\|_F \le C\,\frac{\beta^4}{\sigma^4}\,\|v\|_2^{4}.
\end{align}
This proves the remainder bound in \eqref{eq:J-second-order-main}.
\end{proof}

Next, we derive a closed-form for the quadratic coefficient $K_{2,\mathrm{resp}}$.

\begin{lem}[Closed form of the second–order term $K_{2,\mathrm{resp}}$]
\label{lem:Delta2-closed}
Let $v\in\mathbb{R}^d$ and decompose $v=\bar{v}+\tilde{v}$ with $\bar{v}=\Pi_{\mathrm{mean}}v$ and $\tilde{v}=P_0v=v-\bar{v}$.
Then the coefficient $K_{2,\mathrm{resp}}$ defined in Lemma~\ref{lem:J-up-to-Delta2-correct} admits the explicit representation
\begin{align}
    K_{2,\mathrm{resp}} = -\frac{1}{d}\left[\sum_{s=0}^{d-1}\mathcal{T}_s \tilde{v} \tilde{v}^\top \mathcal{T}_s^\top + \sum_{s=0}^{d-1}(\mathcal{T}_s\tilde{v}) \tilde{v}^\top \mathcal{T}_s + \sum_{s=0}^{d-1}\big\langle \tilde{v},\mathcal{T}_s \tilde{v}\big\rangle \mathcal{T}_s \right].
    \label{eq:K2resp-closed}
\end{align}
\end{lem}

\begin{proof}[Proof of Lemma~\ref{lem:Delta2-closed}]
Recall the definition from Lemma~\ref{lem:J-up-to-Delta2-correct} (using the overline conventions):
\begin{align}
    K_{2,\mathrm{resp}} = \mathbb{E}_\xi \left[\frac{1}{2d}\sum_{\ell=0}^{d-1} \p{(\eta_\ell^{(0)})^2-\overline{(\eta^{(0)})^2}} \p{ \mathcal{T}_\ell^{-1}\xi}\p{ \mathcal{T}_\ell^{-1}\xi}^\top - \p{\overline{\eta^{(0)}  \mathcal{T}^{-1}\xi}}\p{\overline{\eta^{(0)} \mathcal{T}^{-1}\xi}}^\top \right],
    \label{eq:K2resp-def}
\end{align}
where
\begin{align}
    \eta_\ell^{(0)} &= \frac{1}{\sigma^2}\p{\langle \xi,\mathcal{T}_\ell v\rangle-\overline{\langle \xi,\mathcal{T} v\rangle}} = \frac{1}{\sigma^2} \big\langle \xi,\mathcal{T}_\ell \tilde{v}\big\rangle .
\end{align}
Introduce the standardized noise $g_\ell\triangleq \sigma^{-1} \mathcal{T}_\ell^{-1}\xi\sim\mathcal{N}(0,I_d)$. Then
\begin{align}
    \eta_\ell^{(0)}=\frac{1}{\sigma} \big\langle g_\ell,\tilde{v}\big\rangle,
    \qquad \mathcal{T}_\ell^{-1}\xi = \sigma g_\ell.
\end{align}
Substituting these into \eqref{eq:K2resp-def}, all factors of $\sigma$ cancel, and $K_{2,\mathrm{resp}}$ reduces to expectations of quadratic and bilinear Gaussian forms in $\{g_\ell\}$.

Define the first term in~\eqref{eq:K2resp-def},
\begin{align}
    \mathsf{A} & \triangleq \mathbb{E} \left[\frac{1}{2d}\sum_{\ell=0}^{d-1}\p{\langle g_\ell,\tilde{v}\rangle^2-\overline{\langle g,\tilde{v}\rangle^2}} g_\ell g_\ell^\top\right] \notag\\
    &=\frac{1}{2d}\sum_{\ell=0}^{d-1}\left\{\mathbb{E} \pp{g_\ell g_\ell^\top \langle g_\ell,\tilde{v}\rangle^2} -\mathbb{E} \pp{g_\ell g_\ell^\top  \overline{\langle g,\tilde{v}\rangle^2}}\right\}.
    \label{eq:A-split}
\end{align}
For each fixed $\ell$, using the shift relation $g_r \stackrel{\mathcal{D}}{=} \mathcal{T}_{r-\ell}^{-1}g_\ell$ we have
\begin{align}
    \overline{\langle g,\tilde{v}\rangle^2} = \frac{1}{d}\sum_{r=0}^{d-1}\langle g_r,\tilde{v}\rangle^2 \stackrel{\mathcal{D}}{=}\frac{1}{d}\sum_{s=0}^{d-1}\langle g_\ell,\mathcal{T}_s \tilde{v}\rangle^2,
\end{align}
hence
\begin{align}
\mathbb{E} \pp{g_\ell g_\ell^\top  \overline{\langle g,\tilde{v}\rangle^2}}
=\frac{1}{d}\sum_{s=0}^{d-1}\mathbb{E} \pp{g_\ell g_\ell^\top \langle g_\ell,\mathcal{T}_s \tilde{v}\rangle^2}.
\label{eq:centered-to-s}
\end{align}
By Isserlis's theorem, for any $u\in\mathbb{R}^d$ and $g\sim\mathcal{N}(0,I_d)$,
\begin{align}
    \mathbb{E} \pp{g g^\top \langle g,u\rangle^2}=2uu^\top+\|u\|^2I,
    \qquad \mathbb{E}[g g^\top]=I.
\end{align}
Applying this with $u=\tilde{v}$ and $u=\mathcal{T}_s\tilde{v}$ yields
\begin{align}
    \mathbb{E} \pp{g_\ell g_\ell^\top \langle g_\ell,\tilde{v}\rangle^2} &=2\tilde{v}\tilde{v}^\top+\|\tilde{v}\|^2 I, \label{eq:wick-first}\\
    \mathbb{E} \pp{g_\ell g_\ell^\top \langle g_\ell,\mathcal{T}_s\tilde{v}\rangle^2} &=2(\mathcal{T}_s\tilde{v})(\mathcal{T}_s\tilde{v})^\top+\|\tilde{v}\|^2 I. \label{eq:wick-second}
\end{align}
Plugging \eqref{eq:centered-to-s}–\eqref{eq:wick-second} into \eqref{eq:A-split},
\begin{align}
    \mathsf{A} &\triangleq\frac{1}{2}\p{2\tilde{v}\tilde{v}^\top+\|\tilde{v}\|^2 I} -\frac{1}{2d}\sum_{s=0}^{d-1}\p{2(\mathcal{T}_s\tilde{v})(\mathcal{T}_s\tilde{v})^\top+\|\tilde{v}\|^2 I} \notag\\    &=\tilde{v}\tilde{v}^\top+\frac{\|\tilde{v}\|^2}{2}I -\frac{1}{d}\sum_{s=0}^{d-1}(\mathcal{T}_s\tilde{v})(\mathcal{T}_s\tilde{v})^\top -\frac{\|\tilde{v}\|^2}{2}I \notag\\
    &=\tilde{v}\tilde{v}^\top-\frac{1}{d}\sum_{s=0}^{d-1}\mathcal{T}_s \tilde{v}\tilde{v}^\top \mathcal{T}_s^\top.
\end{align}
Next, define the second term in~\eqref{eq:K2resp-def},
\begin{align}
    \mathsf{B} &\triangleq \mathbb{E} \left[\left(\frac{1}{d}\sum_{\ell=0}^{d-1}\langle g_\ell,\tilde{v}\rangle g_\ell\right) \left(\frac{1}{d}\sum_{r=0}^{d-1}\langle g_r,\tilde{v}\rangle g_r\right)^\top\right]
    \notag\\
    &=
    \frac{1}{d^2}\sum_{\ell=0}^{d-1}\sum_{r=0}^{d-1} \mathbb{E} \pp{ \langle g_\ell,\tilde{v}\rangle \langle g_r,\tilde{v}\rangle  g_\ell g_r^\top }.
    \label{eq:B-double-sum}
\end{align}
Fix $\ell$ and set $s\equiv r-\ell \ (\mathrm{mod}\ d)$.
By circular–shift invariance, $(g_\ell,g_r)\stackrel{\mathcal{D}}{=}(g_0,g_s)$ with $g_s=\mathcal{T}_s^{-1}g_0$.
Hence $\langle g_r,\tilde{v}\rangle=\langle g_0,\mathcal{T}_s\tilde{v}\rangle$, and $g_r=\mathcal{T}_s^{-1}g_0$.
Since the summand in \eqref{eq:B-double-sum} no longer depends on $\ell$, summing over $\ell$ yields a factor $d$:
\begin{align}
    \mathsf{B}
    &=
    \frac{1}{d}\sum_{s=0}^{d-1} \mathbb{E} \pp{ \langle g_0,\tilde{v}\rangle \langle g_0,\mathcal{T}_s\tilde{v}\rangle  g_0(\mathcal{T}_s^{-1}g_0)^\top }
    \notag\\
    &=
    \frac{1}{d}\sum_{s=0}^{d-1} \mathbb{E} \pp{ \langle g,\tilde{v}\rangle \langle g,\mathcal{T}_s\tilde{v}\rangle  g g^\top } \mathcal{T}_s,
    \label{eq:B-reduced}
\end{align}
where in the last step we used $(\mathcal{T}_s^{-1}g)^\top = g^\top \mathcal{T}_s$ and wrote $g\equiv g_0\sim\mathcal{N}(0,I_d)$.
By Isserlis's theorem, for $g\sim\mathcal{N}(0,I_d)$ and any $u,v\in\mathbb{R}^d$,
\begin{align}
    \mathbb{E} \pp{g g^\top  \langle g,u\rangle \langle g,v\rangle} &= \mathbb{E} \pp{(g g^\top)(g^\top u)(g^\top v)} = u v^\top + v u^\top + \langle u,v\rangle I.
    \label{eq:wick-bilinear}
\end{align}
From \eqref{eq:B-reduced}–\eqref{eq:wick-bilinear},
\begin{align}
    \mathsf{B} &= \frac{1}{d}\sum_{s=0}^{d-1} \p{\tilde{v}(\mathcal{T}_s\tilde{v})^\top + (\mathcal{T}_s\tilde{v})\tilde{v}^\top + \langle \tilde{v},\mathcal{T}_s\tilde{v}\rangle I} \mathcal{T}_s.
    \label{eq:B-after-wick}
\end{align}
Since $(\mathcal{T}_s\tilde{v})^\top \mathcal{T}_s=\tilde{v}^\top$, we can rewrite the first term as:
\begin{align}
    \tilde{v}(\mathcal{T}_s\tilde{v})^\top \mathcal{T}_s=\tilde{v} \tilde{v}^\top.
\end{align}
Thus \eqref{eq:B-after-wick} becomes
\begin{align}
    \mathsf{B} &= \frac{1}{d}\sum_{s=0}^{d-1} \p{\tilde{v}\tilde{v}^\top  +  (\mathcal{T}_s\tilde{v}) \tilde{v}^\top \mathcal{T}_s + \langle \tilde{v},\mathcal{T}_s\tilde{v}\rangle \mathcal{T}_s}.
\end{align}
Finally, we combine back the terms $\mathsf{A} - \mathsf{B}$. Since $\frac{1}{d}\sum_{s=0}^{d-1}\tilde{v}\tilde{v}^\top=\tilde{v}\tilde{v}^\top$, the $\tilde{v}\tilde{v}^\top$ terms cancel between $\mathsf{A}$ and $\mathsf{B}$, giving
\begin{align}
    K_{2,\mathrm{resp}} = \mathsf{A}-\mathsf{B} = -\frac{1}{d}\left[\sum_{s=0}^{d-1}\mathcal{T}_s \tilde{v} \tilde{v}^\top \mathcal{T}_s^\top + \sum_{s=0}^{d-1}(\mathcal{T}_s\tilde{v}) \tilde{v}^\top \mathcal{T}_s + \sum_{s=0}^{d-1}\langle \tilde{v},\mathcal{T}_s \tilde{v}\rangle \mathcal{T}_s\right],
\end{align}
which is precisely \eqref{eq:K2resp-closed}.
\end{proof}

\paragraph{Combining $K_{2,\mathrm{sig}}$ and $K_{2,\mathrm{resp}}$ into the Jacobian expansion $J(\beta)$.} 
After deriving $K_{2,\mathrm{resp}}$ in Lemma~\ref{lem:Delta2-closed}, now we combine it with $K_{2,\mathrm{sig}}$. Recall from~\eqref{eq:K2sig},
\begin{align}
    K_{2,\mathrm{sig}} &= \frac{1}{d}\sum_{s=0}^{d-1} (\mathcal{T}_s v)(\mathcal{T}_s v)^\top - (\Pi_{\mathrm{mean}}v)(\Pi_{\mathrm{mean}}v)^\top,
\end{align}
and write $v=\bar{v}+\tilde{v}$ with $\bar{v}=\Pi_{\mathrm{mean}}v$ and $\tilde{v}=P_0 v$.
Since $\mathcal{T}_s\bar{v}=\bar{v}$ and $\frac{1}{d}\sum_s \mathcal{T}_s \tilde{v} = \Pi_{\mathrm{mean}}\tilde{v}=0$, we have
\begin{align}
    K_{2,\mathrm{sig}} = \frac{1}{d}\sum_{s=0}^{d-1} \mathcal{T}_s \tilde{v} \tilde{v}^\top \mathcal{T}_s^\top. \label{eqn:K2-sig-simplified}
\end{align}
From Lemma~\ref{lem:Delta2-closed},
\begin{align}
    K_{2,\mathrm{resp}} = -\frac{1}{d}\left[ \sum_{s=0}^{d-1}\mathcal{T}_s \tilde{v} \tilde{v}^\top \mathcal{T}_s^\top + \sum_{s=0}^{d-1}(\mathcal{T}_s\tilde{v}) \tilde{v}^\top \mathcal{T}_s + \sum_{s=0}^{d-1}\big\langle \tilde{v},\mathcal{T}_s \tilde{v}\big\rangle \mathcal{T}_s \right]. \label{K2-resp-simplified}
\end{align}
Adding~\eqref{eqn:K2-sig-simplified} and~\eqref{K2-resp-simplified}, and noting that the first sums cancel, we obtain the combined quadratic coefficient
\begin{align}
    K_{2,\mathrm{sig}} + K_{2,\mathrm{resp}} = -\frac{1}{d}\left[\sum_{s=0}^{d-1}(\mathcal{T}_s\tilde{v}) \tilde{v}^\top \mathcal{T}_s + \sum_{s=0}^{d-1}\big\langle \tilde{v},\mathcal{T}_s \tilde{v}\big\rangle \mathcal{T}_s\right]. \label{eqn:explcit-sum-of_Ks}
\end{align}
Then, combining~\eqref{eqn:explcit-sum-of_Ks} with~\eqref{eq:J-second-order-main} results,
\begin{align}
    J(\beta) &= J(0) + \frac{\beta^2}{\sigma^2}\p{K_{2,\mathrm{sig}}+K_{2,\mathrm{resp}}} + O \p{\frac{\beta^4}{\sigma^4}}
    \\
    &= (I-\Pi_{\mathrm{mean}}) - \frac{\beta^2}{\sigma^2} \frac{1}{d} \left[\sum_{s=0}^{d-1}\p{\mathcal{T}_s \tilde{v}} \tilde{v}^\top \mathcal{T}_s + \sum_{s=0}^{d-1}\big\langle \tilde{v},\mathcal{T}_s \tilde{v}\big\rangle \mathcal{T}_s \right] + O \p{\frac{\beta^4}{\sigma^4}},
\end{align}
proving~\eqref{eqn:Jacobian-spatial-domain-exp}.

\subsection{Proof of Proposition~\ref{prop:K2-MRA}}
In the previous subsection we showed that the real-space expansion of the Jacobian is
\begin{align}\label{eq:J-real-again}
    J(\beta) = (I-\Pi_{\mathrm{mean}}) - \frac{\beta^2}{\sigma^2} \frac{1}{d}  \left[\sum_{s=0}^{d-1}\p{\mathcal{T}_s \tilde{v}} \tilde{v}^\top \mathcal{T}_s + \sum_{s=0}^{d-1}\big\langle \tilde{v},\mathcal{T}_s \tilde{v}\big\rangle \mathcal{T}_s\right] + O \p{\frac{\beta^4}{\sigma^4}}.
\end{align}
Define
\begin{align}
    S_0 \triangleq \frac{1}{d}\sum_{s=0}^{d-1} (\mathcal{T}_s \tilde{v}) \tilde{v}^{\top}\mathcal{T}_s, \label{eqn:app-C72}
    \\
    S_1 \triangleq \frac{1}{d}\sum_{s=0}^{d-1}\langle \tilde{v},\mathcal{T}_s \tilde{v}\rangle \mathcal{T}_s, \label{eqn:app-C73}
\end{align}
so that the order-$\beta^2$ term in \eqref{eq:J-real-again} equals $-\frac{\beta^2}{\sigma^2}(S_0+S_1)$.

Let $F\in\mathbb{C}^{d\times d}$ be the unitary DFT matrix with columns $\{f_k\}_{k=0}^{d-1}$, as defined in~\eqref{eqn:DFT-col}, so $F^\ast F=FF^\ast=I$ and $f_0=d^{-1/2}\mathbf{1}$. Write $\s{V} \triangleq F^\ast v \in \mathbb{C}^d$ and $\s{\tilde{V}} \triangleq F^\ast \tilde{v}$; then $\s{V}_k = \s{\tilde{V}}_k$ for every $k \neq 0$, and $\s{\tilde{V}}_0 = 0$. For a vector $u$, let $\diag(u)$ denote the diagonal matrix with the entries of $u$ on its diagonal. Then, the circular shifts $\{\mathcal{T}_s\}_{s=0}^{d-1}$ diagonalize in the DFT basis:
\begin{align}\label{eq:shift-diag}
    F^\ast \mathcal{T}_s F = \Lambda_s, \qquad \Lambda_s = \diag \p{e^{-2\pi isk/d}}_{k=0}^{d-1}.
\end{align}

\paragraph{Fourier image of $S_1$.}
Using \eqref{eq:shift-diag} and $F^\ast \tilde{v} = \s{\tilde{V}}$, we have
\begin{align}
    F^\ast S_1 F = \frac{1}{d}\sum_{s=0}^{d-1} \langle F^\ast\tilde{v}, F^\ast \mathcal{T}_s \tilde{v}\rangle  F^\ast \mathcal{T}_s F
    & = \frac{1}{d}\sum_{s=0}^{d-1} \langle \s{\tilde{V}}, F^\ast \mathcal{T}_s F F^\ast \tilde{v} \rangle \Lambda_s 
    \\ & = 
    \frac{1}{d}\sum_{s=0}^{d-1} \langle \s{\tilde{V}}, \Lambda_s \s{\tilde{V}}\rangle \Lambda_s. \label{eqn:app-C75}
\end{align}
Now $\langle \s{\tilde{V}},\Lambda_s \s{\tilde{V}}\rangle=\sum_{m=0}^{d-1} |\s{\tilde{V}}_m|^2 e^{-2\pi i sm/d}$. Thus the $(k,k)$ entry of $F^\ast S_1 F$ is
\begin{align}
    \frac{1}{d}\sum_{s=0}^{d-1} \sum_{m=0}^{d-1} |\s{\tilde{V}}_m|^2 e^{-2\pi i s(m+k)/d} = \sum_{m=0}^{d-1} |\s{\tilde{V}}_m|^2 \underbrace{\p{\frac{1}{d}\sum_{s=0}^{d-1}e^{-2\pi is(m+k)/d}}}_{= \mathbf{1}\{m+k\equiv 0 \pmod d\}} = |\s{\tilde{V}}_{-k}|^2. \label{eqn:app-C76}
\end{align}
Since $v$ is real and $|\s{\tilde{V}}_{-k}|=|\s{\tilde{V}}_k|$, \eqref{eqn:app-C76} equals $|\s{\tilde{V}}_k|^2$. Off-diagonal entries (i.e., $(k_1, k_2)$, for $k_1 \neq k_2$) vanish by the same orthogonality identity.
Hence, combining~\eqref{eqn:app-C75}–\eqref{eqn:app-C76} yields
\begin{align}\label{eq:S3-hat}
    F^\ast S_1 F = \diag \p{|\s{\tilde{V}}|^2}.
\end{align}

\paragraph{Fourier image of $S_0$.}
By \eqref{eq:shift-diag},
\begin{align}
    F^\ast S_0 F &= \frac{1}{d}\sum_{s=0}^{d-1} \p{F^\ast \mathcal{T}_s \tilde{v}} \p{F^\ast \tilde{v}}^\ast  \p{F^\ast \mathcal{T}_s F} = \frac{1}{d}\sum_{s=0}^{d-1} (\Lambda_s \s{\tilde{V}}) \s{\tilde{V}}^\ast \Lambda_s.
\end{align}
Its $(k,\ell)$ entry equals
\begin{align}
    \frac{1}{d}\sum_{s=0}^{d-1} (\Lambda_s \s{\tilde{V}})_k  \overline{\s{\tilde{V}}_\ell}  (\Lambda_s)_{\ell\ell} = \frac{1}{d}\sum_{s=0}^{d-1} e^{-2\pi is(k+\ell)/d} \s{\tilde{V}}_k \overline{\s{\tilde{V}}_\ell} = \s{\tilde{V}}_k \overline{\s{\tilde{V}}_\ell} \, \mathbf{1}\{k+\ell\equiv 0 \pmod d\}. \label{eqn:app-C80}
\end{align}
Equivalently,
\begin{align} \label{eq:S2-hat}
    (F^{\ast} S_{0} F)_{k\ell} =
    \begin{cases}
        \s{\tilde{V}}_k\,\overline{\s{\tilde{V}}_\ell}, & k+\ell \equiv 0 \pmod d,\\ 0, & \text{otherwise},
    \end{cases}
\end{align}
so $S_0$ couples only the frequency pairs $\{k,-k\}$.

\paragraph{Fourier image of $J(\beta)$.}
Combining~\eqref{eq:S3-hat}–\eqref{eq:S2-hat} with \eqref{eq:J-real-again}, and using that $F^\ast (I-\Pi_{\mathrm{mean}})F$ is $1$ on for all $k \neq 0$ frequencies and $0$ for $k = 0$, yields
\begin{align}
    \widehat{J}_{kk} = 1-\frac{\beta^2}{\sigma^2}\,|\s{{V}}_k|^2 + O \left(\frac{\beta^4}{\sigma^4}\right) \triangleq a_k + O \left(\frac{\beta^4}{\sigma^4}\right),
\end{align}
for every $k \neq 0$. For $\ell\equiv -k \pmod d$ with $k\neq 0$, only $S_0$ contributes, giving
\begin{align}
    \widehat{J}_{k,-k} = -\frac{\beta^2}{\sigma^2}\,\s{V}_k^2 + O \left(\frac{\beta^4}{\sigma^4}\right) \triangleq b_k + O \left(\frac{\beta^4}{\sigma^4}\right),
\end{align}
and the symmetric entry satisfies
\begin{align}
    \widehat{J}_{-k,k} = \overline{b_k} + O \left(\frac{\beta^4}{\sigma^4}\right).
\end{align}
All other entries are zero to this order. This coincides with the block--pair matrix in \eqref{eq:Jhat_explicit_matrix_prop}.

\subsection{Proof of Corollary~\ref{cor:block-spectral}} \label{sec:proof-of-eigenstructure}
Recall that we assume that $d$ is odd. By~\eqref{eq:block-2x2}, $\widehat{J}(\beta)$ is block–diagonal with respect to
\begin{align}
    \mathbb{R}^d = \mathrm{span}\{f_0\}\ \oplus \bigoplus_{k=1}^{(d-1)/2}\mathrm{span}\{f_k,f_{-k}\},
\end{align}
where the Jacobian block $\{k,-k\}$, we have,
\begin{align}
    \widehat{J}(\beta)\big|_{\{k,-k\}} =  I_2 - \frac{\beta^2}{\sigma^2}
    \begin{bmatrix}
    |\s{V}_k|^2 & \s{V}_k^2\\[2pt]
    \overline{\s{V}_k^2} & |\s{V}_k|^2
    \end{bmatrix}
     + \widehat{R}^{(k)}(\beta), \label{eqn:app-C85}
\end{align}
where $\|\widehat{R}^{(k)}(\beta)\|_F=O (\frac{\beta^4}{\sigma^4})$. 

\paragraph{Eigen-structure of each non-mean component pair.}
Denote by $K_2$ the second order expansion matrix in~\eqref{eqn:app-C85}, that is,
\begin{align}
    K_2 \triangleq 
    \begin{bmatrix}
    |\s{V}_k|^2 & \s{V}_k^2\\[2pt]
    \overline{\s{V}_k^2} & |\s{V}_k|^2
    \end{bmatrix}
    = |\s{V}_k|^2
    \begin{bmatrix}
    1 & e^{2i\phi_k}\\[2pt]
    e^{-2i\phi_k} & 1
\end{bmatrix}.
\end{align}
where $\phi_k \triangleq \phi_{\s{V}}[k]$. Using the characteristic polynomial of $K_2$ yields,
\begin{align}
    \det(K_2-\lambda I) =\p{|\s{V}_k|^2-\lambda}^2-|\s{V}_k|^4 = \lambda\p{\lambda-2|\s{V}_k|^2},
\end{align}
so the eigenvalues are $\{0, 2|\s{V}_k|^2\}$. 
Next, we derive the eigenvectors. Set
\begin{align}
    u_k=\frac{1}{\sqrt2} \begin{bmatrix}e^{i\phi_k}\\ e^{-i\phi_k}\end{bmatrix},
    \qquad
    w_k=\frac{1}{i\sqrt2} \begin{bmatrix}e^{i\phi_k}\\ - e^{-i\phi_k}\end{bmatrix}.
\end{align}
They are orthonormal since
\begin{align}
    \|u_k\|^2=\frac{1}{2} \left(|e^{i\phi_k}|^2+|e^{-i\phi_k}|^2\right)=1,
    \quad \|w_k\|^2=\frac{1}{2} \left(|e^{i\phi_k}|^2+|e^{-i\phi_k}|^2\right)=1,
\end{align}
and
\begin{align}
    \langle u_k,w_k\rangle = \frac{1}{2i}\p{e^{-i\phi_k}e^{i\phi_k}-e^{i\phi_k}e^{-i\phi_k}}=0.
\end{align}
Using $\s{V}_k^2=|\s{V}_k|^2 e^{2i\phi_k}$ and $\overline{\s{V}_k^2}=|\s{V}_k|^2 e^{-2i\phi_k}$,
\begin{align}
    K_2 u_k
    &=\frac{1}{\sqrt2}
    \begin{bmatrix}
    |\s{V}_k|^2 e^{i\phi_k}+\s{V}_k^2 e^{-i\phi_k}\\[4pt]
    \overline{\s{V}_k^2} e^{i\phi_k}+|\s{V}_k|^2 e^{-i\phi_k}
    \end{bmatrix}
    =
    \frac{|\s{V}_k|^2}{\sqrt2}
    \begin{bmatrix}
    e^{i\phi_k}+e^{2i\phi_k}e^{-i\phi_k}\\[2pt]
    e^{-2i\phi_k}e^{i\phi_k}+e^{-i\phi_k}
    \end{bmatrix}\\
    &=
    \frac{|\s{V}_k|^2}{\sqrt2}
    \begin{bmatrix}
    2e^{i\phi_k}\\[2pt]
    2e^{-i\phi_k}
    \end{bmatrix}
    =2|\s{V}_k|^2 u_k.
\end{align}
Hence $\lambda_{u_k}=2|\s{V}_k|^2$. For the action on $w_k$, we have,
\begin{align}
    K_2 w_k
    &=\frac{1}{i\sqrt2}
    \begin{bmatrix}
    |\s{V}_k|^2 e^{i\phi_k}-\s{V}_k^2 e^{-i\phi_k}\\[4pt]
    \overline{\s{V}_k^2} e^{i\phi_k}-|\s{V}_k|^2 e^{-i\phi_k}
    \end{bmatrix}
    =
    \frac{|\s{V}_k|^2}{i\sqrt2}
    \begin{bmatrix}
    e^{i\phi_k}-e^{2i\phi_k}e^{-i\phi_k}\\[2pt]
    e^{-2i\phi_k}e^{i\phi_k}-e^{-i\phi_k}
    \end{bmatrix}
    = 0.
\end{align}
Hence $\lambda_{w_k}=0$.
Thus, $K_2$ is an Hermitian matrix with eigenvectors
\begin{align}
    u_k=\frac{1}{\sqrt{2}} \begin{bmatrix}e^{i\phi_k}\\ \  e^{-i\phi_k}\end{bmatrix},
    \qquad
    w_k=\frac{1}{i\sqrt{2}} \begin{bmatrix}e^{i\phi_k}\\ -e^{-i\phi_k}\end{bmatrix},
\end{align}
and corresponding eigenvalues,
\begin{align}
    \lambda_{u_k}=2|\s{V}_k|^2,
    \qquad \lambda_{w_k}=0.
\end{align}
Substituting $K_2$ back into~\eqref{eqn:app-C85} gives
\begin{align}
    F^\ast J(\beta)F u_k &   =  \p{1-2 \frac{\beta^2}{\sigma^2}|\s{V}_k|^2+O(\frac{\beta^4}{\sigma^4})}u_k
    \\ 
    F^\ast J(\beta)F w_k &   =  \p{1+O(\frac{\beta^4}{\sigma^4})}w_k,
\end{align}
where the remainder is given by using Weyl’s inequality applied to the Hermitian perturbation $\widehat{R}^{(k)}(\beta)$.

\paragraph{Mean component.}
By the Jacobian expansion representation in~\eqref{eq:Jhat_explicit_matrix_prop}, it is clear that the mean component $k = 0$ satisfies,
\begin{align}
    \widehat{J}(\beta) f_0 =  0 + O \p{\frac{\beta^4}{\sigma^4}},
\end{align}
since $(I-f_0 f_0^\ast)f_0=0$ and the quadratic term vanishes on $f_0$. In addition, Theorem~\ref{thm:mra-lowSNR-iteration-bias-to-init} implies that the mean component is recovered correctly after a single iteration, independently of the SNR, that is, the empirical Jacobian annihilates the mean direction $\widehat{J}(\beta) f_0 = 0$.

\paragraph{Spectral radius.}
Finally, the spectral radius follows because each non–mean $2\times2$ block has one eigenvalue
$1+O(\beta^4/\sigma^4)$, while the mean eigenvalue is $0+O(\beta^4/\sigma^4)$. Hence
\begin{align}
    \rho(J(\beta)) = 1 + O(\beta^4/\sigma^4).
\end{align}
Unitary conjugation by $F$ preserves eigenvalues, completing the proof.

\subsection{Proof of Theorem~\ref{thm:lowSNR-two-phase}} \label{sec:proof-of-two-phases}
Set $x^\star=\beta v$ with $\|v\|_2=1$ and work in the neighborhood $\mathcal{U}_\beta$ given by Theorem~\ref{thm:EM-MRA-spectral-radius}.  
Let $e^{(t)} \triangleq \hat{x}^{(t)}-x^\star$. Since $M\in C^\infty$ and $J(\cdot)$ is continuous, a mean–value expansion around $x^\star$ yields
\begin{align}
    e^{(t+1)} = J(\beta) e^{(t)} + R \p{e^{(t)}}, \qquad \|R(e)\|_2 = O({\|e\|_2^2)}, \label{eqn:error-recursion}
\end{align}
where we abbreviate $J(\beta)\equiv J(x^\star)$. Recall the definition of the subspaces $\mathcal{V}_{\mathrm{ctr}}$ and $\mathcal{V}_{\mathrm{flat}}$ in~\eqref{eqn:Vctr-def}--\eqref{eqn:Vflat-def}. Because these are orthogonal invariant subspaces for the leading blocks of $J(\beta)$ in~\eqref{eqn:error-recursion}, we obtain
\begin{align}
    e_{\mathrm{ctr}}^{(t+1)} &  =  J(\beta) e_{\mathrm{ctr}}^{(t)}  +  \Pi_{\mathrm{ctr}} R \p{e^{(t)}},
    \\
    e_{\mathrm{flat}}^{(t+1)} &  =  J(\beta) e_{\mathrm{flat}}^{(t)}  +  \Pi_{\mathrm{flat}} R \p{e^{(t)}},
    \\
    e_{\mathrm{mean}}^{(t+1)} &  =  \Pi_{\mathrm{mean}} R  \p{e^{(t)}},
\end{align}
where $\Pi_{\mathrm{ctr}},\Pi_{\mathrm{flat}},\Pi_{\mathrm{mean}}$ are the orthogonal projectors onto the respective subspaces.

\paragraph{(1) Early transient on $\mathcal{V}_{\mathrm{ctr}}$.}
By Proposition~\ref{prop:K2-MRA} and Corollary~\ref{cor:block-spectral}, on $\mathcal{V}_{\mathrm{ctr}}$, the eigenvalues of $J(\beta)$ are given by
\begin{align}
    \lambda_{u_k} = 1-2 \frac{\beta^2}{\sigma^2}|\s{V}_k|^2 + O \p{\frac{\beta^4}{\sigma^4}},\label{eq:ctr-eigs}
\end{align}
for $k \neq 0$. Restrict $J(\beta)$ to $\mathcal{V}_{\mathrm{ctr}}$. Since the restriction is diagonalizable over $\mathbb{R}$ with an orthonormal eigenbasis $\{F u_k\}$, its operator norm is controlled by the extremal eigenvalues on $\mathcal{V}_{\mathrm{ctr}}$. From~\eqref{eq:ctr-eigs}, there exists $C>0$ such that for all $z\in\mathcal{V}_{\mathrm{ctr}}$,
\begin{align}
    \p{1-2 \frac{\beta^2}{\sigma^2}\max_{k\ne 0}|\s{V}_k|^2 - C\frac{\beta^4}{\sigma^4}} \|z\|_2 \leq  \|J(\beta)z\|_2 \leq \p{1-2 \frac{\beta^2}{\sigma^2}\min_{k\ne 0}|\s{V}_k|^2 + C\frac{\beta^4}{\sigma^4}} \|z\|_2.
    \label{eq:ctr-ratio-band}
\end{align}
Using invariance of $\mathcal{V}_{\mathrm{ctr}}$ for $J(\beta)$, we have $J(\beta)e_{\mathrm{ctr}}^{(t)}\in\mathcal{V}_{\mathrm{ctr}}$ and $\Pi_{\mathrm{ctr}}J(\beta)e^{(t)} = J(\beta)e_{\mathrm{ctr}}^{(t)}$.
Thus
\begin{align}\label{eq:ctr-one-step-with-r}
    \|e_{\mathrm{ctr}}^{(t+1)}\|_2 \le \|J(\beta)e_{\mathrm{ctr}}^{(t)}\|_2 + \|R(e^{(t)})\|_2,
\end{align}
and similarly a one–step inequality from below by projecting onto a least–contracting eigenvector in $\mathcal{V}_{\mathrm{ctr}}$.

Next we bound the remainder $R(e^{(t)})$. By~\eqref{eq:remainder-quadratic}, there exist constants $c>0$ and $r_0$ such that for all $\| e\| \le r_0$,
\begin{align}
    \|R(e)\|_2 \leq  c\|e\|_2^2. \label{eqn:app-C114}
\end{align}
Fix $\delta_0\in(0,1)$ and define $t_\ast$ as in~\eqref{eqn:t_ast_def}. For $0 \le t<t_\ast$ we have, by definition,
\begin{align}
    \|e_{\mathrm{ctr}}^{(t)}\|_2  \ge  \delta_0 \|e^{(t)}\|_2.
\end{align}
Choose a radius $r_0(\beta)\le r_0$ such that
\begin{align}\label{eq:r0-beta-choice}
    \|e^{(0)}\|_2 \le  r_0(\beta), \qquad c\,r_0(\beta) \le c_1\,\delta_0\,\frac{\beta^4}{\sigma^4},
\end{align}
where $c$ is the constant from~\eqref{eqn:app-C114}, and $c_1 > 0$ is a constant depending only on $(v,d)$. Choosing $\|e^{(0)}\|_2$ that satisfy~\eqref{eq:r0-beta-choice} is possible by shrinking the initialization radius.
By assumption, the initialization lies in $\mathcal{U}_\beta$ which is contracting, thus, $\|e^{(t)}\|_2 \le \|e^{(0)}\|_2$, and from~\eqref{eq:r0-beta-choice} $\|e^{(t)}\|_2 \le r_0(\beta)$, for all $t < t_\ast$

Using~\eqref{eqn:app-C114}, \eqref{eq:r0-beta-choice}, and the dominance $\|e^{(t)}\|_2 \le \delta_0^{-1} \|e_{\mathrm{ctr}}^{(t)}\|_2$ for $t<t_\ast$, we obtain
\begin{align}
    \|R(e^{(t)})\|_2 \le c\|e^{(t)}\|_2^2 \le c\,\delta_0^{-1}\,r_0(\beta)\,\|e_{\mathrm{ctr}}^{(t)}\|_2 \le c_1\,\frac{\beta^4}{\sigma^4}\,\|e_{\mathrm{ctr}}^{(t)}\|_2.
\end{align}
Combining this with the linear bounds on $\|J(\beta)e_{\mathrm{ctr}}^{(t)}\|_2$ and the lower one–step estimate yields the two–sided ratio band~\eqref{eq:ctr-ratio-band} for all $0 \le t<t_\ast$, after absorbing $c_1(\beta^4/\sigma^4)$ into the $O(\beta^4/\sigma^4)$ term, giving
\begin{align}
    \frac{\|e_{\mathrm{ctr}}^{(t+1)}\|}{\|e_{\mathrm{ctr}}^{(t)}\|} \in \pp{ 1-2 \frac{\beta^2}{\sigma^2}\max_{k\neq 0}|\s{V}_k|^2 - C'\frac{\beta^4}{\sigma^4} ,\ 1-2 \frac{\beta^2}{\sigma^2}\min_{k\neq 0}|\s{V}_k|^2 + C'\frac{\beta^4}{\sigma^4} }.
\end{align}

\paragraph{(2) Asymptotic tail on $\mathcal{V}_{\mathrm{flat}}$.} 
By Theorem~\ref{thm:EM-MRA-spectral-radius}, the population EM operator $M$ is $C^2$ in a neighborhood of $x^\star$, and its linearization at $x^\star$ is given by the symmetric Jacobian $J(\beta)$.  
By Corollary~\ref{cor:block-spectral}, $F w_{k_{\max}}$ is an eigenvector of $J(\beta)$ associated with the eigenvalue $\rho(J(\beta))$, and our assumption $\langle e^{(0)},F w_{k_{\max}}\rangle\neq 0$ ensures a nonzero projection on the dominant eigendirection.

Applying Theorem~\ref{thm:Lyap-symmetric} with $u_1 = F w_{k_{\max}}$ therefore yields
\begin{align}
    \lim_{t\to\infty} \limsup_{\substack{e^{(0)}\to 0\\ \langle e^{(0)},F w_{k_{\max}}\rangle\neq 0}} \frac{\|e^{(t+1)}\|_2}{\|e^{(t)}\|_2} = \rho(J(\beta)).
\end{align}
Finally, Corollary~\ref{cor:block-spectral}(3) provides the low-SNR expansion of the spectral radius,
\begin{align}
    \rho(J(\beta)) = 1 - \kappa_{\max}\frac{\beta^4}{\sigma^4} + O  \p{\frac{\beta^6}{\sigma^6}},
\end{align}
for some $\kappa_{\max}\ge 0$ depending only on $v$ and $d$.  Combining the two displays gives
\begin{align}
    \lim_{t \to \infty} \lim_{e^{(0)} \to 0} \frac{\|e^{(t+1)}\|_2}{\|e^{(t)}\|_2} = 1-\kappa_{\max} \frac{\beta^4}{\sigma^4} + O  \p{\frac{\beta^6}{\sigma^6}},
\end{align}
which is the claimed asymptotic tail behavior.
Combining the two parts proves the claimed two–phase behavior.

\subsection{Proof of Corollary~\ref{cor:flat-iter-lower}} \label{sec:proof-of-iteration-complexity}
Work in the setting of Theorem~\ref{thm:lowSNR-two-phase} and let $J(\beta)$ denote the Jacobian of the population EM operator at $x^\star$. By Corollary~\ref{cor:block-spectral}(3), there exists $\kappa_{\max}\ge 0$ such that
\begin{align}
    \rho \p{J(\beta)} = 1 - \kappa_{\max}\frac{\beta^{4}}{\sigma^{4}} + O  \p{\frac{\beta^{6}}{\sigma^{6}}}.
\end{align}
The above expansion implies, by a second–order Taylor expansion of the logarithm at~$1$,
\begin{align}
    \log \rho \p{J(\beta)} = -\kappa_{\max}\frac{\beta^{4}}{\sigma^{4}} + O \p{\frac{\beta^{6}}{\sigma^{6}}}.
\end{align}
Hence there exist constants $\beta_0>0$ and $c_1,c_2>0$ (depending only on $v,d$) such that for all $0<\beta\le \beta_0$,
\begin{align}\label{eq:log-rho-bounds}
    c_1\,\frac{\beta^{4}}{\sigma^{4}} \le \bigl|\log \rho \p{J(\beta)}\bigr| \le c_2\,\frac{\beta^{4}}{\sigma^{4}}.
\end{align}

By Theorem~\ref{thm:lowSNR-two-phase}, after the initial transient the EM dynamics enter a neighborhood of $x^\star$ in which the linear approximation governed by $J(\beta)$ is valid. Under the genericity assumption that the initial error $e^{(0)}=\hat{x}^{(0)}-x^\star$ has nonzero projection on the eigenspace corresponding to $\rho \p{J(\beta)}$, we may apply Corollary~\ref{cor:iter-complex-flow} with $J=J(\beta)$ to the tail regime. Thus, for $\|e^{(0)}\|_2$ and $\varepsilon_{\mathrm{abs}}$ small enough, any $t$ satisfying $\|e^{(t)}\|_2\le \varepsilon_{\mathrm{abs}}$ must obey
\begin{align}\label{eq:t-lb-rho}
    t \gtrsim \frac{1}{\bigl|\log \rho \p{J(\beta)}\bigr|} \,\log  \p{\frac{\|e^{(0)}\|_2}{\varepsilon_{\mathrm{abs}}}},
\end{align}
where the implicit constant is absolute (independent of $\beta$, $e^{(0)}$, and $\varepsilon_{\mathrm{abs}}$).

Combining \eqref{eq:t-lb-rho} with the upper bound in \eqref{eq:log-rho-bounds} yields
\begin{align}
    t \ge \frac{1}{c_+}\,\frac{\sigma^{4}}{\beta^{4}} \,\log  \p{\frac{\|e^{(0)}\|_2}{\varepsilon_{\mathrm{abs}}}},
\end{align}
for all $0<\beta\le\beta_0$, where $c_{+}>0$ is a constant depending only on $v,d$ (through $\kappa_{\max}$ and the absolute constant in~\eqref{eq:t-lb-rho}). This is exactly \eqref{eq:necessary-iter-flat}, and completes the proof.

\subsection{Proof of Proposition~\ref{prop:mag-phase-first-order}}\label{sec:proof-Fourier-phases-magnitudes-convergence}

Fix $k \in \{1,\dots,(d-1)/2\}$ and write $\s{X}^\star[k] = r_k e^{i\phi_k}$, with $r_k \triangleq |\s{X}^\star[k]| > 0$, and $\phi_k \triangleq \phi_{\s{X}^\star}[k]$.
Denote the difference of the $k$-th Fourier component by 
$\delta_k = \s{X}[k] - \s{X}^\star[k]$, and assume it satisfies
\begin{align}
    |\delta_k|\le c\,r_k,
    \label{eq:polar-lemma-smallness}
\end{align}
for some fixed $0<c<1$.
Write
\begin{align}
    e^{-i\phi_k}\delta_k = a_k+ib_k, \qquad a_k,b_k\in\mathbb{R}.
    \label{eq:polar-lemma-ab-def}
\end{align}
Throughout the proof we use
\begin{align}
    |\delta_k|^2 = a_k^2 + b_k^2,
    \label{eq:polar-lemma-delta-a-b}
\end{align}
which follows from \eqref{eq:polar-lemma-ab-def}.

Define the magnitude and phase errors
\begin{align}
    m_k(\delta_k) &\triangleq \big|\s{X}^\star[k]+\delta_k\big|-r_k = \big|\s{X}[k]\big|-\big|\s{X}^\star[k]\big|,
    \label{eq:polar-lemma-m-def}\\
    p_k(\delta_k) &\triangleq \arg\p{\s{X}^\star[k]+\delta_k }-\phi_k = \phi_{\s{X}}[k] - \phi_{\s{X}^\star}[k],
    \label{eq:polar-lemma-p-def}
\end{align}
where $\arg(\cdot)\in(-\pi,\pi]$ denotes the principal argument.
Before proving the proposition, we state and prove an auxiliary lemma.

\begin{lem}[Polar expansion around a nonzero Fourier coefficient]
\label{lem:polar-expansion}
There exists a constant $C>0$, depending only on $c$, such that
\begin{align}
    \big|m_k(\delta_k) - a_k\big| &\le C\,\frac{|\delta_k|^2}{r_k},
    \label{eq:polar-lemma-mag}\\[4pt]
    \big|p_k(\delta_k) - b_k/r_k\big| &\le C\,\frac{|\delta_k|^2}{r_k^2}.
    \label{eq:polar-lemma-phase}
\end{align}
\end{lem}

\begin{proof}[Proof of Lemma~\ref{lem:polar-expansion}]
We prove the lemma in three steps.

\paragraph{Step 1: Rotation to the real axis.}
Set
\begin{align}
    z^\star &\triangleq \s{X}^\star[k] = r_k e^{i\phi_k}, \\ z &\triangleq z^\star + \delta_k = \s{X}[k],
\end{align}
and define the rotated variable
\begin{align}
    \tilde{z} \triangleq e^{-i\phi_k} z = e^{-i\phi_k}\p{\s{X}^\star[k]+\delta_k } = r_k + e^{-i\phi_k}\delta_k.
    \label{eq:polar-lemma-ztilde-def}
\end{align}
Writing $e^{-i\phi_k}\delta_k = a_k+ib_k$ as in \eqref{eq:polar-lemma-ab-def}, we have
\begin{align}
    \tilde{z} = r_k + a_k + ib_k.
    \label{eq:polar-lemma-ztilde-ab}
\end{align}
Since rotation preserves magnitudes,
\begin{align}
    |\tilde{z}| &= |z|, \qquad
    \arg(\tilde{z}) = \arg(z)-\phi_k.
    \label{eq:polar-lemma-rot-invariance}
\end{align}
Combining \eqref{eq:polar-lemma-m-def}, \eqref{eq:polar-lemma-p-def} and
\eqref{eq:polar-lemma-rot-invariance}, we may rewrite the magnitude
and phase errors as
\begin{align}
    m_k(\delta_k) &= |\tilde{z}|-r_k,
    \label{eq:polar-lemma-m-ztilde}\\
    p_k(\delta_k) &= \arg(\tilde{z}).
    \label{eq:polar-lemma-p-ztilde}
\end{align}
Hence it suffices to analyze the real function of $(a_k,b_k)$ associated with $\tilde{z} = r_k + a_k + ib_k$.

\paragraph{Step 2: Magnitude expansion.}
Define
\begin{align}
    f(a_k,b_k) & \triangleq |\tilde{z}| = \sqrt{(r_k+a_k)^2 + b_k^2}.
    \label{eq:polar-lemma-f-def}
\end{align}
By Taylor's theorem with integral remainder applied at $(0,0)$, there exists $(\xi_a,\xi_b)$ on the segment joining $(0,0)$ and $(a,b)$ such that
\begin{align}
    f(a_k,b_k) & = f(0,0) + \nabla f(0,0)\cdot (a_k,b_k)^\top + \frac{1}{2} (a_k,b_k)\,\nabla^2 f(\xi_a,\xi_b)\,(a_k,b_k)^\top.
    \label{eq:polar-lemma-f-taylor}
\end{align}
As $f(0,0)=r_k$, and a direct computation gives
\begin{align}
    \nabla f(a_k,b_k) &= \frac{1}{\sqrt{(r_k+a_k)^2+b_k^2}}
       \begin{bmatrix}
           r_k + a_k \\[2pt] b_k
       \end{bmatrix},
    \label{eq:polar-lemma-grad-f}\\
    \nabla f(0,0)
    &= \frac{1}{r_k}
       \begin{bmatrix}
           r_k \\[2pt] 0
       \end{bmatrix}
     = \begin{bmatrix} 1 \\[2pt] 0 \end{bmatrix}.
    \label{eq:polar-lemma-grad-f-zero}
\end{align}
Thus, substituting~\eqref{eq:polar-lemma-grad-f-zero} into~\eqref{eq:polar-lemma-f-taylor} results that the linearization of $f$ at $(0,0)$ is $r_k + a_k$.

To control the nonlinear remainder, we use the condition~\eqref{eq:polar-lemma-smallness}. For $|\delta_k|=\sqrt{a_k^2+b_k^2}\le c r_k$ we have
\begin{align}
    r_k + a_k &\ge r_k - |a_k| \ge r_k - |\delta_k| \ge (1-c) r_k > 0,
    \label{eq:polar-lemma-lower-bound-real-part}
\end{align}
and hence
\begin{align}
    (r_k+a_k)^2 + b_k^2 \ge (1-c)^2 r_k^2.
    \label{eq:polar-lemma-denominator-bound}
\end{align}
A direct computation of the Hessian $\nabla^2 f(a_k,b_k)$ shows that its entries are rational functions of $(r_k+a_k,b_k)$ whose denominators are powers of $((r_k+a_k)^2+b_k^2)^{3/2}$. Using \eqref{eq:polar-lemma-denominator-bound}, we deduce that there exists a constant $C_1>0$ (depending only on $c$) such that
\begin{align}
    \|\nabla^2 f(a_k,b_k)\| \le \frac{C_1}{r_k}.
    \label{eq:polar-lemma-hess-f-bound}
\end{align}
for all $a_k^2+b_k^2 \le c^2 r_k^2$.
Using \eqref{eq:polar-lemma-grad-f-zero},
\eqref{eq:polar-lemma-hess-f-bound}, and \eqref{eq:polar-lemma-delta-a-b} in
\eqref{eq:polar-lemma-f-taylor} yields
\begin{align}
    \big|f(a_k,b_k)-(r_k+a_k)\big| &\le \frac{1}{2} \|\nabla^2 f(\xi_a,\xi_b)\|(a_k^2+b_k^2) \notag\\ &\le C\,\frac{a_k^2+b_k^2}{r_k} =  C\,\frac{|\delta_k|^2}{r_k},
    \label{eq:polar-lemma-mag-remainder}
\end{align}
for some constant $C>0$ depending only on $c$. Combining \eqref{eq:polar-lemma-m-ztilde}, \eqref{eq:polar-lemma-f-def}, and \eqref{eq:polar-lemma-mag-remainder}, we obtain
\begin{align}
    \big|m_k(\delta_k)-a_k\big| &= \big|(f(a_k,b_k)-r_k)-a_k \big| \le C\,\frac{|\delta_k|^2}{r_k},
\end{align}
which proves \eqref{eq:polar-lemma-mag}.

\paragraph{Step 3: Phase expansion.}
Define
\begin{align}
    g(a_k,b_k) &\triangleq \arg(\tilde{z}) = \arg\p{(r_k+a_k)+ib_k} = \arctan \p{\frac{b_k}{r_k+a_k}},
    \label{eq:polar-lemma-g-def}
\end{align}
where we work in the neighborhood $\sqrt{a_k^2+b_k^2}\le c r_k$ so that the branch of $\arg$ is single–valued and smooth. 
Applying Taylor's theorem at $(0,0)$, there exists $(\zeta_a,\zeta_b)$ on the line segment between $(0,0)$ and $(a,b)$ such that
\begin{align}
    g(a_k,b_k) &= g(0,0) + \nabla g(0,0)\cdot(a_k,b_k)^\top + \frac{1}{2} (a_k,b_k)\,\nabla^2 g(\zeta_a,\zeta_b)\,(a_k,b_k)^\top.
    \label{eq:polar-lemma-g-taylor}
\end{align}
As $g(0,0)=0$, and
\begin{align}
    \frac{\partial g}{\partial a_k}(a_k,b_k) &= -\,\frac{b_k}{(r_k+a_k)^2 + b_k^2},
    \label{eq:polar-lemma-ga}\\
    \frac{\partial g}{\partial b_k}(a_k,b_k) &= \frac{r_k+a_k}{(r_k+a_k)^2 + b_k^2}.
    \label{eq:polar-lemma-gb}
\end{align}
Evaluating at $(0,0)$ gives
\begin{align}
    \nabla g(0,0)
    &= \begin{bmatrix}
        0 \\[2pt] 1/r_k
       \end{bmatrix}.
    \label{eq:polar-lemma-grad-g-zero}
\end{align}
Thus, substituting~\eqref{eq:polar-lemma-grad-g-zero} into~\eqref{eq:polar-lemma-g-taylor} results that the linearization of $g$ at $(0,0)$ is $\frac{b_k}{r_k}$.

As before, \eqref{eq:polar-lemma-smallness} and \eqref{eq:polar-lemma-denominator-bound} ensure that $(r_k+a_k)^2+b_k^2$ is uniformly bounded away from zero on the ball $a_k^2+b_k^2\le c^2 r_k^2$. A direct computation of the Hessian $\nabla^2 g(a_k,b_k)$ shows that its entries are rational functions whose denominators are powers of $(r_k+a_k)^2+b_k^2$. Consequently, there exists $C_2>0$ (depending only on $c$) such that
\begin{align}
    \|\nabla^2 g(a_k,b_k)\| \le \frac{C_2}{r_k^2},
    \label{eq:polar-lemma-hess-g-bound}
\end{align}
for all $a_k^2+b_k^2 \le c^2 r_k^2$

Using \eqref{eq:polar-lemma-grad-g-zero}, \eqref{eq:polar-lemma-hess-g-bound}, and \eqref{eq:polar-lemma-delta-a-b} in \eqref{eq:polar-lemma-g-taylor} yields
\begin{align}
    \Big|g(a_k,b_k)-\frac{b_k}{r_k}\Big| &\le \frac{1}{2} \|\nabla^2 g(\zeta_a,\zeta_b)\|(a_k^2+b_k^2) \notag\\ &\le C\,\frac{a_k^2+b_k^2}{r_k^2} = C\,\frac{|\delta_k|^2}{r_k^2},
    \label{eq:polar-lemma-phase-remainder}
\end{align}
for some constant $C>0$ depending only on $c$.
Finally, recalling from \eqref{eq:polar-lemma-p-ztilde} that $p_k(\delta_k)=g(a_k,b_k)$, inequality \eqref{eq:polar-lemma-phase-remainder} is exactly \eqref{eq:polar-lemma-phase}.
This completes the proof.
\end{proof}

Now we are ready to prove the proposition.
Fix $k\in\{1,\dots,(d-1)/2\}$.
Set $e\triangleq x-x^\star$ and recall that 
\begin{align}
    \delta_k \triangleq \s{X}[k]-\s{X}^\star[k] = (F^\ast e)[k]
\end{align}
is the perturbation of the $k$-th Fourier coefficient.

\paragraph{Step 1: Decomposition of $\delta_k$ in the $\{u_k,w_k\}$ block.}
By Corollary~\ref{cor:block-spectral}, the vectors $u_k,w_k$ span the
$2$-dimensional invariant subspace corresponding to the frequency
pair $\{k,-k\}$, and in the $\{f_k,f_{-k}\}$ Fourier basis they are given by
\begin{align}
    u_k = \frac{1}{\sqrt{2}}
       \begin{bmatrix}
           e^{i\phi_k}\\[2pt] e^{-i\phi_k}
       \end{bmatrix},
    \qquad
    w_k = \frac{1}{i\sqrt{2}}
       \begin{bmatrix}
           e^{i\phi_k}\\[2pt] -\,e^{-i\phi_k}
       \end{bmatrix}.
    \label{eq:Fourier-uk-wk}
\end{align}

Since $F$ is unitary and $\{u_k,w_k\}_{k}$ are orthonormal and supported on
disjoint frequency pairs, we may decompose
\begin{align}
    e = \sum_{j=1}^{(d-1)/2} \alpha_j u_j + \beta_j w_j + e_{\mathrm{mean}},
\end{align}
where $e_{\mathrm{mean}}$ is the mean component, and we have defined,
\begin{align}
    \alpha_j = \langle e,u_j\rangle, \qquad \beta_j = \langle e,w_j\rangle.
\end{align}
Applying $F^\ast$ and taking the $k$-th coordinate, only the $j=k$ terms contribute at frequency $k$, yielding
\begin{align}
    \delta_k = (F^\ast e)[k] &= \alpha_k (F^\ast u_k)[k] + \beta_k (F^\ast w_k)[k].
\end{align}
Using~\eqref{eq:Fourier-uk-wk}, we obtain
\begin{align}
    \delta_k &= \alpha_k \,\frac{e^{i\phi_k}}{\sqrt{2}} + \beta_k \,\frac{-\,i e^{i\phi_k}}{\sqrt{2}} = \frac{e^{i\phi_k}}{\sqrt{2}}\p{\alpha_k - i\beta_k}.
\label{eq:delta-k-alpha-beta}
\end{align}

\paragraph{Step 2: Rotation and application of the polar expansion lemma.}
Consider the rotated perturbation
\begin{align}
    e^{-i\phi_k}\delta_k = a_k + ib_k
\end{align}
as in Lemma~\ref{lem:polar-expansion}.  From \eqref{eq:delta-k-alpha-beta} we deduce
\begin{align}
    e^{-i\phi_k}\delta_k = \frac{1}{\sqrt{2}}\p{\alpha_k - i\beta_k},
\end{align}
so that
\begin{align}
    a_k = \frac{\alpha_k}{\sqrt{2}} = \frac{1}{\sqrt{2}}\langle e,u_k\rangle, \qquad
    b_k = -\,\frac{\beta_k}{\sqrt{2}} = -\,\frac{1}{\sqrt{2}}\langle e,w_k\rangle.
    \label{eq:a-b-alpha-beta}
\end{align}

Lemma~\ref{lem:polar-expansion} yields
\begin{align}
    m_k(\delta_k) = a_k + R^{(\mathrm{mag})}_k(e),
    \label{eq:mag-lemma-remainder}\\[4pt]
    p_k(\delta_k) = \frac{b_k}{r_k} + R^{(\mathrm{ph})}_k(e),
    \label{eq:phase-lemma-remainder}
\end{align}
where,
\begin{align}
    \big|R^{(\mathrm{mag})}_k(e)\big| &\le C_1\,\frac{|\delta_k|^2}{r_k}, \\
    \big|R^{(\mathrm{ph})}_k(e)\big| &\le C_2\,\frac{|\delta_k|^2}{(r_k)^2},    
\end{align}
for some constants $C_1,C_2>0$ depending only on the parameter $c$ in Lemma~\ref{lem:polar-expansion}. Substituting \eqref{eq:a-b-alpha-beta} into \eqref{eq:mag-lemma-remainder} and \eqref{eq:phase-lemma-remainder} gives
\begin{align}
    m_k(\delta_k) &= \frac{1}{\sqrt{2}}\langle e,u_k\rangle + R^{(\mathrm{mag})}_k(e),
    \label{eq:mag-alpha-form}\\
    p_k(\delta_k) &= -\,\frac{1}{\sqrt{2}\,r_k}\langle e,w_k\rangle + R^{(\mathrm{ph})}_k(e),
    \label{eq:phase-beta-form}
\end{align}
with remainders as in \eqref{eq:mag-lemma-remainder}–\eqref{eq:phase-lemma-remainder}.

\paragraph{Step 3: Uniform $O(\|e\|^2)$ bounds.}
Since $F^\ast$ is unitary, we have $|\delta_k| = |(F^\ast e)[k]|\le \|e\|_2$, hence
\begin{align}
    \big|R^{(\mathrm{mag})}_k(e)\big| &\le C_1\,\frac{\|e\|_2^2}{r_k},\\
    \big|R^{(\mathrm{ph})}_k(e)\big| &\le C_2\,\frac{\|e\|_2^2}{r_k^2}.
\end{align}
By assumption $\s{X}^\star[k]\neq 0$ for $1\le k\le (d-1)/2$, and there are only finitely many such $k$, so 
\begin{align}
    r_{\min} \triangleq \min_{1\le k\le (d-1)/2} r_k > 0.
\end{align}
Thus, for all such $k$,
\begin{align}
    \big|R^{(\mathrm{mag})}_k(e)\big| &\le \frac{C_1}{r_{\min}}\,\|e\|_2^2,\\ \big|R^{(\mathrm{ph})}_k(e)\big| &\le \frac{C_2}{r_{\min}^2}\,\|e\|_2^2.
\end{align}
Absorbing these into a single constant $C>0$ (depending only on $x^\star$ and $d$) yields the desired uniform $O(\|e\|_2^2)$ bounds in $k$.

Finally, since $F^\ast$ is unitary, for every $k$ we have $|\delta_k|=|(F^\ast e)[k]|\le \|e\|_2$.
Let $r_{\min}\triangleq \min_{1\le k\le (d-1)/2}|\s{X}^\star[k]|>0$.
Choose $\rho\le c\,r_{\min}$, so that $\|e\|_2\le \rho$ implies $|\delta_k|\le c\,|\s{X}^\star[k]|$ for all $k$, and therefore the assumptions of Lemma~\ref{lem:polar-expansion} hold uniformly in $k$.
Combining \eqref{eq:mag-alpha-form}--\eqref{eq:phase-beta-form} with the bounds
$|\delta_k|\le \|e\|_2$ and $r_k\ge r_{\min}$ yields
\begin{align}
    \max_{1\le k\le (d-1)/2}\Big(|R_k^{\mathrm{mag}}(e)|+|R_k^{\mathrm{ph}}(e)|\Big)\le C\|e\|_2^2
\end{align}
for a constant $C$ depending only on $x^\star$ and $d$.
This proves \eqref{eq:delta-rk-alpha-statement}--\eqref{eq:delta-phik-beta-statement} and completes the proof.

\section{Bias to initialization}
\label{sec:appendix-bias-to-initalization}

\subsection{Oddness of the population EM map}
\begin{lem}[Oddness of the population EM map]
\label{lem:M-odd-in-beta}
Consider the MRA model in~\eqref{eqn:mainModel1D}, parameterized in the low-SNR regime in $\beta \in \mathbb{R}$, as in~\eqref{eq:mra-beta-scaling}
\begin{align}
    Y_\beta = \mathcal{T}_S(\beta v) + \xi,
\end{align}
where $S$ is uniform on $\{0,\dots,d-1\}$ and $\xi \sim \mathcal N(0,\sigma^2 I_d)$.
Recall the population EM map~\eqref{eq:Phi-pop-1d}, parameterized for each $\beta \in \mathbb{R}$, 
\begin{align}
    M_\beta(x)\;\triangleq\;\mathbb{E}\Big[\sum_{\ell=0}^{d-1}\gamma_\ell(x;Y_\beta)\,z_\ell(Y_\beta)\Big], \qquad z_\ell(Y)\triangleq \mathcal{T}_\ell^{-1}Y,
\end{align}
Then for all $\beta$ and all $x$,
\begin{align}\label{eq:M-odd-identity}
    M_{-\beta}(-x)\;=\;-\,M_\beta(x).
\end{align}
In particular, fixing $v,\hat v$ and setting $x=\beta\hat v$ and $x^\star=\beta v$, the map $\beta\mapsto M_\beta(\beta\hat v)$ is odd, hence its Taylor expansion at $0$ contains only odd powers of $\beta$.
\end{lem}

\begin{proof}[Proof of Lemma~\ref{lem:M-odd-in-beta}]
Since $\xi\stackrel{\mathcal{D}}{=}-\xi$ and $\mathcal{T}_S$ is linear,
\begin{align}
    Y_{-\beta}=\mathcal{T}_S(-\beta v)+\xi = -\mathcal{T}_S(\beta v)+\xi \stackrel{\mathcal{D}}{=} -\mathcal{T}_S(\beta v)-\xi = -Y_\beta.
\end{align}
Thus $Y_{-\beta}\stackrel{d}{=}-Y_\beta$.

Because $z_\ell(\cdot)$ is linear, $z_\ell(-Y)=-z_\ell(Y)$.
Moreover, under $(x,Y)\mapsto(-x,-Y)$ the logits $s_\ell$~\eqref{eq:s-logits} satisfy
\begin{align}
    s_\ell(-x;-Y)=\frac{1}{\sigma^2}\langle z_\ell(-Y),-x\rangle = \frac{1}{\sigma^2}\langle z_\ell(Y),x\rangle = s_\ell(x;Y),
\end{align}
so the softmax responsibilities are unchanged:
$\gamma_\ell(-x;-Y)=\gamma_\ell(x;Y)$ for all $\ell$.
Thus, we have,
\begin{align}
    M_{-\beta}(-x) & =\mathbb{E}\Big[\sum_{\ell}\gamma_\ell(-x;Y_{-\beta})\,z_\ell(Y_{-\beta})\Big] = \mathbb{E}\Big[\sum_{\ell}\gamma_\ell(-x;-Y_\beta)\,z_\ell(-Y_\beta)\Big]\\    &=\mathbb{E}\Big[\sum_{\ell}\gamma_\ell(x;Y_\beta)\,(-z_\ell(Y_\beta))\Big] = -\,\mathbb{E}\Big[\sum_{\ell}\gamma_\ell(x;Y_\beta)\,z_\ell(Y_\beta)\Big] = -\,M_\beta(x),
\end{align}
which proves~\eqref{eq:M-odd-identity}. The final claim about vanishing even powers follows from the oddness.
\end{proof}

\subsection{Proof of Theorem~\ref{thm:mra-lowSNR-iteration-bias-to-init}} \label{sec:proof-of-bias-to-initalization}

Recall the definition of the population EM map from~\eqref{eq:Phi-pop-1d}, 
\begin{align} \label{eqn:EM-population-map-bias}
    M(x)=\mathbb{E} \Big[\sum_{\ell=0}^{d-1} \gamma_\ell(x;Y)\,z_\ell(Y)\Big], \qquad z_\ell(Y) \triangleq \mathcal{T}_\ell^{-1}Y,
\end{align}
where the expectation is over $(S,\xi)$ in the MRA model $Y=\mathcal{T}_S x^\star+\xi$~\eqref{eqn:mainModel1D} with $S$ uniform on $\{0,\dots,d-1\}$ and $\xi \sim \mathcal{N}(0,\sigma^2 I_d)$.

\paragraph{Step 1: Mean component is corrected in one step.}
Define $\bar{z}(Y)\triangleq \frac{1}{d}\sum_{\ell=0}^{d-1} z_\ell(Y)$.
Since shifting and averaging produces the constant vector of sample mean,
\begin{align}\label{eq:shift-average-is-mean}
\bar{z}(Y)=\frac{1}{d}\sum_{\ell=0}^{d-1}\mathcal{T}_\ell^{-1}Y
=\Pi_{\mathrm{mean}}Y.
\end{align}
Using the facts that $\sum_{\ell}\gamma_\ell(x;Y)=1$ and $\Pi_{\mathrm{mean}}z_\ell(Y)=\Pi_{\mathrm{mean}}Y$, which follow because the mean is invariant under cyclic shifts, we get
\begin{align}
    \Pi_{\mathrm{mean}}M(x)  =\mathbb{E}\Big[\sum_{\ell}\gamma_\ell(x;Y)\,\Pi_{\mathrm{mean}}z_\ell(Y)\Big] = \mathbb{E}\big[\Pi_{\mathrm{mean}}Y\big] = \Pi_{\mathrm{mean}}\mathbb{E}[Y].
\end{align}
Because $S$ is uniform, $\mathbb{E}[\mathcal{T}_S x^\star]=\Pi_{\mathrm{mean}}x^\star$ and $\mathbb{E}[\xi]=0$, hence $\mathbb{E}[Y]=\Pi_{\mathrm{mean}}x^\star$ and therefore
\begin{align}\label{eq:mean-fixed-point-exact}
    \Pi_{\mathrm{mean}}M(x)=\Pi_{\mathrm{mean}}x^\star \qquad \text{for all }x.
\end{align}
In particular, $\Pi_{\mathrm{mean}}\hat{x}^{(1)}=\Pi_{\mathrm{mean}}x^\star$ for every initialization $\hat{x}^{(0)}$, proving~\eqref{eq:mean-one-step}.
Moreover, applying~\eqref{eq:mean-fixed-point-exact} inductively yields
$\Pi_{\mathrm{mean}}\hat{x}^{(t)}=\Pi_{\mathrm{mean}}x^\star$ for all $t\ge1$.

\paragraph{Step 2: A cubic drift bound for the non-mean component.}
We next show that, after removing the mean correction, the population EM map is identity up to cubic order:
\begin{align}\label{eq:key-cubic-expansion}
    M(\beta \hat{v}) =\Pi_{\mathrm{mean}}(\beta v) + (I-\Pi_{\mathrm{mean}})(\beta \hat{v}) +E(\beta;\hat{v},v),
\end{align}
where
\begin{align}
    \|E(\beta;\hat{v},v)\|\leq C (d,P) \frac{\beta^3}{\sigma^2},
\end{align}
uniformly in $\|v\|,\|\hat{v}\|\le P$. Fix $Y$ and apply Lemma~\ref{lem:resp-near-uniform-MRA-gamma} to $\gamma_\ell(x;Y)$:
\begin{align} \label{eqn:gamma-expansion-bias}
    \gamma_\ell(x;Y)=\frac{1}{d}+\frac{1}{d}\eta_\ell(Y) +\frac{1}{2d}\big(\eta_\ell(Y)^2-\overline{\eta^2}(Y)\big)+R_\ell(\beta;Y),
\end{align}
where $\eta_\ell(Y)$ are the centered logits~\eqref{eq:eta-def}, $\overline{\eta^2}(Y)=\frac{1}{d}\sum_j \eta_j(Y)^2$, and $R_\ell(\beta; Y)$ is the uniform remainder given by~\eqref{eq:MRA-R-bound}.
Let us define,
\begin{align}
    A_1(x) & \triangleq \mathbb{E}\Big[\frac{1}{d}\sum_{\ell}\eta_\ell(Y)\,z_\ell(Y)\Big],
    \\  
    A_2(x) & \triangleq \mathbb{E}\Big[\frac{1}{2d}\sum_{\ell}\big(\eta_\ell(Y)^2-\overline{\eta^2}(Y)\big)\,z_\ell(Y)\Big],
    \\ 
    A_3(x) & \triangleq \mathbb{E}\Big[\sum_{\ell}R_\ell(\beta;Y)\,z_\ell(Y)\Big].
\end{align}
Then, substituting~\eqref{eqn:gamma-expansion-bias} into $M(x)$~\eqref{eqn:EM-population-map-bias} and using~\eqref{eq:shift-average-is-mean} yields the decomposition
\begin{align}
    M(x) = \Pi_{\mathrm{mean}}x^\star + A_1(x) + A_2(x) + A_3(x),
    \label{eq:M-decomp-A123}
\end{align}
where we used $\mathbb{E}[\bar{z}(Y)]=\Pi_{\mathrm{mean}}\mathbb{E}[Y]=\Pi_{\mathrm{mean}}x^\star$. Clearly, $\|A_3(x)\|\le C_3 (d,P)\,\beta^3/\sigma^2$ from the bound of the remainder $R_\ell (\beta; Y)$~\eqref{eq:MRA-R-bound}. It remains to bound the terms $A_1, A_2$.

\medskip \noindent 
\textit {\underline{Claim 1: $A_1(x)=(I-\Pi_{\mathrm{mean}})x + O(\beta^3 / \sigma^2)$.}}
For the MRA logits~\eqref{eq:eta-def}, we have
\begin{align}
    \eta_\ell(Y) = \frac{1}{\sigma^2}\Big(\langle z_\ell(Y),x\rangle-\frac{1}{d}\sum_{j=0}^{d-1}\langle z_j(Y),x\rangle\Big) = \frac{1}{\sigma^2}\,\langle z_\ell(Y)-\bar{z}(Y),\,x\rangle.
\end{align}
Therefore
\begin{align} 
    A_1(x)
    &=\frac{1}{d\sigma^2}\mathbb{E}\Big[\sum_{\ell}\langle z_\ell-\bar{z},x\rangle\,z_\ell\Big] =\frac{1}{d\sigma^2}\mathbb{E}\Big[\Big(\sum_{\ell} z_\ell z_\ell^\top - d\,\bar{z}\,\bar{z}^\top\Big)x\Big]
    \\ &= \frac{1}{\sigma^2} \Big[ \mathbb{E}[YY^\top]-\Pi_{\mathrm{mean}}\mathbb{E}[YY^\top]\Pi_{\mathrm{mean}} \Big] x, \label{eqn:eqn:app_E15}
\end{align}
where we have used $\mathbb{E}[z_\ell z_\ell^\top]=\mathbb{E}[YY^\top]$ because $S$ is uniform. 
Since $Y=\mathcal{T}_S(\beta v)+\xi$ with $\mathbb{E}[\xi\xi^\top]=\sigma^2 I_d$ and $\mathbb{E}[\xi]=0$,
\begin{align}\label{eqn:app_E_16}
    \mathbb{E}[YY^\top]=\sigma^2 I_d + \beta^2\,\mathbb{E}\big[\mathcal{T}_S v v^\top \mathcal{T}_S^\top\big].
\end{align}
The second term has operator norm at most $\beta^2\|v\|^2\le \beta^2 P^2$. Therefore, substituting~\eqref{eqn:app_E_16} into~\eqref{eqn:eqn:app_E15} 
\begin{align}
    A_1(x) = (I-\Pi_{\mathrm{mean}})x+\frac{1}{\sigma^2}\Delta_C x,
\end{align}
where,
\begin{align}
    \|\Delta_C\|_{\operatorname{op}}\le C_1\,\beta^2,
\end{align}
with $C_1$ independent of $\beta$ and $\sigma$. Therefore, as $x = \beta \hat{v}$, with $\|v \| \leq P$,
\begin{align}
    A_1(x) = (I-\Pi_{\mathrm{mean}})x+O(\beta^3 / \sigma^2),
\end{align}
which proves Claim~1.

\medskip \noindent
\textit {\underline{Claim 2: $\|A_2(x)\|\le C_2(d, P) \beta^3 / \sigma^2$.}}
Using $\eta_\ell(Y)=\sigma^{-2}\langle z_\ell-\bar{z},x\rangle$ and $x=\beta\hat{v}$, each term $\big(\eta_\ell(Y)^2-\overline{\eta^2}(Y)\big)z_\ell(Y)$ is cubic in $Y$, that is, of the form
\begin{align}
    \frac{\beta^2}{\sigma^4}  \mathcal{L}_\ell(Y,Y,Y)
\end{align}
where $\mathcal{L}_\ell$ is a trilinear form in $Y$. 
Clearly, since $\mathcal{L}_\ell$ is a trilinear form in $Y$ given in \eqref{eqn:mainModel1D}, and due to the oddness of the noise term $\xi$, the expectation of $\mathcal{L}_\ell(Y,Y,Y)$ has no constant term in $\beta$, and the leading order is $\beta$, that is,
\begin{align}
    \big\|\mathbb{E}[\mathcal{L}_\ell(Y,Y,Y)]\big\|_{\mathrm{op}} \leq \tilde{C}_2 (d,P) \sigma^2 \beta.
\end{align}
Therefore, by linearity of expectation and the triangle inequality for operator norms, we have,
\begin{align}\label{eq:A2-third-moment-reduction}
    \|A_2(x)\| \le C_2(d, P) \,\frac{\beta^3}{\sigma^2}.
\end{align}
Combining Claims~1--2 with~\eqref{eq:M-decomp-A123} yields the expansion~\eqref{eq:key-cubic-expansion}.

\paragraph{Step 3: Iterate the cubic drift bound and conclude.}
Let $\hat{x}^{(t+1)}=M(\hat{x}^{(t)})$ be the population iterates. Fix $T\in\mathbb{N}$ and apply~\eqref{eq:key-cubic-expansion} at $x=\hat{x}^{(t)}=\beta\hat{v}^{(t)}$. Recall by assumption that $\|\hat{v}^{(0)}\|\le P$ and $\|v\|\le P$. 
By induction, there exists $\beta_0=\beta_0(d,P,T)>0$ such that for all $\beta\in(0,\beta_0]$,
\begin{align}\label{eq:vt-bounded-horizon}
    \max_{0\le t\le T}\|\hat{v}^{(t)}\| \le 2P.
\end{align}
Consequently, for all $t\le T-1$ and all $\beta\le\beta_0$, we may apply~\eqref{eq:key-cubic-expansion} at
$x=\hat{x}^{(t)}=\beta\hat{v}^{(t)}$ (with $\|\hat{v}^{(t)}\|\le 2P$) to obtain
\begin{align}\label{eq:one-step-drift}
    \hat{x}^{(t+1)} = \Pi_{\mathrm{mean}}x^\star + (I-\Pi_{\mathrm{mean}})\hat{x}^{(t)} + E_t, \qquad \|E_t\|\le C(d,P) \beta^3 / \sigma^2.
\end{align}
Subtract $\hat{x}^{(t)}$ from both sides:
\begin{align}
    \hat{x}^{(t+1)}-\hat{x}^{(t)} = \Pi_{\mathrm{mean}}\big(x^\star- \hat{x}^{(t)}\big)+E_t.
\end{align}
Summing over $t\in\{0,\ldots,T-1\}$ and using telescoping Sum gives
\begin{align}
    \hat{x}^{(T)}-\hat{x}^{(0)} - \Pi_{\mathrm{mean}}\big(x^\star - \hat{x}^{(0)}\big) &=\sum_{t=0}^{T-1}E_t +\sum_{t=1}^{T-1} \Pi_{\mathrm{mean}} \big(x^\star-\hat{x}^{(t)}\big).
\end{align}
By Step~1, $\Pi_{\mathrm{mean}}\hat{x}^{(t)}=\Pi_{\mathrm{mean}}x^\star$ for all $t\ge1$, so the second sum is identically zero.
Therefore
\begin{align}
    \Big\|\hat{x}^{(T)}-\hat{x}^{(0)} - \Pi_{\mathrm{mean}}\big(x^\star - \hat{x}^{(0)}\big)\Big\| \le \sum_{t=0}^{T-1}\|E_t\| \le C(d,P) T \frac{\beta^3}{\sigma^2}.
\end{align}
Dividing both sides by $\|\hat{x}^{(0)}\|=\beta\|\hat{v}^{(0)}\|$ and by $\mathrm{SNR}\asymp \beta^2 / \sigma^2$ to obtain
\begin{align}
    \frac{1}{\mathrm{SNR}}\frac{\big\|\hat{x}^{(T)}-\hat{x}^{(0)}- \Pi_{\mathrm{mean}}\big(x^\star - \hat{x}^{(0)}\big)\big\|}{\|\hat{x}^{(0)}\|} \le C(d,P) T,
\end{align}
and taking $\mathrm{SNR}\to 0$ yields~\eqref{eq:mra-iter-bound}. Finally, if $\Pi_{\mathrm{mean}}\hat{x}^{(0)}=\Pi_{\mathrm{mean}}x^\star$, then the mean offset vanishes and the same argument yields~\eqref{eq:mra-aligned-bound}.

\section{Preliminaries: Finite-sample analysis} \label{sec:preliminaries-finite-sample}

This section collects preliminaries for the finite-sample analysis leading to Proposition~\ref{lem:subgaussian-epsn}. Our arguments rely on standard non-asymptotic tools from empirical process theory and concentration of measure, see, e.g., \cite{van1996weak}.

\paragraph{Notations.}
Let $Y=\mathcal{T}_S x^\star+\xi \in \mathbb{R}^d$ with $S$ uniform on $\{0,\dots,d-1\}$, where $\{\mathcal{T}_\ell\}_{\ell=0}^{d-1}$ are the circular shifts on $\mathbb{R}^d$, and $\xi\sim\mathcal{N}(0,\sigma^2 I_d)$. Let $M(x)\triangleq  \mathbb{E} \phi_x(Y)$ and $M_n(x)\triangleq  \frac{1}{n}\sum_{i=1}^{n} \phi_x(Y_i)$ be the population and empirical EM updates, where $\phi_x(Y)$
\begin{align}
    \phi_x(Y)\triangleq \sum_{\ell=0}^{d-1}\gamma_\ell(x;Y) \mathcal{T}_\ell^{-1}Y,
    \qquad
    \gamma_\ell(x;Y)=\frac{\exp(\langle Y, \mathcal{T}_\ell x\rangle/\sigma^2)}{\sum_{r=0}^{d-1}\exp(\langle Y,\mathcal{T}_r x\rangle/\sigma^2)}. \label{eqn:app-gamma-ell-def}
\end{align}
Throughout, $\|\cdot\|$ denotes the Euclidean norm on $\mathbb{R}^d$, and for a matrix $A\in\mathbb{R}^{d\times d}$ we write
\begin{align}
    \|A\|_{\mathrm{op}} \triangleq \sup_{\|u\|_2=1} \|Au\|_2
\end{align}
for the operator norm induced by the Euclidean norm (equivalently, the largest singular value of $A$).

\begin{remark}
We mention here that the proof of Proposition~\ref{lem:subgaussian-epsn} relies only on the orthogonality of $\{\mathcal{T}_\ell\}$, and as so extends to any finite orthogonal group action.\end{remark}

\paragraph{Diameter, covering numbers, and Dudley’s entropy bound.}
We collect standard definitions and results from empirical process theory (see, e.g.,~\cite{van1996weak}).
Throughout, $(T,d)$ denotes a nonempty separable metric space.

\begin{definition}[Diameter and Euclidean diameter]\label{def:diameters}
The diameter of $(T,d)$ is
\begin{align}
    \mathrm{diam}(T,d)\triangleq \sup_{s,t\in T} d(s,t)\ \in[0,\infty].
    \label{eq:def-diam-Td}
\end{align}
For a bounded set $\mathcal{B}\subset\mathbb{R}^p$, its Euclidean diameter is
\begin{align}
    \mathrm{diam}(\mathcal{B})\triangleq \sup_{x,z\in\mathcal{B}}\|x-z\|_2.
    \label{eq:def-diam-B}
\end{align}
\end{definition}

\begin{definition}[$\varepsilon$-net and covering number~\cite{van1996weak}]\label{def:covering-number}
Fix $\varepsilon>0$. A set $\mathcal{N}\subseteq T$ is an \emph{$\varepsilon$-net} of $(T,d)$ if for every $t\in T$ there exists $s\in\mathcal{N}$ such that $d(t,s)\le \varepsilon$. The \emph{covering number} of $(T,d)$ at scale $\varepsilon$ is
\begin{align}
    N(T,d;\varepsilon)\triangleq \min\big\{|\mathcal{N}|:\ \mathcal{N}\subseteq T\ \text{is an $\varepsilon$-net of }(T,d)\big\}.
    \label{eq:def-covering-number}
\end{align}
\end{definition}

\begin{lem}[Covering bound in Euclidean space~\cite{vershynin2018high}]\label{lem:euclid-covering}
There exists a universal constant $c_0>0$ such that for every bounded $\mathcal{B}\subset\mathbb{R}^p$ and every $\varepsilon>0$,
\begin{align}
    N\!\left(\mathcal{B},\|\cdot\|_2; \varepsilon\right) \leq   \left(\frac{c_0\,\mathrm{diam}(\mathcal{B})}{\varepsilon}\right)^{p}.
    \label{eq:euclid-covering-diam}
\end{align}
In particular, if $\mathcal{B}\subseteq\{x\in\mathbb{R}^p:\|x\|_2\le R\}$ then $\mathrm{diam}(\mathcal{B})\le 2R$, and hence
\begin{align}
    N\!\left(\mathcal{B},\|\cdot\|_2; \varepsilon\right) \leq   \left(\frac{2c_0 R}{\varepsilon}\right)^{p}.
    \label{eq:euclid-covering-R}
\end{align}
\end{lem}

\begin{definition}[Sub-Gaussian increments]\label{def:subg-increments}
A stochastic process $\{X_t\}_{t\in T}$ is said to have \emph{sub-Gaussian increments} with respect to $d$ if
for all $s,t\in T$ and all $\lambda\in\mathbb{R}$,
\begin{align}
    \mathbb{E}\pp{e^{\lambda(X_t-X_s)}} \leq  \exp\left(\frac{\lambda^2}{2}\,d(t,s)^2\right).
    \label{eq:def-subg-increments}
\end{align}
\end{definition}

\begin{thm}[Dudley’s entropy integral bound~\cite{van1996weak}]\label{thm:dudley-entropy}
Let $\{X_t\}_{t\in T}$ be a separable stochastic process with sub-Gaussian increments with respect to the metric $d$.
Then there exists a universal constant $C>0$ such that for any fixed $t_0\in T$,
\begin{align}
    \mathbb{E}\sup_{t\in T}\big|X_t-X_{t_0}\big| \leq C \int_{0}^{\mathrm{diam}(T,d)} \sqrt{\log N(T,d;\varepsilon)}\,\mathrm{d}\varepsilon.
    \label{eq:dudley-entropy}
\end{align}
\end{thm}

\subsection{Lipschitz property of the EM map }
\begin{lem}[Lipschitz property of the EM map] \label{lem:Lipschitz-in-x}
Recall the definitions of $\gamma_\ell(x;Y)$ and $\phi_x(Y)$ in~\eqref{eqn:app-gamma-ell-def}. Then, for all $x,z\in\mathbb{R}^d$ and all $Y\in\mathbb{R}^d$,
\begin{align}\label{eq:LipParam}
    \|\phi_x(Y)-\phi_z(Y)\| \leq \frac{2}{\sigma^2} \|Y\|^2 \|x-z\|,
\end{align}
and for every $x\in\mathbb{R}^d$,
\begin{align}\label{eq:JacBound}
    \big\|\nabla_x \phi_x(Y)\big\|_{\mathrm{op}} \leq \frac{2}{\sigma^2} \|Y\|^2.
\end{align}
\end{lem}

\begin{proof}[Proof of Lemma~\ref{lem:Lipschitz-in-x}]
First, we note that only the weights $\gamma_\ell(x;Y)$ depend on the parameter $x$; thus, differentiating $\phi_x(Y)=\sum_\ell \gamma_\ell(x;Y) \mathcal{T}_\ell^{-1}Y$ with respect to $x$ yields,
\begin{align}
    \nabla_x \phi_x(Y) = \sum_{\ell=0}^{d-1} (\mathcal{T}_\ell^{-1}Y)\big(\nabla_x \gamma_\ell(x;Y)\big)^\top.
    \label{eqn:gradient-of-phi}
\end{align}
 Define $s_\ell(x;Y)  \triangleq \frac{1}{\sigma^2} \langle Y,\mathcal{T}_\ell x\rangle$. 
Thus,
\begin{align}
    \nabla_x s_\ell(x;Y)  = \frac{1}{\sigma^2} \mathcal{T}_\ell^{-1} Y.
\end{align}
Then, the softmax gradient satisfies
\begin{align}
    \nabla_x \gamma_\ell(x;Y) & = \gamma_\ell(x;Y) \left(\nabla_x s_\ell(x;Y)-\sum_{r=0}^{d-1}\gamma_r(x;Y) \nabla_x s_r(x;Y)\right) \\ & = \frac{1}{\sigma^2} \gamma_\ell(x;Y) \left(\mathcal{T}_\ell^{-1} Y-\sum_{r=0}^{d-1}\gamma_r(x;Y) \mathcal{T}_r^{-1} Y\right). \label{eqn:app-gradient-gamma}
\end{align}
Using triangle inequality, $\|\mathcal{T}_r^{-1} Y\|=\|Y\|$ and $\sum_r\gamma_r=1$, we have,
\begin{align}
    \left\|\mathcal{T}_\ell^{-1} Y-\sum_{r}\gamma_r \mathcal{T}_r^{-1} Y\right\| & \le \|\mathcal{T}_\ell^{-1} Y\|+\left\|\sum_{r}\gamma_r \mathcal{T}_r^{-1} Y\right\| \\ & \le\|Y\|+\sum_{r}\gamma_r\|\mathcal{T}_r^{-1} Y\|\le 2\|Y\|. \label{eqn:triangle-inequality-1}
\end{align}
Therefore, substituting~\eqref{eqn:triangle-inequality-1} into~\eqref{eqn:app-gradient-gamma} yields,
\begin{align}\label{eq:grad-gamma-bound}
    \big\|\nabla_x \gamma_\ell(x;Y)\big\| \leq \frac{2}{\sigma^2} \gamma_\ell(x;Y) \|Y\|.
\end{align}
Using~\eqref{eqn:gradient-of-phi},  \eqref{eq:grad-gamma-bound}, triangle-inequality, 
and $\|\mathcal{T}_\ell^{-1}Y\|=\|Y\|$ yields,
\begin{align}
    \big\|\nabla_x \phi_x(Y)\big\|_{\mathrm{op}} & \; \le \; \sum_{\ell=0}^{d-1}\big\|\nabla_x \gamma_\ell(x;Y)\big\| \|\mathcal{T}_\ell^{-1}Y\|
    \\ & \; \le \; \sum_{\ell=0}^{d-1}\frac{2}{\sigma^2} \gamma_\ell(x;Y) \|Y\| \|Y\| = \frac{2}{\sigma^2} \|Y\|^2,
\end{align}
which is \eqref{eq:JacBound}. The mean-value theorem along $x_t=z+t(x-z)$, $t\in[0,1]$, gives
\begin{align}
    \|\phi_x(Y)-\phi_z(Y)\| \; \leq \; \int_0^1 \big\|\nabla_x \phi_{x_t}(Y)\big\|_{\mathrm{op}} \|x-z\| dt \; \leq \; \frac{2}{\sigma^2} \|Y\|^2 \|x-z\|,
\end{align}
establishing \eqref{eq:LipParam}.
\end{proof}

\subsection{Expectation and tail bounds for the empirical EM}
\begin{lem}[Expectation and tail bounds for the empirical EM]
\label{lem:EF-bounds-MRA-fixedx}
Let $Y=\mathcal{T}_S x^\star+\xi \in \mathbb{R}^d$ be an MRA observation, with $S$ uniform on $\{0,\dots,d-1\}$, and $\xi\sim\mathcal{N}(0,\sigma^2 I_d)$. Recall the definitions of $\gamma_\ell(x;Y)$ and $\phi_x(Y)$ in~\eqref{eqn:app-gamma-ell-def}. 
Define for the i.i.d. observations $\{ Y_i\}_{i=1}^{n} \sim Y$:
\begin{align}
    F(Y_1,\ldots,Y_n) & \triangleq \Big\|\frac{1}{n}\sum_{i=1}^{n} \phi_x(Y_i) - \mathbb{E} \phi_x(Y)\Big\|_2,
    \\
    Z & \triangleq\phi_x(Y)-\mathbb{E} \phi_x(Y),
    \\\Sigma_x & \triangleq\operatorname{Cov}(Z) = \mathbb{E} \pp{ZZ^\top}.
\end{align}
Set $V \triangleq \|x^\star\|_2^2+\sigma^2 d$. Then, the following statements hold.
\begin{enumerate}
\item \emph{(Second moment and expectation.)} For all $n\ge1$,
\begin{align}\label{eq:EF-upper-second-moment}
    \mathbb{E} F^2   =  \frac{1}{n} \mathrm{tr} \Sigma_x \leq \frac{V}{n},
\end{align}
and,
\begin{align}
    \mathbb{E} F \leq  \sqrt{\frac{\mathrm{tr} \Sigma_x}{n}} \leq \sqrt{\frac{V}{n}}. \label{eqn:expectation-of-F}
\end{align}

\item \emph{(Directional scalar deviation.)}
For any fixed $u\in\mathbb{S}^{d-1}$ and any $\eta\in(0,1)$,
\begin{align}\label{eq:scalar-quantile}
    \Big|\Big\langle u,\ \frac{1}{n}\sum_{i=1}^{n} \phi_x(Y_i) - \mathbb{E} \phi_x(Y)\Big\rangle\Big| \leq C_1 \sqrt{\frac{V}{n}} \sqrt{\log\frac{2}{\eta}},
\end{align}
with probability $\ge 1-\eta$ for a global constant $C_1$.

\item \emph{(Uniform bound.)}
Let $\mathcal{N}\subset\mathbb{S}^{d-1}$ be any $1/2$-net with $\log(|\mathcal{N}|)\le Cd$, for a constant $C$.
Then for any $\delta\in(0,1)$, with probability at least $1-\delta$,
\begin{align}\label{eq:anchor-term}
    \max_{u\in\mathcal{N}} \Big|\Big\langle u,\ \frac{1}{n}\sum_{i=1}^{n} \phi_x(Y_i) - \mathbb{E} \phi_x(Y)\Big\rangle\Big|\leq C_2 \sqrt{\frac{V}{n}}\ \sqrt{ d+\log\frac{2}{\delta} }.
\end{align}
and,
\begin{align}\label{eq:vector-fixed-x}
   \mathbb{P}\ \p{F \le\ C_3 \sqrt{\frac{V}{n}}\ \sqrt{ d+\log\frac{2}{\delta} }} \geq 1-\delta,
\end{align}
for global constants $C_2, C_3$.
\end{enumerate}
\end{lem}

\begin{proof}[Proof of Lemma~\ref{lem:EF-bounds-MRA-fixedx}]
We prove each item of Lemma~\ref{lem:EF-bounds-MRA-fixedx} separately. 

\paragraph{(1) Second moment and expectation.}
Let $Z_i=\phi_x(Y_i)-\mathbb{E}\phi_x(Y)$, which are i.i.d. by assumption and satisfy $\mathbb{E}Z_i=0$. Then
\begin{align}
    \mathbb{E}F^2  &= \mathbb{E}\Big\|\frac{1}{n}\sum_{i=1}^n Z_i\Big\|^2 = \frac{1}{n^2}\sum_{i=1}^n \mathbb{E}\|Z_i\|^2 = \frac{1}{n}\operatorname{tr}\Sigma_x,
\end{align}
where the cross-terms vanish for $i\neq j$ by independence and centering; indeed, note that $\mathbb{E}\langle Z_i,Z_j\rangle=\langle \mathbb{E}Z_i,\mathbb{E}Z_j\rangle=0$.
Next, using the triangle inequality, $\|\mathcal{T}_\ell^{-1}Y\|=\|Y\|$, and $\sum_\ell \gamma_\ell(x;Y)=1$, we have
\begin{align}
    \|\phi_x(Y)\| \le \sum_{\ell}\gamma_\ell(x;Y)\,\|\mathcal{T}_\ell^{-1}Y\| = \|Y\|.
\end{align}
Therefore,
\begin{align}
    \operatorname{tr}\Sigma_x = \mathbb{E}\|Z\|^2 \le \mathbb{E}\|\phi_x(Y)\|^2 \le \mathbb{E}\|Y\|^2 = \|x^\star\|^2+\sigma^2 d = V,
\end{align}
which yields~\eqref{eq:EF-upper-second-moment}. For the expectation, Jensen's inequality gives
\begin{align}
    (\mathbb{E}F)^2\le \mathbb{E}F^2, \label{eqn:-jensen-inequality-quadratic}
\end{align}
and combining with~\eqref{eq:EF-upper-second-moment} yields~\eqref{eqn:expectation-of-F}.

\paragraph{(2) Directional deviation.}
Fix $u\in\mathbb{S}^{d-1}$ and define
\begin{align}
    X_i \triangleq \Big\langle u,\ \phi_x(Y_i)-\mathbb{E}\phi_x(Y)\Big\rangle.
\end{align}
By definition, $\{X_i\}_{i=1}^n$ are i.i.d. and centered. As above, we have $|\langle u,\phi_x(Y)\rangle|\le \|\phi_x(Y)\|\le \|Y\|$, and hence
\begin{align}\label{eq:dom}
    |X_i| \le \|\phi_x(Y_i)\|+\mathbb{E}\|\phi_x(Y)\| \le \|Y_i\|+\mathbb{E}\|Y\| \triangleq \|Y_i\|+m .
\end{align}
Moreover, by Cauchy--Schwarz inequality,
\begin{align}\label{eq:Ez}
m=\mathbb{E}\|Y\|\le \sqrt{\mathbb{E}\|Y\|^2}=\sqrt{V}.
\end{align}
Since $\xi\sim\mathcal{N}(0,\sigma^2 I_d)$, there exists a universal constant $C>0$ such that for all $t\ge 0$,
\begin{align}\label{eq:norm-conc}
    \mathbb{P}\big(|\|Y\|-\mathbb{E}\|Y\||\ge t\big) \le 2\exp\left(-\frac{t^2}{C\sigma^2}\right).
\end{align}
By~\eqref{eq:dom}, for every $t\ge 0$,
\begin{align}
    \mathbb{P}\left(|X_i|\ge t\right) \le \mathbb{P}\left(\|Y\|\ge t-m\right) & = \mathbb{P}\left(\|Y\|-\mathbb{E}\|Y\|\ge t-2m\right)
    \\ & \le \mathbb{P}\left(\|Y\|-\mathbb{E}\|Y\|\ge (t-2m)_+\right),
\end{align}
where $(a)_+\triangleq \max\{a,0\}$. Combining this with~\eqref{eq:norm-conc} yields
\begin{align}\label{eq:Xi-tail-intermediate}
    \mathbb{P}\left(|X_i|\ge t\right) \le 2\exp\left(-\frac{(t-2m)_+^2}{C\sigma^2}\right).
\end{align}
Now, to upper bound the right-hand side of \eqref{eq:Xi-tail-intermediate} by a sub-Gaussian tail in $t$, we split into two regimes. First, if $t\le 4m$, then choosing $C'>0$ sufficiently large makes
$2\exp\big(-t^2/(C'(\sigma+m)^2)\big)\ge 1$ and the desired bound is trivial. Second, if $t>4m$, then $(t-2m)_+\ge t/2$ and $\sigma\le \sigma+m$, so~\eqref{eq:Xi-tail-intermediate} implies
\begin{align}
    \mathbb{P}\left(|X_i|\ge t\right) \le 2\exp\left(-\frac{t^2}{4C\sigma^2}\right) \le 2\exp\left(-\frac{t^2}{C'(\sigma+m)^2}\right).
\end{align}
Combining, for all $t\ge 0$,
\begin{align}\label{eq:Xi-subg}
    \mathbb{P}\left(|X_i|\ge t\right) \le 2\exp\left(-\frac{t^2}{C'(\sigma+m)^2}\right),
\end{align}
so each $X_i$ is sub-Gaussian with parameter
\begin{align}\label{eq:tau-final}
    \tau \le C_0(\sigma+m) \le C_0' \sqrt{V}.
\end{align}
Therefore, by the standard concentration bound for i.i.d. centered sub-Gaussian random variables, we have for all $t>0$
\begin{align}\label{eq:subg-average}
    \mathbb{P}\left(\left|\frac{1}{n}\sum_{i=1}^{n} X_i\right|\ge t\right) \le 2\exp\left(-c\,\frac{n t^2}{\tau^2}\right),
\end{align}
for some universal constant $c>0$. Equivalently, with $s=\log\frac{2}{\eta}$,
\begin{align}\label{eq:avg-quantile}
    \mathbb{P}\left(\left|\frac{1}{n}\sum_{i=1}^{n} X_i\right| \le C\frac{\tau}{\sqrt{n}}\sqrt{s}\right)\ge 1-\eta.
\end{align}
Plugging~\eqref{eq:tau-final} and $\tau\le C_0'\sqrt{V}$ into~\eqref{eq:avg-quantile} yields~\eqref{eq:scalar-quantile}.

\paragraph{(3) Union over a $1/2$-net.}
Apply~\eqref{eq:scalar-quantile} to each $u\in\mathcal{N}$ with $\eta=\delta/|\mathcal{N}|$ and use a union bound. Since $\log|\mathcal{N}|\le Cd$, we have $|\mathcal{N}|\le r^d$ for $r=e^{C}$, hence with probability at least $1-\delta$,
\begin{align}
    \max_{u\in\mathcal{N}} \Big|\Big\langle u,\ \frac{1}{n}\sum_{i=1}^{n}\phi_x(Y_i)-\mathbb{E}\phi_x(Y)\Big\rangle\Big|  &\le  C_2\sqrt{\frac{V}{n}}\sqrt{\log\frac{2|\mathcal{N}|}{\delta}}
    \\
    &\le C_3\sqrt{\frac{V}{n}}\sqrt{d+\log\frac{2}{\delta}},
\end{align}
which is~\eqref{eq:anchor-term}. Finally, for any $v\in\mathbb{R}^d$, since $\mathcal{N}$ is a $1/2$-net of $\mathbb{S}^{d-1}$,
\begin{align}
    \|v\|\le 2\max_{u\in\mathcal{N}}|\langle u,v\rangle|,
\end{align}
and applying this to $v=\frac{1}{n}\sum_{i=1}^{n}\phi_x(Y_i)-\mathbb{E}\phi_x(Y)$ yields~\eqref{eq:vector-fixed-x}.
\end{proof}

\subsection{Symmetrization and chaining bound}

\begin{lem}[Symmetrization and chaining bound]\label{lem:symm-chaining-osc}
Let $Y=\mathcal{T}_S x^\star+\xi\in\mathbb{R}^d$ be an MRA observation~\eqref{eqn:mainModel1D}, where $S$ is uniform on $\{0,\dots,d-1\}$
and $\xi\sim\mathcal{N}(0,\sigma^2 I_d)$. Recall the definitions of $\gamma_\ell(x;Y)$ and $\phi_x(Y)$ in~\eqref{eqn:app-gamma-ell-def}.
Let $\mathcal{B}\subset\mathbb{R}^d$ be compact and satisfy $\mathcal{B}\subseteq\{x:\|x\|_2\le R\}$, and fix an anchor $x_0\in\mathcal{B}$.
For any $u\in\mathbb{S}^{d-1}$, conditionally on i.i.d. observations $\{Y_i\}_{i=1}^n\sim Y$, define
\begin{align}\label{eq:def-Wxu}
    W_x^{(u)} \triangleq \frac{1}{n}\sum_{i=1}^n \varepsilon_i\,\big\langle u,\phi_x(Y_i)\big\rangle,
    \qquad x\in\mathcal{B},
\end{align}
where $\{\varepsilon_i\}_{i=1}^n$ are i.i.d. Rademacher signs~\cite{koltchinskii2011oracle,ledoux2013probability},
independent of $\{Y_i\}_{i=1}^n$. Introduce the random scale factor
\begin{align}\label{eq:def-Ln}
    L_n \triangleq \frac{2}{\sigma^2}\left(\frac{1}{n}\sum_{i=1}^n \|Y_i\|_2^4\right)^{1/2}.
\end{align}
Then the following statements hold for every fixed $u\in\mathbb{S}^{d-1}$:
\begin{enumerate}
\item \emph{(Sub-Gaussian increments.)}
Conditionally on $\{Y_i\}_{i=1}^n$, the process $\{W_x^{(u)}\}_{x\in\mathcal{B}}$ has sub-Gaussian increments in the sense of Definition~\ref{def:subg-increments} with respect to the metric
\begin{align}
    d(x,z)\triangleq \frac{L_n}{\sqrt{n}}\|x-z\|_2,\qquad x,z\in\mathcal{B}.
\end{align}
Equivalently, for all $\lambda\in\mathbb{R}$ and all $x,z\in\mathcal{B}$,
\begin{align}\label{eq:sg-increments}
    \mathbb{E}_\varepsilon \exp\Big\{\lambda\big(W_x^{(u)}-W_z^{(u)}\big)\Big\} \leq \exp\left(\frac{\lambda^2}{2}\,\frac{L_n^2}{n}\,\|x-z\|_2^2\right).
\end{align}

\item \emph{(Dudley chaining.)}
Conditionally on $\{Y_i\}_{i=1}^n$,
\begin{align}\label{eq:anchored-chaining-simplified}
    \mathbb{E}_\varepsilon \sup_{x\in\mathcal{B}}\big|W_x^{(u)}-W_{x_0}^{(u)}\big| \ \le\ C\,\frac{L_n}{\sqrt{n}}\,R\,\sqrt{d},
\end{align}
for a universal constant $C>0$.
\end{enumerate}
\end{lem}

\begin{proof}[Proof of Lemma~\ref{lem:symm-chaining-osc}]
We condition throughout on the observations $\{Y_i\}_{i=1}^n$.

\paragraph{Step 1: sub-Gaussian increments.}
By Lemma~\ref{lem:Lipschitz-in-x}, for all $x,z\in\mathbb{R}^d$ and all $y\in\mathbb{R}^d$,
\begin{align}\label{eq:Lipx}
    \big\|\phi_x(y)-\phi_z(y)\big\|_2 \leq \frac{2}{\sigma^2}\,\|y\|_2^2\,\|x-z\|_2.
\end{align}
Fix $x,z\in\mathcal{B}$. Writing
\begin{align}
    W_x^{(u)}-W_z^{(u)} = \frac{1}{n}\sum_{i=1}^n \varepsilon_i\,\big\langle u,\phi_x(Y_i)-\phi_z(Y_i)\big\rangle \triangleq \sum_{i=1}^n \varepsilon_i a_i,
\end{align}
we have, since $\|u\|_2=1$,
\begin{align}\label{eq:def-ai-bound}
    |a_i| \leq \frac{1}{n}\big\|\phi_x(Y_i)-\phi_z(Y_i)\big\|_2 \ \overset{\eqref{eq:Lipx}}{\le}\ \frac{2}{\sigma^2 n}\,\|Y_i\|_2^2\,\|x-z\|_2.
\end{align}
By Hoeffding's lemma for Rademacher sums (see, e.g.,~\cite[Lemma~2.2.8]{van1996weak}),
\begin{align}\label{eq:hoeffding-rademacher}
    \mathbb{E}_\varepsilon \exp\left\{\lambda\sum_{i=1}^n \varepsilon_i a_i\right\} \leq \exp\left(\frac{\lambda^2}{2}\sum_{i=1}^n a_i^2\right).
\end{align}
Combining~\eqref{eq:def-ai-bound}--\eqref{eq:hoeffding-rademacher} yields
\begin{align}
    \sum_{i=1}^n a_i^2 \leq  \frac{4}{\sigma^4 n^2}\,\|x-z\|_2^2 \sum_{i=1}^n \|Y_i\|_2^4= \frac{L_n^2}{n}\,\|x-z\|_2^2, \label{eqn:app-E58}
\end{align}
where we used the definition of $L_n$ in~\eqref{eq:def-Ln}. Substituting ~\eqref{eqn:app-E58} into~\eqref{eq:hoeffding-rademacher} proves~\eqref{eq:sg-increments}, which is the sub-Gaussian increment condition of Definition~\ref{def:subg-increments} for the metric $d(x,z)=(L_n/\sqrt{n})\|x-z\|_2$.

\paragraph{Step 2: Dudley’s entropy bound.}
Since $\{W_x^{(u)}\}_{x\in\mathcal{B}}$ has sub-Gaussian increments with respect to the metric
$d(x,z)=(L_n/\sqrt{n})\|x-z\|_2$, we may apply Dudley’s entropy integral bound
(Theorem~\ref{thm:dudley-entropy}) on the index set $T=\mathcal{B}$ with anchor $t_0=x_0$. Using the definitions of diameter
and covering number in~\eqref{eq:def-diam-Td} and~\eqref{eq:def-covering-number}, we obtain
\begin{align}\label{eq:dudley-applied}
    \mathbb{E}_\varepsilon \sup_{x\in\mathcal{B}} \big|W_x^{(u)}-W_{x_0}^{(u)}\big| \leq C \int_{0}^{\mathrm{diam}(\mathcal{B},d)} \sqrt{\log N(\mathcal{B},d;\varepsilon)}\,\mathrm{d}\varepsilon,
\end{align}
where $\mathrm{diam}(\mathcal{B},d)=\sup_{x,z\in\mathcal{B}} d(x,z)$.

We bound the entropy integral by reducing to Euclidean covering numbers (Definition~\ref{def:covering-number}). Since $d(x,z)=(L_n/\sqrt{n})\|x-z\|_2$,
a scaling argument gives
\begin{align}\label{eq:cover-scale}
    N(\mathcal{B},d;\varepsilon)= N\left(\mathcal{B},\|\cdot\|_2;\frac{\sqrt{n}}{L_n}\varepsilon\right).
\end{align}
Applying Lemma~\ref{lem:euclid-covering} and using $\mathrm{diam}(\mathcal{B})\le 2R$ yields, for all $\varepsilon>0$,
\begin{align}\label{eq:cover-bound-d}
    N(\mathcal{B},d;\varepsilon) \leq \left(\frac{c_0\,\mathrm{diam}(\mathcal{B})}{(\sqrt{n}/L_n)\varepsilon}\right)^d \leq \left(\frac{2c_0\,R\,L_n}{\sqrt{n}\,\varepsilon}\right)^d.
\end{align}
Moreover,
\begin{align}\label{eq:diam-d}
    \mathrm{diam}(\mathcal{B},d) =  \frac{L_n}{\sqrt{n}}\,\mathrm{diam}(\mathcal{B}) \leq \frac{2RL_n}{\sqrt{n}}.
\end{align}
Substituting~\eqref{eq:cover-bound-d} into~\eqref{eq:dudley-applied} and performing the change of variables
$\varepsilon=(L_n/\sqrt{n})t$ with $t\in[0,\mathrm{diam}(\mathcal{B})]$, we obtain
\begin{align}
    \nonumber \int_{0}^{\mathrm{diam}(\mathcal{B},d)} \sqrt{\log N(\mathcal{B},d;\varepsilon)}\,\mathrm{d}\varepsilon
    &\le
    \int_{0}^{2RL_n/\sqrt{n}} \sqrt{\,d\,\log\left(\frac{2c_0RL_n}{\sqrt{n}\,\varepsilon}\right)}\,\mathrm{d}\varepsilon
    \\ \nonumber
    &=
    \frac{L_n}{\sqrt{n}}\sqrt{d}\int_{0}^{2R} \sqrt{\log\left(\frac{2c_0R}{t}\right)}\,\mathrm{d}t
    \\
    &\le C'\,\frac{L_n}{\sqrt{n}}\,R\,\sqrt{d},
\end{align}
where the last inequality uses $\int_0^{1}\sqrt{\log(c/u)}\,\mathrm{d}u<\infty$ for any fixed $c>0$ and absorbs constants into $C'$.
Combining with~\eqref{eq:dudley-applied} proves~\eqref{eq:anchored-chaining-simplified}.
\end{proof}

\subsection{Finite-sample uniform deviation for the EM operator}

\begin{proposition}
[Finite-sample uniform deviation for the EM operator]
\label{prop:finite-max-u}
Let $Y=\mathcal{T}_S x^\star+\xi \in \mathbb{R}^d$ be an MRA observation~\eqref{eqn:mainModel1D}, with $S$ uniform on $\{0,\dots,d-1\}$, and $\xi\sim\mathcal{N}(0,\sigma^2 I_d)$. Recall the definitions of $\gamma_\ell(x;Y)$ and $\phi_x(Y)$ in~\eqref{eqn:app-gamma-ell-def}. 
Let $\mathcal{B}\subset\{x\in\mathbb{R}^d:\|x\|_2\le R\}$ be compact, fix $x_0\in\mathcal{B}$, and let $M_n$ and $M$ denote the empirical and population EM operators, respectively. Assume that $\|x^\star \| \leq C\sigma \sqrt{d}$, and that $n \geq C'\p{\log(2/\delta)}^{3}$, for some constant $C, C'$.
Then, there exist universal constants $C_1,C_2>0$ such that
with probability at least $1-\delta$,
\begin{align}
   \sup_{u \in \mathbb{S}^{d-1}} & \sup_{x\in\mathcal{B}} \abs{\big\langle u, (M_n-M)(x)-(M_n-M)(x_0)\big\rangle} \\ & \leq C_1 \frac{R\sqrt d}{\sqrt n} d \p{1+C_2\sqrt{\frac{\log(2/\delta)}{n}}}. \label{eq:finite-max-union}
\end{align}

\end{proposition}

\begin{proof} [Proof of Proposition~\ref{prop:finite-max-u}]
Let $M(x)\triangleq  \mathbb{E} \phi_x(Y)$ and $M_n(x)\triangleq  \frac{1}{n}\sum_{i=1}^{n} \phi_x(Y_i)$ be the population and empirical EM updates in the MRA model. Then, with this notation, we have,

\begin{align}
    \sup_{x\in\mathcal{B}} & \ \abs{\big\langle u,(M_n-M)(x)-(M_n-M)(x_0)\big\rangle} \\ & \qquad = \sup_{x\in\mathcal{B}} \abs{ \big\langle u,\frac{1}{n}\sum_{i=1}^{n}\p{\phi_x(Y_i)-\mathbb{E}\phi_x(Y)} -\frac{1}{n}\sum_{i=1}^{n}\p{\phi_{x_0}(Y_i)-\mathbb{E}\phi_{x_0}(Y)}\big\rangle}.
\end{align}

\paragraph{Step 1: Conditional ghost--sample symmetrization and chaining.}
Fix the observations $\{Y_i\}_{i=1}^n$. For $u\in\mathbb{S}^{d-1}$ and an anchor $x_0\in\mathcal{B}$, define
\begin{align}
    \Delta_x(y) \triangleq \big\langle u,\ \phi_x(y)-\phi_{x_0}(y)\big\rangle, \qquad x\in\mathcal{B},\ y\in\mathbb{R}^d,
\end{align}
and write $P f=\mathbb{E}[f(Y)]$, $P_n f=\frac{1}{n}\sum_{i=1}^n f(Y_i)$. Let $\{Y'_i\}_{i=1}^n$ be an i.i.d. ghost sample independent of $\{ Y_i\}_{i=1}^{n}$, and let $\{\varepsilon_i\}_{i=1}^n$ be i.i.d. Rademacher signs, independent of both $\{Y_i\}$, and $\{Y'_i\}$. Then, conditioning on $\{Y_i\}$,
\begin{align}
    \sup_{x\in\mathcal{B}} & \abs{\Big\langle u,\frac{1}{n}\sum_{i=1}^{n}\p{\phi_x(Y_i)-\mathbb{E}\phi_x(Y)} -\frac{1}{n}\sum_{i=1}^{n}\p{\phi_{x_0}(Y_i)-\mathbb{E}\phi_{x_0}(Y)}\Big\rangle}
    \\ & = \sup_{x\in\mathcal{B}}\abs{(P_n-P)\Delta_x}
    \\ & = \sup_{x\in\mathcal{B}}\abs{(P_n-\mathbb{E}_{Y'} [P'_n])\Delta_x}
\end{align}
where $P_n' f=\frac{1}{n}\sum_{i=1}^n f(Y'_i)$. Then, by supremum–expectation inequality,
\begin{align}
    \sup_{x\in\mathcal{B}}  \abs{(P_n-\mathbb{E}_{Y'} [P'_n])\Delta_x} \le \mathbb{E}_{Y'}\ \sup_{x\in\mathcal{B}}\abs{(P_n-P_n')\Delta_x}
    \label{eq:ghost-pointwise-1}
\end{align}
Now apply standard Rademacher symmetrization~\cite{koltchinskii2011oracle, ledoux2013probability}, conditionally on $\{Y_i\}_{i=1}^{n}$:
\begin{align}
    \mathbb{E}_{Y'}\ \sup_{x\in\mathcal{B}}\abs{(P_n-P_n')\Delta_x} \le \mathbb{E}_{Y',\varepsilon}\ \sup_{x\in\mathcal{B}}\abs{\frac{1}{n}\sum_{i=1}^n \varepsilon_i\p{\Delta_x(Y_i)-\Delta_x(Y'_i)}}.
    \label{eq:ghost-pointwise-2}
\end{align}
Using the triangle inequality in \eqref{eq:ghost-pointwise-2} yields
\begin{align}
    \sup_{x\in\mathcal{B}}\abs{(P_n-P)\Delta_x} \leq \mathbb{E}_{\varepsilon}\ \sup_{x\in\mathcal{B}}\abs{\frac{1}{n}\sum_{i=1}^n \varepsilon_i \Delta_x(Y_i)} + \mathbb{E}_{Y'}\mathbb{E}_{\varepsilon}\ \sup_{x\in\mathcal{B}}\abs{\frac{1}{n}\sum_{i=1}^n \varepsilon_i \Delta_x(Y'_i)}.
    \label{eq:ghost-decomp}
\end{align}
Define, for the fixed sample and the ghost sample respectively,
\begin{align}
    W_x^{(u)} &   \triangleq  \frac{1}{n}\sum_{i=1}^n \varepsilon_i \langle u,\phi_x(Y_i)\rangle,
    \\
    W_x^{(u)\prime} &   \triangleq  \frac{1}{n}\sum_{i=1}^n \varepsilon_i \langle u,\phi_x(Y'_i)\rangle .
\end{align}
Then \eqref{eq:ghost-decomp} becomes
\begin{align}
    \sup_{x\in\mathcal{B}}\abs{(P_n-P)\Delta_x} \leq \mathbb{E}_{\varepsilon}\ \sup_{x\in\mathcal{B}} \abs{\p{W_x^{(u)}-W_{x_0}^{(u)}}} + \mathbb{E}_{Y'}\mathbb{E}_{\varepsilon}\ \sup_{x\in\mathcal{B}}\abs{W_x^{(u)\prime}-W_{x_0}^{(u)\prime}}.
    \label{eq:ghost-decomp-W}
\end{align}

In addition, by Lemma~\ref{lem:symm-chaining-osc}, conditionally on the respective samples,
\begin{align}
    \mathbb{E}_{\varepsilon}\ & \sup_{x\in\mathcal{B}}\abs{W_x^{(u)}-W_{x_0}^{(u)}} \leq  C \frac{L_n}{\sqrt n} R \sqrt d,
    \\
    \mathbb{E}_{\varepsilon} & \ \sup_{x\in\mathcal{B}}\abs{W_x^{(u)\prime}-W_{x_0}^{(u)\prime}} \leq C \frac{L_n'}{\sqrt n} R \sqrt d,
    \label{eq:anchored-chaining-two}
\end{align}
where
\begin{align}
    L_n \triangleq \frac{2}{\sigma^2}\p{\frac{1}{n}\sum_{i=1}^{n}\|Y_i\|^4}^{1/2},
    \qquad L_n' \triangleq \frac{2}{\sigma^2}\p{\frac{1}{n}\sum_{i=1}^{n}\|Y'_i\|^4}^{1/2}.
\end{align}

\paragraph{Step 2: Bounding $\mathbb{E} [L_n']$.}
Taking expectation in the second inequality of \eqref{eq:anchored-chaining-two} with respect to $\{Y'_i\}$ and using Jensen’s inequality (concavity of $t\mapsto\sqrt t$),
\begin{align}
    \mathbb{E}_{Y'}\pp{L_n'} \leq \frac{2}{\sigma^2} \sqrt{\mathbb{E}\|Y\|^4}.
    \label{eq:Lnprime-expect-raw}
\end{align}
For $Y=\mathcal{T}_S x^\star+\xi$, $\xi\sim\mathcal{N}(0,\sigma^2 I_d)$, we have,
\begin{align}
    \mathbb{E}\|Y\|^4   =  \|x^\star\|^4 + 2\sigma^2(d+2)\|x^\star\|^2 + \sigma^4 d(d+2),
    \label{eq:Y4-moment}
\end{align}
so that under the assumption that $\|x^\star  \|\leq  C\sigma \sqrt{d}$,
\begin{align}
    \mathbb{E}_{Y'}[L_n']  \leq  \frac{2}{\sigma^2}\sqrt{\mathbb{E}\|Y\|^4} \leq  \tilde{C}d,
    \label{eq:Lnprime-expect}
\end{align}
for a global constant $\tilde{C}$. 
Combining \eqref{eq:ghost-decomp-W}, \eqref{eq:anchored-chaining-two}, and \eqref{eq:Lnprime-expect} yields the pointwise (conditional on $\{Y_i\}$) bound
\begin{align}
    \sup_{x\in\mathcal{B}} & \abs{\Big\langle u,\frac{1}{n}\sum_{i=1}^{n}\p{\phi_x(Y_i)-\mathbb{E}\phi_x(Y)} -\frac{1}{n}\sum_{i=1}^{n}\p{\phi_{x_0}(Y_i)-\mathbb{E}\phi_{x_0}(Y)}\Big\rangle}
    \\ &  \leq  C \frac{R\sqrt d}{\sqrt n} \p{L_n + \mathbb{E}_{Y'}[L_n']} \leq  C \frac{R\sqrt d}{\sqrt n} \p{L_n + \tilde{C} d}.
    \label{eq:step1-final}
\end{align}

\paragraph{Step 3: Removal of conditioning (control of $L_n$).}
Set
\begin{align}\label{eq:def-mu-Delta}
    \mu \triangleq \mathbb{E}\|Y\|_2^4, \qquad \Delta_n \triangleq \frac{1}{n}\sum_{i=1}^n \|Y_i\|_2^4-\mu,
\end{align}
so that
\begin{align}\label{eq:Ln-mu-Delta}
    L_n = \frac{2}{\sigma^2}\sqrt{\mu+\Delta_n}.
\end{align}

Let $Y=\mathcal{T}_Sx^\star+\xi$ with $\xi\sim\mathcal{N}(0,\sigma^2I_d)$.
Since $\mathcal{T}_Sx^\star$ is deterministic conditional on $S$, Gaussian concentration implies that for all $t\ge 0$,
\begin{align}\label{eq:gauss-conc-norm}
    \mathbb{P}\Big(\|Y\|_2 \ge \mathbb{E}\|Y\|_2 + \sigma\sqrt{2t}\Big) \leq  e^{-t}.
\end{align}
Using $\mathbb{E}\|Y\|_2\le \|\mathcal{T}_Sx^\star\|_2+\mathbb{E}\|\xi\|_2 \le \|x^\star\|_2+\sigma\sqrt{d}$, we obtain from \eqref{eq:gauss-conc-norm} that for all $t\ge 0$,
\begin{align}\label{eq:norm-upper-prob}
    \mathbb{P}\Big(\|Y\|_2 \ge \|x^\star\|_2+\sigma\sqrt{d}+\sigma\sqrt{2t}\Big) \le e^{-t}.
\end{align}
Consequently, for all $t\ge 0$,
\begin{align}\label{eq:tail-Y4}
    \mathbb{P} \Big(\|Y\|_2^4 \ge \big(\|x^\star\|_2+\sigma\sqrt{d}+\sigma\sqrt{2t}\big)^4\Big) \leq  e^{-t}.
\end{align}
In particular, under the standing assumption $\|x^\star\|_2\le C\sigma\sqrt{d}$, there exists a universal constant $A>0$ such that
\begin{align}\label{eq:tail-Y4-simplified}
    \mathbb{P}\Big(\|Y\|_2^4 \ge A\sigma^4(d+t)^2\Big) \leq  e^{-t}
    \qquad\text{for all }t\ge 0.
\end{align}
This tail behavior is sub-Weibull($1/2$): there exists a universal constant $v_4\asymp \sigma^4 d^2$ such that the centered random variable $X\triangleq \|Y\|_2^4-\mu$ satisfies a Bernstein-type tail bound of the form (see, e.g., \cite{bong2023tight})
\begin{align}\label{eq:X-subweibull-tail}
    \mathbb{P}\big(|X|\ge r\big) \leq 2\exp\left(-c\,\min\left\{\frac{r^2}{v_4^2},\;\Big(\frac{r}{v_4}\Big)^{1/2}\right\}\right)
    \qquad\text{for all }r\ge 0,
\end{align}
for a universal constant $c>0$.

Let $X_i\triangleq \|Y_i\|_2^4-\mu$. Then $\{X_i\}_{i=1}^n$ are independent, mean-zero, and satisfy \eqref{eq:X-subweibull-tail}
with the same scale $v_4$.
Applying the Bernstein-type concentration inequality for weighted sums of independent sub-Weibull($\alpha$) variables with $\alpha=\frac{1}{2}$ to $\sum_{i=1}^n a_iX_i$ with $a_i\equiv 1/n$ yields (see e.g.,~\cite{bong2023tight}): for any $\delta\in(0,1)$ and $s\triangleq \log\frac{2}{\delta}$, with probability at least $1-\delta$,
\begin{align}\label{eq:Delta-bernstein-subweibull}
    |\Delta_n| =\left|\frac{1}{n}\sum_{i=1}^n X_i\right| \le C\,v_4\,\max\left\{\sqrt{\frac{s}{n}}, \frac{s^2}{n}\right\},
\end{align}
for a universal constant $C>0$.
Moreover, if $n\ge c\,s^3$ (so that $s^2/n\le \sqrt{s/n}$), then \eqref{eq:Delta-bernstein-subweibull} simplifies to
\begin{align}\label{eq:Delta-bernstein-simplified}
    |\Delta_n| \le C\,v_4\,\sqrt{\frac{s}{n}},
\end{align}
with probability at least $1-\delta$.

On the event \eqref{eq:Delta-bernstein-subweibull}, use the elementary inequality $\sqrt{a+b}\le \sqrt{a}+\frac{|b|}{2\sqrt{a}}$ valid for $a>0$ to obtain from \eqref{eq:Ln-mu-Delta}:
\begin{align}\label{eq:Ln-dev-no-psi}
    L_n \le\frac{2}{\sigma^2}\sqrt{\mu} +\frac{1}{\sigma^2}\frac{|\Delta_n|}{\sqrt{\mu}}.
\end{align}
Finally, the explicit fourth-moment identity yields,
\begin{align}
    \mu=\mathbb{E}\|Y\|_2^4  = \|x^\star\|_2^4 + 2\sigma^2(d+2)\|x^\star\|_2^2 + \sigma^4 d(d+2).
\end{align}
Recall that $v_4 \asymp \sigma^4 d^2$, and that by assumption, we have $\sqrt{\mu}\asymp \sigma^2 d$,  under $\|x^\star\|_2\le C\sigma\sqrt d$, hence $\sqrt{\mu} / \sigma^2 \le \tilde C d$ and $\frac{v_4}{\sigma^2\sqrt{\mu}}\lesssim d$.
Combining these estimates with \eqref{eq:Ln-dev-no-psi} and \eqref{eq:Delta-bernstein-subweibull} gives, with probability at least $1-\delta$,
\begin{align}
    L_n \le \tilde C d \Big(1 + C'\max\{\sqrt{s/n},\,s^2/n\}\Big),
\end{align}
where $s=\log\frac{2}{\delta}$. In particular, if $n\ge c\,s^3$ then
\begin{align}
    L_n \le \tilde C d \Big(1 + C'\sqrt{s/n}\Big).
\end{align}
Substituting this bound on $L_n$ into \eqref{eq:step1-final} yields the desired high-probability deviation bound.

\paragraph{Step 4: Supremum over $u \in \mathbb{S}^{d-1}$.}
Observe that the high-probability event used in~\eqref{eq:step1-final} to remove conditioning, depends only on the sample $\{Y_i\}_{i=1}^n$ (equivalently, only on $\frac{1}{n}\sum_{i=1}^n \|Y_i\|_2^4$), and is therefore independent of the choice of direction $u$.
Moreover, the symmetrization and chaining bound from Lemma~\ref{lem:symm-chaining-osc} is uniform over all unit directions in the sense that its proof uses only $\|u\|_2=1$ and introduces no constants depending on $u$. Consequently, the pointwise bound~\eqref{eq:step1-final} holds simultaneously for every $u\in\mathbb{S}^{d-1}$.
Therefore, with probability at least $1-\delta$,
\begin{align}
    \sup_{u \in \mathbb{S}^{d-1}}\ \sup_{x\in\mathcal{B}} \abs{\big\langle u,(M_n-M)(x)-(M_n-M)(x_0)\big\rangle} \leq C_1 \frac{R\sqrt d}{\sqrt n}\left[d + C_2 d \sqrt{\frac{s}{n}}\right],
    \label{eq:maxK-no-union}
\end{align}
where $s=\log\frac{2}{\delta}$.
This completes the proof.
\end{proof}

\section{Finite-sample analysis} \label{sec:finite-sample-analysis-results}
This appendix proves the finite–sample main results for EM in MRA at low SNR and assembles them into the main lower–bound conclusions.

\subsection{Proof of Theorem~\ref{thm:finite-sample-tracking-groundtruth}} \label{sec:proofOfFiniteSampleTracking}

We prove Theorem~\ref{thm:finite-sample-tracking-groundtruth} in two steps.
\paragraph{Step 1: One–step deviation recursion.}
Let $e_t\triangleq \|\hat{x}^{(t)} - x^\star\|_2$. Using the triangle inequality and then Assumption~\ref{assump:finite-sample-EM} (A1)--(A2), 
\begin{align}
    e_{t+1}
    &=  \big\| M_n(\hat{x}^{(t)};\mathcal{Y}) - x^\star \big\|_2 
    \\
    &\le \big\|M_n(\hat{x}^{(t)};\mathcal{Y}) - M(\hat{x}^{(t)}) \big\|_2
      + \big\|M(\hat{x}^{(t)}) - M(x^\star)\big\|_2
      + \big\|M(x^\star) - x^\star\big\|_2 \\
    &\le \underbrace{\sup_{x\in\mathcal{B}}  \big\| M_n(x;\mathcal{Y}) - M(x)\big\|}_{\le  \varepsilon_M(n, \delta)\ \text{ on }\mathcal{E}}
       \ +\ \underbrace{\big\|M(\hat{x}^{(t)}) - M(x^\star)\big\|}_{\le  \kappa \big\|\hat{x}^{(t)} - x^\star\big \| \ \text{ by (A1)}}
       \ +\ \underbrace{\big\|M(x^\star) - x^\star\big\|}_{= 0} \\
    &\le \kappa  e_t + \varepsilon_M(n, \delta).
\end{align}

\paragraph{Step 2: Unrolling the recursion.}
By induction,
\begin{align}
    e_t \leq \kappa^t e_0 + \sum_{s=0}^{t-1} \kappa^s  \varepsilon_M(n, \delta)= \kappa^t e_0 + \frac{1-\kappa^t}{1-\kappa} \varepsilon_M(n, \delta),
    \qquad t=0,1,\dots,T,
\end{align}
which is exactly \eqref{eq:tracking-to-groundtruth}. Taking $\limsup_{t\to\infty}$ and using $\kappa\in[0,1)$ yields
\begin{align}
    \limsup_{t\to\infty} e_t \leq \frac{\varepsilon_M(n, \delta)}{1-\kappa},
\end{align}
establishing \eqref{eq:asymptotic-bias-floor}.
Finally, by Assumption~\ref{assump:finite-sample-EM} (A3) and the definition of the iteration, $\hat{x}^{(t)}\in\mathcal{B}$ for all $t\le T$ on $\mathcal{E}$, so that Assumption~\ref{assump:finite-sample-EM} (A1)–(A2) are applicable at each step. This completes the proof.

\subsection{Proof of Proposition~\ref{lem:subgaussian-epsn}} \label{sec:proofOfUniformBoundGauss}

The proof of the proposition is based on the results obtained in Appendix~\ref{sec:preliminaries-finite-sample}. We prove the proposition in three steps.
Write $M(x)\triangleq \mathbb{E}[\phi_x(Y)]$ and $M_n(x)\triangleq \frac{1}{n}\sum_{i=1}^n \phi_x(Y_i)$.
Fix an anchor $x_0\in\mathcal{B}$. By the triangle inequality,
\begin{align}
    \sup_{x\in\mathcal{B}} & \|M_n(x)-M(x)\|_2 \\ & \le \|M_n(x_0)-M(x_0)\|_2 + \sup_{x\in\mathcal{B}}\|(M_n-M)(x)-(M_n-M)(x_0)\|_2.
    \label{eq:decomp-anchor-osc}
\end{align}
Let us upper bound each one of the two terms at the right-hand side of \eqref{eq:decomp-anchor-osc}, starting with the first term. Let $V\triangleq (M_n-M)(x_0)\in\mathbb{R}^d$.
Using $\|V\|_2=\sup_{u\in\mathbb{S}^{d-1}}|\langle u,V\rangle|$, we may apply Lemma~\ref{lem:EF-bounds-MRA-fixedx} at the fixed point $x_0$. Concretely, there exists a universal constant $C>0$ such that,
with probability at least $1-\delta/2$,
\begin{align}
    \|M_n(x_0)-M(x_0)\|_2 \leq C\,\sigma\sqrt{\frac{d}{n}}\sqrt{d+\log\frac{2}{\delta}}.
    \label{eq:anchor-term-bound}
\end{align}
As for the second term at the right-hand side of \eqref{eq:decomp-anchor-osc}, we have,
\begin{align}
    \|(M_n-M)(x)-(M_n-M)(x_0)\|_2 = \sup_{u\in\mathbb{S}^{d-1}} \left|\big\langle u,(M_n-M)(x)-(M_n-M)(x_0)\big\rangle\right|.
\end{align}
Therefore,
\begin{align}
    \nonumber \sup_{x\in\mathcal{B}} & \|(M_n-M)(x)-(M_n-M)(x_0)\|_2 \\ & = \sup_{u\in\mathbb{S}^{d-1}}\sup_{x\in\mathcal{B}} \left|\big\langle u,(M_n-M)(x)-(M_n-M)(x_0)\big\rangle\right|.
    \label{eq:osc-dual}
\end{align}
Applying Proposition~\ref{prop:finite-max-u} with confidence level $\delta/2$, we obtain that with probability at least $1-\delta/2$,
\begin{align}
    \sup_{x\in\mathcal{B}}\|(M_n-M)(x)-(M_n-M)(x_0)\|_2 \leq C_1\,\frac{R\sqrt{d}}{\sqrt{n}}\,d\left(1+C_2\sqrt{\frac{\log(2/\delta)}{n}}\right).
    \label{eq:osc-term-bound}
\end{align}
Using \eqref{eq:decomp-anchor-osc}, \eqref{eq:anchor-term-bound}, and~\eqref{eq:osc-term-bound}, with probability at least $1-\delta$, we have
\begin{align}
    \sup_{x\in\mathcal{B}}\|M_n(x)-M(x)\|_2 \le C\,\sigma\sqrt{\frac{d}{n}}\sqrt{d+\log\frac{2}{\delta}} + C_1\,\frac{R\,d^{3/2}}{\sqrt{n}}\left(1+C_2\sqrt{\frac{\log(2/\delta)}{n}}\right).\label{eq:final-term-bound}
\end{align}
Finally, for $n \geq C_2^2 \log(2/\delta)$, we can simplify the right-hand side of \eqref{eq:final-term-bound} as follows
\begin{align}
    \sup_{x\in\mathcal{B}}\|M_n(x)-M(x)\|_2 \le C\,\sigma\sqrt{\frac{d}{n}}\sqrt{d+\log\frac{2}{\delta}} + 2C_1\frac{R\,d^{3/2}}{\sqrt{n}}.
\end{align}

\subsection{Necessary conditions on the EM sample complexity} \label{sec:app-sample-complexity}
In this subsection, we prove a general necessary condition for the sample complexity of empirical EM fixed points. Under mild few assumptions listed below, we prove that no empirical fixed point inside the basin can approach $x^\star$ beyond the intrinsic statistical noise floor. In particular, even if EM converges algorithmically, it cannot beat the finite-sample floor unless $n$ is large enough.

To capture a lower bound on the sample complexity, it is necessary to track EM along its \emph{slowest} contractive direction. Let $\lambda_{\max}<1$ be the largest non-mean eigenvalue of the population Jacobian $J(x^\star)$ with unit eigenvector $u_{\max}\perp\mathbf{1}$, and define the rank-one projector
\begin{align}
    P \triangleq u_{\max}u_{\max}^\top .
\end{align}
The subspace $T=\mathrm{range}(P)$ is $J(x^\star)$-invariant and the linearized dynamics reduce to
\begin{align}
    P J(x^\star) P = \lambda_{\max} P.
\end{align}
Therefore, a lower bound on the projected fluctuation $P\big(M_n(x^\star;\mathcal{Y})-M(x^\star)\big)$ yields, via an amplification factor $(1-\lambda_{\max})^{-1}$, a corresponding lower bound on the attainable estimation error.

\begin{assum}[Conditions for the sample-complexity bound]
\label{assum:gap-sb-mr-P-general}
Let $\mathcal{B}\subset\mathbb{R}^d$ be a compact basin containing $x^\star$.
Let $M:\mathbb{R}^d\to\mathbb{R}^d$ and $M_n(\cdot;\mathcal{Y})$ denote the population and empirical EM operators, respectively. Let $J(x^\star)$ denote the population Jacobian at $x^\star$, and assume $J(x^\star)=J(x^\star)^\top\succeq 0$.
Let $\lambda_{\max}$ be the largest non-mean eigenvalue of $J(x^\star)$ with corresponding unit eigenvector $u_{\max}\perp\mathbf{1}$, and assume $\lambda_{\max}<1$. Set $P\triangleq u_{\max}u_{\max}^\top$. Fix $p_0>0$, and assume that on an event $\Omega$ with $\mathbb{P}(\Omega)\ge p_0$ the following hold:

\begin{enumerate}
\item \emph{(B0) Slow-mode gap scaling.}
There exist constants $c_\lambda,C_\lambda>0$ and an integer $k \ge 1$ such that
\begin{align}
    c_\lambda\,\mathrm{SNR}^{k} \leq  1-\lambda_{\max} \leq C_\lambda\,\mathrm{SNR}^{k}.
    \label{eq:gap-scaling-general}
\end{align}

\item \emph{(B1) Pointwise fluctuation at the truth.}
There exist constants $c_0>0$ and $\alpha\ge 0$ such that
\begin{align}\label{eq:fluctuations-lower-bound-general}
    \frac{\big\|P\big(M_n(x^\star;\mathcal{Y})-M(x^\star)\big)\big\|}{\|x^\star\|} \geq     c_0\,\frac{\mathrm{SNR}^{\alpha}}{\sqrt{n}} \qquad\text{on }\Omega.
\end{align}

\item \emph{(B2) Regularity near $x^\star$.}
There exist constants $K_R,K_E>0$ and a radius $R_0>0$ such that for all $\Delta$ with
$\|\Delta\|\le R_0 \|x^\star\|$ and $x^\star+\Delta\in\mathcal{B}$,
\begin{align}
    \|M(x^\star+\Delta)-M(x^\star)-J(x^\star)\Delta\|
    &\le K_R\|\Delta\|^2,
    \label{eq:reg-pop-general}\\
    \big\| [M_n(x^\star+\Delta;\mathcal{Y})-M(x^\star+\Delta)] - [M_n(x^\star;\mathcal{Y})-M(x^\star)] \big\|
    &\le \frac{K_E}{\sqrt{n}} \|\Delta\|.
    \label{eq:reg-emp-general}
\end{align}
\end{enumerate}
\end{assum}

\begin{thm}[Sample complexity lower bound (general scaling)]
\label{thm:necessary-small-target-P-general}
Assume Assumption~\ref{assum:gap-sb-mr-P-general} holds and $\|x^\star\|>0$.
Choose $R_0>0$ such that
\begin{align}
    0<R_0 \leq \min\left\{\frac{c_0\,\mathrm{SNR}^{\alpha}}{4K_E},\ \frac{1-\lambda_{\max}}{2K_R\,\|x^\star\|}\right\}.
    \label{eq:smallness-P-general}
\end{align}
Then, on the event $\Omega$ from Assumption~\ref{assum:gap-sb-mr-P-general} (so $\mathbb{P}(\Omega)\ge p_0>0$), every empirical fixed point $x_n^\star\in\mathcal{B}$ (i.e., $M_n(x_n^\star)=x_n^\star$) satisfying
$\|x_n^\star-x^\star\|/\|x^\star\|\le R_0$
obeys the normalized error floor
\begin{align}
\label{eq:local-floor-general}
    \frac{\|x_n^\star-x^\star\|}{\|x^\star\|} \geq  \frac{c_0}{2} \frac{\mathrm{SNR}^{\alpha}}{1-\lambda_{\max}}\frac{1}{\sqrt{n}}.
\end{align}
Consequently, to achieve any target relative accuracy $\|x_n^\star-x^\star\|/\|x^\star\|\le \varepsilon_{\mathrm{abs}}\in(0,R_0]$, it is necessary that
\begin{align}
\label{eq:necessary-clean-P-general}
    n \geq  \frac{C}{\varepsilon_{\mathrm{abs}}^{2}}\cdot \frac{\mathrm{SNR}^{2\alpha}}{(1-\lambda_{\max})^2},
\end{align}
for a constant $C>0$ independent of $n$ and $\mathrm{SNR}$. Equivalently, under the gap scaling Assumption~\ref{assum:gap-sb-mr-P-general}(B0), 
\begin{align}
\label{eq:necessary-clean-P-general-explicit}
    n \gtrsim \mathrm{SNR}^{-2(k-\alpha)}\varepsilon_{\mathrm{abs}}^{-2}.
\end{align}
\end{thm}

\begin{proof}[Proof of Theorem~\ref{thm:necessary-small-target-P-general}]
We prove the statement in steps.
Fix an empirical fixed point $x_n^\star\in\mathcal{B}$, set $\Delta \triangleq x_n^\star-x^\star$, and recall that $M(x^\star)=x^\star$.

\paragraph{Step 1: Fixed-point decomposition.}
From $M_n(x_n^\star)=x_n^\star$ we have
\begin{align}
    P\Delta &= P\big(M_n(x_n^\star)-x^\star\big) \notag\\
    &= P\big(M_n(x_n^\star)-M(x_n^\star)\big) + P\big(M(x_n^\star)-M(x^\star)\big).
\end{align}
Writing $x_n^\star=x^\star+\Delta$ and adding--subtracting $J(x^\star)\Delta$ gives
\begin{align}
    P\Delta &= P\big(M_n(x^\star+\Delta)-M(x^\star+\Delta)\big) + P\big(M(x^\star+\Delta)-M(x^\star)-J(x^\star)\Delta\big) + PJ(x^\star)\Delta.
\end{align}
Rearranging,
\begin{align}
    P\big(I-J(x^\star)\big)\Delta = P\big(M_n(x^\star+\Delta)-M(x^\star+\Delta)\big) + PR(\Delta),
\end{align}
where $PR(\Delta)\triangleq P\big(M(x^\star+\Delta)-M(x^\star)-J(x^\star)\Delta\big)$.
Finally, add and subtract the fluctuation at the truth to obtain
\begin{align}
    P\big(I-J(x^\star)\big)\Delta &= T_P + PR(\Delta) + PE(\Delta),
    \label{eq:fpid-expanded-P-general}
\end{align}
where
\begin{align}
    T_P \triangleq P\big(M_n(x^\star)-M(x^\star)\big), \quad PE(\Delta)\triangleq P\Big([M_n-M](x^\star+\Delta)-[M_n-M](x^\star)\Big).
\end{align}

\paragraph{Step 2: Projected bounds.}
Taking norms and applying the triangle inequality to~\eqref{eq:fpid-expanded-P-general} yields
\begin{align}
    \|T_P\| &\le \|P(I-J(x^\star))\Delta\|+\|PR(\Delta)\|+\|PE(\Delta)\|.
    \label{eq:step1-P-general}
\end{align}
By $PJ(x^\star)P=\lambda_{\max}P$,
\begin{align}
    \|P(I-J(x^\star))\Delta\| = (1-\lambda_{\max})\|P\Delta\| \le (1-\lambda_{\max})\|\Delta\|.
    \label{eq:gap-ineq-P-up-general}
\end{align}
On the event $\Omega$, Assumption (B1) gives
\begin{align}
    \|T_P\| = \big\|P(M_n(x^\star)-M(x^\star))\big\| \ge c_0\,\mathrm{SNR}^{\alpha}\,\frac{\|x^\star\|}{\sqrt{n}}.
    \label{eq:SB-P-general}
\end{align}
Moreover, Assumption (B2) implies that for all $\|\Delta\|\le R_0\|x^\star\|$ with $x^\star+\Delta\in\mathcal{B}$,
\begin{align}
    \|PR(\Delta)\| &\le K_R\|\Delta\|^2, \\ \|PE(\Delta)\| &\le \frac{K_E}{\sqrt{n}}\|\Delta\|.
    \label{eq:MR-P-general}
\end{align}
Plugging~\eqref{eq:gap-ineq-P-up-general}, \eqref{eq:MR-P-general}, and \eqref{eq:SB-P-general} into~\eqref{eq:step1-P-general} yields
\begin{align}
    c_0\,\mathrm{SNR}^{\alpha}\,\frac{\|x^\star\|}{\sqrt{n}} \le (1-\lambda_{\max})\|\Delta\| + K_R\|\Delta\|^2 + \frac{K_E}{\sqrt{n}}\|\Delta\|.
    \label{eq:master-again-P-general}
\end{align}

\paragraph{Step 3: Error floor and necessary sample size.}
Assume the local regime $\|\Delta\|\le R_0\|x^\star\|$ with $R_0$ satisfying~\eqref{eq:smallness-P-general}.
Then
\begin{align}
    K_R\|\Delta\|^2 \le K_R(R_0\|x^\star\|)\|\Delta\| \le \frac{1-\lambda_{\max}}{2}\,\|\Delta\|,
    \\
    \frac{K_E}{\sqrt{n}}\|\Delta\| \le \frac{K_E}{\sqrt{n}}(R_0\|x^\star\|) \le \frac{c_0}{4}\,\mathrm{SNR}^{\alpha}\,\frac{\|x^\star\|}{\sqrt{n}}.
\end{align}
Using these bounds in~\eqref{eq:master-again-P-general} gives
\begin{align}
    c_0\,\mathrm{SNR}^{\alpha}\,\frac{\|x^\star\|}{\sqrt{n}} &\le \frac{3}{2}(1-\lambda_{\max})\|\Delta\| + \frac{c_0}{4}\,\mathrm{SNR}^{\alpha}\,\frac{\|x^\star\|}{\sqrt{n}}.
\end{align}
Rearranging yields
\begin{align}
    \frac{\|\Delta\|}{\|x^\star\|} \ge \frac{c_0}{2}\cdot \frac{\mathrm{SNR}^{\alpha}}{1-\lambda_{\max}}\cdot\frac{1}{\sqrt{n}},
\end{align}
which is~\eqref{eq:local-floor-general}.
Finally, imposing $\|\Delta\|/\|x^\star\|\le \varepsilon_{\mathrm{abs}}$ implies
\begin{align}
    \sqrt{n}\geq  \frac{c_0}{2}\cdot \frac{\mathrm{SNR}^{\alpha}}{(1-\lambda_{\max})\,\varepsilon_{\mathrm{abs}}},
\end{align}
equivalently
\begin{align}
    n \geq  \frac{c_0^2}{4}\cdot \frac{\mathrm{SNR}^{2\alpha}}{(1-\lambda_{\max})^2}\cdot \varepsilon_{\mathrm{abs}}^{-2},
\end{align}
which is~\eqref{eq:necessary-clean-P-general} after absorbing constants into $C$.
If Assumption~\ref{assum:gap-sb-mr-P-general}(B0) holds, substituting $1-\lambda_{\max}\asymp \mathrm{SNR}^{k}$ yields~\eqref{eq:necessary-clean-P-general-explicit}.
\end{proof}

\subsection{Empirical fluctuations lower bound for the MRA}

Next we verify that Assumption~\ref{assum:gap-sb-mr-P-general}(B1) holds in the MRA setting. 
As in Section~\ref{subsec:local-convergence}, we work in the low-SNR regime $x^\star=\beta v$ and $\|v\|_2=1$, with $\beta/\sigma\ll 1$.

\paragraph{Real-space representation of the $\{k,-k\}$ Fourier block.}
Throughout this subsection, we work with the real-space representation of the Jacobian eigenvectors on each non-mean Fourier block (equivalently to the complex $\{k,-k\}$ block vectors in~\eqref{eq:uk-wk-def}). 
Fix a non-mean frequency $k\in\{1,\dots,(d-1)/2\}$ and write the DFT coefficient of $v$ at frequency $k$ as $\s{V}[k]=|\s{V}[k]|e^{i\phi_k}$.
The conjugate-symmetric pair $\{k,-k\}$ corresponds to a real two-dimensional invariant subspace of $\mathbb{R}^d$. Taking the inverse Fourier transform of the $\{k,-k\}$-block basis in~\eqref{eq:uk-wk-def} yields a cosine--sine basis for this subspace. 
Concretely, define the real vectors $u_k,w_k\in\mathbb{R}^d$ by
\begin{align}\label{eq:ukwk-realspace}
    u_k[m] \triangleq \sqrt{\frac{2}{d}}\cos\Big(\frac{2\pi k m}{d}+\phi_k\Big),
    \qquad w_k[m] \triangleq \sqrt{\frac{2}{d}}\sin\Big(\frac{2\pi k m}{d}+\phi_k\Big),
\end{align}
for $m=0,1,\dots,d-1$, so that $\{u_k,w_k\}$ is an orthonormal basis of this real block, and agrees with the complex two-coordinate representation in~\eqref{eq:uk-wk-def}.

Let $\{\mathcal{T}_\ell\}_{\ell=0}^{d-1}$ be the circular shifts, $(\mathcal{T}_\ell x)[t]=x[t-\ell]$ (indices modulo $d$), and set $\theta_\ell \triangleq \frac{2\pi k\ell}{d}$.
Then shifting $u_k$ and $w_k$ in~\eqref{eq:ukwk-realspace} simply subtracts $\theta_\ell$ from the phase. Using the trigonometric identities $\cos(a-b)=\cos a\cos b+\sin a\sin b$ and $\sin(a-b)=\sin a\cos b-\cos a\sin b$ with $a=\frac{2\pi k m}{d}+\phi_k$ and $b=\theta_\ell$, we obtain for every $\ell$:
\begin{align}\label{eq:shift-acts-rotation}
    \mathcal{T}_\ell u_k &= \cos(\theta_\ell)\,u_k + \sin(\theta_\ell)\,w_k,\\
    \mathcal{T}_\ell w_k &= -\sin(\theta_\ell)\,u_k + \cos(\theta_\ell)\,w_k.
\end{align}
Thus, on the $2$D subspace $\mathrm{span}\{u_k,w_k\}$, the shift operator $\mathcal{T}_\ell$ acts as a real planar rotation by angle $\theta_\ell$.

\begin{lem}
\label{lem:first-order-cancel-wk}
Fix $Y\in\mathbb{R}^d$. With $w_k$ defined by \eqref{eq:ukwk-realspace}, we have for every $k\in\{1,\dots,(d-1)/2\}$,
\begin{align}
    \label{eq:first-order-cancel}
    \sum_{\ell=0}^{d-1}\Big\langle \mathcal{T}_\ell^{-1}Y,v\Big\rangle \Big\langle \mathcal{T}_\ell^{-1}Y,w_k\Big\rangle = 0.
\end{align}
\end{lem}

\begin{proof}[Proof of Lemma~\ref{lem:first-order-cancel-wk}]
Define
\begin{align}
\label{eq:AY-def}
    A(Y)\triangleq \frac{1}{d} \sum_{\ell=0}^{d-1}\mathcal{T}_\ell^{-1}(YY^\top)\mathcal{T}_\ell .
\end{align}
Using the identity $(a^\top b)(a^\top c)=b^\top(aa^\top)c$ with
$a=\mathcal{T}_\ell^{-1}Y$, $b=v$, and $c=w_k$, we obtain
\begin{align}
    \label{eq:sum-as-quadratic}
    \frac{1}{d}\sum_{\ell=0}^{d-1}\Big\langle \mathcal{T}_\ell^{-1}Y,v\Big\rangle \Big\langle \mathcal{T}_\ell^{-1}Y,w_k\Big\rangle
    &= \frac{1}{d} \sum_{\ell=0}^{d-1} v^\top\Big((\mathcal{T}_\ell^{-1}Y)(\mathcal{T}_\ell^{-1}Y)^\top\Big)w_k \\
    &= v^\top\Big(\frac{1}{d}\sum_{\ell=0}^{d-1}\mathcal{T}_\ell^{-1}(YY^\top)\mathcal{T}_\ell\Big)w_k = v^\top A(Y)\,w_k . \nonumber
\end{align}

The matrix $A(Y)$ commutes with every shift $\mathcal{T}_r$ (since the sum in \eqref{eq:AY-def} is taken over all shifts),
hence $A(Y)$ is circulant and is diagonalized by the DFT basis $\{f_m\}_{m=0}^{d-1}$:
\begin{align}
\label{eq:AY-eigs}
    A(Y)f_m = \lambda_m f_m,  \qquad \lambda_m = \,|\s Y[m]|^2 .
\end{align}
Because $Y\in\mathbb{R}^d$, its Fourier coefficients satisfy conjugate symmetry, so $\lambda_k=\lambda_{-k}$.
Therefore $A(Y)$ acts as a scalar on the real two-dimensional invariant subspace
$\mathrm{span}\{f_k,f_{-k}\}$, and since $w_k\in \mathrm{span}\{f_k,f_{-k}\}$ (by construction in \eqref{eq:ukwk-realspace}),
\begin{align}
    \label{eq:AY-on-wk}
    A(Y)w_k = \lambda_k w_k = |\s Y[k]|^2\,w_k .
\end{align}
Combining \eqref{eq:sum-as-quadratic} and \eqref{eq:AY-on-wk} yields
\begin{align}
\label{eq:reduce-to-vwk}
    v^\top A(Y)w_k = |\s Y[k]|^2\,\langle v,w_k\rangle .
\end{align}

Finally, the real block basis $\{u_k,w_k\}$ in \eqref{eq:ukwk-realspace} is chosen so that the $\{\pm k\}$ Fourier block of $v$ is aligned with $u_k$ (equivalently, \eqref{eq:uk-wk-def}), hence $\langle v,w_k\rangle=0$. Substituting into \eqref{eq:reduce-to-vwk} gives $v^\top A(Y)w_k=0$, which together with \eqref{eq:sum-as-quadratic} proves \eqref{eq:first-order-cancel}.
\end{proof}

Recall the definitions~\eqref{eq:s-logits}--\eqref{eq:eta-def},
\begin{align}
    s_\ell(Y)=\frac{\beta}{\sigma^2}\langle z_\ell,v\rangle, \qquad \eta_\ell(Y)=s_\ell(Y)-\overline{s}(Y), \qquad \overline{\eta(Y)z}=\frac{1}{d}\sum_{\ell=0}^{d-1}\eta_\ell(Y)\,z_\ell. \label{eq:eta-def-lem}
\end{align}

\begin{proposition}
\label{prop:uk-plateau-updated}
Let $M$ and $M_n(\cdot;\mathcal{Y})$ be the population and empirical EM operators, and define $ Z \triangleq M_n(x^\star;\mathcal{Y})-M(x^\star)$.
Assume the low-SNR scaling $x^\star=\beta v$ with $\|v\|_2=1$, and fix $\sigma>0$. Assume moreover that the Fourier spectrum of $v$ is nonvanishing on non-mean frequencies:
\begin{align}
    \label{eq:V-nonvanish-prop}
    \s{V}[k]\neq 0 \qquad \text{for all } k\in\{1,\dots,(d-1)/2\}.
\end{align}
Let $K\subset\{1,\dots,(d-1)/2\}$ and let $P$ be the orthogonal projector onto $\mathrm{span}\{w_k:\,k\in K\}$, where $\{w_k\}$ are defined by \eqref{eq:ukwk-realspace}.
Then there exists $\beta_0>0$ such that for all $0<\beta\le \beta_0$ and every fixed $\eta\in(0,1)$, there exist constants $c_2,C_4>0$ for which the following holds: for any choice of $k_0\in K$,
\begin{align}
    \label{eq:Prob-lb-updated}
    \mathbb{P}\left(\frac{\|PZ\|_2}{\|x^\star\|_2} \geq \frac{\eta\sqrt{c_2}\,\tau_{k_0}}{\sqrt{n}}\cdot\frac{\beta}{\sigma} \right) \geq \frac{(1-\eta^2)^2}{C_4},
\end{align}
where $\tau_{k_0}^2 > 0$.
\end{proposition}

\begin{proof}[Proof of Proposition~\ref{prop:uk-plateau-updated}]
We split the argument into three steps.

\paragraph{Step 1: Second-order responsibility expansion and first-order cancellation on $w_k$.}
Recall the per-sample EM update
\begin{align}
    \label{eq:phi-def}
    \phi_{x^\star}(Y)\;\triangleq\;\sum_{\ell=0}^{d-1}\gamma_\ell^{(\beta)}(Y)\,z_\ell,
    \qquad z_\ell\triangleq \mathcal{T}_\ell^{-1}Y, \qquad M_n(x^\star;\mathcal{Y})=\frac{1}{n}\sum_{i=1}^n \phi_{x^\star}(Y_i).
\end{align}
By Lemma~\ref{lem:resp-near-uniform-MRA-gamma}, for each fixed $Y$,
\begin{align}
    \label{eq:gamma-2nd}
    \gamma_\ell^{(\beta)}(Y) = \frac{1}{d}+\frac{1}{d}\eta_\ell(Y) + \frac{1}{2d}\big(\eta_\ell(Y)^2-\overline{\eta^2}(Y)\big)+R_\ell(\beta;Y),
\end{align}
with the uniform remainder bound \eqref{eq:MRA-R-bound}. Substituting \eqref{eq:gamma-2nd} into \eqref{eq:phi-def} and using $\frac{1}{d}\sum_{\ell=0}^{d-1}z_\ell=\Pi_{\mathrm{mean}}Y$ yields
\begin{align}
    \label{eq:phi-exp-2nd}
    \phi_{x^\star}(Y) &= \Pi_{\mathrm{mean}}Y +\overline{\eta(Y)z} +\frac{1}{2}\,\overline{\big(\eta(Y)^2-\overline{\eta^2}(Y)\big)z} + R(Y),
\end{align}
where $R(Y)\triangleq \sum_{\ell=0}^{d-1}R_\ell(\beta;Y)\,z_\ell$.

Fix $k\in\{1,\dots,(d-1)/2\}$. Since $w_k$ is a non-mean Fourier vector, we have
\begin{align}
    \label{eq:mean-orth-wk}
    \langle \Pi_{\mathrm{mean}}Y,w_k\rangle=0.
\end{align}
Moreover, by Lemma~\ref{lem:first-order-cancel-wk},
\begin{align}
    \big\langle \overline{\eta(Y)z},\,w_k\big\rangle \equiv 0.
\end{align}
Combining \eqref{eq:phi-exp-2nd}--\eqref{eq:first-order-cancel}, the leading contribution to $\langle \phi_{x^\star}(Y),w_k\rangle$ comes from the quadratic term $\frac{1}{2}\,\overline{\big(\eta(Y)^2-\overline{\eta^2}(Y)\big)z}$.

\paragraph{Step 2: Leading cubic Gaussian term and its variance.}
Write $Y=\mathcal{T}_S(\beta v)+\xi$ with $\xi=\sigma g$ and $g\sim\mathcal{N}(0,I_d)$.
From the definition of $\eta_\ell(Y)$ (via $s_\ell(Y)$ and centering), we have the low-SNR expansion
\begin{align}
    \label{eq:eta-exp}
    \eta_\ell(Y) &=\frac{\beta}{\sigma}\,A_\ell(g) + O\left(\frac{\beta^2}{\sigma^2}\right),
    \qquad A_\ell(g)\triangleq \big\langle g,\mathcal{T}_\ell v\big\rangle -\frac{1}{d}\sum_{r=0}^{d-1}\big\langle g,\mathcal{T}_r v\big\rangle,
\end{align}
where the $O(\beta^2/\sigma^2)$ term is uniform in $\ell$ for fixed $d$. Also,
\begin{align}
    \label{eq:z-exp}
    z_\ell=\mathcal{T}_\ell^{-1}Y = \sigma\,\mathcal{T}_\ell^{-1}g + O(\beta),
\end{align}
uniformly in $\ell$. Plugging \eqref{eq:eta-exp}--\eqref{eq:z-exp} into the quadratic term of \eqref{eq:phi-exp-2nd} gives
\begin{align}
    \label{eq:wk-leading-cubic}
    \big\langle \phi_{x^\star}(Y),w_k\big\rangle &= \frac{\beta^2}{\sigma}\,G_k(g) + r_k,
\end{align}
where
\begin{align}
    \label{eq:Gk-def}
    G_k(g) &\triangleq \left\langle \frac{1}{2d}\sum_{\ell=0}^{d-1}\Big(A_\ell(g)^2-\overline{A(g)^2}\Big)\,\mathcal{T}_\ell^{-1}g,\; w_k \right\rangle,
\end{align}
and the remainder satisfies, for fixed $d$ and $k$,
\begin{align}
    \label{eq:rk-bound}
    \mathbb{E}[r_k^2] &=\; O\left(\frac{\beta^6}{\sigma^2}\right) \qquad\text{as }\beta/\sigma\to 0.
\end{align}
The random variable $G_k(g)$ is a centered polynomial of total degree $3$ in the Gaussian vector $g$; define
\begin{align}
    \label{eq:tau-def}
    \tau_k^2 \triangleq \mathrm{Var}(G_k(g)).
\end{align}
By Lemma~\ref{lem:tau-positive} and the nonvanishing spectrum assumption on $v$, we have $\tau_k^2>0$ for every
$k\in\{1,\dots,(d-1)/2\}$.

\paragraph{Step 3: Paley--Zygmund lower bound for $\|PZ\|_2$.}
Fix $k_0\in K$ and set
\begin{align}
    \label{eq:Q-def}
    Q \triangleq \langle w_{k_0},Z\rangle = \frac{1}{n}\sum_{i=1}^n\Big(\langle \phi_{x^\star}(Y_i),w_{k_0}\rangle-\mathbb{E}\langle \phi_{x^\star}(Y),w_{k_0}\rangle\Big).
\end{align}
Using \eqref{eq:wk-leading-cubic}--\eqref{eq:rk-bound} and independence of the samples,
\begin{align}
    \label{eq:Q2-scale}
    \mathbb{E}[Q^2] &=\frac{1}{n}\mathrm{Var}\big(\langle \phi_{x^\star}(Y),w_{k_0}\rangle\big) = \frac{\beta^4}{\sigma^2}\cdot\frac{\tau_{k_0}^2}{n}\,\big(1+o(1)\big)
\qquad\text{as }\beta/\sigma\to 0.
\end{align}
Moreover, since $Q$ is an average of degree-$3$ Gaussian polynomials, hypercontractivity implies the moment comparison
\begin{align}
    \label{eq:Q4-bound}
    \mathbb{E}[Q^4] \leq  C_4\,\mathbb{E}[Q^2]^2,
\end{align}
where $C_4>0$ depends only on $d$ and $v$ (in particular, not on $n,\beta,\sigma$).
Applying Paley--Zygmund to $U=Q^2$ gives, for any $\eta\in(0,1)$,
\begin{align}
    \label{eq:PZ}
    \mathbb{P}\left(Q^2\ge \eta^2\,\mathbb{E}[Q^2]\right) &\ge (1-\eta^2)^2\,\frac{\mathbb{E}[Q^2]^2}{\mathbb{E}[Q^4]} \geq \frac{(1-\eta^2)^2}{C_4}.
\end{align}
On this event, using $\|PZ\|_2\ge |\langle w_{k_0},PZ\rangle|=|Q|$ and $\|x^\star\|_2=\beta$, we obtain
\begin{align}
    \label{eq:finish}
    \frac{\|PZ\|_2}{\|x^\star\|_2} \geq \frac{|Q|}{\beta} \geq \frac{\eta}{\beta}\sqrt{\mathbb{E}[Q^2]} \geq \frac{\eta\sqrt{c_2}\,\tau_{k_0}}{\sqrt{n}}\cdot\frac{\beta}{\sigma},
\end{align}
for some constant $c_2>0$ (absorbing the $(1+o(1))$ factor from \eqref{eq:Q2-scale} into $c_2$ by choosing $\beta_0$ small enough). Combining \eqref{eq:PZ} and \eqref{eq:finish} yields \eqref{eq:Prob-lb-updated}.
\end{proof}

\begin{lem}
\label{lem:tau-positive}
Fix $k\in\{1,\dots,(d-1)/2\}$ and let $g\sim\mathcal{N}(0,I_d)$.
Assume that the Fourier spectrum of $v$ is nonvanishing on non-mean frequencies~\eqref{eq:V-nonvanish-prop}. Then $\tau_k^2\triangleq \mathrm{Var}(G_k(g))>0$.
\end{lem}

\begin{proof}[Proof of Lemma~\ref{lem:tau-positive}]
Recall $A_\ell(g)=\langle g,\mathcal{T}_\ell v\rangle-\overline{\langle g,\mathcal{T} v\rangle}$ and
\begin{align}
\label{eq:Gk-repeat}
G_k(g)=\left\langle \frac{1}{2d}\sum_{\ell=0}^{d-1}\Big(A_\ell(g)^2-\overline{A(g)^2}\Big)\,\mathcal{T}_\ell^{-1}g, \; w_k \right\rangle.
\end{align}
Consider the DFT over $\ell$ of the sequence $\{A_\ell(g)\}_{\ell=0}^{d-1}$.
By standard correlation identities,
\begin{align}
    \label{eq:Ahat-linear}
    \s{A}[m] = \overline{\s{g}[m]}\,\s{V}[m] \quad \text{for } m\neq 0,
    \qquad
    \s{A}[0]=0.
\end{align} 
Hence the $k$-th Fourier coefficient of the squared sequence
$A_\ell(g)^2$ is
\begin{align}
\label{eq:A2hat-conv}
\s{(A^2)}[k] \;=\;\sum_{m\in\Z_d}\s{A}[m]\;\s{A}[k-m].
\end{align}
Moreover, since $w_k$ is supported on the frequency pair $\{\pm k\}$, the inner product in \eqref{eq:Gk-repeat}
extracts precisely the $\pm k$ Fourier modes in $\ell$ from the weights $\{A_\ell(g)^2-\overline{A(g)^2}\}$, and yields the representation
\begin{align}
\label{eq:Gk-fourier-form}
G_k(g)
=\;c\Big(\overline{\s{g}[k]}\,\s{(A^2)}[k]\Big)+\overline{c}\Big(\overline{\s{g}[-k]}\,\s{(A^2)}[-k]\Big),
\end{align}
for some constant $c\neq 0$ depending only on the convention and on $\s W_k[\pm k]$ (in particular $c\neq 0$ since $\s W_k[k]\neq 0$ by construction).

Choose any index $m_\star\in\Z_d\setminus\{0,k\}$ (e.g. $m_\star=1$ if $k\neq 1$, and $m_\star=2$ if $k=1$).
Then the single summand $m=m_\star$ in \eqref{eq:A2hat-conv} contributes the monomial
\begin{align}
\label{eq:nonzero-monomial}
\overline{\s{g}[k]\,\s{g}[m_\star]\,\s{g}[k-m_\star]}\;\s{V}[m_\star]\s{V}[k-m_\star]
\end{align}
to $\overline{\s{g}[k]}\,\s{(A^2)}[k]$. Under \eqref{eq:V-nonvanish-prop} we have
$\s{V}[m_\star]\s{V}[k-m_\star]\neq 0$, so the coefficient of the monomial in \eqref{eq:nonzero-monomial}
inside $G_k(g)$ is nonzero (since also $c\neq 0$ in \eqref{eq:Gk-fourier-form}).
Therefore $G_k(g)$ is a nonzero polynomial in the Gaussian Fourier coordinates of $g$.

Since the law of $g$ has a density on $\mathbb{R}^d$, a nonzero polynomial cannot be almost surely constant;
hence $\mathrm{Var}(G_k(g))>0$, i.e. $\tau_k^2>0$.
\end{proof}

\subsection{Proof of Theorem~\ref{thm:tight-1d-branch}}
\label{sec:proof-of-sample-complexity-mra}
The proof verifies Assumption~\ref{assum:gap-sb-mr-P-general} with $k=2$ and $\alpha=1/2$ and then applies Theorem~\ref{thm:necessary-small-target-P-general}.

\paragraph{(1) Pointwise fluctuation at the truth (Assumption~\ref{assum:gap-sb-mr-P-general}(B1)).}
Let
\begin{align}
\label{eq:Z-def-proof}
Z \;\triangleq\; M_n(x^\star;\mathcal{Y})-M(x^\star).
\end{align}
Let $k_{\max}$ denote a non-mean block attaining the largest non-mean eigenvalue $\lambda_{\max}$ of
the population Jacobian $J(x^\star)$, and let $u_{\max}$ be the associated unit eigenvector.
In the MRA model, $u_{\max}$ lies in the $w$-direction of the $\{\pm k_{\max}\}$ block, hence we may take
\begin{align}
\label{eq:P-def-proof}
u_{\max}=w_{k_{\max}},
\qquad
P\;\triangleq\;u_{\max}u_{\max}^\top \;=\; w_{k_{\max}}w_{k_{\max}}^\top,
\end{align}
see Corollary~\ref{cor:block-spectral}.

Assume the nonvanishing-spectrum condition $\s{V}[k]\neq 0$ for all $k\in\{1,\dots,(d-1)/2\}$.
Then Proposition~\ref{prop:uk-plateau-updated} (applied with $K=\{k_{\max}\}$) implies that for every fixed
$\eta\in(0,1)$ there exist constants $\beta_0>0$, $c_2>0$, $C_4>0$ and $\tau_{k_{\max}}>0$ such that for all
$0<\beta/\sigma\le \beta_0$,
\begin{align}
    \label{eq:lb-fluct-new}    \mathbb{P}\left(\frac{\|PZ\|_2}{\|x^\star\|_2}  \geq     \frac{\eta\sqrt{c_2}\,\tau_{k_{\max}}}{\sqrt{n}}\cdot\frac{\beta}{\sigma}\right)  \geq  p_{\mathrm{fluc}}, \qquad p_{\mathrm{fluc}}\triangleq \frac{(1-\eta^2)^2}{C_4}.
\end{align}
Thus Assumption~\ref{assum:gap-sb-mr-P-general}(B1) holds with exponent $\alpha=1/2$.

\paragraph{(2) Local regularity near $x^\star$ (Assumption~\ref{assum:gap-sb-mr-P-general}(B2)).}
As established in the same manner as in Section~\ref{subsec:local-convergence},
there exist constants $K_R,K_E>0$ and a radius $R_0>0$ such that for all $\Delta$ with
$\|\Delta\|\le R_0\|x^\star\|$ and $x^\star+\Delta\in\mathcal{B}$,
\begin{align}
\label{eq:reg-pop-proof}
\|M(x^\star+\Delta)-M(x^\star)-J(x^\star)\Delta\|
&\le K_R\|\Delta\|^2,
\\
\label{eq:reg-emp-proof}
\big\| [M_n-M](x^\star+\Delta)-[M_n-M](x^\star) \big\|
&\le \frac{K_E}{\sqrt{n}}\|\Delta\|
\qquad \text{on an event $\Omega_E$ with $\mathbb{P}(\Omega_E)\ge 1-\delta$,}
\end{align}
for any fixed $\delta>0$, using Proposition~\ref{prop:finite-max-u}.
This verifies Assumption~\ref{assum:gap-sb-mr-P-general}(B2).

\paragraph{(3) Error floor from Theorem~\ref{thm:necessary-small-target-P-general}.}
Let $\widehat x$ be any empirical fixed point in the basin, $M_n(\widehat x;\mathcal{Y})=\widehat x$, satisfying
$\|\widehat x-x^\star\|/\|x^\star\|\le R_0$.
Define the event
\begin{align}
\label{eq:Omega-def-proof}
\Omega \;\triangleq\; \Omega_E \cap
\left\{
\frac{\|PZ\|_2}{\|x^\star\|_2}
 \geq 
\frac{\eta\sqrt{c_2}\,\tau_{k_{\max}}}{\sqrt{n}}\cdot\frac{\beta}{\sigma}
\right\}.
\end{align}
By \eqref{eq:lb-fluct-new} and $\mathbb{P}(\Omega_E)\ge 1-\delta$,
\begin{align}
\label{eq:Omega-prob-proof}
\mathbb{P}(\Omega) \geq  p_{\mathrm{fluc}}-\delta \;\triangleq\; p_0 \;>\;0
\end{align}
for any $\delta\in(0,p_{\mathrm{fluc}})$.
On $\Omega$, Theorem~\ref{thm:necessary-small-target-P-general} (with $\alpha=1/2$) yields
\begin{align}
\label{eq:floor-proof}
\frac{\|\widehat x-x^\star\|_2}{\|x^\star\|_2}
 \geq 
\frac{c_0}{2}\cdot \frac{\mathrm{SNR}}{1-\lambda_{\max}}\cdot\frac{1}{\sqrt n},
\qquad
c_0\triangleq \eta\sqrt{c_2}\,\tau_{k_{\max}}.
\end{align}

\paragraph{(4) Low-SNR scaling of the spectral gap (Assumption~\ref{assum:gap-sb-mr-P-general}(B0)).}
By Proposition~\ref{prop:K2-MRA} and Corollary~\ref{cor:block-spectral},
\begin{align}
\label{eq:gap-scaling-proof}
1-\lambda_{\max}
\;=\;
\kappa_{\max}\frac{\beta^{4}}{\sigma^{4}}
+o\left(\frac{\beta^{4}}{\sigma^{4}}\right)
\qquad \text{as }\beta/\sigma\to 0,
\end{align}
with $\kappa_{\max}>0$ for the MRA model under the standing nondegeneracy assumptions.
Substituting \eqref{eq:gap-scaling-proof} into \eqref{eq:floor-proof} gives, for $\beta/\sigma$ small enough,
\begin{align}
\label{eq:floor-final-proof}
\frac{\|\widehat x-x^\star\|_2}{\|x^\star\|_2}
\;\gtrsim\;
\frac{1}{\sqrt n}\cdot
\frac{\beta/\sigma}{\beta^4/\sigma^4}
\;=\;
\frac{1}{\sqrt n}\cdot\frac{\sigma^3}{\beta^3}.
\end{align}
Therefore, if $\|\widehat x-x^\star\|/\|x^\star\|\le \varepsilon_{\mathrm{abs}}$ with
$\varepsilon_{\mathrm{abs}}\in(0,R_0]$, then \eqref{eq:floor-final-proof} forces
\begin{align}
\varepsilon_{\mathrm{abs}}
\;\gtrsim\;
\frac{1}{\sqrt n}\cdot\frac{\sigma^3}{\beta^3}
\qquad\Longrightarrow\qquad
n\;\gtrsim\; \frac{\sigma^6}{\beta^6}\cdot \frac{1}{\varepsilon_{\mathrm{abs}}^2}.
\end{align}
Absorbing all fixed constants (depending only on $d$ and the normalized spectrum of $v$) into $C_{\mathrm{nec}}$
yields \eqref{eq:necessary-sample-1d}.

\section{High-SNR regime: Proof of Proposition \ref{prop:relationBetweenSoftAndHard}} \label{sec:highSNRasymptotics}

We will prove Proposition \ref{prop:relationBetweenSoftAndHard} under the general MRA model in \eqref{eqn:mainModel}, for a general compact group $G$ acting on a finite-dimensional vector space $V$. The results for $G = \mathbb{Z}_d$ then follow as a special case. Analogous to the definition of the hard-assignment estimator update in \eqref{eqn:generalizedEfNEqn}, we define the update rule at iteration $t+1$ as,
    \begin{align}
        \hat{x}_\mathrm{hard}^{(t+1)}\triangleq \frac{1}{n} \sum_{i=1}^{n} \hat{g}_{i,t}^{-1} \cdot y_i, \label{eqn:E1}
    \end{align}
where in accordance to \eqref{eqn:OptShiftRealSpace}, the group element that best aligns with $y_i$ is given by,
     \begin{align}
        {\hat{g}}_{i,t}\triangleq \underset{g \in G}{\argmax} {    \langle{y_i}, g \cdot \hat{x}^{(t)} \rangle}.
    \label{eqn:E2}
    \end{align}
Since $G$ is a compact group, the maximizer over the compact group $G$ is attained. For the EM estimator, similarly to the update rule in \eqref{eqn:generalizedEfNEqnSoft}, we define,
    \begin{align}
        \hat{x}_{\mathrm{soft}}^{(t+1)}\triangleq \frac{1}{n}\sum_{i=1}^{n} \int_{g \in G}{ \gamma_{i,t} \p{g} \p{g^{-1} \cdot y_i}} \mathrm{d} \mu \p{g} \label{eqn:E3},
    \end{align}
where $\mu\p{g}$ denotes the uniform Haar measure on $G$. The soft-assignment weights $\gamma_{i,t} \p{g}$ are defined analogously to \eqref{eqn:softmaxGamma} as,
    \begin{align} 
        \gamma_{i,t}(g) \triangleq \frac{\exp\left( y_i^\top (g \cdot \hat{x}^{(t)}) / \sigma^2 \right)}{\int_{g' \in G} \exp\left( y_i^\top (g' \cdot \hat{x}^{(t)}) / \sigma^2 \right) \mathrm{d}\mu(g')}. \label{eqn:E4} 
    \end{align}

\paragraph{Proof of \eqref{eqn:equivalenceSoftAndHardHighSNR}.} 
We begin by proving the first part of the proposition, which establishes the equivalence between the EM and hard-assignment estimators in the high SNR regime, i.e., as $\sigma \to 0$, as stated in \eqref{eqn:equivalenceSoftAndHardHighSNR}. Specifically, we aim to show that,
\begin{align}
    \nonumber \lim_{\sigma \to 0} & \ \int_{g \in G} { \gamma_{i,t} \p{g} \p{g^{-1} \cdot y_i}} \ \mathrm{d}g 
    \\ & = \lim_{\sigma \to 0} \ \frac{ \int_{g \in G} \exp \p{y_i^\top \p{g \cdot \hat{x}^{(t)}}/\sigma^2} \p{g^{-1} \cdot y_i} \ \mathrm{d}g}{\int_{g \in G} \exp \p{y_i^\top \p{g \cdot \hat{x}^{(t)}} / \sigma^2 } \ \mathrm{d}g} = \hat{g}_{i,t}^{-1} \cdot \p{\lim_{\sigma \to 0}  y_i}, \label{eqn:E5}
\end{align}
for every $i \in \pp{n}$, and for every $t$. To prove this, we first state the following result.
\begin{thm}[{\cite[Theorem 5.10]{robert1999monte}}]\label{thm:auxThm}
Consider $h$ a real-valued function defined on a closed and bounded set $\Theta$. If there exists a unique solution $\theta^{*}$ satisfying 
\begin{align}
    \theta^{*} = \underset{\theta \in \Theta} {\argmax} \ h(\theta),
\end{align}
then,
\begin{align}
     \underset{\lambda \to \infty} {\lim} \frac{\int_{\Theta} \theta e^{\lambda h(\theta)}d\theta}{\int_{\Theta} e^{\lambda h(\theta)}d\theta} = \theta^{*}.
\end{align}
\end{thm} 

Define the function $h_{i,t}: G \to \mathbb{R}$ as,
\begin{align}
    \nonumber h_{i,t} \p{g} & \triangleq \langle {y_i}, g \cdot \hat{x}^{(t)} \rangle \\ & = \langle {\xi_i}, g \cdot \hat{x}^{(t)} \rangle + \langle {g_{i}} \cdot x, g \cdot \hat{x}^{(t)} \rangle,
\end{align}
where we have used $y_i = g_i \cdot x + \xi_i$, as defined in \eqref{eqn:mainModel}.
By this definition, and the definition of ${\hat{g}}_{i,t}$ in 
\eqref{eqn:E2}, we have,
\begin{align}
    {\hat{g}}_{i,t} \triangleq \underset{g \in G}{\argmax} \  h_{i,t} \p{g}.
\end{align}
By the assumption of the proposition, $\hat{g}_{i,t}$ is unique, and as  $G$ is compact, the conditions of Theorem~\ref{thm:auxThm} are satisfied with $\Theta = G$, $\theta^\star = \hat{g}_{i,t}$, and $h = h_{i,t}$, for every realization of the noise vector $\xi_i$. Therefore, applying Theorem~\ref{thm:auxThm}, we obtain,
\begin{align}
    \nonumber \lim_{\sigma \to 0} &\int_{g \in G} { \gamma_{i,t} \p{g}g^{-1}}\mathrm{d}g 
    \\ & = \lim_{\sigma \to 0} \ \frac{ \int_{g \in G} \exp \p{y_i^\top \p{g \cdot \hat{x}^{(t)}}/\sigma^2} g^{-1} \ \mathrm{d}g}{\int_{g \in G} \exp \p{y_i^\top \p{g \cdot \hat{x}^{(t)}} / \sigma^2 } \ \mathrm{d}g} \label{eqn:C9}
    \\ & = \hat{g}_{i,t}^{-1}, \label{eqn:C10}
\end{align}
where \eqref{eqn:C9} follows from the definition of $\gamma_{i,t} \p{g}$ in \eqref{eqn:E4}, and follows \eqref{eqn:C10} from Theorem~\ref{thm:auxThm}. As \eqref{eqn:C10} holds for every realization of the noise vector $\xi_i$, this convergence is almost sure.

Finally, substituting the expression in \eqref{eqn:C10} involving the observed data $y_i$ yields,
\begin{align}
     \lim_{\sigma \to 0} & \ \int_{g \in G} { \gamma_{i,t} \p{g}\p{g^{-1} \cdot y_i}}\mathrm{d}g = \hat{g}_{i,t}^{-1} \cdot \p{ \lim_{\sigma \to 0} y_i}, \label{eqn:E11}
\end{align}
which proves \eqref{eqn:E5}. Since the left-hand side of \eqref{eqn:E11} corresponds to the EM update in \eqref{eqn:E3}, and the right-hand side corresponds to the hard-assignment update in \eqref{eqn:E1}, this establishes the equivalence in the high-SNR limit as stated in \eqref{eqn:equivalenceSoftAndHardHighSNR}.

\paragraph{Proof of \eqref{eqn:highSNRconvergneToSignal}.} Our next goal is to establish the convergence result stated in \eqref{eqn:highSNRconvergneToSignal}, which, using the notation of the general MRA model, can be rewritten as,
\begin{align} 
    \lim_{\sigma \to 0} \ \hat{x}_{\mathrm{hard}}^{(1)} = \lim_{\sigma \to 0} \ \frac{1}{n} \sum_{i=1}^{n} \hat{g}_{i,0}^{-1} \cdot y_i = \tilde{g} \cdot x, \label{eqn:E14_2} 
\end{align} 
for some element $\tilde{g} \in G$
, where ${\hat{g}}_{i,t}$ is defined in in \eqref{eqn:E2}.
In the high-SNR regime, \eqref{eqn:E2} simplifies to,
    \begin{align}
        \lim_{\sigma \to 0} \ {\hat{g}}_{i,t} & = \lim_{\sigma \to 0} \ \underset{g \in G}{\argmax} \pp{\langle {g_i \cdot x}, g \cdot \hat{x}^{(t)} \rangle + \langle  \xi_i,  g \cdot \hat{x}^{(t)}  \rangle}
        \\ & = \underset{g \in G}{\argmax} \langle {g_i \cdot x}, g \cdot \hat{x}^{(t)} \rangle \label{eqn:E14}
        \\ & \triangleq \hat{g}_{i,t,\sigma=0},  \label{eqn:E15} 
    \end{align}
where \eqref{eqn:E14} follows from the almost sure convergence of $\xi_i \to 0$, as $\sigma \to 0$. Under the assumption stated in Proposition \ref{prop:relationBetweenSoftAndHard}, the cross-correlation between $x$ and $\hat{x}^{(t)}$ has a unique maximizer, and so it follows from the definition in \eqref{eqn:E15} that the aligned versions of the observations are identical. Specifically, for any $i_1,i_2 \in \pp{n}$, we have,
\begin{align}
    \hat{g}_{i_1,t,\sigma=0} ^{-1} \cdot \p{g_{i_1} \cdot x} = \hat{g}_{i_2,t,\sigma=0} ^{-1} \cdot \p{g_{i_2} \cdot x}.
\end{align}
Thus,
\begin{align}
       \lim_{\sigma \to 0} \ \frac{1}{n} \sum_{i=1}^{n} {\hat{g}}_{i,0}^{-1} \cdot y_i & = \ \hat{g}_{1,0,\sigma=0} ^{-1} \cdot \p{g_1 \cdot x}  
       \\ & = \p{\hat{g}_{1,0,\sigma=0} ^{-1} \cdot g_1} \cdot x = \tilde{g} \cdot x,
    \end{align}
where $\tilde{g} \triangleq \hat{g}_{1,0,\sigma=0} ^{-1} \cdot g_1$. This completes the proof of \eqref{eqn:highSNRconvergneToSignal}.

\section{Population EM approximation in low SNR } \label{sec:lowSNRapproxMain}

\subsection{Global population EM approximation} \label{subsec:global-convergence}
Here, we show an approximation (developed in detail in Appendix~\ref{sec:lowSNRapprox}), which is not restricted only to a basin of attraction around $x^\star$, as in the local (Jacobian-based) analysis. Instead, it provides a general low–SNR expansion of the population EM map \eqref{eq:Phi-pop-1d}, obtained by a Taylor expansion of an expectation of a ratio. 
The expansion is carried out to second order in $\| \hat{x}^{(t)} \| / \sigma$ and $\| x^\star \| /\sigma$. As we demonstrate empirically in Figure \ref{fig:8}(a), this approximation closely matches the true EM update as the number of observations $n$ increases.

To represent the resulting approximation, we introduce the notation $\hat{x}_\ell^{(t)} \triangleq \mathcal{T}_\ell  \hat{x}^{(t)}$, denoting the cyclic shift of the current estimate at iteration $t$. In this context, we define the softmax weights with respect to the true underlying signal $x^\star$ and the estimate $\hat{x}^{(t)}$ as follows,
\begin{align}
    w_\ell^{(t)} \triangleq \frac{\exp\left((x^\star)^\top \hat{x}_\ell^{(t)}\right)}{\sum_{r=0}^{d-1} \exp\left((x^\star)^\top \hat{x}_r^{(t)}\right)}, \label{eqn:softMaxWeights}
\end{align}
for each $\ell \in \{0, \dots, d-1\}$. Note that $\sum_{\ell=0}^{d-1} w_\ell^{(t)} = 1$. Under the assumptions and notation above, the EM update can be approximated as,
\begin{align}
    \hat{x}^{(t+1)} \approx 
    \underbrace{\hat{x}^{(t)} + \sum_{\ell=0}^{d-1} w_\ell^{(t)} \left( \mathcal{T}_\ell^{-1} x^\star \right)}_{\text{First-order terms}} 
    + 
    \underbrace{h_1\left( \hat{x}^{(t)} \right) \sum_{\ell=0}^{d-1} w_\ell^{(t)} \left( \mathcal{T}_\ell^{-1} x^\star \right) - h_2\left( \hat{x}^{(t)} \right)}_{\text{Second-order correction terms}}, \label{eqn:2.12}
\end{align}
where $h_1, h_2$ are given by,
\begin{align}
    h_1\left( \hat{x}^{(t)} \right) \triangleq \sum_{r_1, r_2 = 0}^{d-1} w_{r_1}^{(t)} w_{r_2}^{(t)} \exp \left( \langle \hat{x}_{r_1}^{(t)}, \hat{x}_{r_2}^{(t)} \rangle \right) 
    - \sum_{r=0}^{d-1} w_r^{(t)} \exp \left( \langle \hat{x}_\ell^{(t)}, \hat{x}_r^{(t)} \rangle \right), \label{eqn:h1Def}
\end{align}
and
\begin{align}
    h_2\left( \hat{x}^{(t)} \right) \triangleq \sum_{\ell, r = 0}^{d-1} w_\ell^{(t)} w_r^{(t)} \left( \mathcal{T}_\ell^{-1}  \hat{x}_r^{(t)} \right) \exp \left( \langle \hat{x}_\ell^{(t)}, \hat{x}_r^{(t)} \rangle \right). \label{eqn:h2Def}
\end{align}
The first-order terms correspond to the linear contribution of the weights $w_\ell^{(t)}$ in the EM update, while higher-order terms involve nonlinear (e.g., quadratic or higher) combinations of these weights. Specifically, the first-order term can be interpreted as updating the current estimate by adding a weighted average of cyclically shifted versions of the true signal $x^\star$, where the weights $w_\ell^{(t)}$ are defined by the softmax in~\eqref{eqn:softMaxWeights}.

\paragraph{Empirical validation of the approximation.}  
Figure~\ref{fig:8}(a) compares the low-SNR analytic approximation $x_{\s{approx}}^{(t+1)}$, derived in~\eqref{eqn:2.12}, with empirical estimates obtained via Monte Carlo simulations. Specifically, it shows the approximation error between $x_{\s{approx}}^{(t+1)}$ and the EM estimator $x_{\mathrm{soft}}^{(t+1)}$, defined in~\eqref{eqn:generalizedEfNEqnSoft}, as a function of the number of observations $n$.

Since the analytic approximation is derived in the asymptotic regime $n \to \infty$, while simulations are performed at finite $n$, the observed discrepancy can be attributed to two main sources: (1) the inherent error in the analytic expansion itself, and (2) finite-sample effects arising from a limited number of observations. The latter is expected to introduce an error of order $O(1/n)$ per iteration. As these errors accumulate across iterations, the total approximation error grows approximately as $O(t^2/n)$, a behavior similar to the error accumulation seen in the Einstein from Noise phenomenon discussed in Section~\ref{sec:EfN}.

The results in Figure~\ref{fig:8}(a) show that the approximation becomes increasingly accurate as $n$ grows, indicating that the dominant source of error is due to finite-sample effects rather than limitations of the analytic expansion. Moreover, within the range of $n$ considered, no saturation in this trend is observed, further supporting the validity of the analytic approximation in the large-$n$ regime.

\begin{figure*}[!t]
    \centering
    \includegraphics[width=1.0 \linewidth]{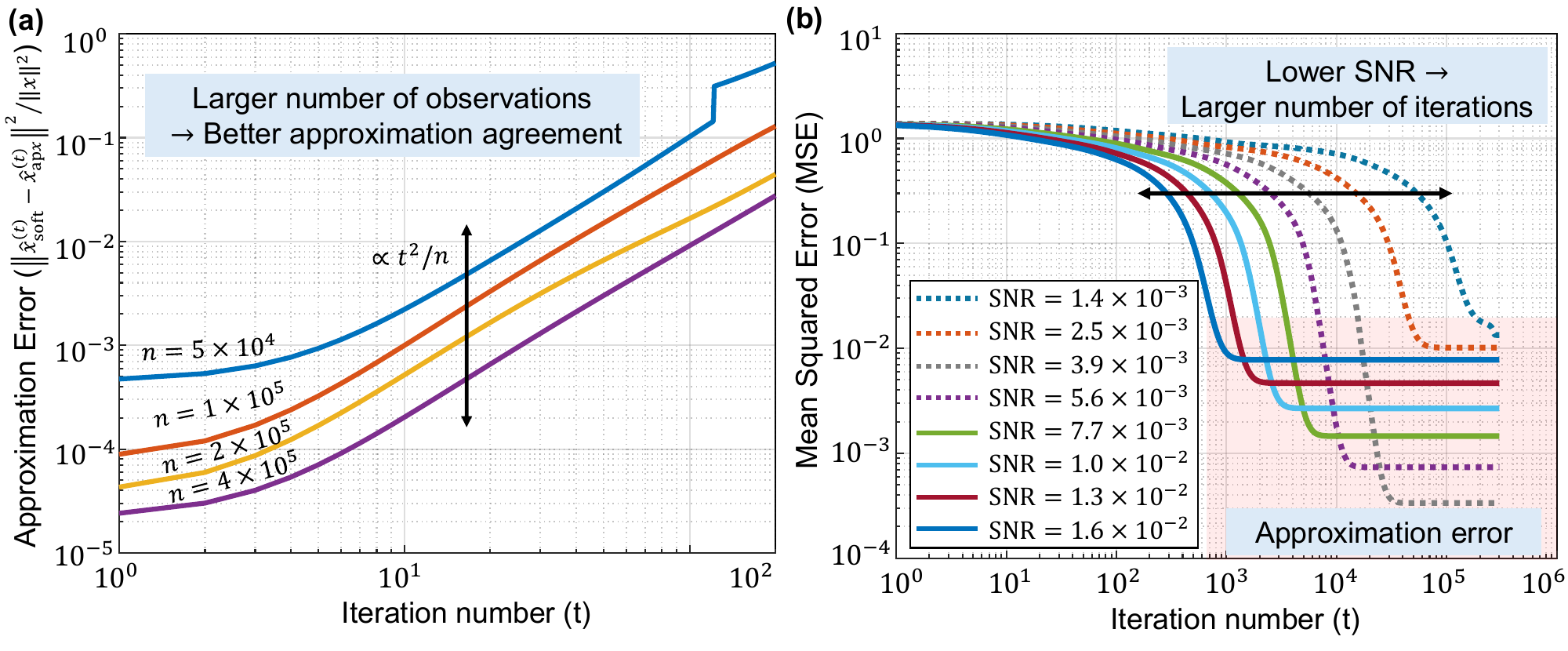}
    \caption{\textbf{Low signal-to-noise ratio: Analytic approximation and computational cost of the expectation-maximization (EM) algorithm.} \textbf{(a)} Approximation error $\| \hat{x}_{\mathrm{soft}}^{(t)} - \hat{x}_{\s{approx}}^{(t)} \|^2$ between the Monte Carlo estimate $\hat{x}_{\mathrm{soft}}^{(t)}$ (from Algorithm~\ref{alg:generalizedEfNsoft}) and the analytic approximated EM update $\hat{x}_{\s{approx}}^{(t)}$ from~\eqref{eqn:2.12}, which is valid as $n \to \infty$. As $n$ increases, the empirical estimates approach the analytic prediction, with the error decreasing proportionally to $t^2/n$, where $t$ is the iteration number, indicating that finite sample size is the primary source of deviation. Simulations were run with $d = 8 \times 8$, $\mathrm{SNR} = 2 \times 10^{-3}$, and averaged over 32 trials. \textbf{(b)} Computational cost in the low-SNR regime. As SNR decreases, convergence requires substantially more iterations. Results are based on the analytic update from~\eqref{eqn:2.12}, corresponding to the $n \to \infty$ regime, with the same image size $d = 8 \times 8$.}    
    \label{fig:8} 
\end{figure*}

\paragraph{Computational complexity in the low-SNR regime.} 
In Figure~\ref{fig:8}(b), we use the analytic approximation from \eqref{eqn:2.12} to investigate the convergence behavior of the EM algorithm as $n \to \infty$ in the low-SNR regime. Starting from a random initialization, we simulate EM updates using~\eqref{eqn:2.12} across several SNR levels. The results show that the number of iterations needed for convergence grows rapidly as the SNR decreases. In some cases, it requires hundreds of thousands of iterations to converge. Consequently, computational constraints may pose a significant barrier to applying the EM algorithm effectively in low-SNR settings.

\subsection{Population EM approximation derivation} \label{sec:lowSNRapprox}
Here, we assume that $\sigma^2 = 1$ is fixed, and treat the norm of $x$ as the variable of interest. Recall that by SLLN, the update rule for EM can be represented by
\begin{align}
    \hat{x}^{(t+1)} = \frac{1}{n}\sum_{i=1}^{n} \sum_{\ell=0}^{d-1}{ \gamma_\ell(\hat{x}^{(t)}, y_i) \p{\mathcal{T}_{\ell}^{-1} y_i}} & \xrightarrow[]{\s{a.s.}} \sum_{\ell=0}^{d-1}\mathbb{E} \pp{{ \gamma_{1,t}^{(\ell)} \p{\mathcal{T}_{\ell}^{-1} (\mathcal{T}_{\ell_1} y_1)}}} \label{eqn:C0}
    \\ & = \sum_{\ell=0}^{d-1}\mathbb{E} \pp{{ \gamma_{1,t}^{(\ell)} \p{\mathcal{T}_{\ell}^{-1} y_1}}}, \label{eqn:C1}
\end{align}
as $n \to \infty$, where $\gamma_{1,t}^{(\ell)}$ is defined in \eqref{eqn:softmaxGamma}, and $y_i = x + \xi_i \sim \mathcal{N}\p{x, I_{d \times d}}$. The transition from~\eqref{eqn:C0} to~\eqref{eqn:C1}, in which the translation operator $\mathcal{T}_{\ell_1}$  is omitted, is justified by the cyclic symmetry of the model: since the distribution of the observations is invariant under cyclic shifts, the expectation remains unchanged. 
Our aim is to approximate the expectation value in the right-hand-side of \eqref{eqn:C1} for the low SNR regime.
For that purpose,  we use an approximation that is based on the following approximation of the expected value of the ratio between two random variables $A$ and $B$, which is second-order approximation of the expectation of a ratio of random variables (see, e.g., second-order Taylor expansions for ratios \cite{stuart2010kendall}):
\begin{align}
   \mathbb{E} \pp{\frac{A}{B}} \approx \frac{\mu_A}{\mu_B} \p{1 - \frac{\mathrm{Cov} \p{A,B}}{\mu_A \mu_B} + \frac{\sigma_B^2}{\mu_B^2}} = \frac{\mu_A}{\mu_B} - \frac{\mathrm{Cov}\p{A,B}}{\mu_B^2} + \frac{\sigma_B^2 \mu_A}{\mu_B^3}. \label{eqn:ratioExpectationApproximation}
\end{align}

To apply this, we define:
\begin{align}
    A_\ell \triangleq y_1\exp(y_1^\top \hat{x}_\ell^{(t)}), \label{eqn:C3}
\end{align}
for each term in the numerator of \eqref{eqn:C1}, and 
\begin{align}
     B = \sum_{r=0}^{d-1} \exp(y_1^\top \hat{x}_r^{(t)}), \label{eqn:C4}
\end{align}
for the denominator in \eqref{eqn:C1}.
In addition, we define $ \hat{x}_\ell^{(t)} \triangleq \mathcal{T}_\ell  \hat{x}^{(t)} $ as the cyclic shift of the estimate $\hat{x}^{(t)}$ at iteration $t$. Note that as the cyclic shift preserves the norm, we have,
\begin{align}
    \|\hat{x}_\ell^{(t)}\| = \|\hat{x}^{(t)}\|, \label{eqn:C4_2}
\end{align}
for every $\ell \in \pp{d}$. In addition, define,
\begin{align}
    w_\ell^{(t)} \triangleq \frac{\exp(x^\top \hat{x}_\ell^{(t)})}{\sum_{r=0}^{d-1} \exp(x^\top \hat{x}_r^{(t)})}, \label{eqn:C5}
\end{align}
for every $\ell \in \pp{d}$, which is the softmax weight that appears when computing expectations over random rotations. Note that $\sum_{\ell=0}^{d-1} w_\ell^{(t)} = 1$. For notations simplicity, we will omit the $(t)$ superscript.

For the following computations, we use the following properties of a Gaussian expectation of a vector $z \sim \mathcal{N}(x, I_{d \times d})$:
\begin{align}
    \mathbb{E}[z \exp(z^\top t)] = (t + x) \exp(x^\top t + \norm{t}^2/2), \label{eqn:C6}  
\end{align}
and
\begin{align}
    \mathbb{E}[\exp(z^\top t)] = \exp(x^\top t + \norm{t}^2/2). \label{eqn:C6_2}  
\end{align}

\paragraph{Step 1: Compute $\mu_{A_\ell} = \mathbb{E}[A_\ell]$.}
Applying \eqref{eqn:C6} on $\mu_{A_\ell}$ yields, together with \eqref{eqn:C4_2}:
\begin{align}
    \mu_{A_\ell} &= \mathbb{E}\pp{ y_1\exp(y_1^\top \hat{x}_\ell^{(t)})} 
    = (\hat{x}_\ell^{(t)} + x) \cdot \exp\p{x^\top \hat{x}_\ell^{(t)} + \|\hat{x}^{(t)}\|^2/2}. \label{eqn:C7}  
\end{align}

\paragraph{Step 2: Compute $\mu_B = \mathbb{E}[B]$.}
Each term in the sum in \eqref{eqn:C4} has expectation $\exp(x^\top \hat{x}_r^{(t)} + \|\hat{x}_r^{(t)}\|^2/2)$, so:
\begin{align}
    \mu_B = \sum_{r=0}^{d-1} \exp(x^\top \hat{x}_r^{(t)} + \|\hat{x}^{(t)}\|^2/2). \label{eqn:C8} 
\end{align}
We obtain:
\begin{align}
    \frac{\mu_{A_\ell}}{\mu_B} = \frac{\p{\hat{x}_\ell^{(t)} + x} \cdot \exp \p{\norm{\hat{x}^{(t)}}^2/2} \exp \ppp{x^\top \hat{x}_\ell^{(t)}} }{ \exp \p{\norm{\hat{x}^{(t)}}^2/2} \sum_{r=0}^{d-1}\exp \ppp{x^\top \hat{x}_r^{(t)}}} = w_\ell \p{\hat{x}_\ell^{(t)} + x}, \label{eqn:C8_2}
\end{align}
where $w_\ell$ is defined in \eqref{eqn:C5}.

\paragraph{Step 3: Compute $\sigma_B^2 = \mathrm{Var}(B)$.}
We expand:
\begin{align}
    \sigma_B^2 &= \mathbb{E}[B^2] - \mu_B^2. \label{eqn:C11}
\end{align}
From linearity of the expectation, and applying \eqref{eqn:C6_2}, we have,
\begin{align}
    \mathbb{E} \pp{B^2} & = \mathbb{E} \pp{\ppp{\sum_{r=0}^{d-1} \exp \p{y_1^\top \hat{x}_{r}^{(t)}} }^2}  
    \\ & = \sum_{r_1,r_2=0}^{d-1} \mathbb{E} \pp{\exp \p{y_1^\top \p{\hat{x}_{r_1}^{(t)}+\hat{x}_{r_2}^{(t)}}} }
    \\ & = \sum_{r_1,r_2=0}^{d-1} \exp \ppp{x^\top \p{\hat{x}_{r_1}^{(t)} + \hat{x}_{r_2}^{(t)}} +\|\hat{x}_{r_1}^{(t)} + \hat{x}_{r_2}^{(t)}\|^2/2 }
    . \label{eqn:C12}
\end{align}
Then, substituting \eqref{eqn:C12}, and \eqref{eqn:C8} into \eqref{eqn:C11},
\begin{align}
    \sigma_B^2 &= \mathbb{E}[B^2] - \mu_B^2 \\
    &=  \exp( \|\hat{x}^{(t)}\|^2) \sum_{r_1, r_2=0}^{d-1} \exp\left(x^\top(\hat{x}_{r_1}^{(t)} + \hat{x}_{r_2}^{(t)})\right) \left( \exp(\langle \hat{x}_{r_1}^{(t)}, \hat{x}_{r_2}^{(t)} \rangle) - 1 \right). \label{eqn:C14}
\end{align}
Normalizing by $\mu_B^2$ gives:
\begin{align}
    \frac{\sigma_B^2}{\mu_B^2} = & \frac{\sum_{r_1,r_2=0}^{d-1} \exp \ppp{x^\top \p{\hat{x}_{r_1}^{(t)} + \hat{x}_{r_2}^{(t)}}} \p{\exp  \ppp{\langle \hat{x}_{r_1}^{(t)}, \hat{x}_{r_2}^{(t)} \rangle}-1}}{\sum_{r_1,r_2=0}^{d-1}\exp \ppp{x^\top \p{\hat{x}_{r_1}^{(t)} + \hat{x}_{r_2}^{(t)}}}}
    \\ & = \sum_{r_1,r_2 = 0}^{d-1} w_{r_1}w_{r_2} \p{\exp  \ppp{\langle \hat{x}_{r_1}^{(t)}, \hat{x}_{r_2}^{(t)} \rangle}-1}, \label{eqn:C15}
\end{align}
where $w_\ell$ is defined in \eqref{eqn:C5}.

\paragraph{Step 4: Compute $\s{Cov}(A_\ell,B)$.}
By definition,
\begin{align}
    \s{Cov}(A_\ell,B) &= \mathbb{E}[A_\ell B] - \mu_{A_\ell} \mu_B. \label{eqn:C16}
\end{align}
Then, applying \eqref{eqn:C6} on the first term in the right-hand-side of \eqref{eqn:C16}, yields,
\begin{align}
    \mathbb{E}[A_\ell B] &= \sum_{r=0}^{d-1} \mathbb{E}\left[ y_1 \exp(y_1^\top(\hat{x}_\ell^{(t)} + \hat{x}_r^{(t)})) \right] \\
    &= \sum_{r=0}^{d-1} (\hat{x}_\ell^{(t)} + \hat{x}_r^{(t)} + x) \exp(x^\top(\hat{x}_\ell^{(t)} + \hat{x}_r^{(t)})) \exp(\|\hat{x}_\ell^{(t)} + \hat{x}_r^{(t)}\|^2/2). \label{eqn:C17}
\end{align}
Combining \eqref{eqn:C7} and \eqref{eqn:C8}:
\begin{align}
    \mu_{A_\ell}\mu_{B} & = (\hat{x}_\ell^{(t)} + x) \cdot \exp\p{x^\top \hat{x}_\ell^{(t)} + \|\hat{x}^{(t)}\|^2/2} \cdot \sum_{r=0}^{d-1} \exp(x^\top \hat{x}_r^{(t)} + \|\hat{x}^{(t)}\|^2/2)
    \\ & = \exp( \|\hat{x}^{(t)}\|^2) \cdot \sum_{r=0}^{d-1} (\hat{x}_\ell^{(t)} + x) \exp(x^\top(\hat{x}_\ell^{(t)} + \hat{x}_r^{(t)})). \label{eqn:C18}
\end{align}
Combining \eqref{eqn:C17}, \eqref{eqn:C18}, into \eqref{eqn:C16}, and dividing by $\mu_B^2$, we obtain:
\begin{align}
    \frac{\s{Cov}(A_\ell,B)}{\mu_B^2} &= w_\ell (\hat{x}_\ell^{(t)} + x) \sum_{r=0}^{d-1} w_r \left( \exp(\langle \hat{x}_\ell^{(t)}, \hat{x}_r^{(t)} \rangle) - 1 \right) \\
    &\quad + w_\ell \sum_{r=0}^{d-1} w_r \hat{x}_r^{(t)} \exp(\langle \hat{x}_\ell^{(t)}, \hat{x}_r^{(t)} \rangle). \label{eqn:C19}
\end{align}

\paragraph{Step 5: Plug into Approximation \eqref{eqn:ratioExpectationApproximation}.} Substituting the approximations into the right-hand-side of \eqref{eqn:C1}
\begin{align}
    \sum_{\ell=0}^{d-1} \mathcal{T}_{\ell}^{-1} \p{\mathbb{E} \pp{{ y_1\gamma_{1,t}^{(\ell)}}}}.
\end{align}
For the first term in \eqref{eqn:ratioExpectationApproximation}, we obtain from \eqref{eqn:C8_2},
\begin{align}
    \sum_{\ell=0}^{d-1} \mathcal{T}_\ell^{-1}  \p{\frac{\mu_{A_\ell}}{\mu_B}} = \hat{x}^{(t)} + \sum_{\ell=0}^{d-1}w_\ell \p{\mathcal{T}_\ell^{-1}  x}, \label{eqn:C20}
\end{align}
where we have used $\sum_{\ell=0}^{d-1} w_\ell = 1$.

For the second term in \eqref{eqn:ratioExpectationApproximation}, we obtain from \eqref{eqn:C19},
\begin{align}
    \sum_{\ell=0}^{d-1} \mathcal{T}_\ell^{-1}  \p{\frac{\s{cov}\p{A_\ell,B}}{\mu_B^2}} = \sum_{\ell,r=0}^{d-1} & w_\ell w_r \p{\hat{x}^{(t)} + \mathcal{T}_\ell^{-1}  x} \p{\exp \ppp{\langle \hat{x}_\ell^{(t)}, \hat{x}_r^{(t)} \rangle} - 1}
\nonumber    \\ & + \sum_{\ell,r=0}^{d-1}  w_\ell w_r  \p{\mathcal{T}_\ell^{-1}  \hat{x}_r^{(t)}} \exp \ppp{\langle \hat{x}_\ell^{(t)}, \hat{x}_r^{(t)} \rangle}. \label{eqn:C21}
\end{align}
Finally, for the third term in \eqref{eqn:ratioExpectationApproximation}, we obtain from \eqref{eqn:C15}:
\begin{align}
    \sum_{\ell=0}^{d-1} & \mathcal{T}_\ell^{-1}  \p{\frac{\sigma_B^2 \mu_{A_\ell}}{\mu_B^3}} 
    \nonumber\\ & = \p{\hat{x}^{(t)} + \sum_{\ell=0}^{d-1}w_\ell \p{\mathcal{T}_\ell^{-1}  x}} \p{\sum_{r_1,r_2 = 0}^{d-1} w_{r_1}w_{r_2} \p{\exp  \ppp{\langle \hat{x}_{r_1}^{(t)}, \hat{x}_{r_2}^{(t)} \rangle}-1}}. \label{eqn:C22}
\end{align}
Substituting \eqref{eqn:C20}--\eqref{eqn:C22} into \eqref{eqn:ratioExpectationApproximation}, yields,
\begin{align}
    \sum_{\ell=0}^{d-1} \mathcal{T}_\ell^{-1}  \mathbb{E}\left[\frac{A_\ell}{B}\right] 
    &\approx \sum_{\ell=0}^{d-1} \mathcal{T}_\ell^{-1}  \left( \frac{\mu_{A_\ell}}{\mu_B} \right)
    - \sum_{\ell=0}^{d-1} \mathcal{T}_\ell^{-1}  \left( \frac{\s{Cov}(A_\ell, B)}{\mu_B^2} \right)
    + \sum_{\ell=0}^{d-1} \mathcal{T}_\ell^{-1}  \left( \frac{\sigma_B^2 \mu_{A_\ell}}{\mu_B^3} \right) \nonumber\\
    &= \hat{x}^{(t)} + \sum_{\ell=0}^{d-1} w_\ell \left( \mathcal{T}_\ell^{-1}  x \right) \\
    &\quad - \sum_{\ell, r=0}^{d-1} w_\ell w_r 
    \left( \hat{x}^{(t)} + \mathcal{T}_\ell^{-1}  x \right) 
    \left( \exp\left( \left\langle \hat{x}_\ell^{(t)}, \hat{x}_r^{(t)} \right\rangle \right) - 1 \right) \nonumber \\
    &\quad - \sum_{\ell, r=0}^{d-1} w_\ell w_r 
    \left( \mathcal{T}_\ell^{-1}  \hat{x}_r^{(t)} \right) 
    \exp\left( \left\langle \hat{x}_\ell^{(t)}, \hat{x}_r^{(t)} \right\rangle \right) \nonumber \\
    &\quad + \left( \hat{x}^{(t)} + \sum_{\ell=0}^{d-1} w_\ell \left( \mathcal{T}_\ell^{-1}  x \right) \right)
    \left( \sum_{r_1, r_2 = 0}^{d-1} w_{r_1} w_{r_2} 
    \left( \exp\left( \left\langle \hat{x}_{r_1}^{(t)}, \hat{x}_{r_2}^{(t)} \right\rangle \right) - 1 \right) \right). \nonumber \label{eqn:C23}
\end{align}
Using simple algebra steps, and using the fact that $\sum_{\ell=0}^{d-1}w_\ell = 1$, we obtain,
\begin{align}
    \hat{x}^{(t+1)} \approx & \ \hat{x}^{(t)} + \sum_{\ell=0}^{d-1}w_\ell \p{\mathcal{T}_\ell^{-1}  x}
    \\ & + \sum_{\ell=0}^{d-1}w_\ell \p{\mathcal{T}_\ell^{-1}  x} \pp{\sum_{r_1,r_2 = 0}^{d-1} w_{r_1}w_{r_2} \exp  \ppp{\langle \hat{x}_{r_1}^{(t)}, \hat{x}_{r_2}^{(t)} \rangle} - \sum_{r=0}^{d-1}w_r\exp \ppp{\langle \hat{x}_\ell^{(t)}, \hat{x}_r^{(t)} \rangle}} \nonumber
    \\ & - \sum_{\ell,r=0}^{d-1}  w_\ell w_r  \p{\mathcal{T}_\ell^{-1}  \hat{x}_r^{(t)}} \exp \ppp{\langle \hat{x}_\ell^{(t)}, \hat{x}_r^{(t)} \rangle}. \nonumber
\end{align}
Identifying the terms $h_1, h_2$, defined in \eqref{eqn:h1Def}, \eqref{eqn:h2Def}, respectively, completes the approximation derivation.

\section{Preliminaries: Einstein from Noise} \label{sec:EfNpreliminaries}
Before we delve into the proofs of Theorems~\ref{thm:1} and \ref{thm:2}, we fix several notations and definitions, and prove auxiliary results that will be used in the proofs. For simplicity of notations, in the single-iteration analysis, we omit the explicit dependence on the iteration index and denote $x = \hat{x}^{(t)}$, and $\hat{x} = \hat{x}^{(t+1)}$ for two successive iterations.  Furthermore, throughout the Einstein from Noise analysis, we will assume that $\sigma^2 = 1$.

\subsection{Notations} 

We derive the EM update rule in the Einstein from Noise regime, when the SNR is zero. Recall the definitions of $\gamma_\ell(\hat{x}^{(t)}, y_i)$, and $\hat{x}^{(t)} = x$ in \eqref{eqn:softmaxGamma} and
\eqref{eqn:generalizedEfNEqnSoft}, respectively. In the absence of a true signal (i.e., $\mathrm{SNR} = 0$), we have $y_i = \xi_i$, and the the update equations of the soft-assignment weights simplifies as follows,
    \begin{align}
        \gamma_{i,\ell} \triangleq \gamma_\ell(x, \xi_i) =  \frac{\exp \p{\xi_i^\top \p{\mathcal{T}_\ell  {x}}/\sigma^2}}{\sum_{r=0}^{d-1} \exp \p{\xi_i^\top \p{\mathcal{T}_r  {x}}/\sigma^2 }}, \label{eqn:appx_E1}
    \end{align}
for $i=0,\ldots,n-1$, and $\ell = 0, \dots, d-1$. Thus, the corresponding estimator $\hat{x}$ is then given by,
    \begin{align}
        \hat{x}\triangleq \frac{1}{n}\sum_{i=1}^{n} \sum_{\ell=0}^{d-1}{ \gamma_{i,\ell} \, \mathcal{T}_{\ell}^{-1} \xi_i} \label{eqn:appx_E2}.
    \end{align}

\paragraph{Fourier space conjugate-symmetry relation.} Recall the definitions of the Fourier transforms of $x$ and $\xi_i$ from~\eqref{FourierSpace}, and recall that the signal length $d$ is assumed to be even. Note that since $x$ and $\xi_i$ are real-valued, their Fourier coefficients satisfy the conjugate-symmetry relation: 
\begin{align}
    \s{X}[k]=\overline{\s{X}[d-k]},\quad 
    \s{N}_i[k]=\overline{\s{N}_i[d-k]}. \label{eqn:conjRelation}
\end{align}
In particular, $|\s{N}_i[k]| = |\s{N}_i[d-k]|$ and $\phi_{\s{N}_i}[k] = -\phi_{\s{N}_i}[d-k]$, which implies that only the first $d/2 + 1$ components of $\s{N}[k]$ are statistically independent.

\paragraph{The estimator in Fourier space.} 
With the definitions above, the estimation process can be expressed conveniently in the Fourier domain. Since a cyclic shift in the signal domain corresponds to multiplication by a linear phase factor in the Fourier domain, the EM update in~\eqref{eqn:generalizedEfNEqnSoft} admits an equivalent frequency-domain representation. In particular, by linearity of the Fourier transform and its shift property, the $k^{\mathrm{th}}$ Fourier coefficient of the updated estimate can be written as:
\begin{align}
    \hat{\s{X}}[k] = \frac{1}{n} \sum_{i=1}^{n} \sum_{\ell=0}^{d-1} \gamma_{i,\ell} \cdot \s{N}_i[k] \exp \ppp{j\frac{2\pi k \ell}{d}}.
\end{align}
or equivalently,
\begin{align}
    \s{\hat{X}}[k] & =  \frac{1}{n} \sum_{i=1}^{n} \sum_{\ell=0}^{d-1} \gamma_{i, \ell} \cdot \abs{\s{N}_i[k]} e^{j\phi_{\s{N}_i}[k]} e^{j\frac{2\pi k}{d}\ell},        \label{eqn:estimatorFourierRepresentation_pre}
\end{align}
for $k \in \pp{d}$. It is important to note that $\gamma_{i, \ell}$ captures the dependency on the estimates from the previous iteration ${x}$, as well as the connections between the different spectral components. 

Our goal is to analyze the phase and magnitude of the estimator $\hat{\s{X}}$ in \eqref{eqn:estimatorFourierRepresentation_pre}. Simple manipulations reveal that, for any $0\leq k\leq d-1$, the estimator's phases are given by,
\begin{align}
    \phi_{\s{\hat{X}}}[k] = \phi_{\s{X}}[k] + \arctan \left( \frac{\sum_{i=1}^{n}\abs{\s{N}_i[k]} \sum_{\ell=0}^{d-1} \gamma_{i,\ell} \sin \left( \phi_{i,\ell}[k] \right)}{\sum_{i=1}^{n}\abs{\s{N}_i[k]} \sum_{\ell=0}^{d-1} \gamma_{i, \ell} \cos \left( \phi_{i,\ell}[k] \right)} \right ),\label{eqn:estimatorPhase}
\end{align}
where we have defined,
\begin{align}
    \phi_{i,\ell}[k] \triangleq \frac{2\pi k \ell} {d} +  \phi_{\s{N}_i}[k] - \phi_{\s{X}}[k].    \label{eqn:phaseDifferenceTerm}
\end{align}

\paragraph{The estimator as $n \to \infty$.} Recall the definition of $\phi_{i,\ell}[k]$ in \eqref{eqn:phaseDifferenceTerm}. Then, following \eqref{eqn:estimatorFourierRepresentation_pre}, and simple algebraic manipulation,
    \begin{align}
        \hat{\s{X}}[k] & = \frac{1}{n} \sum_{i=1}^{n} \sum_{\ell=0}^{d-1} \gamma_{i,\ell} \, \s{N}_i[k] \exp \ppp{j\frac{2\pi k \ell}{d}} \nonumber
        \\ & = \frac{1}{n} \sum_{i=1}^{n} \sum_{\ell=0}^{d-1} \gamma_{i,\ell} \abs{\s{N}_i[k]} e^{j\phi_{\s{N}_i}[k]} e^{j\frac{2\pi k}{d} \ell} \nonumber
        \\ & = \frac{e^{j\phi_{\s{{X}}}[k]}}{n} \sum_{i=1}^{n} \sum_{\ell=0}^{d-1} \gamma_{i,\ell} \abs{\s{N}_i[k]} e^{j\phi_{\s{N}_i}[k]} e^{j\frac{2\pi k}{d}\ell} e^{-j\phi_{\s{X}}[k]}
        \nonumber\\ & = \frac{e^{j\phi_{\s{{X}}}[k]}}{n} \sum_{i=1}^{n} \sum_{\ell=0}^{d-1} \gamma_{i,\ell} \abs{\s{N}_i[k]} e^{j\phi_{i,\ell}[k]}. 
        \label{eqn:A11} 
    \end{align}
By applying the strong law of large numbers (SLLN) on the right-hand-side of \eqref{eqn:A11}, for $n \to \infty$, we have,
\begin{align}
   \s{\hat{X}}[k] e^{-j\phi_{\s{X}}[k]} & =  \frac{1}{n} { \sum_{i=1}^{n} \sum_{\ell=0}^{d-1} \gamma_{i,\ell}  \abs{\s{N}_i[k]} e^{j\phi_{i,\ell}[k]}} 
   \nonumber \\ & \xrightarrow[]{\s{a.s.}} 
    \mathbb{E} \left[ \abs{\s{N}[k]} \sum_{\ell=0}^{d-1} \gamma_{\ell} \cos\p{\phi_{\ell}[k]} \right] + j \mathbb{E} \left[ \abs{\s{N}[k]} \sum_{\ell=0}^{d-1} \gamma_{\ell} \sin\p{\phi_{\ell}[k]} \right],\label{eqn:strongLLN}
\end{align}
where we have used the fact that the sequences of random variables $\{\abs{\s{N}_i[k]} \sin \left( \phi_{i,\ell}[k] \right)\}_{i=1}^{n}$ and
$\{\abs{\s{N}_i[k]} \cos \left( \phi_{i,\ell}[k] \right)\}_{i=1}^{n}$ are i.i.d. with finite mean and variances. From now on, for notational simplicity, we will omit the index ``$i$" when taking expectations and define $\gamma_\ell \triangleq \gamma_{1,\ell}$, $\xi \triangleq \xi_1$, and $\phi_{\ell}[k] \triangleq \phi_{i,\ell}[k]$.

\subsection{Cyclo-stationary correlations} \label{sec:cylocStationaryProcess}
We will rely on the observation that correlation terms $\xi^\top \p{\mathcal{T}_\ell{x}}$ between the noise $\xi$ and the different shifts of the signal $x$, form a Gaussian random vector, which we denote by $\s{S}$. Specifically, the $r^{\s{th}}$ entry of this vector is defined as,
 \begin{align}
    {\s{S}_{r}} &\triangleq \langle \xi, {\mathcal{T}_r{x}} \rangle , \label{eqn:appx_E4}
\end{align}
for $0\leq r\leq d-1$. In the Fourier domain, this expression becomes,
 \begin{align}
    {\s{S}_{r}} &\triangleq \langle \xi, {\mathcal{T}_r  {x}} \rangle \nonumber \\ & = \langle \mathcal{F} \ppp{\xi}, \mathcal{F} \ppp{\mathcal{T}_r x} \rangle\nonumber
    \\ & =  \sum_{k=0}^{d-1} {\abs{{\s{X}}[k]} \abs{\s{N}[k]}   \cos \left( \frac{2\pi kr}{d} +  \phi_{\s{N}}[k] - \phi_{\s{X}}[k] \right)} \nonumber
    \\ & = \sum_{k=0}^{d-1} {\abs{{\s{X}}[k]} \abs{\s{N}[k]} \cos \left(\phi_{r}[k] \right)}, \label{eqn:appx_E5}
\end{align}
where $\phi_{r}$ is defined in \eqref{eqn:phaseDifferenceTerm}. The vector $\s{S} \triangleq (\s{S}_{0},\s{S}_{1},\ldots,\s{S}_{d-1})^T$ follows a multivariate Gaussian distribution $\mathcal{N} \p{0, \Sigma}$ where the covariance matrix $\Sigma$ is circulant. Consequently, $\s{S}$ forms a cyclo-stationary Gaussian process. 
Moreover, the eigenvalues of $\Sigma$ are defined by $\ppp{\abs{\s{X}}[k]^2}_{k \in \pp{d}}$ \cite{balanov2024einstein}. Therefore, under the assumption that $\s{X}[k] \neq 0$ for every $k \in \pp{d}$, the covariance matrix $\Sigma$ is full-rank with rank $d$.

\subsection{Properties of Gaussian random vectors}

In this subsection, we state and prove auxiliary results about certain properties of Gaussian random vectors. Let us denote the function $\gamma_{\ell}^{(\beta)}:\mathbb{R}^d \to \mathbb{R}$, parameterized by $\beta\in\mathbb{R}_+$, as follows,
\begin{align}
   \gamma_{\ell}^{(\beta)}\p{\s{Q}} \triangleq \frac{ \exp{\p{\beta Q_\ell}}}{\sum_{r=0}^{d-1}\exp{\p{\beta Q_r}}}, \label{eqn:softMaxFuncDef1}
\end{align}
for $\ell\in[d]$, where $\s{Q} = (Q_0,Q_1,\ldots,Q_{d-1})\in\mathbb{R}^d$. 

\begin{lem} \label{lem:cycloStationary}
    Let ${\s{X}} = \p{X_0, X_1,\ldots, X_{d-1}}^T$ be a zero-mean cyclo-stationary Gaussian vector, with  ${\s{X}} \sim \calN \p{0, \Sigma_{\s{X}}}$, such that,
        \begin{align}
            \mathbb{E}\p{X_{\ell_1} X_{\ell_2}} = \rho_{\abs{{\ell_1}-{\ell_2}}\s{mod} \ d}. \label{eqn:cycloStationartyRelation}
        \end{align}
    Recall the definition of $\gamma_{\ell}^{(\beta)}$ in \eqref{eqn:softMaxFuncDef1}. Then,
        \begin{enumerate}
        \item For every $\beta>0$ and $\ell \in \pp{d}$, 
        \begin{align}
            \mathbb{E}\pp{\gamma_{\ell}^{(\beta)}\p{\s{X}}}= \frac{1}{d}. \label{eqn:expectedProbabilityUniform}
        \end{align}
        \item For every $\beta>0$, $\ell_1, \ell_2,\tau \in \pp{d}$,
        \begin{align}            \mathbb{E}\pp{\gamma_{\ell_1}^{(\beta)}\p{\s{X}}\cdot \gamma_{\ell_2}^{(\beta)}\p{\s{X}}} = \mathbb{E}\pp{\gamma_{\ell_1+\tau}^{(\beta)}\p{\s{X}}\cdot \gamma_{\ell_2+\tau}^{(\beta)}\p{\s{X}}}, \label{eqn:equalNumeratorCyclic}
        \end{align}
        where the sums $\ell + \tau$ are taken modulo $d$.
    \end{enumerate}
\end{lem}

\begin{proof}[Proof of Lemma~\ref{lem:cycloStationary}]    

By definition, due to \eqref{eqn:cycloStationartyRelation}, the Gaussian vector $\s{X}$ is cyclo-stationary. Therefore, by the definition of cyclo-stationary Gaussian vectors, its cumulative distribution function $F_{\s{X}}$ is invariant under cyclic shifts \cite{durrett2019probability, balanov2024einstein}, i.e.,
    \begin{align}   
         F_{\s{X}}\left(z_0, z_1,\ldots,z_{d-1} \right) = F_{\s{X}}\left( z_\tau, z_{\tau+1},\ldots,z_{\tau+d-1} \right),
         \label{eqn:stationarityTimeShift}
    \end{align}
for any $\tau \in \mathbb{Z}$, where the indices are taken modulo $d$. Therefore, the following holds for any $\tau \in \pp{d}$,
    \begin{align}
        \mathbb{E}\pp{\gamma_{\ell}^{(\beta)} \p{\s{X}}} & = \mathbb{E} \pp{\frac{\exp\p{\beta X_\ell}}{\sum_{r = 0}^{d-1} \exp\p{\beta X_r}}} \nonumber
        \\ & = \mathbb{E} \pp{\frac{\exp\p{\beta X_{\ell+\tau}}}{\sum_{r = 0}^{d-1} \exp\p{\beta X_{r+\tau}}}} \nonumber 
        \\ & = \mathbb{E} \pp{\frac{\exp\p{\beta X_{\ell+\tau}}}{\sum_{r = 0}^{d-1} \exp\p{\beta X_{r}}}} \nonumber \\
        &= \mathbb{E}\pp{\gamma_{\ell+\tau}\p{\s{X}}},
    \end{align}
where the second equality is due to the cyclo-stationary invariance property of $\s{X}$, and the third equality is due to the fact that sum in the denominator is over all the entries of $\s{X}$. This proves \eqref{eqn:equalNumeratorCyclic}. For \eqref{eqn:expectedProbabilityUniform}, we note that  since $\mathbb{E}\pp{\gamma_{\ell}^{(\beta)} \p{\s{X}}} = \mathbb{E}\pp{\gamma_{\ell+\tau} \p{\s{X}}}$, for every $\tau \in \pp{d}$, as well as due to the property that $\sum_{\ell=0}^{d-1} \gamma_{\ell}^{(\beta)} \p{\s{X}} = 1$, we get,
\begin{align}
    1 = \mathbb{E}\pp{\sum_{\ell=0}^{d-1} \gamma_{\ell}^{(\beta)} \p{\s{X}}} = d \cdot \mathbb{E}\pp{\gamma_{\ell}^{(\beta)} \p{\s{X}}},
\end{align}
for every $\ell\in[d]$, as claimed.

Finally, the second part of the lemma is obtained using similar arguments. Specifically, following \eqref{eqn:stationarityTimeShift}, we have,
\begin{align}
    \mathbb{E}\pp{\gamma_{\ell_1}^{(\beta)}\p{\s{X}}\cdot \gamma_{\ell_2}^{(\beta)}\p{\s{X}}} & = \mathbb{E} \pp{\frac{\exp\p{\beta X_{\ell_1}}}{\sum_{r = 0}^{d-1} \exp\p{\beta X_r}}\frac{\exp\p{\beta X_{\ell_2}}}{\sum_{r = 0}^{d-1} \exp\p{\beta X_r}}} \nonumber
    \\ & = \mathbb{E} \pp{\frac{\exp\p{\beta X_{\ell_1+\tau}}}{\sum_{r = 0}^{d-1} \exp\p{\beta X_{r+\tau}}}\frac{\exp\p{\beta X_{\ell_2+\tau}}}{\sum_{r = 0}^{d-1} \exp\p{\beta X_{r+\tau}}}} \nonumber
    \\ & = \mathbb{E} \pp{\frac{\exp\p{\beta X_{\ell_1+\tau}}}{\sum_{r = 0}^{d-1} \exp\p{\beta X_{r}}}\frac{\exp\p{\beta X_{\ell_2+\tau}}}{\sum_{r = 0}^{d-1} \exp\p{\beta X_{r}}}} \nonumber
    \\ & = \mathbb{E}\pp{\gamma_{\ell_1+\tau}^{(\beta)}\p{\s{X}}\cdot \gamma_{\ell_2+\tau}^{(\beta)}\p{\s{X}}},
\end{align}
for every $\ell_1, \ell_2, \tau \in \pp{d}$.

\end{proof}

\begin{lem} \label{lem:positivityOfSoftMax}
    Fix $d \geq 2$. Let ${\s{X}} = \p{X_0, X_1,\ldots, X_{d-1}}^T$ be a zero-mean Gaussian vector, with  ${\s{X}} \sim \calN \p{0, \Sigma_{\s{X}}}$. Recall the definition of $\gamma_{\ell}^{(\beta)}$ in \eqref{eqn:softMaxFuncDef1}. Then, for every $\beta > 0$, and $\ell_1, \ell_2 \in \pp{d}$,
        \begin{align}
            0 < \mathbb{E}\pp{\gamma_{{\ell_1}}^{(\beta)} \p{\s{X}} \cdot \gamma_{{\ell_2}}^{(\beta)}\p{\s{X}}} < 1. \label{eqn:app_E33_1}
        \end{align}
\end{lem}
\begin{proof}[Proof of Lemma~\ref{lem:positivityOfSoftMax}]
By the definition of $\gamma_\ell^{(\beta)}$, 
\begin{align} 
    \gamma_{\ell}^{(\beta)}(\s{X}) \triangleq \frac{ \exp{\left(\beta X_\ell\right)}}{\sum_{r=0}^{d-1} \exp{\left(\beta X_r\right)}}. 
\end{align} 
For any realization of the Gaussian vector $\s{X}$, it is clear that
\begin{align} 
    0 < \gamma_{\ell}^{(\beta)}(\s{X}) < 1, \quad \text{for each} , \ell \in {0, 1, \ldots, d-1}. 
\end{align} 
Hence, for any realization of $\s{X}$, we have 
\begin{align} 
    0 < \gamma_{{\ell_1}}^{(\beta)}(\s{X}) \cdot \gamma_{{\ell_2}}^{(\beta)}(\s{X}) < 1. \label{eqn:app_E36} 
\end{align} 
Taking the expectation of both sides in \eqref{eqn:app_E36}, we obtain the desired result.

\end{proof}
Finally, we state the following known Gaussian integration by parts lemma \cite{ross2011fundamentals}, which will be used to prove Proposition \ref{prop:realSpaceConvergence}.

\begin{lem} \label{lemma:3} Let $F: \mathbb{R}^d \to \mathbb{R}$ be a $C^1$ function, and ${\s{X}} = \p{X_0, X_1,\ldots, X_{d-1}}^T$ be a zero-mean Gaussian random vector. Then, for every $0 \leq i \leq d-1$,
    \begin{align}
       \mathbb{E} \pp{X_i F\p{\s{X}}} = \sum_{j=0}^{d-1}\mathbb{E}\p{X_iX_j} \mathbb{E}\p{\frac{\partial F}{\partial x_j} \p{\s{X}}}.
    \end{align}
\end{lem}

\subsection{The convergence of the Einstein from Noise estimator} \label{sec:convergenceOfEfNestimator}
We now derive the asymptotic form of the estimator $\hat{x}$, as $n \to \infty$.
Applying the SLLN to \eqref{eqn:appx_E2}, we have,
\begin{align}
    \hat{x} = \frac{1}{n}\sum_{i=1}^{n} \sum_{\ell=0}^{d-1}{ \gamma_{i,\ell}  \mathcal{T}_{\ell}^{-1} \xi_i} \xrightarrow[]{\s{a.s.}} \mathbb{E} \pp{\sum_{\ell=0}^{d-1}{ \gamma_{\ell} \mathcal{T}_{\ell}^{-1} \xi}}, \label{eqn:app_E10}
\end{align}
as $n \to \infty$. 
We now present our first result.
\begin{proposition} [Convergence in real space]
\label{prop:realSpaceConvergence}
    Recall the definition of $\gamma_{\ell}$ in \eqref{eqn:appx_E1}. Then, as $n \to \infty$, 
    \begin{align}
        \hat{x} \xrightarrow[]{\s{a.s.}}  x - \sum_{r,\ell=0}^{d-1} \p{\mathcal{T}_{r-\ell} \cdot x} \mathbb{E} \pp{\gamma_{\ell} \gamma_{r}}. \label{eqn:realSpaceConvergence}
    \end{align}
\end{proposition} 

\begin{proof}[Proof of Proposition~\ref{prop:realSpaceConvergence}]
To prove Proposition~\ref{prop:realSpaceConvergence}, we begin by applying Lemma~\ref{lemma:3} to the right-hand side of \eqref{eqn:app_E10}, yielding
\begin{align}
    \mathbb{E}\pp{\xi \gamma_{\ell}} =  \sum_{j=0}^{d-1}\mathbb{E}\p{\xi_i\xi_j} \mathbb{E}\p{\frac{\partial \gamma_{\ell}}{\partial \xi_j}}. \label{eqn:K4}
\end{align}
Since the noise vector $\xi$ has identity covariance, we have $\mathbb{E}\pp{\xi_i \xi_j} = \delta_{ij}$, and thus
\begin{align}
    \mathbb{E}\pp{\xi \gamma_{\ell}} =   \mathbb{E}\pp{\frac{\partial \gamma_{\ell}}{\partial \xi}}. \label{eqn:K5}
\end{align}
A straightforward calculation shows that 
\begin{align}
    \frac{\partial \gamma_{\ell}}{\partial \xi} = \p{\mathcal{T}_\ell  x}  \gamma_{\ell} - \sum_{r=0}^{d-1} \p{\mathcal{T}_r  x}   \gamma_{\ell} \gamma_{r}. \label{eqn:K6}
\end{align}
Substituting \eqref{eqn:K5}-\eqref{eqn:K6} into \eqref{eqn:K4}, we obtain
\begin{align}
    \mathbb{E}\pp{\xi \gamma_{\ell}} = \p{\mathcal{T}_\ell  x}  \mathbb{E} \pp{\gamma_{\ell}} - \sum_{r=0}^{d-1} \p{\mathcal{T}_r  x} \mathbb{E}\pp{\gamma_{\ell} \gamma_{r}}. \label{eqn:K7}
\end{align}
Substituting this result into the right-hand side of \eqref{eqn:app_E10} gives
\begin{align}
     \mathbb{E} & \pp{\sum_{\ell=0}^{d-1}{ \gamma_{\ell} \p{\mathcal{T}_{\ell}^{-1} \xi}}}  \nonumber
     \\ & = \sum_{\ell=0}^{d-1} \mathcal{T}_\ell^{-1}  \pp{\p{\mathcal{T}_\ell  x} \mathbb{E} \pp{\gamma_\ell} } - \sum_{r,\ell=0}^{d-1} \mathcal{T}_\ell^{-1}  \pp{\p{\mathcal{T}_r x} \mathbb{E} \pp{\gamma_\ell \gamma_r}} \nonumber
     \\ & =  x - \sum_{r,\ell=0}^{d-1} \p{\mathcal{T}_{r-\ell} x} \mathbb{E} \pp{\gamma_{\ell} \gamma_{r}}, \label{eqn:app_E19}
\end{align}
where in the final step we use the identity $\sum_{\ell=0}^{d-1}\gamma_\ell = 1$, and $\mathcal{T}_\ell^{-1}  \p{\mathcal{T}_r  x} = \mathcal{T}_{r-\ell} x$. This completes the proof.

\end{proof}

\subsection{Fourier space convergence}
Following Proposition \ref{prop:realSpaceConvergence}, we have the following convergence in Fourier space.
\begin{proposition} [Convergence in Fourier space]
\label{prop:fourierSpaceConvergence}
    Recall the definition of $\gamma_{\ell}$ in \eqref{eqn:appx_E1}. Then, for every $k \in \pp{d}$,
    \begin{align}
        \hat{\s{X}}[k] \xrightarrow[]{\s{a.s.}}  \s{X}[k] \pp{1 - d \, \sum_{\ell=0}^{d-1} \cos \p{\frac{2\pi k \ell}{d}} \mathbb{E} \pp{\gamma_{0} \gamma_{\ell}}}, \label{eqn:fourierSpaceConvergence}
    \end{align}
    as $n \to \infty$.
\end{proposition} 

\begin{proof}[Proof of Proposition~\ref{prop:fourierSpaceConvergence}]
Taking the Fourier transform of both sides of \eqref{eqn:realSpaceConvergence}, and applying the continuous mapping theorem, we obtain: 
\begin{align} 
    \hat{\s{X}}[k] \xrightarrow[]{\text{a.s.}} \s{X}[k] - \s{X}[k] \sum_{r,\ell=0}^{d-1} \exp\left( j \frac{2\pi k (r - \ell)}{d} \right) \mathbb{E} \left[ \gamma_r \gamma_\ell \right]. \label{eqn:app_E38} 
\end{align}
We now apply Lemma~\ref{lem:cycloStationary}. Recall the definition of the Gaussian vector $\s{S}$ in \eqref{eqn:appx_E4}. By this definition, the softmax weights are given by
\begin{align}
    \gamma_{\ell}\triangleq \frac{\exp \p{\s{S}_{\ell}}}{\sum_{r=0}^{d-1} \exp \p{\s{S}_{r}}}. 
\end{align}
This is equivalent to the softmax function $\gamma_\ell^{(\beta=1)} \p{\s{S}}$ defined in \eqref{eqn:softMaxFuncDef1}, with $\beta = 1$. In particular, $\s{S}$ is cyclo-stationary Gaussian vector, thus is satisfies the conditions of \ref{lem:cycloStationary} (Appendix \ref{sec:cylocStationaryProcess}). Therefore, according to Lemma~\ref{lem:cycloStationary}, we have the identity 
\begin{align} 
    \mathbb{E} \left[ \gamma_r \gamma_\ell \right] = \mathbb{E} \left[ \gamma_{r+\tau} \gamma_{\ell+\tau} \right], 
\end{align} 
for all $\ell, r, \tau \in \ppp{0, \dots, d-1}$ (with indices taken modulo $d$). Using this shift-invariance, the double sum in \eqref{eqn:app_E38} simplifies to,
\begin{align} 
    \sum_{r,\ell=0}^{d-1} \exp\left( j \frac{2\pi k (r - \ell)}{d} \right) \mathbb{E} \left[ \gamma_r \gamma_\ell \right] = d  \sum_{\tau=0}^{d-1} \exp\left( j \frac{2\pi k \tau}{d} \right) \mathbb{E} \left[ \gamma_0  \gamma_\tau \right]. \label{eqn:app_E40}
\end{align}

\paragraph{The imaginary part vanishes.}
We now show that the imaginary part of the expression in \eqref{eqn:app_E40} vanishes, 
\begin{align}
    \sum_{\tau=0}^{d-1} \sin \p{\frac{2\pi k \tau}{d}} \mathbb{E} \pp{\gamma_{0} \gamma_{\tau}} = 0. \label{eqn:app_E41}
\end{align}
From Lemma~\ref{lem:cycloStationary}, we have the following symmetry property,
\begin{align}
    \mathbb{E} \pp{\gamma_{0} \gamma_{\tau}} = \mathbb{E} \pp{\gamma_{0} \gamma_{-\tau}},  \label{eqn:app_E42}
\end{align}
for every $\tau \in \pp{d}$. Using the odd symmetry of the sine function, we compute,
\begin{align}
    \sum_{\tau=0}^{d-1} \sin & \p{\frac{2\pi k \tau}{d}} \mathbb{E} \pp{\gamma_{0} \gamma_{\tau}} \label{eqn:app_E43}
    \\ & = \sum_{\tau=1}^{d/2-1} \sin \p{\frac{2\pi k \tau}{d}} \mathbb{E} \pp{\gamma_{0} \gamma_{\tau}} + \sum_{\tau=d/2+1}^{d-1} \sin \p{\frac{2\pi k \tau}{d}} \mathbb{E} \pp{\gamma_{0} \gamma_{\tau}} \label{eqn:app_E44}
    \\ & = \sum_{\tau=1}^{d/2-1} \sin \p{\frac{2\pi k \tau}{d}} \mathbb{E} \pp{\gamma_{0}  \gamma_{\tau}} - \sum_{\tau=1}^{d/2-1} \sin \p{\frac{2\pi k \tau}{d}} \mathbb{E} \pp{\gamma_{0} \gamma_{-\tau}} \label{eqn:app_E45}
    \\ & = \sum_{\tau=1}^{d/2-1} \sin \p{\frac{2\pi k \tau}{d}} \mathbb{E} \pp{\gamma_{0}  \gamma_{\tau}} - \sum_{\tau=1}^{d/2-1} \sin \p{\frac{2\pi k \tau}{d}} \mathbb{E} \pp{\gamma_{0} \gamma_{\tau}}  \label{eqn:app_E46}
    \\ & = 0. \label{eqn:app_E47}
\end{align}
where the second equality is because the sine function is odd, and the third equality follows from the symmetry in \eqref{eqn:app_E42}. Therefore, substituting \eqref{eqn:app_E40} and \eqref{eqn:app_E41} into \eqref{eqn:app_E38} completes the proof.

\end{proof}

\subsection{Positive definiteness of the responsibility covariance matrix}

\begin{lem}[Positive definiteness of the responsibility covariance matrix]
\label{lem:gamma-cov-PD}
Fix $d \ge 2$ and let $x\in\mathbb{R}^d$ be a nonzero signal with DFT coefficients $\s{X}[k]$. 
Assume that $\s{X}[k] \neq 0$ for all $k = 0,\dots,d-1$. 
Recall the definition $\gamma_{\ell} \triangleq \gamma_\ell(x,\xi)$ from~\eqref{eqn:appx_E1}, and define the matrix $C(x) \in \mathbb{R}^{d \times d}$ by
\begin{align}
    [C(x)]_{\ell_1, \ell_2} 
     \triangleq  
    \mathbb{E}\pp{ \gamma_{\ell_1}(x,\xi)\,\gamma_{\ell_2}(x,\xi) },
    \qquad \ell_1, \ell_2 \in \{0,1,\ldots,d-1\}.
    \label{eq:def-Cx}
\end{align}
Then $C(x)$ is symmetric positive definite; that is,
\begin{align}
    v^\top C(x)\, v > 0, 
    \qquad \text{for all } v\in\mathbb{R}^d\setminus\{0\}.
\end{align}
Moreover, $C(x)$ is circulant with strictly positive eigenvalues, and in particular
\begin{align}
    \sum_{\tau=0}^{d-1} 
    \cos \left(\frac{2\pi k \tau}{d} \right) 
    \,\mathbb{E}\pp{ \gamma_0 \gamma_\tau } 
     >  0,
    \qquad k=0,\dots,d-1.
    \label{eqn:resposnibilities-covaraince-matrix-positive-eigenvalues}
\end{align}
\end{lem}

\begin{proof}[Proof of Lemma~\ref{lem:gamma-cov-PD}]
Define the random vector $\gamma(x,\xi)\in\mathbb{R}^d$ as the concatenation of the responsibilities:
\begin{align}
    \gamma(x,\xi) 
     \triangleq  (\gamma_0, \gamma_1 , \ldots, \gamma_{d-1}).
    \label{eqn:gamma-conc}
\end{align}
By construction,
\begin{align}
    v^\top C(x)\, v
     =  
    \mathbb{E}\pp{(v^\top \gamma(x,\xi))^2}  
     \ge  0,
\end{align}
so $C(x)$ is always symmetric positive semidefinite. To establish positive definiteness, it suffices to show that
\begin{align}
    v^\top C(x)\, v = 0 
    \quad \Longrightarrow \quad v=0.
\end{align}

Suppose $v^\top C(x)\,v = 0$ for some $v\in\mathbb{R}^d$. Then
\begin{align}
    0 
     =  v^\top C(x)\, v
     =  \mathbb{E}\pp{(v^\top \gamma(x,\xi))^2}
\end{align}
implies that $v^\top \gamma(x,\xi) = 0$ almost surely in $\xi$.

\paragraph{Positive definiteness of $C(x)$.}
Define the score vector $z(\xi)\in\mathbb{R}^d$ by
\begin{align}
    z_\ell(\xi) \triangleq \,\xi^\top \mathcal{T}_\ell x, \qquad \ell=0,\dots,d-1,
    \label{eq:def-z}
\end{align}
so that $z(\xi)$ is a linear transform of $\xi$:
\begin{align}
    z(\xi) = L\,\xi,
\end{align}
where the $\ell$-th row of $L$ is $(\mathcal{T}_\ell x)^\top$. The columns of $L^\top$ are the shifted copies $\mathcal{T}_\ell x$, so $L$ is precisely the circulant matrix generated by $x$.  
By our assumption that $\s{X}[k]\neq 0$ for all $k$, this circulant matrix is invertible (see Section~\ref{sec:cylocStationaryProcess}). Consequently,
\begin{align}
    z(\xi) = L\xi \sim \mathcal{N}(0,\Sigma_z),
    \qquad \Sigma_z = L L^\top \succ 0,
\end{align}
and $z$ has a Gaussian density that is strictly positive everywhere on $\mathbb{R}^d$. In particular, the support of $z$ is all of $\mathbb{R}^d$.

By definition~\eqref{eqn:appx_E1}, each $\gamma_\ell(x,\xi)$ is a smooth function of the scores $\{z_\ell(\xi)\}_{\ell=0}^{d-1}$ (via a softmax transformation), and for every realization of $z$ we have $\gamma_\ell(x,\xi) > 0$ and $\sum_{\ell=0}^{d-1} \gamma_\ell(x,\xi) = 1$.
As $z$ ranges over $\mathbb{R}^d$, the vector $\gamma(x,\xi)$ visits a subset of the probability simplex with nonempty interior. Therefore, if a linear form $v^\top \gamma(x,\xi)$ vanishes almost surely, it must vanish on a set with nonempty interior in that simplex, which forces $v=0$. Hence $v^\top C(x)\,v>0$ for every nonzero $v$, so $C(x)$ is positive definite.

\paragraph{Circulant matrix with positive eigenvalues.}
Finally, by the shift-invariance of the construction in $\ell$ (see Section~\ref{sec:cylocStationaryProcess}), the process $\{\gamma_\ell\}_{\ell=0}^{d-1}$ is stationary, and thus $C(x)$ is circulant with entries
\begin{align}
    [C(x)]_{\ell_1,\ell_2}
    = \mathbb{E}[\gamma_{\ell_1}\gamma_{\ell_2}]
    = \mathbb{E}[\gamma_0 \gamma_{\ell_2-\ell_1}],
\end{align}
where indices are taken modulo $d$. The eigenvalues of a real symmetric circulant matrix are given by
\begin{align}
    \lambda_k = \sum_{\tau=0}^{d-1} \cos \p{\frac{2\pi k \tau}{d}}\,\mathbb{E}[\gamma_0 \gamma_\tau],
    \qquad k=0,\dots,d-1.
\end{align}
Since $C(x)\succ 0$, all $\lambda_k$ are strictly positive, yielding
\begin{align}
    \sum_{\tau=0}^{d-1} 
    \cos \left(\frac{2\pi k \tau}{d}\right)\mathbb{E}[\gamma_0 \gamma_\tau] > 0,
    \qquad k=0,\dots,d-1,
\end{align}
which is exactly~\eqref{eqn:resposnibilities-covaraince-matrix-positive-eigenvalues}.
\end{proof}

\subsection{Convergence of the Fourier phases} 

\begin{proposition}\label{prop:2} 
Let $d \geq 2$, and suppose that $\s{X}[k] \neq 0$ for all $0 < k \leq d - 1$. Then, for every $0 < k \leq d - 1$, we have 
\begin{align} 
    \s{\hat{X}}[k] \xrightarrow[]{\s{a.s.}} \alpha_k \s{X}[k], 
\end{align} 
where $\alpha_k \in (0,1)$ is a real constant. 
\end{proposition}

This result implies that the Fourier phases of $\s{\hat{X}}[k]$ converge to those of $\s{X}[k]$, since $\alpha_k$ is strictly positive and real.
\begin{proof}[Proof of Proposition~\ref{prop:2}] 
By \eqref{eqn:fourierSpaceConvergence}, we have,
\begin{align} 
    \hat{\s{X}}[k] \xrightarrow[]{\text{a.s.}} \s{X}[k] \pp{1 - d \, \sum_{\tau=0}^{d-1} \cos\left(\frac{2\pi k \tau}{d} \right) \mathbb{E} \left[ \gamma_0 \gamma_\tau \right]}. \label{eqn:app_G2} 
\end{align}
We now apply Lemma~\ref{lem:positivityOfSoftMax}. Recall the definition of the Gaussian vector $\s{S}$ in \eqref{eqn:appx_E4}. By this definition, the softmax weights are given by,
\begin{align}
    \gamma_{\ell}\triangleq \frac{\exp \p{\s{S}_{\ell}}}{\sum_{r=0}^{d-1} \exp \p{\s{S}_{r}}}. \label{eqn:appx_E57}
\end{align}
This is equivalent to the softmax function $\gamma_\ell^{(\beta=1)} \p{\s{S}}$ defined in \eqref{eqn:softMaxFuncDef1}, with $\beta = 1$. 
Then, by Lemma \ref{lem:positivityOfSoftMax}, $\mathbb{E} \pp{\gamma_0 \gamma_\tau} > 0$, for every $\tau \in \pp{d}$. This implies that,
\begin{align} 
    \left| \sum_{\tau=0}^{d-1} \cos \left( j \frac{2\pi k \tau}{d} \right) \mathbb{E} \left[ \gamma_{0} \gamma_{\tau} \right] \right| < \sum_{\tau=0}^{d-1} \mathbb{E} \left[ \gamma_{0}  \gamma_{\tau} \right] = \mathbb{E}[\gamma_0] = \frac{1}{d}. 
\end{align} 
In particular, we have,
\begin{align} 
    1 - d  \, \sum_{\ell=0}^{d-1} \cos \left( \frac{2\pi k \ell}{d} \right) \mathbb{E} \left[ \gamma_0 \gamma_\ell \right] > 0. \label{eqn:app_E7} 
\end{align} 
Combining the convergence in~\eqref{eqn:app_G2} with the positivity in~\eqref{eqn:app_E7}, we conclude that for every $k \in \pp{d}$, 
\begin{align} 
    \hat{\s{X}}[k] \xrightarrow[]{\s{a.s.}} \alpha_k \s{X}[k], 
\end{align} 
where 
\begin{align}
    \alpha_k = 1 - d \, \sum_{\ell=0}^{d-1} \cos \left( \frac{2\pi k \ell}{d} \right) \mathbb{E} \left[ \gamma_0 \gamma_\ell \right] > 0, \label{eqn:app-A65}
\end{align}
is a real constant. 

Next, to show that $\alpha_k < 1$, we invoke Lemma~\ref{lem:gamma-cov-PD} and~\eqref{eqn:resposnibilities-covaraince-matrix-positive-eigenvalues}
\begin{align}
    \sum_{\tau=0}^{d-1} \cos\left(\frac{2\pi k \tau}{d} \right) \mathbb{E} \left[ \gamma_0 \gamma_\tau \right] > 0. \label{eqn:app-A66}
\end{align}
Combining~\eqref{eqn:app-A66} with~\eqref{eqn:app-A65} shows that $
\alpha_k < 1$. This completes the proof and establishes the convergence of the Fourier phases. 
\end{proof}

\subsection{Positive correlation} 
\begin{lem} \label{lem:positiveCorrelation}
Let $d \in \mathbb{N}$, and suppose that $\s{X}[k] \neq 0$ for all $0 < k \leq d - 1$. Then,  
\begin{align} 
   \langle \hat{x}, x \rangle > 0, 
\end{align} 
almost surely, as $n \to \infty$. 
\end{lem}

\begin{proof}[Proof of Proposition~\ref{lem:positiveCorrelation}] 
By the definition of the inner product and the unitarity of the discrete Fourier transform, we have
    \begin{align}
        \langle \hat{x}, x \rangle = \langle \mathcal{F} \ppp{\hat{x}}, \mathcal{F} \ppp{x} \rangle = \sum_{k=0}^{d-1} \s{\hat{X}}[k] \overline{\s{{X}}[k]}.
        \label{eqn:app_E64}
    \end{align}
From Proposition~\ref{prop:2}, we know that $\s{\hat{X}}[k] \xrightarrow[]{\s{a.s.}} \alpha_k \s{X}[k]$, where each $\alpha_k > 0$ is a real constant. Applying the continuous mapping theorem, it follows that
\begin{align}
    \sum_{k=0}^{d-1} \s{\hat{X}}[k] \overline{\s{{X}}[k]} \xrightarrow[]{\s{a.s.}} \sum_{k=0}^{d-1} \alpha_k | \s{X}[k]|^2 > 0. \label{eqn:app_E65}
\end{align}
Combining \eqref{eqn:app_E64}, and \eqref{eqn:app_E65} completes the proof.
\end{proof}

\subsection{Small-signal expansion}

\begin{lem}[Small-signal expansion of $\alpha_k$]
\label{prop:alpha-k-small-signal}
Recall the definition of $\alpha_k$ from~\eqref{eqn:app-A65}
\begin{align}
    \alpha_k = 1 - d \, \sum_{\ell=0}^{d-1} \cos \left( \frac{2\pi k \ell}{d} \right) \mathbb{E} \left[ \gamma_0 \gamma_\ell \right].
\end{align}
Let $x=\beta v$ with $\beta\to 0$ and $\xi\sim\mathcal{N}(0,I_d)$, so that
\begin{align}
    \gamma_\ell \triangleq \gamma_\ell^{(\beta)}(\xi) =\frac{\exp(\beta\langle \xi,\mathcal{T}_\ell v\rangle)} {\sum_{r=0}^{d-1}\exp(\beta\langle \xi,\mathcal{T}_r v\rangle)} .
\end{align}
Then for every $k \neq 0$, as $\beta \to 0$,
\begin{align}
    \alpha_k = 1-|\s{X}[k]|^2 + O (\beta^4).
    \label{eq:alpha-k-X-expansion-unitary}
\end{align}
\end{lem}

\begin{proof}[Proof of Lemma~\ref{prop:alpha-k-small-signal}]
Since $\xi\sim\cN(0,I_d)$ is symmetric, $\xi\stackrel{\mathcal{D}}= -\xi$. Therefore,
\begin{align}
    \mathbb{E} \pp{\gamma_0^{(\beta)}(\xi)\gamma_\ell^{(\beta)}(\xi)}
    = \mathbb{E} \pp{\gamma_0^{(\beta)}(-\xi)\gamma_\ell^{(\beta)}(-\xi)}
    = \mathbb{E} \pp{\gamma_0^{(-\beta)}(\xi)\gamma_\ell^{(-\beta)}(\xi)}.
\end{align}
Hence $\mathbb{E}[\gamma_0^{(\beta)}\gamma_\ell^{(\beta)}]$ is an even function of $\beta$, and so is $\alpha_k(\beta)$. Consequently its Taylor expansion contains only even powers:
\begin{align}
    \alpha_k(\beta)= 1 -c_k\beta^2+O(\beta^4).
\end{align}
It remains to derive the $c_k$ coefficient. By Lemma~\ref{lem:resp-near-uniform-MRA-gamma},
\begin{align}
    \gamma_\ell = \frac{1}{d}
    +\frac{1}{d}\,\eta_\ell
    +\frac{1}{2d}\p{\eta_\ell^2-\overline{\eta^2}}
    +R_\ell,
    \qquad 
    |R_\ell| \le C|\beta|^3\|\xi\|^3\|v\|^3 ,
    \label{eq:gamma-expansion-unitary-compact}
\end{align}
where $\eta_\ell=s_\ell-\bar{s}$ with $s_\ell=\beta\langle \xi,\mathcal{T}_\ell v\rangle$.
Expanding \eqref{eq:gamma-expansion-unitary-compact} to second order and using that $\eta_\ell=O(\beta)$, we obtain
\begin{align}
    \gamma_0\gamma_\ell &=\Big(\frac{1}{d}+\frac{1}{d}\eta_0+\frac{1}{2d}(\eta_0^2-\overline{\eta^2})+O(\beta^3)\Big) \Big(\frac{1}{d}+\frac{1}{d}\eta_\ell+\frac{1}{2d}(\eta_\ell^2-\overline{\eta^2})+O(\beta^3)\Big) \\
    &=\frac{1}{d^2} + \frac{1}{d^2}(\eta_0+\eta_\ell) + \frac{1}{d^2}\eta_0\eta_\ell + \frac{1}{2d^2}(\eta_0^2-\overline{\eta^2}) + \frac{1}{2d^2}(\eta_\ell^2-\overline{\eta^2}) + O(\beta^3),
\end{align}
where all remaining products are of order at least $\beta^3$.
Taking expectations, the linear terms vanish since $\mathbb{E}[\eta_\ell]=0$. Moreover,
\begin{align}
    \mathbb{E}[\eta_\ell^2-\overline{\eta^2}] =\mathbb{E}[\eta_\ell^2]-\frac{1}{d}\sum_{j=0}^{d-1}\mathbb{E}[\eta_j^2]=0,
\end{align}
because the vector $\eta$ is shift-invariant (its covariance is circulant), hence $\mathbb{E}[\eta_j^2]$ is the same for all $j$.
Therefore,
\begin{align}
    \mathbb{E}[\gamma_0\gamma_\ell] =\frac{1}{d^2}+\frac{1}{d^2}\mathbb{E}[\eta_0\eta_\ell]+O(\beta^3).
\end{align}
Substituting into the definition of $\alpha_k$ and using $\sum_{\ell=0}^{d-1}\cos(2\pi k\ell/d)=0$ for $k\neq 0$, the constant term cancels and we obtain
\begin{align}
    \alpha_k &=1-d\sum_{\ell=0}^{d-1}\cos\Big(\frac{2\pi k\ell}{d}\Big)\mathbb{E}[\gamma_0\gamma_\ell] \\
    &=1-\frac{1}{d}\sum_{\ell=0}^{d-1}\cos\Big(\frac{2\pi k\ell}{d}\Big)\mathbb{E}[\eta_0\eta_\ell]+O(\beta^4),
    \qquad k\neq 0. 
    \label{eq:alpha-k-M-compact}
\end{align}

Now $\Cov(s_\ell,s_m)=\beta^2\rho_{\ell-m}$ with $\rho_r=\langle v,\mathcal{T}_r v\rangle$, so $\Cov(s)=\beta^2\Sigma_v$ where $\Sigma_v$ is circulant. Since $\eta=\Pi_{\mathrm{mean}} s$ with $\Pi_{\mathrm{mean}}=I-\frac{1}{d}\mathbf 1\mathbf 1^\top$ (which is also circulant),
\begin{align}
    \Cov(\eta) = \beta^2\Pi_{\mathrm{mean}}\Sigma_v\Pi_{\mathrm{mean}}
\end{align}
is circulant. Under the unitary DFT on the shift index, the eigenvalues of $\Sigma_v$ are $d|\s{V}[k]|^2$, and $\Pi_{\mathrm{mean}}$ projects away the $k=0$ mode and acts as identity on $k\neq 0$. Therefore,
\begin{align}
    \lambda_\eta[k] =\beta^2 d |\s{V}[k]|^2,\qquad k\neq 0.
\end{align}
For a real symmetric circulant matrix, the cosine transform of its first row equals its Fourier eigenvalues, so the sum in~\eqref{eq:alpha-k-M-compact} equals $\lambda_\eta[k]$.
Thus, for $k\neq 0$,
\begin{align}
    \alpha_k = 1-\frac{1}{d}\,\lambda_\eta[k]+O(\beta^4) = 1-\beta^2|\s{V}[k]|^2+O(\beta^4).
\end{align}
Since $\s{X}[k]=\beta\s{V}[k]$, we have $\beta^2|\s{V}[k]|^2=|\s{X}[k]|^2$, giving
\begin{align}
    \alpha_k = 1-|\s{X}[k]|^2+O(\beta^4).
\end{align}
This proves~\eqref{eq:alpha-k-X-expansion-unitary}.
\end{proof}

\subsection{High-dimentional settings}
\begin{proposition} \label{prop:auxToTheorem2Main}
Denote by $\hat{x}$ the estimator as defined in \eqref{eqn:appx_E2}, and define
\begin{align}
    \rho_{\abs{\ell_1-\ell_2}} \triangleq \langle \mathcal{T}_{\ell_1} x, \mathcal{T}_{\ell_2} x \rangle. \label{eqn:app_H1}
\end{align}
Assume that $\rho_\ell \to 0$, as $\ell \to \infty$.
Then,
    \begin{align}
        \lim_{d \to \infty} \lim_{n \to \infty} \, (\hat{x} - x) \to 0,
    \end{align}
    where the convergence is in probability. 
\end{proposition}

\begin{proof}[Proof of Proposition~\ref{prop:auxToTheorem2Main}]    
To prove the proposition, we use a similar technique as in \cite{balanov2025confirmation}. Applying SLLN on $\hat{x}$ defined in \eqref{eqn:appx_E2} results:
\begin{align}
    \hat{x}\triangleq \frac{1}{n}\sum_{i=1}^{n} \sum_{\ell=0}^{d-1}{ \gamma_{i,\ell} \mathcal{T}_{\ell}^{-1} \xi_i} \xrightarrow[]{\text{a.s.}} \mathbb{E} \pp{\sum_{\ell=0}^{d-1}{ \gamma_{\ell} \mathcal{T}_{\ell}^{-1} \xi}}. \label{eqn:app_H3}
\end{align}
We next show that under the assumptions of the proposition, we have,
\begin{align}
    \lim_{d \to \infty} \mathbb{E} \pp{\gamma_\ell  \xi} - \frac{1}{d} \mathcal{T}_\ell  x = 0. \label{eqn:app_H4}
\end{align}
To prove this, we will use Bernstein's law of large numbers \cite{durrett2019probability}. 
\begin{proposition} \label{prop:4}(Bernstein's LLN)
Let $ Y_1, Y_2, \ldots $ be a sequence of random variables with finite expectation $\mathbb{E}\p{Y_j} = \mu < \infty$, and uniformly bounded variance $\text{Var} \p{Y_j} < K < \infty$ for every $j \geq 1$, and $\text{Cov} \p{Y_i, Y_j} \to 0$, as $\abs{i-j} \to \infty$. Then,
    \begin{align}
       \frac{1}{n}\sum_{i=1}^{n} Y_i \xrightarrow[]{\calP} \mu,
    \end{align}
as $n \to \infty$.
\end{proposition}

We also use the following result, which we prove later on in this subsection.
\begin{lem} \label{lemma:5}
Assume $\rho_\ell \to 0$, as $\ell \to \infty$. Then,
    \begin{align}
       \frac{1}{d}\sum_{r=0}^{d-1} e^{{ \langle{\xi}, \mathcal{T}_r  x \rangle}} - \mathbb{E}\pp{e^{ \langle{\xi}, \mathcal{T}_1 x \rangle}} \xrightarrow[]{\calP} 0 \quad \Longrightarrow 
       \quad \frac{1}{d}\sum_{r=0}^{d-1} e^{{ \langle{\xi}, \mathcal{T}_r  x \rangle}} - {e^{\frac{\norm{x}^2}{2}}} \xrightarrow[]{\calP} 0,
       \label{eqn:lemma5}
    \end{align}
as $d \to \infty$.
\end{lem}
Define the sequence of random variables,
    \begin{align}
      Z_d \triangleq \xi \gamma_\ell = \frac{1}{d} { \frac{\xi e^{{ \langle{\xi}, \mathcal{T}_\ell  x \rangle}}}{\frac{1}{d}\sum_{r=0}^{d-1} e^{{ \langle{\xi}, \mathcal{T}_r  x \rangle}}}}, \label{eqn:asymptoticLsoftAssignA164}
    \end{align}
indexed by $d$. Lemma \ref{lemma:5} implies that the denominator in \eqref{eqn:asymptoticLsoftAssignA164} converges, as $d\to\infty$, to ${e^{\frac{\norm{x}_2^2}{2}}} > 0$ in probability. Thus, applying the continuous mapping theorem, 
    \begin{align}
      Z_d - \frac{1}{d} \xi e^{{ \langle{\xi}, \mathcal{T}_\ell  x  \rangle}-\frac{\norm{x}^2}{2}} \xrightarrow[]{\calP} 0,
    \end{align}
as $d \to \infty$. Now, since the sequence $\{Z_d\}_d$ is uniformly integrable for every $d$ (since $\abs{Z_d} \leq \abs{\xi}$ and $\mathbb{E}\abs{\xi} < \infty$), we also have that $\{Z_d\}_d$ converges in expectation, namely, 
    \begin{align}
      \lim_{d \to \infty} \mathbb{E} \pp{Z_d} - \frac{1}{d} \mathbb{E} \pp{\xi e^{{ \langle{\xi}, \mathcal{T}_\ell x \rangle}-\frac{\norm{x}_2^2}{2}}} = 0. \label{eqn:nominatorConvergenceForLargeL_0}
    \end{align}
As $\mathbb{E} \pp{\xi e^{{ \langle{\xi}, \mathcal{T}_\ell x \rangle}-\frac{\norm{x}_2^2}{2}}} = \mathcal{T}_\ell x$, we have,
    \begin{align}
      \lim_{d \to \infty} \mathbb{E} \pp{Z_d} - \frac{1}{d} \mathcal{T}_\ell x  = 0. \label{eqn:nominatorConvergenceForLargeL}
    \end{align}
Thus, combining \eqref{eqn:asymptoticLsoftAssignA164} and \eqref{eqn:nominatorConvergenceForLargeL} proves \eqref{eqn:app_H4}.
Thus, substituting \eqref{eqn:nominatorConvergenceForLargeL} into \eqref{eqn:app_H3}, we get,
    \begin{align}
       \lim_{d \to \infty} \lim_{n \to \infty} \hat{x} = \lim_{d \to \infty} \sum_{\ell=0}^{d-1}  \mathbb{E} \pp{ \gamma_{\ell} \p{\mathcal{T}_{\ell}^{-1} \xi}} = \lim_{d \to \infty} \frac{1}{d} \sum_{\ell = 0}^{d-1} \mathcal{T}_\ell^{-1}  \p{\mathcal{T}_\ell  x}.
    \end{align}
As $\mathcal{T}_\ell^{-1} (\mathcal{T}_\ell x) = x$, for every $x$ and $\ell \in \pp{d}$, we have,
\begin{align}
    \lim_{d \to \infty} \lim_{n \to \infty} \, (\hat{x}  - x) = 0,
\end{align}
as claimed. 
\end{proof}

It is left to prove Lemma \ref{lemma:5}.
\begin{proof}[Proof of Lemma \ref{lemma:5}]
Let us denote $Y_\ell = e^{{ \langle{\xi}, \mathcal{T}_\ell  x \rangle}}$. In order to apply Proposition \ref{prop:4} in our case, we need to show that the expectation of $Y_\ell$ is finite, its variance is uniformly bounded, and that the covariance decays to zero, as $d \to \infty$. For the expectation, we have for every $\ell \in \pp{d}$,
    \begin{align}
       \mathbb{E}\pp{e^{{ \langle{\xi}, \mathcal{T}_\ell x \rangle}}} = e^{\frac{\norm{x}^2}{2}} < \infty.
    \end{align}
The variance is given by,
    \begin{align}
       \s{Var}\pp{e^{{ \langle{\xi}, \mathcal{T}_\ell x \rangle}}} = e^{2\norm{x}^2} - e^{\norm{x}^2} < \infty.
    \end{align}
Finally, for the covariance we have,
    \begin{align}
       \nonumber \s{cov} \p{e^{{ \langle{\xi}, \mathcal{T}_{\ell_1} x \rangle}},e^{{ \langle{\xi}, \mathcal{T}_{\ell_2} x \rangle}}} & = \mathbb{E}\pp{e^{{ \langle{\xi}, \mathcal{T}_{\ell_1} x  + \mathcal{T}_{\ell_2} x  \rangle}}} - \mathbb{E}\pp{e^{{ \langle{\xi}, \mathcal{T}_{\ell_1} x  \rangle}}}\mathbb{E}\pp{e^{{ \langle{\xi}, \mathcal{T}_{\ell_2} x  \rangle}}} \\ 
       & = e^{{{\norm{x}^2 + { \langle{\mathcal{T}_{\ell_1} x , \mathcal{T}_{\ell_2} x  \rangle}}}}} - e^{{{\norm{x}^2}}}.
    \end{align}
By assumption, we have ${ \langle{\mathcal{T}_{\ell_1} x , \mathcal{T}_{\ell_2} x \rangle}} \to 0$, as $\abs{\ell_1 - \ell_2} \to \infty$. Thus,
    \begin{align}
       \nonumber \s{cov} \p{e^{{ \langle{\xi}, \mathcal{T}_{\ell_1} x  \rangle}},e^{{ \langle{\xi}, \mathcal{T}_{\ell_2} x  \rangle}}} \to 0.
    \end{align}
Therefore, all the assumptions of Proposition \ref{prop:4} are satisfied, which proves \eqref{eqn:lemma5}, as required.
\end{proof}

\section{ Finite-dimensional Einstein from Noise} \label{sec:finite-dimentional-EfN}

\subsection{Proof of Theorem~\ref{thm:1}}\label{thm:proofs1}

\paragraph{Contraction of Fourier components.}
We begin by establishing the convergence of the magnitudes of the estimator. By Proposition~\ref{prop:2}, whose assumptions are satisfied in our setting, we have,
\begin{align}
    \s{\hat{X}}^{(t+1)}[k] \xrightarrow[]{\text{a.s.}} \alpha_k^{(t)} \, \s{\hat{X}}^{(t)}[k],
\end{align}
for a real constant $\alpha_k^{(t)} \in (0,1)$. This directly proves the magnitude convergence result stated in \eqref{eqn:magnitudeConvergenceAsymptoticM}.

\paragraph{Fourier phases preservation.}
Next, we address the convergence of the Fourier phases. Since $\alpha_k^{(t)} > 0$, the phase of $\s{\hat{X}}^{(t+1)}[k]$ converges to the phase of $\s{\hat{X}}^{(t)}[k]$. Specifically, we obtain,
\begin{align}
    \phi_{\s{\hat{X}}^{(t+1)}}[k] \xrightarrow[]{\text{a.s.}} \phi_{\s{\hat{X}}^{(t)}}[k],
\end{align}
which establishes \eqref{eqn:fourierPhasesConvergence}.

\paragraph{Positive correlation.}
Finally, by Lemma~\ref{lem:positiveCorrelation}, which holds under the same assumptions, we have,
\begin{align}
    \langle \hat{x}^{(t+1)}, \hat{x}^{(t)} \rangle > 0,
\end{align}
which confirms \eqref{eqn:positiveCorrelation} and completes the proof of Theorem~\ref{thm:1}.

\subsection{Proof of Corollary~\ref{thm:3.5}}\label{sec:proofOfMultiIterationEfN}

\paragraph{Multi-iteration Fourier phases.}
Denote by $\hat{x}^{(t)}$ the $t^{\s{th}}$ iteration of the population EM map in~\eqref{eq:Phi-pop-1d}, and denote the initialization by $\hat{x}^{(0)}$. According to Theorem \ref{thm:1}, and under the assumptions of the proposition, for consecutive iterations, as $n \to \infty$,
    \begin{align}
        \phi_{{\s{\hat{X}}}^{(t+1)}}[k] \xrightarrow[]{\s{a.s.}}  \phi_{{\s{\hat{X}}}^{(t)}}[k].
    \end{align}
Therefore, by induction, we have,
    \begin{align}
        \phi_{{\s{\hat{X}}}^{(T)}}[k] - \phi_{{\s{\hat{X}}}^{(0)}}[k] \xrightarrow[]{\s{a.s.}} 0,
    \end{align}
for every $T<\infty$.

\paragraph{Magnitudes converge to zero as $t\to\infty$.}
Fix a frequency $k\in\{1,\dots,d-1\}$ and define $r_t \triangleq  \big|\s{\hat{X}}^{(t)}[k]\big|^2$. By Theorem~\ref{thm:1}(1), for every $t \ge 0$ there exists $\alpha_k^{(t)}\in(0,1)$ such that $\s{\hat{X}}^{(t+1)}[k]=\alpha_k^{(t)}\s{\hat{X}}^{(t)}[k]$, hence $r_{t+1}=(\alpha_k^{(t)})^2 r_t$. In particular, $(r_t)_{t\ge0}$ is strictly decreasing and bounded below by $0$, and therefore converges to a limit $r_\infty\ge 0$.

We claim that $r_\infty=0$. Suppose for contradiction that $r_\infty>0$.
Then $r_t\ge r_\infty/2$ for all sufficiently large $t$. Since the population EM map is continuous and the population iterates remain bounded (indeed, by Theorem~\ref{thm:1}(3) the iterates form a positively correlated contracting trajectory, so $\|\hat{x}^{(t)}\|$ is nonincreasing and thus uniformly bounded), the sequence $\{\hat{x}^{(t)}\}_{t\ge0}$ has at least one limit point, say $x^\infty$ along a subsequence.
By construction,
\begin{align}
    \big|\s{\hat{X}}^\infty[k]\big|^2 = r_\infty>0.
\end{align}
The contraction factor $\alpha_k(x)$ defining the population update at frequency $k$ is a continuous function of $x$, and by Theorem~\ref{thm:1}(1) satisfies $\alpha_k(x)\in(0,1)$ whenever $\s{X}[k]\neq 0$. Hence $\alpha_k(x^\infty)<1$, and by continuity there exist $\eta>0$ and $t_0$ such that
\begin{align}
    \alpha_k^{(t)} = \alpha_k(x^{(t)})  \le  1-\eta,\qquad\forall\, t\ge t_0 .
\end{align}
Consequently,
\begin{align}
    r_t \le (1-\eta)^{2(t-t_0)}\,r_{t_0}\xrightarrow[t\to\infty]{}0,
\end{align}
contradicting $r_\infty>0$. Therefore $r_\infty=0$, i.e.
\begin{align}
    \lim_{t\to\infty}\big|\s{X}^{(t)}[k]\big|=0.
\end{align}

\subsection{Proof of Theorem~\ref{thm:low-mag-rate}} \label{sec:proofOfEfNlowSNR}

Before proving the Theorem, we prove the following auxiliary proposition.

\begin{proposition}[Multi--iteration decay]
\label{prop:multi_iter_logistic}
Fix $d\ge 2$ and a frequency $k\in\{1,\dots,d-1\}$. 
Along the population EM trajectory, the $k$-th Fourier coefficient obeys the scalar recursion
\begin{align}
    \s{\hat{X}}^{(t+1)}[k] = \alpha_k^{(t)}\,\s{\hat{X}}^{(t)}[k], 
    \qquad t=0,1,2,\dots,
    \label{eqn:popultation-EM-trajectory}
\end{align}
where $\alpha_k^{(t)}$ is the population contraction factor evaluated at $x^{(t)}$. Let $x=\beta v$ with $\beta\to 0$, and assume the iterates remain in a low--signal neighborhood so that Lemma~\ref{prop:alpha-k-small-signal} applies uniformly, namely for all $t \ge 0$,
\begin{align}
    \alpha_k^{(t)}  =  1-|\s{\hat{X}}^{(t)}[k]|^2 + O(\beta^4).
    \label{eqn:alpha_k-low-SNR}
\end{align}
Let $r_t\triangleq|\s{\hat{X}}^{(t)}[k]|^2$.  
Then for every $\varepsilon\in(0,1)$ there exists $\delta(\varepsilon)>0$ such that, whenever $r_0\le \delta(\varepsilon)$, the sequence $(r_t)$ satisfies for all integers $T\ge 0$
\begin{align}
    \frac{r_0}{1+2(1+\varepsilon)r_0T}
     \le r_T \le 
    \frac{r_0}{1+2(1-\varepsilon)r_0T}.
    \label{eq:logistic-sandwich-unitary}
\end{align}
Consequently,
\begin{align}
    \frac{|\s{\hat{X}}^{(T)}[k]|}{|\s{\hat{X}}^{(0)}[k]|} \asymp \frac{1}{\sqrt{1+2T\,|\s{\hat{X}}^{(0)}[k]|^2}},
\end{align}
uniformly in $T \ge 0$ as $r_0 \to 0$.
\end{proposition}

\begin{proof}[Proof of Proposition~\ref{prop:multi_iter_logistic}]
From~\eqref{eqn:popultation-EM-trajectory} and~\eqref{eqn:alpha_k-low-SNR}, write $\alpha_k^{(t)}=1-r_t+E_t$ with $|E_t|\le C r_t^2$ for a uniform constant $C>0$. Then,
\begin{align}
    r_{t+1}
    &= |\s{\hat{X}}^{(t+1)}[k]|^2
     = (\alpha_k^{(t)})^2 r_t
     = (1-r_t+E_t)^2 r_t
     \nonumber\\
    &= r_t - 2r_t^2 + \Delta_t,
    \label{eq:r-rec-unitary}
\end{align}
where $\Delta_t=r_t^3+2E_t r_t(1-r_t)+E_t^2 r_t$ satisfies $|\Delta_t|\le c_1 r_t^3$
for a constant $c_1$ depending only on $C$ and the low--signal radius.

Fix $\varepsilon\in(0,1)$. For $\delta(\varepsilon)$ small enough the map is contracting, 
so $r_t\le r_0$ for all $t$. Choosing $\delta(\varepsilon)$ additionally so that 
$c_1 r_0\le \varepsilon$ yields
\begin{align}
    |\Delta_t|\le 2 \varepsilon r_t^2,
\end{align}
and therefore
\begin{align}
    2(1-\varepsilon)r_t^2 \le r_t-r_{t+1} \le 2(1+\varepsilon)r_t^2,
    \qquad t\ge 0.
    \label{eq:disc-ineq-multi-unitary}
\end{align}

Next, introduce the comparison sequences $(u_t)$ and $(v_t)$ 
\begin{align}
    \frac{1}{u_{t+1}}-\frac{1}{u_t} & =2 (1-\varepsilon) ,\quad u_0=r_0,
    \\
    \frac{1}{v_{t+1}}-\frac{1}{v_t} & = 2(1+\varepsilon) ,\quad v_0=r_0.
\end{align}
These sequences solve explicitly to
\begin{align}
    u_t=\frac{r_0}{1+2(1-\varepsilon)r_0 t}, \qquad v_t=\frac{r_0}{1+2(1+\varepsilon)r_0 t}.
\end{align}
Using~\eqref{eq:disc-ineq-multi-unitary} and a standard induction argument,
$v_t \le r_t \le u_t$ for all $t \ge 0$, and evaluating at $t=T$
gives~\eqref{eq:logistic-sandwich-unitary}. The final asymptotic relation for $|\s{\hat{X}}^{(T)}[k]|$ follows immediately by taking square roots.
\end{proof}

We are now ready to prove the Theorem. By Corollary~\ref{thm:3.5}, we have for every $t\ge 0$ and $k\neq 0$,
\begin{align}
    |\s{\hat{X}}^{(t+1)}[k]|\le |\s{\hat{X}}^{(t)}[k]|.
\end{align}
Hence $\|\hat{x}^{(t+1)}\|_2\le \|\hat{x}^{(t)}\|_2$ by Parseval. 
Therefore, if $\|\hat{x}^{(0)}\|\le \delta$, then $\|\hat{x}^{(t)}\|\le \delta$ for all $t\ge 0$.
In particular, provided $\delta$ is chosen sufficiently small, every iterate lies in the small--signal neighborhood where Lemma~\ref{prop:alpha-k-small-signal} applies.

\paragraph{Successive-step expansion.}
Fix $t \ge 0$ and $1\le k\le d-1$. Apply Lemma~\ref{prop:alpha-k-small-signal} to $x=\hat{x}^{(t)}$. Since $\|\hat{x}^{(t)}\|\le \delta$ and $\delta$ is small, the lemma yields
\begin{align}
    \alpha_k^{(t)}
    = 1-|\s{\hat{X}}^{(t)}[k]|^2 + O (\beta^4),
    \qquad k\neq 0,
\end{align}
Writing the remainder as $E_{t,k}$ gives \eqref{eqn:alpha-low-mag-expansion} with $|E_{t,k}|\le C|\s{\hat{X}}^{(t)}[k]|^4$.

\paragraph{Uniform multi-iteration bounds.}
Fix $k\neq 0$ and set $r_t\triangleq|\s{ \hat{X}}^{(t)}[k]|^2$.  
By the recursion \eqref{eqn:popultation-EM-trajectory} from Proposition~\ref{prop:multi_iter_logistic} and the expansion in part (1), the assumptions of Proposition~\ref{prop:multi_iter_logistic} are satisfied whenever $r_0$ is small enough. Therefore, for any $\varepsilon\in(0,1)$ there exists $\delta(\varepsilon)>0$ such that if $r_0\le \delta(\varepsilon)$, then for all $T\ge 0$,
\begin{align}
    \frac{r_0}{1+2(1+\varepsilon)r_0T}
    \le r_T \le 
    \frac{r_0}{1+2(1-\varepsilon)r_0T}.
\end{align}
Taking square roots and recalling $r_t=|\s{\hat{X}}^{(t)}[k]|^2$ yields exactly~\eqref{eqn:multi-iter-logistic-bounds}.

\section{High-dimensional Einstein from Noise}\label{sec:proofOfHighDimensionEfN}

The proof of the theorem is based on invoking Proposition~\ref{prop:auxToTheorem2Main} over $T$ successive iterations. Fix $1 \le T<\infty$. For each dimension $d$ and iteration index $t$, define the population EM iterates by
\begin{align}
     \hat{x}_d^{(t+1)} = M_d \big(\hat{x}_d^{(t)}\big),
\end{align}
where $M_d$ denotes the population EM operator in dimension $d$ (i.e., the $n\to\infty$ limit of one empirical EM step). Recall the shift--autocorrelations from~\eqref{eqn:app_H1_1_rewrite}
\begin{align}
    \rho^{(t)}_{d,\ell} \triangleq \big\langle \mathcal{T}_\ell \hat{x}_d^{(t)}, \hat{x}_d^{(t)}\big\rangle,
    \qquad \ell=0,\dots,d-1.
\end{align}

\paragraph{Step 1: Successive iterations.}
Fix $t\in\{0,\dots,T-1\}$ and consider one population EM update in dimension $d$ with current signal $x=\hat{x}_d^{(t)}$. Let $\hat{x}$ denote the empirical one-step estimator defined in~\eqref{eqn:appx_E2} based on $n$ samples. By definition of the population map,
\begin{align}
    M_d \big(\hat{x}_d^{(t)}\big) = \lim_{n\to\infty}\hat{x}.
\end{align}

To invoke Proposition~\ref{prop:auxToTheorem2Main}, note that for any $\ell_1,\ell_2$,
\begin{align}
    \rho_{|\ell_1-\ell_2|} \triangleq \big\langle \mathcal{T}_{\ell_1}x,\mathcal{T}_{\ell_2}x\big\rangle = \big\langle \mathcal{T}_{|\ell_1-\ell_2|}x,x\big\rangle,
\end{align}
so $\rho_\ell$ is precisely the inner product between $x$ and its $\ell$-shift.
Therefore, when $x=\hat{x}_d^{(t)}$, the correlation sequence appearing in~\eqref{eqn:app_H1} is exactly
\begin{align}
    \rho^{(t)}_{d,\ell} = \big\langle \mathcal{T}_{\ell}\hat{x}^{(t)}_d,\hat{x}^{(t)}_d\big\rangle,
    \qquad \ell\in\{0,\dots,d-1\}.
    \label{eqn:app_H1_1_rewrite-2}
\end{align}
By Assumption~\eqref{eqn:rho_decay_assump}, the decay condition $\rho_\ell\to 0$ as $\ell\to\infty$ in~\eqref{eqn:app_H1} holds along the sequence $\{\rho^{(t)}_{d,\ell}\}$ as $d\to\infty$; thus, Proposition~\ref{prop:auxToTheorem2Main} yields
\begin{align}
    \lim_{d\to\infty}\lim_{n\to\infty}\,(\hat{x}-x)=0
    \qquad\text{in probability}.
\end{align}
Since $x=\hat{x}_d^{(t)}$ and $\hat{x}_d^{(t+1)}=M_d(\hat{x}_d^{(t)})=\lim_{n\to\infty}\hat{x}$, we conclude that
\begin{align}
    \hat{x}_d^{(t+1)}-\hat{x}_d^{(t)} =\lim_{n\to\infty}(\hat{x}-x) \xrightarrow[d\to\infty]{\mathcal{P}} 0.
\end{align}
Because the left-hand side is deterministic (for fixed $d$), this convergence in probability is equivalent to convergence in norm, and hence
\begin{align}
    \lim_{d\to\infty}\big\|\hat{x}_d^{(t+1)}-\hat{x}_d^{(t)}\big\|_2 = 0,
\end{align}
which proves~\eqref{eqn:succ_iter_highdim}.

\paragraph{Step 2: Multi-iteration bound.}
By the triangle inequality and a telescoping sum,
\begin{align}
    \big\|\hat{x}_d^{(T)}-\hat{x}_d^{(0)}\big\|_2 \le \sum_{t=0}^{T-1}\big\|\hat{x}_d^{(t+1)}-\hat{x}_d^{(t)}\big\|_2 .
\end{align}
By Step~1, for each fixed $t$ the summand converges to $0$ as $d\to\infty$.
Since $T$ is fixed and finite, letting $d\to\infty$ yields
\begin{align}
    \lim_{d\to\infty}\big\|\hat{x}_d^{(T)}-\hat{x}_d^{(0)}\big\|_2 = 0,
\end{align}
establishing~\eqref{eqn:multi_iter_highdim}.

\section{Finite-sample Einstein from Noise}
\label{sec:finite-sample-EfN-Main}

\subsection{Fourier phases convergence rate} \label{sec:FourierPhasesConvergenceRate}

\paragraph{Notation.}
In this section, we prove several auxiliary statements needed in the proof of Theorem \ref{thm:1}. Then, we establish a few additional notations which will be used to prove the convergence rate of the Fourier phases.
Specifically, by \eqref{eqn:strongLLN},
\begin{align}
   \s{\hat{X}}[k] e^{-j\phi_{\s{X}}[k]} \xrightarrow[]{\s{a.s.}} 
    \mathbb{E} \left[ \abs{\s{N}[k]} \sum_{\ell=0}^{d-1} \gamma_{\ell} \cos\p{\phi_{\ell}[k]} \right] + j \mathbb{E} \left[ \abs{\s{N}[k]} \sum_{\ell=0}^{d-1} \gamma_{\ell} \sin\p{\phi_{\ell}[k]} \right],\label{eqn:app_F1}
\end{align}
where $\phi_{i,\ell}[k]$ is defined in \eqref{eqn:phaseDifferenceTerm}. We denote for every $0 \leq k \leq d-1$,
\begin{align}
    \mu_{\s{A},k} &\triangleq \mathbb{E}\left[ \abs{\s{N}[k]} \sum_{\ell=0}^{d-1} \gamma_{\ell} \sin (\phi_{ \ell}[k]) \right],  \label{eqn:muA}
    \\
    \mu_{\s{B},k} &\triangleq  \mathbb{E}\left[ \abs{\s{N}[k]} \sum_{\ell=0}^{d-1} \gamma_{\ell} \cos (\phi_{\ell}[k])  \right], \label{eqn:muB}
\end{align}
the imaginary and real parts of the right-hand-side of \eqref{eqn:app_F1}, respectively. In addition, we denote,
\begin{align}
    \sigma_{\s{A},k}^2 \triangleq \s{Var}\p{\abs{\s{N}[k]} \sum_{\ell=0}^{d-1} \gamma_{\ell} \sin (\phi_{\ell}[k]) },
    \label{eqn:sigmaA}
    \\ \sigma_{\s{B},k}^2 \triangleq \s{Var}\p{\abs{\s{N}[k]} \sum_{\ell=0}^{d-1} \gamma_{\ell} \cos (\phi_{\ell}[k]) }.
    \label{eqn:sigmaB}    
\end{align}
By Proposition \ref{prop:2}, $\s{\hat{X}}[k] \xrightarrow[]{\s{a.s.}} \alpha_k \s{X}[k]$, for a positive real constant $\alpha_k$. This in turn implies that,  $\mu_{\s{A},k} = 0$ while $\mu_{\s{B},k} > 0$. Consequently, by \eqref{eqn:app_F1}, as $n \to \infty$, the estimator converges to a non-vanishing signal, and its Fourier phases converge those of the signal $x$.

Next, for every $0 \leq k \leq d-1$,
\begin{align}
    \s{A}_{n}[k] \triangleq \frac{1}{\sqrt{n}}\sum_{i=1}^{n}\abs{\s{N}_i[k]} \sum_{\ell=0}^{d-1} \gamma_{i,\ell} \sin \left( \phi_{i, \ell}[k] \right), \label{eqn:AMdef}
\end{align}
and,
\begin{align}
    \s{B}_{n}[k] \triangleq \frac{1}{n}\sum_{i=1}^{n}\abs{\s{N}_i[k]} \sum_{\ell = 0}^{d-1} \gamma_{i,\ell} \cos \left( \phi_{i, \ell}[k] \right). \label{eqn:BMdef}
\end{align}
Note that $\s{A}_{n}[k]$ is normalized by $1/\sqrt{n}$ instead of $1/n$, to facilitate the analysis of the convergence rate. Additionally, we define the following Gaussian random variable $\s{Q}_k$,
\begin{align}
   \s{Q}_k \sim \calN \left( 0, \frac{\sigma_{\s{A},k}^2}{\mu_{\s{B},k}^2} \right),
   \label{eqn: A51}
\end{align}
for every $0 \leq k \leq d-1$.

\paragraph{Fourier phases convergence rate.}
\label{sec:convergenceInExpectationOfFourierPhases} 

\begin{proposition}
\label{prop:convergeneInExpectation} 
Recall the definitions of $\mu_{\s{B},k}$, and $\sigma_{\s{A},k}^2$ in \eqref{eqn:muB}, and \eqref{eqn:sigmaB}, respectively. Assume that $\s{X}[k] \neq 0$, for all $0< k \leq d-1$. Then, as $n \to \infty$,
    \begin{align}
       \lim_{n\to\infty} \frac{\mathbb{E} |\phi_{\s{\hat{X}}}[k] - \phi_{\s{X}}[k]|^2}{1/n} = \frac{\sigma_{\s{A},k}^2}{\mu_{\s{B},k}^2}.
        \label{eqn:B19}
    \end{align}
\end{proposition}

\begin{proof}[Proof of Proposition~\ref{prop:convergeneInExpectation}]
The proof of this proposition is similar to the proof found in \cite{balanov2024einstein}.
Recall the definitions of $\s{A}_{n}[k]$ and $\s{B}_{n}[k]$ in \eqref{eqn:AMdef} and \eqref{eqn:BMdef}, respectively, and of $\s{Q}_k$ in \eqref{eqn: A51}. Then, using the phase difference expression in \eqref{eqn:estimatorPhase}, it follows that establishing \eqref{eqn:B19} is equivalent to proving the following:
\begin{align}
        \lim_{n\to\infty}\frac{ \mathbb{E} \pp{\arctan^2 {\left(\frac{1}{\sqrt{n}}\frac{\s{A}_{n}[k]}{\s{B}_{n}[k]}\right)}}}{ \frac{1}{{n}}\mathbb{E} \pp{ {\s{Q}^2_k}}} = 1, \label{eqn:B20}
    \end{align}
    for every $0 \leq k \leq d-1$. Recall by the definition of $\s{Q}_k$ in \eqref{eqn: A51} that $\mathbb{E} \pp{\s{Q}_k^2} = \sigma_{\s{A},k}^2 / \mu_{\s{B},k}^2$, which is equivalent to the right-hand-side of \eqref{eqn:B19}.

For brevity, we fix $k$, and denote $\s{A}_{n} = \s{A}_{n}[k]$, $\s{B}_{n} = \s{B}_{n}[k]$, $\mu_{\s{B}} = \mu_{\s{B},k}$, $\sigma_{\s{A}}^2 = \sigma_{\s{A},k}^2$.
Using \eqref{eqn:estimatorPhase} it is clear that,
\begin{align}
    \sqrt{n}\cdot\tan\pp{\phi_{\s{\hat{X}}}[k] - \phi_{\s{X}}[k]} = \frac{\frac{1}{\sqrt{n}}\sum_{i=1}^{n}\abs{\s{N}_i[k]} \sum_{\ell=0}^{d-1} \gamma_{i, \ell} \sin \left( \phi_{i, \ell}[k] \right)}{\frac{1}{n}\sum_{i=1}^{n}\abs{\s{N}_i[k]} \sum_{\ell=0}^{d-1} \gamma_{i,\ell} \cos \left( \phi_{i, \ell}[k] \right)}\triangleq\frac{\s{A}_n}{\s{B}_n}. \label{eqn:arctanExpression}
\end{align}
It is important to note that the denominator $\s{B}_n$ can be zero with positive probability, implying that the expression in \eqref{eqn:arctanExpression} may diverge with non-zero probability. Therefore, it is necessary to control the occurrence of such events. To this end, $\s{B}_n \xrightarrow[]{\s{a.s.}}  \mu_\s{B}$, by SLLN (see Section \ref{sec:convergenceOfEfNestimator}), where $\mu_\s{B}$ is defined in \eqref{eqn:muB}. Fix $0 < \epsilon < \mu_\s{B}$, and proceed by decomposing as follows:
\begin{align}
    \mathbb{E} \pp{\arctan^2 {\left(\frac{1}{\sqrt{n}}\frac{\s{A}_n}{\s{B}_n}\right)}}& = \mathbb{E} \pp{\arctan^2 {\left(\frac{1}{\sqrt{n}}\frac{\s{A}_n}{\s{B}_n}\right) \mathbbm{1}_{\abs{\s{B}_n} > \epsilon}}}  \nonumber\\ &\quad\quad+\mathbb{E} \pp{\arctan^2 {\left(\frac{1}{\sqrt{n}}\frac{\s{A}_n}{\s{B}_n}\right) \mathbbm{1}_{\abs{\s{B}_n} < \epsilon}}} \label{eqn:splittingArgTanIntoTwoTerms}.
\end{align}
The next result shows that the second term at the r.h.s. of \eqref{eqn:splittingArgTanIntoTwoTerms} converges to zero with rate $O(1/n^2)$.
\begin{lem} \label{lem:B5}
    The following inequality holds,
    \begin{align}
        \mathbb{E} \pp{\arctan^2 {\left(\frac{1}{\sqrt{n}}\frac{\s{A}_n}{\s{B}_n}\right) \mathbbm{1}_{\abs{\s{B}_n} < \epsilon}}} \leq \frac{D}{n^2}, \label{eqn:B23}
    \end{align}
    for a finite $D > 0$.
\end{lem}
In addition, we have the following asymptotic relation for the last term in \eqref{eqn:splittingArgTanIntoTwoTerms}.
\begin{lem} \label{lem:B6}
    The following asymptotic relation hold,
        \begin{align}
         \lim_{n\to\infty}\frac{ \mathbb{E} \pp{\arctan^2 {\left(\frac{1}{\sqrt{n}}\frac{\s{A}_n}{\s{B}_n}  \right) \mathbbm{1}_{\abs{\s{B}_n} > \epsilon}}}}{ \frac{1}{{n}}\mathbb{E} \pp{ {\s{Q}^2_k}}} = 1 \label{eqn:arctanWithPositiveEventConvergenceToOne1}.
    \end{align}
\end{lem}
The proofs of these lemmas can be found in \cite{balanov2024einstein}.
Substituting \eqref{eqn:B23} and \eqref{eqn:arctanWithPositiveEventConvergenceToOne1} in \eqref{eqn:splittingArgTanIntoTwoTerms}, leads to \eqref{eqn:B20}, which completes the proof.
\end{proof}

\subsection{Fourier phases convergence at low-magnitude expansion}

Here, we approximate~\eqref{eqn: A51} in the low-magnitude approximation when $x = \beta v$, where $\beta \to 0$. Let us define
\begin{align} \label{eqn:app_M15}
    C_k(x) \triangleq \frac{\sigma_{\s{A},k}^2 (x)}{\mu_{\s{B},k}^2 (x)},
\end{align}
as the right-hand-side of~\eqref{eqn:B19}. 

\subsubsection{Scaling of the Fourier phase convergence constant}
Our next lemma shows that $C_k(x) = \Theta(\beta^2)$. 

\begin{lem}[Low-magnitude expansions of $\mu_{\s B,k}$ and $\sigma_{\s{A},k}^2$]
\label{lem:muB-sigmaA-low-magnitude-short}
Assume $\s{X}[k]\neq 0$ for every $k \neq 0$, and fix a nonzero frequency $k\neq 0$. Let $x=\beta v$ with $\beta \to 0$ and $\xi\sim\mathcal{N}(0,I_d)$.
Recall the definitions of $\mu_{\s B,k}(x)$ in~\eqref{eqn:muB} and $\sigma_{\s{A},k}^2(x)$ in~\eqref{eqn:sigmaA}. Then:
\begin{align}
    \mu_{\s B,k}(x) &= |\s{X}[k]| (1-|\s{X}[k]|^2+O(\beta^4)),
    \label{eq:muB_low_signal}
    \\
    \sigma_{\s{A},k}^2 (x) &= \beta^4\,\operatorname{Var}\left( |\s{N}[k]|\sum_{\ell=0}^{d-1}\frac{1}{2d}\,\langle\xi,\mathcal{T}_\ell v\rangle^2\,\sin(\phi_\ell[k])\right) +O(\beta^6), \label{eq:sigmaA_low_signal_short}
\end{align}
where the $O(\beta^4)$ and $O(\beta^6)$ terms are uniform over $\|v\|_2\le 1$ (for fixed $d,k$).
\end{lem}

\begin{proof}[Proof of Lemma~\ref{lem:muB-sigmaA-low-magnitude-short}]
We prove in steps.
\paragraph{Step 1: Expansion of $\mu_{\s B,k}(x)$.}
By~\eqref{eqn:muB} and Proposition~\ref{prop:2}, $\s{\hat{X}}[k]\xrightarrow[]{\mathrm{a.s.}}\alpha_k(x)\s{X}[k]$, hence
\begin{align}
    \mu_{\s B,k}=\mathbb{E}[|\s{\hat{X}}[k]|]=\alpha_k(x)\,|\s{X}[k]|.
    \label{eq:muB_alpha_def}
\end{align}
Lemma~\ref{prop:alpha-k-small-signal} gives, in the low-magnitude regime $x=\beta v$,
\begin{align}
    \alpha_k(x)=1-|\s{X}[k]|^2+O(\beta^4).
    \label{eq:alpha_low_mag}
\end{align}
Combining \eqref{eq:muB_alpha_def}--\eqref{eq:alpha_low_mag} yields \eqref{eq:muB_low_signal}.

\paragraph{Step 2: Softmax expansion and removal of $\ell$--independent terms.}
By~\eqref{eqn:sigmaA},
\begin{align}
    \sigma_{\s{A},k}^2 =\operatorname{Var}\left(|\s{N}[k]|\sum_{\ell=0}^{d-1}\gamma_\ell\,\sin(\phi_\ell[k])\right).    \label{eq:sigmaA_def_start}
\end{align}
Using Lemma~\ref{lem:resp-near-uniform-MRA-gamma}, for $s_\ell=\beta\langle\xi,\mathcal{T}_\ell v\rangle$,
$\bar s=\frac{1}{d}\sum_{\ell=0}^{d-1}s_\ell$, $\eta_\ell=s_\ell-\bar s$, and
$\overline{\eta^2}=\frac{1}{d}\sum_{\ell=0}^{d-1}\eta_\ell^2$, we have
\begin{align}
    \gamma_\ell =\frac{1}{d}+\frac{1}{d}\eta_\ell+\frac{1}{2d}\big(\eta_\ell^2-\overline{\eta^2}\big)+R_\ell, \qquad |R_\ell|\le C|\beta|^3\|\xi\|^3\|v\|^3.    \label{eq:gamma_expansion_app}
\end{align}
Since $k\neq 0$, the sinusoid $\sin(\phi_\ell[k])$ has zero mean over $\ell$, namely
\begin{align}
    \sum_{\ell=0}^{d-1}\sin(\phi_\ell[k])=0.
    \label{eq:sin_sum_zero}
\end{align}
Therefore any $\ell$--independent term in \eqref{eq:gamma_expansion_app} vanishes after summation against
$\sin(\phi_\ell[k])$, and in particular the $\frac{1}{d}$ term and the $-\overline{\eta^2}$ term drop.

\paragraph{Step 3: The linear term cancels.}
Let $\Delta_k\triangleq \phi_{\s{N}}[k]-\phi_{\s{X}}[k]$. Using $\sin(u)=\mathrm{Im}(e^{iu})$,
\begin{align}
    |\s{N}[k]|\sum_{\ell=0}^{d-1}s_\ell\sin\Big(\tfrac{2\pi k\ell}{d}+\Delta_k\Big)
    &= |\s{N}[k]|\;\mathrm{Im}\left(e^{i\Delta_k}\sum_{\ell=0}^{d-1}s_\ell e^{2\pi i k\ell/d} \right).    \label{eq:linear_term_as_imag}
\end{align}
Since $s_\ell=\beta\langle\xi,\mathcal{T}_\ell v\rangle$, we write
\begin{align}
    \sum_{\ell=0}^{d-1}s_\ell e^{2\pi i k\ell/d} = \beta\sum_{\ell=0}^{d-1}\langle\xi,\mathcal{T}_\ell v\rangle e^{2\pi i k\ell/d}.
    \label{eq:s_as_corr}
\end{align}
A standard correlation/DFT identity (for $k\neq 0$) gives
\begin{align}
    \sum_{\ell=0}^{d-1}\langle\xi,\mathcal{T}_\ell v\rangle e^{2\pi i k\ell/d} = d\sqrt d\,\overline{\s{N}[k]}\,\s{V}[k].    \label{eq:corr_dft_identity}
\end{align}
Substituting \eqref{eq:s_as_corr}--\eqref{eq:corr_dft_identity} into \eqref{eq:linear_term_as_imag} yields
\begin{align}
    |\s{N}[k]|\sum_{\ell=0}^{d-1}s_\ell\sin\Big(\tfrac{2\pi k\ell}{d}+\Delta_k\Big) &= |\s{N}[k]|\;\mathrm{Im}\left( e^{i\Delta_k}\cdot \beta\, d\sqrt d\,\overline{\s{N}[k]}\,\s{V}[k]\right)
    \nonumber\\
    &=
    \beta\,d\sqrt d\,|\s{N}[k]|^2\; \mathrm{Im}\left(e^{-i\phi_{\s{X}}[k]}\,\s{V}[k]\right),    \label{eq:linear_term_reduced}
\end{align}
where we used $e^{i\Delta_k}\overline{\s{N}[k]}=e^{-i\phi_{\s{X}}[k]}|\s{N}[k]|$.
Finally, since $\s{X}[k]=\beta\s{V}[k]$, we have $\phi_{\s{X}}[k]=\phi_{\s{V}}[k]$, hence
\begin{align}
    e^{-i\phi_{\s{X}}[k]}\,\s{V}[k]=|\s{V}[k]|\in\mathbb{R}
    \quad\Longrightarrow\quad
    \mathrm{Im}\left(e^{-i\phi_{\s{X}}[k]}\,\s{V}[k]\right)=0.
    \label{eq:imag_zero}
\end{align}
Combining \eqref{eq:linear_term_reduced}--\eqref{eq:imag_zero} shows that the entire linear contribution is $0$ almost surely.

\paragraph{Step 4: Quadratic leading term.}
From \eqref{eq:gamma_expansion_app} together with \eqref{eq:sin_sum_zero} and the linear cancellation above,
\begin{align}
    |\s{N}[k]|\sum_{\ell=0}^{d-1}\gamma_\ell\sin(\phi_\ell[k])
    &=
    |\s{N}[k]|\sum_{\ell=0}^{d-1}\frac{1}{2d}\,\eta_\ell^2\,\sin(\phi_\ell[k]) \;+ \;O(\beta^3\|\xi\|^3).
    \label{eq:quad_term_eta2}
\end{align}
Now expand $\eta_\ell^2=(s_\ell-\bar s)^2=s_\ell^2-2\bar s\,s_\ell+\bar s^2$. Using \eqref{eq:sin_sum_zero}
and the linear cancellation (applied to $\sum_\ell s_\ell\sin(\phi_\ell[k])$), we get
\begin{align}
    \sum_{\ell=0}^{d-1}\eta_\ell^2\sin(\phi_\ell[k]) = \sum_{\ell=0}^{d-1}s_\ell^2\sin(\phi_\ell[k]).
    \label{eq:eta2_to_s2}
\end{align}
Substituting \eqref{eq:eta2_to_s2} into \eqref{eq:quad_term_eta2}, and recalling $s_\ell=\beta\langle\xi,\mathcal{T}_\ell v\rangle$,
we obtain
\begin{align}
    |\s{N}[k]|\sum_{\ell=0}^{d-1}\gamma_\ell\sin(\phi_\ell[k]) &= \beta^2\,|\s{N}[k]|\sum_{\ell=0}^{d-1}\frac{1}{2d}\,\langle\xi,\mathcal{T}_\ell v\rangle^2\,\sin(\phi_\ell[k]) \;+\;O(\beta^3\|\xi\|^3). \label{eq:leading_quadratic_form}
\end{align}

\paragraph{Step 5: Taking variance.}
Applying $\operatorname{Var}(\cdot)$ to \eqref{eq:leading_quadratic_form} and using $\mathbb{E}\|\xi\|^6<\infty$ to control the remainder and cross-terms gives
\begin{align}
    \sigma_{\s{A},k}^2 &= \beta^4\,\operatorname{Var}\left(|\s{N}[k]|\sum_{\ell=0}^{d-1}\frac{1}{2d}\,\langle\xi,\mathcal{T}_\ell v\rangle^2\,\sin(\phi_\ell[k]) \right) +O(\beta^6),
\end{align}
which is exactly \eqref{eq:sigmaA_low_signal_short}, which completes the proof of the statement. 
We note that for $\s{X}[k] \neq 0$, the $\beta^4$ coefficient is non-vanishing, that is,
\begin{align}
    \operatorname{Var}\left(|\s{N}[k]|\sum_{\ell=0}^{d-1}\frac{1}{2d}\,\langle\xi,\mathcal{T}_\ell v\rangle^2\,\sin(\phi_\ell[k])\right) > 0.
\end{align} 
\end{proof}

\subsubsection{Successive population EM iterations}
Next, we analyze the relation between two successive iterations of $C_k(x)$, assuming the successive iterations follow the population EM map. Explicitly, for $\hat{x} = M(x)$, where $M$ is the population EM map we would analyze the scaling of $|C_k(\hat{x}) - C_k(x)|$.

\begin{lem}[One-step population drift of $C_k$ in the low-magnitude regime]
\label{lem:Ck-relative-drift}
For $x\in\mathbb{R}^d$ recall the definition of $C_k(x)$ from~\eqref{eqn:app_M15}. Let $\hat{x}=M(x)$ be the population EM update, and  $x=\beta v$ with $\beta \to 0$, and $\|v \| \leq P$. Assume that $\min_{k \neq 0} |\s{V}[k]| \geq c_0 > 0$. Then there exist constants $\beta_0>0$ and $C<\infty$ (depending only on $d,k,c_0, P$) such that for all $0<\beta\le\beta_0$
\begin{align}
\label{eq:Ck-relative-drift}
    |C_k(M(x)) - C_k(x)| \le C\,\beta^4.
\end{align}
\end{lem}

\begin{proof}[Proof of Lemma~\ref{lem:Ck-relative-drift}]
Let $x=\beta v$ and $\hat{x}=M(x)$ and write $\hat{x}=\beta \hat{v}$.
By \eqref{eq:muB_low_signal}--\eqref{eq:sigmaA_low_signal_short}, 
\begin{align}
    C_k(\beta v)=\frac{\beta^4 a_k(v)+O(\beta^6)}{\beta^2 b_k(v)\,(1+O(\beta^2))} = \beta^2\,\frac{a_k(v)}{b_k(v)} + O(\beta^4). \label{eqn:app_M35}
\end{align}
where we have defined,
\begin{align}
    a_k(v) & \triangleq \operatorname{Var}\left( |\s{N}[k]|\sum_{\ell=0}^{d-1}\frac{1}{2d}\,\langle\xi,\mathcal{T}_\ell v\rangle^2\,\sin(\phi_\ell[k])\right), \label{eqn:ak_v_def}
    \\ 
    b_k(v) & \triangleq |\s{V}[k]|^2. \label{eqn:bk_v_def}
\end{align}
Define $c_k(v)\triangleq a_k(v)/b_k(v)$; by the nondegeneracy assumption, $\min_{k \neq 0} |\s{V}[k]| \geq c_0 > 0$, we have, $c_k(v)\in(0,\infty)$.

The drift bound in Theorem~\ref{thm:mra-lowSNR-iteration-bias-to-init} and~\eqref{eq:mra-iter-bound} implies that the nonzero-frequency part of $\hat{x}-x$ has norm $O(\beta^3)$.
Equivalently, after rescaling by $\beta$,
\begin{align}
    \label{eq:vplus-v}
    \|(I-\Pi_{\mathrm{mean}})(\hat{v}-v)\|\le C_0\,\beta^2.
\end{align}
In particular, for fixed $k\neq 0$,
\begin{align}
    \label{eq:Vk-lip}
    |\s{\hat{V}}[k]-\s{V}[k]|\le \|\hat{v}-v\|\le C_0\,\beta^2.
\end{align}

The map $v\mapsto b_k(v)=|\s{V}[k]|^2$~\eqref{eqn:bk_v_def} is smooth and Lipschitz on $\{ \|v\|\le P\}$.
Moreover, $a_k(v)$~\eqref{eqn:ak_v_def} is a variance of a polynomial functional of $v$, hence it is also smooth and Lipschitz on bounded sets.
Therefore, on any neighborhood where $b_k(v)$ is bounded away from $0$, the ratio $c_k(v)=a_k(v)/b_k(v)$ is Lipschitz.
Since $b_k(v)>0$ is fixed, there exists $\rho>0$ and $L_k<\infty$ such that whenever $\|u-v\|\le \rho$,
\begin{align}
    \label{eq:ck-lip}
    |c_k(u)-c_k(v)|\le L_k\,\|u-v\|.
\end{align}
For $\beta$ small enough, \eqref{eq:vplus-v} ensures $\|\hat{v}-v\|\le \rho$, so \eqref{eq:ck-lip} applies.

By~\eqref{eqn:app_M35}, we have,
\begin{align}
    C_k(\beta \hat{v})-C_k(\beta v) = \beta^2\big(c_k(\hat{v})-c_k(v)\big)+O(\beta^4).
\end{align}
By \eqref{eq:ck-lip} and \eqref{eq:vplus-v}, $|c_k(\hat{v})-c_k(v)|\le L_k\|\hat{v}-v\|\le C\,\beta^2$. Hence
\begin{align}
\label{eq:Ck-diff-abs}
    |C_k(\beta \hat{v})-C_k(\beta v)| \le C\,\beta^4.
\end{align}
which is \eqref{eq:Ck-relative-drift}, completing the proof. 
\end{proof}

\subsubsection{Successive empirical EM iterations}
Next, we analyze the relation between two successive iterations of $C_k(x)$, assuming the successive iterations follow the \emph{empirical} EM map. Explicitly, for $\hat{x} = M_n(x)$, where $M_n$ is the empirical EM map we would analyze the scaling of $|C_k(\hat{x}) - C_k(x)|$.

\begin{lem}[One-step empirical drift of $C_k$ in the low-magnitude regime]
\label{lem:Ck-empirical-drift}
Recall $C_k(\cdot)$ from~\eqref{eqn:app_M15}. Fix $k\neq 0$ and let $x=\beta v$ with $\beta\to 0$ and $\|v\|_2\le P$. Assume $\min_{k\neq 0}|\s V[k]|\ge c_0>0$. Let $\hat{x}_{\mathrm{pop}} \triangleq M(x)$ and $\hat{x} \triangleq M_n(x)$. There exist constants $\beta_0>0$ and $C<\infty$ (depending only on $d,k,P,c_0$) such that for all $0<\beta\le \beta_0$,
\begin{align}
    \label{eq:Ck-empirical-triangle}
    |C_k(\hat{x})-C_k(x)| \le C\,\beta^4 + C\,\beta\,\|\hat{x}-\hat{x}_{\mathrm{pop}}\|_2 .
\end{align}
In particular, on any event on which $\|M_n(x)-M(x)\|_2\le c\,\beta^3$, we have
\begin{align}
\label{eq:Ck-empirical-O4}
    |C_k(M_n(x))-C_k(x)| \le C'\,\beta^4 .
\end{align}
\end{lem}

\begin{proof}[Proof of Lemma~\ref{lem:Ck-empirical-drift}]
Write $\hat{x}_{\mathrm{pop}}=M(x)$ and $\hat{x}=M_n(x)$. By the triangle inequality,
\begin{align}
    |C_k(\hat{x})-C_k(x)| \le |C_k(\hat{x}_{\mathrm{pop}})-C_k(x)| + |C_k(\hat{x})-C_k(\hat{x}_{\mathrm{pop}})|.
\label{eq:triangle_Ck_emp}
\end{align}
The first term is the population drift, bounded by Lemma~\ref{lem:Ck-relative-drift}:
\begin{align}
    |C_k(\hat{x}_{\mathrm{pop}})-C_k(x)| \le C\,\beta^4.
\label{eq:pop_term_beta4}
\end{align}

For the second term, we use the low-magnitude expansion~\eqref{eqn:app_M35}: for any $y$ with $\|y\|_2\le 2\beta$ we may write $y=\|y\|_2\,u$ with $\|u\|_2=1$ and obtain
\begin{align}
    C_k(y)=\|y\|_2^2\,c_k(u)+O(\|y\|_2^4), \qquad c_k(u)=\frac{a_k(u)}{b_k(u)},
\label{eq:Ck_local_form}
\end{align}
where $a_k(\cdot),b_k(\cdot)$ are defined in~\eqref{eqn:ak_v_def}--\eqref{eqn:bk_v_def}.
On the set $\{\|u\|_2\le P,\ b_k(u)\ge c_0^2\}$, the map $u\mapsto c_k(u)$ is Lipschitz, hence $y\mapsto C_k(y)$ is locally Lipschitz with constant $O(\beta)$ on the ball $\{ \|y\|_2\le 2\beta\}$.
Therefore, for $\beta$ small enough (so that $\|\hat{x}\|_2,\|\hat{x}_{\mathrm{pop}}\|_2\le 2\beta$),
\begin{align}
    |C_k(\hat{x})-C_k(\hat{x}_{\mathrm{pop}})| \le C\,\beta\,\|\hat{x}-\hat{x}_{\mathrm{pop}}\|_2,
    \label{eq:Ck_Lip_beta}
\end{align}
for a constant $C=C(d,k,P,c_0)$.
Combining~\eqref{eq:triangle_Ck_emp}--\eqref{eq:Ck_Lip_beta} gives~\eqref{eq:Ck-empirical-triangle}.
If additionally $\|\hat{x}-\hat{x}_{\mathrm{pop}}\|_2\le c\,\beta^3$, then the second term is $O(\beta^4)$ and \eqref{eq:Ck-empirical-O4} follows.
\end{proof}

\subsection{Proof of Proposition~\ref{prop:finite_sample_one_step}} \label{sec:finite_sample_one_step_proof}
To derive the convergence rate in \eqref{eq:phase_mse_conditional_limit}, recall the definitions of $\mu_{\s{B},k}$ and $\sigma_{\s{A},k}^2$ in \eqref{eqn:muB} and \eqref{eqn:sigmaB}. As shown in Appendix~\ref{sec:FourierPhasesConvergenceRate}, $\mu_{\s{B},k} > 0$ and $\sigma_{\s{A},k}^2 < \infty$. Applying Proposition~\ref{prop:convergeneInExpectation}, we obtain, conditioned on $\hat{x}^{(t)}$
\begin{align}
    \lim_{n \to \infty} \frac{\mathbb{E}\left[ \left| \phi_{\s{\hat{X}}^{(t+1)}}[k] - \phi_{\s{\hat{X}}^{(t)}}[k] \right|^2 \mid  \hat{x}^{(t)} \right] }{1/n} = \frac{\sigma_{\s{A},k}^2}{\mu_{\s{B},k}^2} < \infty,
\end{align}
which verifies \eqref{eq:phase_mse_conditional_limit}.

\begin{remark}\label{remark:constant-Ck}
The constant $C_k(\hat{x}^{(t)})$ in \eqref{eq:phase_mse_conditional_limit} is explicitly given by:
\begin{align}
    C_k(\hat{x}^{(t)}) \triangleq \frac{\sigma_{\s{A},k}^2}{\mu_{\s{B},k}^2} = \frac{\operatorname{Var}\left( \abs{\s{N}[k]} \sum_{\ell=0}^{d-1} \gamma_{\ell}(\hat{x}^{(t)}; \xi) \sin \big(\frac{2\pi k \ell} {d} +  \phi_{\s{N}_i}[k] - \phi_{\s{\hat{X}}^{(t)}}[k]\big) \right)}{\left( \mathbb{E}\left[ \abs{\s{N}[k]} \sum_{\ell=0}^{d-1} \gamma_{\ell}(\hat{x}^{(t)}; \xi) \cos \big(\frac{2\pi k \ell} {d} +  \phi_{\s{N}_i}[k] - \phi_{\s{\hat{X}}^{(t)}}[k]\big)  \right] \right)^2}.
\end{align}
\end{remark}

\subsection{Proof of Proposition~\ref{prop:low-magnitude-scaling-of-phases-convergence-rate}}\label{sec:proof-of-low-magnitude-scaling-of-phases}
Fix $t\in\{0,\dots,T-1\}$ and Let $x=\hat{x}^{(t)}$ and write $x=\beta{v}$ with $\beta=\|x\|_2$ and $\|v \|_2=1$.
By Proposition~\ref{prop:finite_sample_one_step} and Remark~\ref{remark:constant-Ck}, the one-step phase-drift constant is
\begin{align}
    C_k(\hat{x}^{(t)})=\frac{\sigma_{\s{A},k}^2(\hat{x}^{(t)})}{\mu_{\s B,k}^2(\hat{x}^{(t)})}.
\end{align}

\paragraph{Step 1: Low--magnitude scaling $C_k(\hat{x}^{(t)})=\Theta(\beta^2)$.}
Under the low--magnitude regime of Theorem~\ref{thm:low-mag-rate}, we have $\beta\to 0$.
Applying Lemma~\ref{lem:muB-sigmaA-low-magnitude-short} with $x=\beta v$ yields
\begin{align}
    \mu_{\s B,k}(\beta v)&=|\s X[k]|\,(1-|\s X[k]|^2+O(\beta^4)),\\
    \sigma_{\s{A},k}^2(\beta v)&=\beta^4\,a_k(v)+O(\beta^6),
\end{align}
where $a_k(v)$ is the variance functional in~\eqref{eq:sigmaA_low_signal_short}.
Since $|\s X[k]|=\beta|\s V[k]|$ and $|\s V[k]|>0$ for the considered frequency, we obtain
\begin{align}
    \mu_{\s B,k}^2(\beta v)=\beta^2|\s V[k]|^2\,(1+O(\beta^2)).
\end{align}
Therefore,
\begin{align}
    C_k(\beta v)=\frac{\beta^4 a_k(v)+O(\beta^6)}{\beta^2|\s V[k]|^2(1+O(\beta^2))} =\beta^2\frac{a_k(v)}{|\s V[k]|^2}+O(\beta^4).
\end{align}
On any set where $|\s V[k]|\ge c_0>0$ and $a_k(v)\in(0,\infty)$ (which holds for all non-degenerate modes), the ratio $a_k(v)/|\s V[k]|^2$ is bounded above and below by positive constants, hence $C_k(\hat{x}^{(t)})=\Theta(\beta^2)=\Theta(\|\hat{x}^{(t)}\|_2^2)$, proving~\eqref{eq:Ck_lowmag_scaling}.

\paragraph{Step 2: Slow variation across empirical iterates.}
Apply Lemma~\ref{lem:Ck-empirical-drift} to the empirical update $\hat{x}=M_n(x)=\hat{x}^{(t+1)}$.
On the event that the sample--population deviation satisfies
\begin{align}
    \|M_n(x)-M(x)\|_2 \le \varepsilon_M(n,\delta)\le c_1\beta^3,
\end{align}
which is equivalent to the assumption of $n \gtrsim \beta^{-6} \log(1/\delta)$ from the expression of $\varepsilon_M(n,\delta)$ in Proposition~\ref{lem:subgaussian-epsn} and~\eqref{eq:MRA-uniform-rate}, Lemma~\ref{lem:Ck-empirical-drift} and~\eqref{eq:Ck-empirical-O4} yields
\begin{align}
    |C_k(\hat{x}^{(t+1)})-C_k(\hat{x}^{(t)})| =|C_k(M_n(x))-C_k(x)| \le C'\,\beta^4 = C'\,\|\hat{x}^{(t)}\|_2^4,
\end{align}
which is exactly~\eqref{eq:slow_variation}.

\subsection{Proof of Proposition \ref{thm:lowerBoundAfterTiterations_rewrite}} \label{sec:proofOfPhaseErrorAfterTiterations}

For every $k \neq 0$ and $ t=0,\dots,T-1$, define the phase increments
\begin{align}
    \Delta_t  \triangleq  \phi_{\s{\hat{X}}^{(t+1)}}[k]-\phi_{\s{\hat{X}}^{(t)}}[k].
\end{align}
By telescoping,
\begin{align}
    \phi_{\s{\hat{X}}^{(T)}}[k]-\phi_{\s{\hat{X}}^{(0)}}[k] = \sum_{t=0}^{T-1}\Delta_t .
\end{align}

\paragraph{General bound.}
Using triangle inequality,
\begin{align}
    \p{\mathbb{E}\big|\phi_{\s{\hat{X}}^{(T)}}[k]-\phi_{\s{\hat{X}}^{(0)}}[k]\big|^2}^{1/2}
    = \p{\mathbb{E}\big|\textstyle\sum_{t=0}^{T-1}\Delta_t\big|^2}^{1/2}
    \le \sum_{t=0}^{T-1}\p{\mathbb{E}|\Delta_t|^2}^{1/2}.
\end{align}
Applying the assumption~\eqref{eqn:eps_assump_rewrite} gives
\begin{align}
    \p{\mathbb{E}\big|\phi_{\s{\hat{X}}^{(T)}}[k]-\phi_{\s{\hat{X}}^{(0)}}[k]\big|^2}^{1/2}  \le \sum_{t=0}^{T-1}\sqrt{\epsilon^{(t)}}, \label{eqn:telescoping-phase-noise-accum}
\end{align}
and squaring yields~\eqref{eqn:lowerBoundMultipleIterations_rewrite}.

\paragraph{Slowly varying envelope.}
Assume~\eqref{eqn:slowly_varying_envelope_rewrite}.  Write
\begin{align}
    \epsilon^{(t)}=\epsilon^{(0)}\p{1+\delta_t},
\end{align}
where,
\begin{align}
    \delta_t  = \frac{\alpha t}{\epsilon^{(0)}}+O \p{\frac{(\alpha t)^2}{\epsilon^{(0)}}},
    \qquad |\delta_t|\ll 1.
\end{align}
A Taylor expansion of the square root gives, uniformly for $\alpha t\ll1$,
\begin{align}
    \sqrt{\epsilon^{(t)}} &= \sqrt{\epsilon^{(0)}}\p{1+\frac{1}{2}\delta_t+O(\delta_t^2)}
    \\
    &= \sqrt{\epsilon^{(0)}} + \frac{\alpha t}{2\sqrt{\epsilon^{(0)}}} + O \p{\frac{(\alpha t)^2}{\sqrt{\epsilon^{(0)}}}}.
\end{align}
Summing from $t=0$ to $T-1$,
\begin{align}
    \sum_{t=0}^{T-1}\sqrt{\epsilon^{(t)}} &= T\sqrt{\epsilon^{(0)}} +\frac{\alpha}{2\sqrt{\epsilon^{(0)}}}\sum_{t=0}^{T-1}t + O \p{\frac{\alpha^2}{\sqrt{\epsilon^{(0)}}}\sum_{t=0}^{T-1}t^2}
    \\
    & = T\sqrt{\epsilon^{(0)}} + \frac{\alpha}{4\sqrt{\epsilon^{(0)}}}T(T-1) + O \p{\alpha^2 T^3}.
\end{align}
Plugging this into~\eqref{eqn:telescoping-phase-noise-accum} gives
\begin{align}
    \mathbb{E}\big|\phi_{\s{\hat{X}}^{(T)}}[k]-\phi_{\s{\hat{X}}^{(0)}}[k]\big|^2
    &\le \p{T\sqrt{\epsilon^{(0)}} + O(\alpha T^2) + O(\alpha^2T^3)}^2
    \\ &= T^2\epsilon^{(0)} + O(\alpha T^3) + O(\alpha^2T^4).
    \label{eq:phase-accum-slow}
\end{align}

\section{Mitigation of initialization bias using mini-batch optimization} \label{sec:adam}

In this section, we analyze mini-batch gradient methods as an alternative optimization approach aimed at mitigating Einstein from Noise phenomenon. Interestingly, while these methods do not exhibit the Ghost of Newton phenomenon (that is, they do not initially converge near Newton’s ground truth image and then diverge), the final reconstructions they produce is similar to the diverged outputs typically seen after EM has drifted away from Newton’s ground truth image.

We emphasize that, although mini-batching is a standard ingredient in stochastic optimization, in low-SNR MRA it mitigates a model-specific failure mechanism. In the Einstein from Noise regime, full-batch EM updates preserve the Fourier phases of the current iterate. This “phase-locking” behavior promotes template imprinting: spurious structure aligned with the initialization can be repeatedly reinforced across iterations even when the observations contain no signal. Mini-batching disrupts this reinforcement by making each update depend on a different random subset of observations, thereby perturbing the empirical alignment statistics from iteration to iteration and weakening the systematic propagation of the template’s Fourier phases. As a result, the reconstruction exhibits substantially less bias toward the initial template in the Einstein-from-Noise regime, consistent with Figure~\ref{fig:9}(b).


\subsection{The mini-batching procedure description}

In the mini-batching approach, only a small portion of the dataset, called a batch, is used at each iteration, rather than the entire dataset.
As analyzed in Section~\ref{sec:EfN}, the MSE of the Einstein from Noise effect, after $t$ iterations, scales as $t^2/n$, where $n$ is the total number of observations. This scaling suggests a dual benefit to using mini-batches: first, each iteration relies on fewer observations, which reduces the cumulative bias from previous estimates; second, more iterations are performed, mitigating the bias even further over time. As a result, the final structure is significantly less biased toward the initial template. Indeed, Figure~\ref{fig:9} demonstrates that, in the Einstein from Noise regime, mini-batching results in a substantially lower bias toward the initial template compared with the EM algorithm.

\paragraph{Mini-batching gradient-descent algorithms.} 
To empirically assess the effectiveness of the mini-batching strategy in mitigating the Einstein from Noise phenomenon, we employ a momentum-based mini-batch optimization procedure, as outlined in Algorithm~\ref{alg:momentumSGD}. This class of optimizers aims to maximize the same log-likelihood function as EM but uses gradients computed from mini-batches instead of the full dataset. In contrast to full-batch EM and hard-assignment methods, which process the entire dataset in each iteration, mini-batch methods update the parameters using only a subset of the data at each step. In our experiments, we focus on stochastic gradient descent (SGD) variants that incorporate momentum and adaptive learning rates, such as ADAM~\cite{kingma2014adam}, Ranger~\cite{wright2021ranger21}, and Yogi~\cite{zaheer2018adaptive}. We emphasize that mitigating the Einstein from Noise effect through the mini-batching optimization technique is a fundamental property, not merely an artifact of implementation choices, as elaborated in the previous paragraph. Further implementation details are provided in Appendix~\ref{sec:empiricalDetails}.

\begin{algorithm}[t!]
  \caption{\texttt{General momentum-based SGD for multi-reference alignment} \label{alg:momentumSGD}}
\textbf{Input:} Initial estimate $\hat{x}^{(0)} = x_{\text{template}}$, observations $\{y_i\}_{i=1}^{n}$, number of iterations $T$, learning rate schedule $\{\alpha_t\}_{t=0}^{T-1}$, momentum parameters $\beta_1$, $\beta_2$ and momentum update function $\mathrm{UpdateMoments}(\cdot)$, small constant $\epsilon > 0$, and mini-batch index sets $\{\mathcal{B}_t\}_{t=0}^{T-1}$ with $\mathcal{B}_t \subset \{0,1,\dots,n-1\}$.\\
\textbf{Output:} Final estimate $\hat{x}^{(T)}$.\\
\textbf{Initialize:} $m^{(0)} = 0$, $v^{(0)} = 0$.\\
\textbf{For each iteration $t = 0, \ldots, T-1$:}
\begin{enumerate}
    \item Sample a mini-batch of observations $\{ y_i \}_{i \in \mathcal{B}_t}$.
    \item For each $i \in \mathcal{B}_t$ and $\ell = 0, \ldots, d-1$, compute the soft-assignment probabilities:
    \begin{align}
        \gamma_\ell(\hat{x}^{(t)}, y_i) \triangleq \frac{\exp\left( y_i^\top \left( \mathcal{T}_\ell \hat{x}^{(t)} \right) \right)}{\sum_{r=0}^{d-1} \exp\left( y_i^\top \left( \mathcal{T}_r \hat{x}^{(t)} \right) \right)}
    \end{align}
    \item Compute the stochastic gradient:
    \begin{align}
        g^{(t)} \triangleq \hat{x}^{(t)} - \frac{1}{|\mathcal{B}_t|} \sum_{i \in \mathcal{B}_t} \sum_{\ell=0}^{d-1} \gamma_\ell(\hat{x}^{(t)}, y_i) \cdot \left( \mathcal{T}_\ell^{-1}  y_i \right)
    \end{align}
    \item Update moments using the optimizer-specific rule:
    \begin{align}
        m^{(t+1)}, v^{(t+1)} \gets \mathrm{UpdateMoments}(m^{(t)}, v^{(t)}, g^{(t)}, \beta_1, \beta_2, t+1)
    \end{align}
    \item Update the estimate:
    \begin{align}
        \hat{x}^{(t+1)} \gets \hat{x}^{(t)} - \alpha_t \cdot \frac{m^{(t+1)}}{\sqrt{v^{(t+1)}} + \epsilon}
    \end{align}
\end{enumerate}
\end{algorithm}

\paragraph{Comparison between EM, hard-assignment, and mini-batch SGD.}
Figure~\ref{fig:9} compares the reconstruction performance of EM, hard-assignment, and a momentum-based mini-batch SGD algorithm across a range of SNR values, illustrating the behavior of each method in the different regimes. For the SGD-based algorithm (the Ranger optimization method), we use mini-batches of size 256 and run for 800 iterations with a learning rate schedule based on exponential decay. 

\begin{figure*}[!t]
    \centering
    \includegraphics[width=0.95 \linewidth]{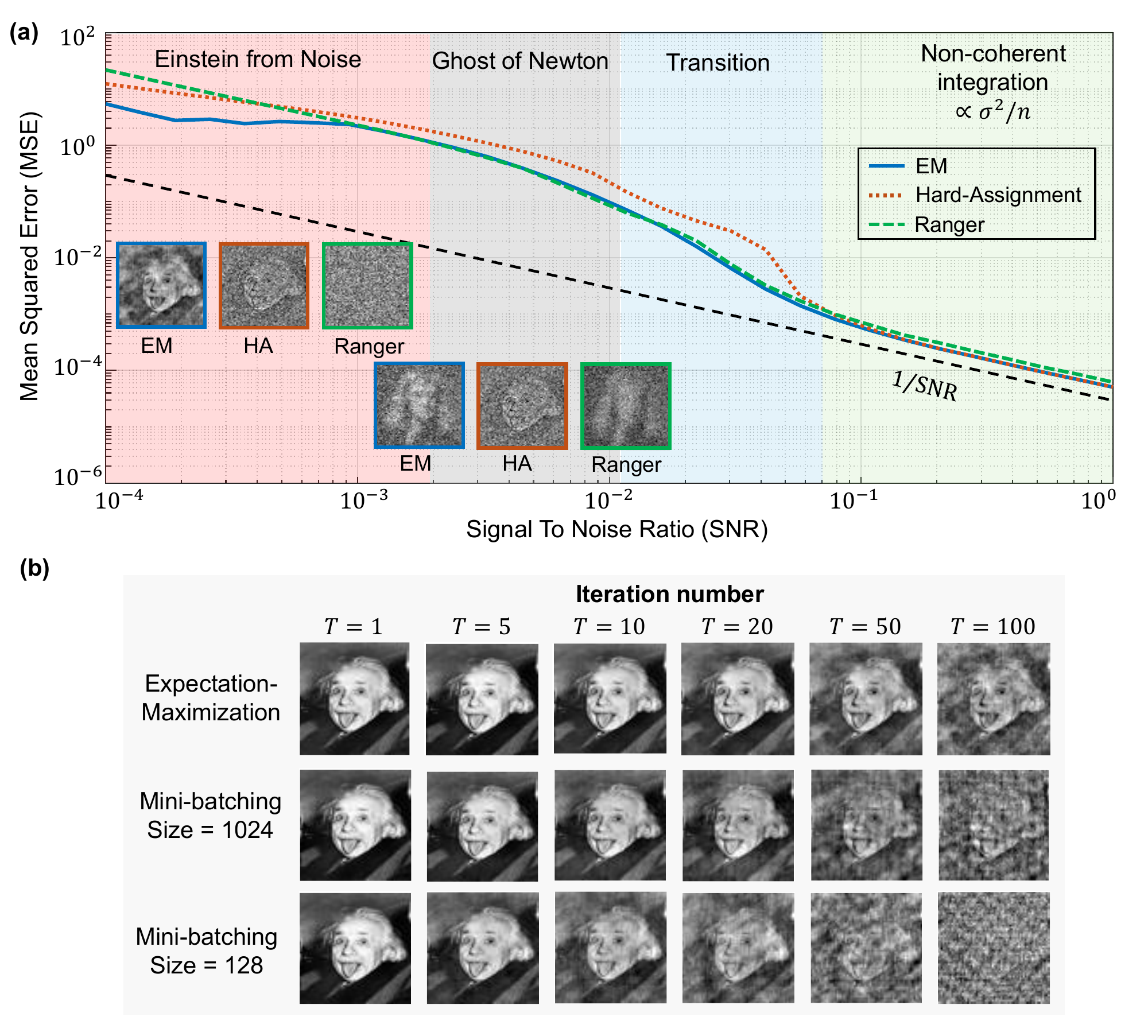}
    \caption{\textbf{Comparison of EM, hard-assignment, and mini-batch SGD-based reconstruction algorithms across varying SNR regimes.} \textbf{(a)} In the Einstein from Noise regime, the Ranger-based method exhibits the least bias toward the initial template (Einstein). In the Ghost-of-Newton regime, both EM and Ranger recover the ground truth structure (Newton), while the hard-assignment method erroneously converges to the initialization (Einstein). 
    Visually, EM achieves higher reconstruction accuracy than Ranger in this setting. In the high-SNR regime, corresponding to the non-coherent integration limit, the EM and hard-assignment methods perform similarly, while Ranger shows slightly reduced reconstruction accuracy. The simulation parameters are: image size of $d = 64 \times 64$, number of observations $n = 2 \times 10^4$, and $T = 200$ iterations. Each data point is averaged over 30 Monte Carlo trials.  \textbf{(b)} Effect of batch size on bias toward the Einstein template in the Einstein from Noise regime. Three setups are compared: EM with $n = 2 \times 10^4$ observations, and ADAM with batch sizes of 1024 and 128. A clear trend emerges: larger batch sizes increase bias toward the initial template. 
    }
    \label{fig:9} 
\end{figure*}

In the low-SNR regime, the mini-batch SGD method performs comparably to EM and significantly outperforms the hard-assignment variant. This is consistent with the fact that both EM and mini-batch SGD aim to optimize the same log-likelihood objective \eqref{eqn:logLikelihoodMRA}, whereas the hard-assignment objective function \eqref{eqn:mraHardTargetFunction} diverges from this objective under high noise due to its reliance on hard label assignments. In the high-SNR regime, where noise plays a less dominant role, the performance of EM and hard-assignment algorithms is equal, while the SGD-based method may exhibit slightly reduced reconstruction accuracy.

Figure~\ref{fig:9}(b) demonstrates that decreasing the batch size at each iteration significantly reduces bias, effectively mitigating the Einstein from Noise effect. This observation is consistent with the trends seen in Figure~\ref{fig:7} and the analysis in Section~\ref{sec:EfN}, which indicate that the Einstein from Noise effect becomes more pronounced as the number of observations increases.

In Appendix~\ref{sec:empiricalDetails}, we compare three mini-batch SGD optimizers: ADAM, Yogi, and Ranger. While all three achieve similar levels of reconstruction accuracy, Ranger shows the lowest bias toward the initial template (Einstein), suggesting greater robustness to initialization. Moreover, none of these optimizers exhibit the Ghost of Newton effect---namely, the non-monotonic MSE behavior observed with EM, where error initially decreases before increasing and drifting away from the true underlying image (Newton).

\paragraph{Computational cost.}
A key advantage of mini-batching is its ability to reduce computational cost significantly. For example, using a batch size of 256 instead of the entire dataset size of $n = 2 \times 10^4$ observations enables reducing the number of observations processed per iteration by two orders of magnitude. In our implementation, the mini-batch SGD method was run for $T = 800$ iterations, compared to $T = 200$ iterations for the full-batch EM algorithm. Despite the increased number of iterations, the overall runtime was reduced by more than an order of magnitude while maintaining comparable reconstruction performance and substantially reducing the Einstein from Noise bias (see additional details on computational costs and runtime in Appendix~\ref{sec:empiricalDetails}).

\paragraph{A hybrid approach.}
The complementary strengths of full-batch and mini-batch methods motivate hybrid strategies that seek to combine their respective advantages. While both full-batch EM and mini-batch SGD aim to maximize the likelihood function in \eqref{eqn:logLikelihoodMRA}, they differ in their trade-offs. Full-batch EM guarantees monotonic improvement of the log-likelihood but is computationally intensive and sensitive to initialization. In contrast, mini-batch methods offer greater efficiency and robustness to initialization, albeit without guarantees of monotonic convergence.

A practical hybrid approach begins with a mini-batch SGD, which reduces bias from the initialization and moves closer to the maximum likelihood solution. The resulting estimate then serves as the initialization for a small number of EM iterations, refining the reconstruction without requiring the full computational cost of EM with multiple iterations. Similarly, the hybrid approach can start with small batches and gradually increase their size as the iterations progress.  
This hybrid strategy is implemented in state-of-the-art cryo-EM software pipelines. For example, the RELION and cryoSPARC software~\cite{kimanius2021new,punjani2017cryosparc} use a two-stage process: first, an SGD-based algorithm with small, adaptively sized batches generates a 3D ab initio model. This initial model is then refined using EM to obtain a high-resolution reconstruction. The approach combines the scalability and reduced bias of mini-batching---effectively mitigating the Einstein from Noise phenomenon---with the statistical rigor and convergence guarantees of EM.

\subsection{Empirical simulation details} \label{sec:empiricalDetails}
Here, we detail the simulation parameters used for the reconstruction experiments employing the SGD-based optimizers ADAM \cite{kingma2014adam}, Ranger \cite{wright2021ranger21}, and Yogi \cite{zaheer2018adaptive}.
For the results shown in Figure~\ref{fig:9}(a), the following settings were used: a batch size of 256, and an exponentially decaying learning rate of the form $\alpha_t = \alpha_0 \exp(-\gamma t)$, with an initial value $\alpha_0 = 0.95$ and decay rate $\gamma = 0.99$. The total number of iterations was set to $T = 800$. The momentum parameters were $\beta_1 = 0.9$ and $\beta_2 = 0.999$. For the Ranger optimizer, we additionally used a look-ahead update every $k=5$ iteration with a look-ahead step size of 0.5.

Figure~\ref{fig:app_SGDcomparison} compares the performance of ADAM, Yogi, and Ranger in terms of their sensitivity to the initialization of the Einstein template. All optimizers were run using the same batch size and hyperparameters as described above. Empirically, Ranger demonstrates the lowest bias toward the initial template, indicating greater robustness to initialization. The reconstruction accuracy (in terms of MSE) of all optimizers is similar.

\begin{figure*}[!t]
    \centering
    \includegraphics[width=0.8 \linewidth]{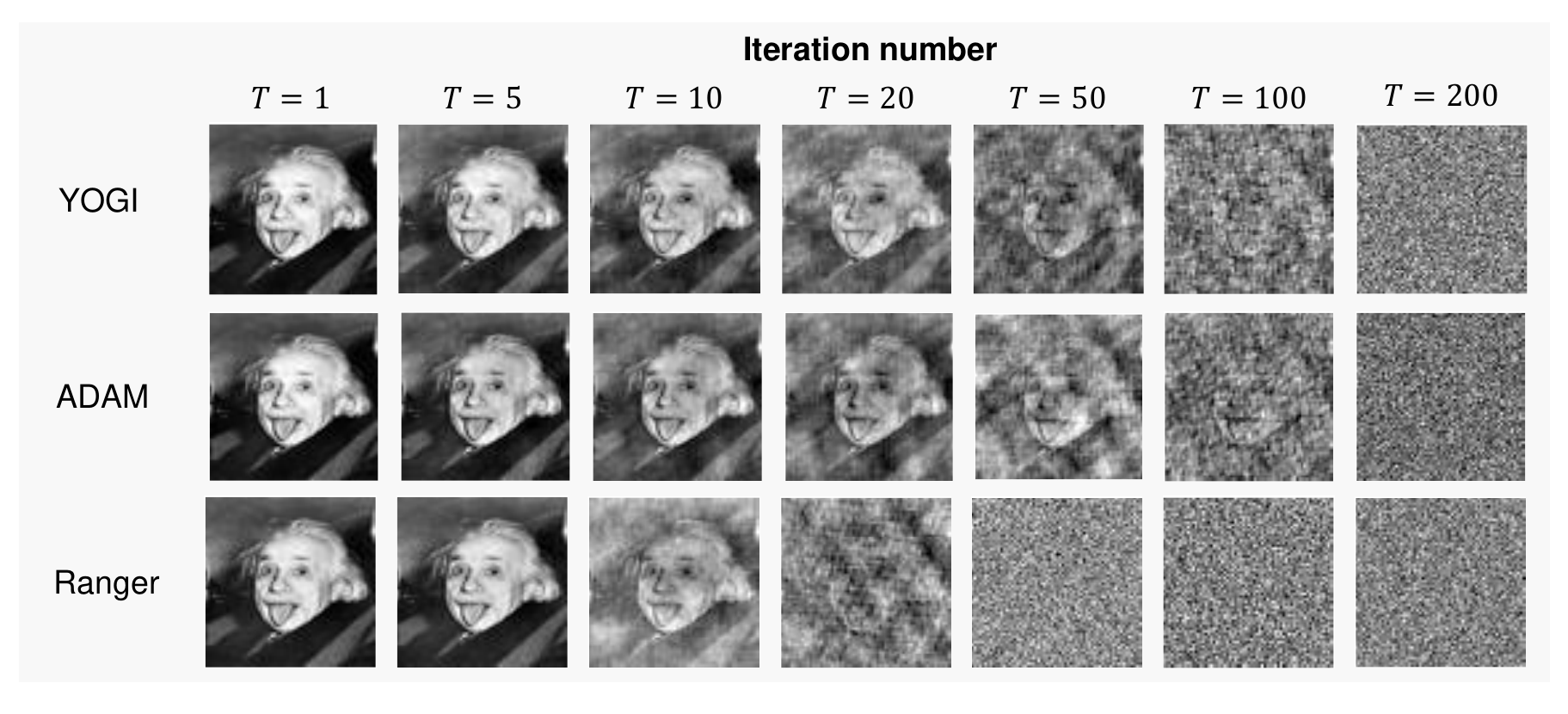}
    \caption{\textbf{Comparison of SGD optimizers and their sensitivity to the initial template.} The Ranger optimizer exhibits the least bias toward the initial Einstein template, compared to the other optimizers.}
    \label{fig:app_SGDcomparison} 
\end{figure*}

\paragraph{Computational running time comparison.}
The EM and hard-assignment algorithms were implemented in MATLAB. For images of size $d = 64 \times 64$, with $n = 2 \times 10^4$ observations and $T = 200$ iterations, a full run required approximately 15 minutes. In contrast, the mini-batch SGD optimizers, using batch sizes of 256 observations and $T = 800$ iterations, completed in about 0.5 minutes under the same conditions. All experiments were performed on a 12th-generation Intel Core i7-12700H CPU with 32 GB of RAM.

\paragraph{Hard-assignment in the intermediate SNR regime.} 
As discussed in the main text, empirical simulations suggest that the Ghost of Newton phenomenon arises exclusively in the EM algorithm. In contrast, the hard-assignment algorithm exhibits a sharp transition: it moves directly from reconstructing the initial Einstein template in the low-SNR regime (Einstein from Noise) to accurately recovering the Newton signal, without undergoing the gradual drift characteristic of the Ghost of Newton effect.

Figure~\ref{fig:GoN_hardAssign} illustrates this comparison between the EM and hard-assignment algorithms in the intermediate SNR regime. While the EM algorithm exhibits the Ghost of Newton behavior, the hard-assignment algorithm does not. Instead, it reconstructs the initial Einstein template, rather than the ground truth image of Newton. This behavior persists across a range of SNR values, as demonstrated in the simulations shown in Figure~\ref{fig:GoN_hardAssign}.

\begin{figure*}[!t]
    \centering
    \includegraphics[width=1.0 \linewidth]{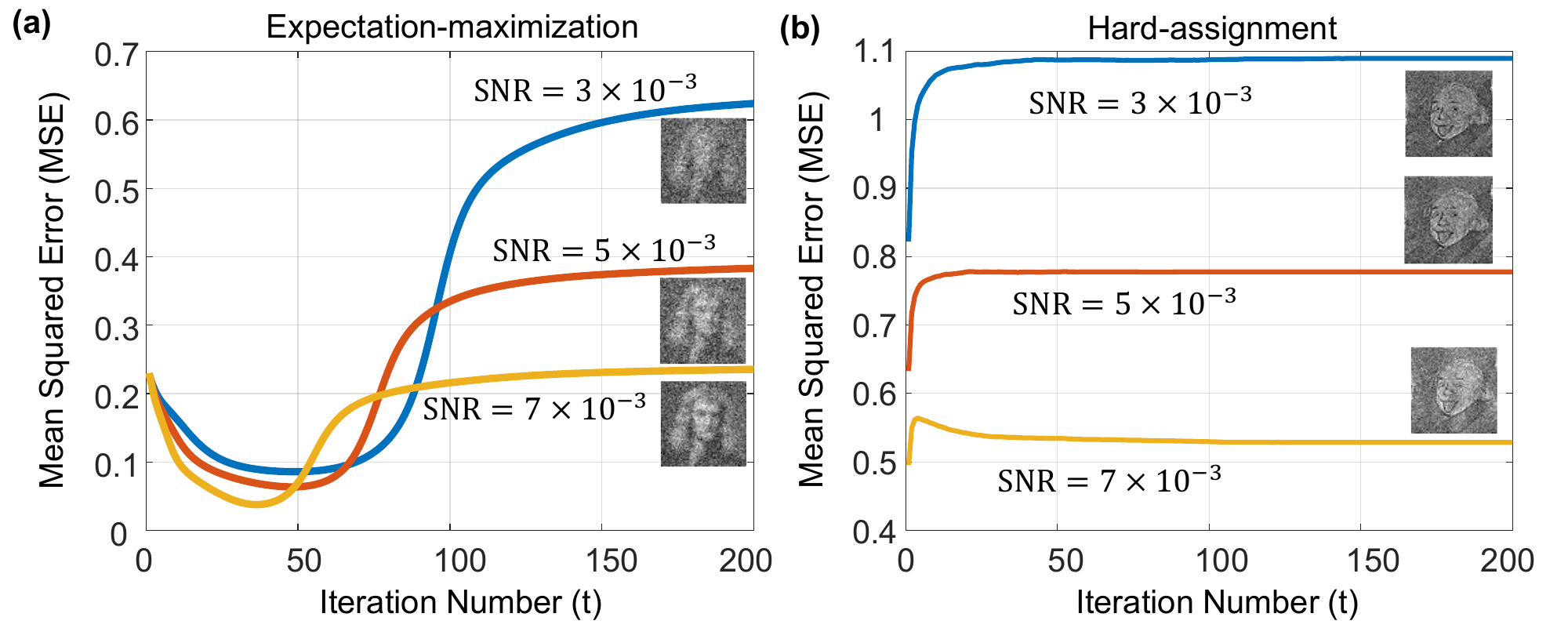}
    \caption{\textbf{Comparison between expectation-maximization (EM) and hard-assignment algorithms in the intermediate SNR regime.} Panels (a) and (b) compare the performance of the EM and hard-assignment algorithms, respectively, in the intermediate SNR regime, where the Ghost of Newton phenomenon emerges in the EM algorithm. Note that the y-axes differ in scale across the two panels. As analyzed in the main text, EM exhibits this effect, producing reconstructions that resemble a noisy version of the ground truth of Newton. In contrast, the hard-assignment algorithm does not display the Ghost of Newton behavior; instead, it exhibits the Einstein from Noise phenomenon, producing reconstructions that resemble the Einstein template rather than the true Newton signal. Simulation parameters: image size $d = 64 \times 64$, number of observations $n = 2 \times 10^4$.}
    \label{fig:GoN_hardAssign} 
\end{figure*}

\subsection{Relation between the Ghost of Newton and Einstein from Noise regimes}

To clarify how the Einstein from Noise and Ghost of Newton effects are connected, it is useful to decompose a single empirical EM update.  
Recall that $y_i=\mathcal{T}_{\ell_i}x^\star+\xi_i$ under~\eqref{eqn:mainModel1D}. Plugging this into the EM update~\eqref{eqn:generalizedEfNEqnSoft} yields the split
\begin{align}
    \hat{x}^{(t+1)} = \hat{x}_{\mathrm{signal}}^{(t+1)}+\hat{x}_{\mathrm{noise}}^{(t+1)},
    \label{eqn:combinedSignalAligned_rewrite}
\end{align}
where
\begin{align}
    \hat{x}_{\mathrm{signal}}^{(t+1)} &\triangleq \frac{1}{n}\sum_{i=1}^n\sum_{\ell=0}^{d-1} \gamma_\ell(\hat{x}^{(t)},y_i)\, \mathcal{T}_\ell^{-1}\p{\mathcal{T}_{\ell_i}x^\star},
    \label{eqn:alignedSignalTerm_rewrite}
    \\
    \hat{x}_{\mathrm{noise}}^{(t+1)} &\triangleq \frac{1}{n}\sum_{i=1}^n\sum_{\ell=0}^{d-1} \gamma_\ell(\hat{x}^{(t)},y_i)\, \mathcal{T}_\ell^{-1}\xi_i .
    \label{eqn:alignedNoiseTerm_rewrite}
\end{align}
Here $\gamma_\ell(\hat{x}^{(t)},y_i)$ are the responsibilities from~\eqref{eqn:softmaxGamma}.  
This decomposition makes explicit that EM simultaneously averages an \emph{aligned signal} term and an \emph{aligned noise} term, with the same data-dependent weights.

Figure~\ref{fig:4}(b-c) visualizes these two contributions across iterations. In the first iterations, the estimate $\hat{x}^{(t)}$ typically exhibits a hybrid appearance: it already contains recognizable Newton-like features from the aligned signal term~\eqref{eqn:alignedSignalTerm_rewrite}, yet it also retains visible Einstein-like structure inherited from the initialization. This mixture is exactly what one would expect from the Einstein from Noise mechanism, namely, the tendency of the responsibilities to align the noise with the current iterate, combined with genuine contraction of the signal toward $x^\star$ at nonzero SNR. As $t$ grows, the signal term continues to sharpen Newton, and the hybrid gradually becomes more Newton-dominated.

The crucial transition underlying the Ghost of Newton is that the aligned noise term~\eqref{eqn:alignedNoiseTerm_rewrite} evolves with the iterate. Once $\hat{x}^{(t)}$ becomes sufficiently Newton-like, the same alignment weights begin to ``organize'' the noise around this Newton-like structure. Consequently, $\hat{x}_{\mathrm{noise}}^{(t+1)}$ no longer resembles the original Einstein template; instead, it forms a structured, sample-dependent perturbation that visually echoes Newton. This structured noise then adds coherently to the update and can pull the iterate away from the ground truth, producing the degradation seen in Figure~\ref{fig:4}(a). 

Viewed through this lens, Einstein from Noise and Ghost of Newton are two manifestations of the same underlying mechanism, alignment of noise to the current estimate, operating in different SNR regimes. In the Einstein from Noise regime ($\mathrm{SNR}\approx 0$), there is no signal term to counteract this bias, so aligned noise alone reconstructs the template. In the intermediate-SNR regime, the signal term initially dominates and drives the estimate toward Newton, but the aligned noise gradually reorients to that evolving structure and eventually induces a drift away from $x^\star$.

\end{appendices}

\end{document}